\documentclass[12pt,twoside]{article}
\usepackage{latexsym, amsthm, amsmath, bbm}
\usepackage{amssymb}
\usepackage{amsfonts}
\usepackage{amstext}
\usepackage{multicol}
\usepackage[table]{xcolor}
\usepackage{graphicx}
\usepackage{grffile}

\usepackage{subcaption}
\usepackage{hyperref}

\usepackage{multirow}

\usepackage{setspace} 

\newtheorem{Theorem}{Theorem}

\newtheorem{Definition}{Definition}

\newtheorem{Corollary}{Corollary}

\theoremstyle{remark}

\newcommand{\Q}{{\mathbb{Q}}}

\usepackage{bm}
\usepackage{MnSymbol}

\extrafloats{100}

\title{Bounds and Bugs: \\ The Limits of Symmetry Metrics to Detect Partisan Gerrymandering}

\author{
Daryl DeFord \\
\small Department of Mathematics and Statistics \\[-.8ex]
\small Vassar College  \\[-.8ex]
\small \tt ddeford@vassar.edu 
\and
Ellen Veomett \\
\small Department of Computer Science \\[-0.8ex]
\small University of San Francisco \\[-0.8ex]
\small \tt eveomett@usfca.edu
}

\begin{document}

\maketitle

\begin{abstract}

We consider two symmetry metrics commonly used to analyze partisan gerrymandering: the Mean-Median Difference (MM) and Partisan Bias (PB). 
Our main results compare, for combinations of seats and votes achievable in districted elections, the number of districts won by each party to the extent of potential deviation from the ideal metric values, taking into account the political geography of the state.   
These comparisons are motivated by examples where the MM and PB have been used in efforts to detect when a districting plan awards extreme number of districts won by some party.  These examples include expert testimony, public-facing apps, recommendations by experts to redistricting commissions, and public policy proposals.

To achieve this goal we perform both theoretical and empirical analyses of the MM and PB.  In our theoretical analysis, we consider vote-share, seat-share pairs $(V, S)$ for which one can construct election data having vote share $V$ and seat share $S$, and turnout is equal in each district.  We calculate the range of values that MM and PB can achieve on that constructed election data.  In the process, we find the range of 
$(V, S)$ pairs 
that achieve $MM=0$, and see that the corresponding range for PB is the same set of $(V,S)$ pairs.  We show how the set of such $(V,S)$ pairs allowing for $MM=0$ (and $PB=0$) changes when turnout in each district is allowed to vary.  
By observing the results of this theoretical analysis, we can show  that the values taken on by these metrics do not necessarily attain more extreme values in plans with more extreme numbers of districts won. We also analyze specific example elections, showing how these metrics can return unintuitive results.
We follow this with an empirical study, where we show  that on 18 different U.S. maps these  metrics can fail to detect extreme seats outcomes. 

\end{abstract}

\section{Introduction}\label{sec:Introduction}

Partisan gerrymandering has recently become a rich topic of study for mathematicians and computer scientists.  The most straightforward and basic tools that have been proposed to detect gerrymandering are metrics that associate a single numerical value with a given combination of districting plan and election data. 
Historically, shape metrics such as the Polsby-Popper,  Reock, and Convex Hull metrics have been used to measure the irregularity of a district's shape (see, for example, \cite{shape_metrics_webpage}).  But recent studies have shown that having reasonable shape metric scores does not preclude a map from being an extreme partisan gerrymander (see, for example, \cite{DukeNC}), and that even calculating shape metric values is subject to manipulation and abuse \cite{PA_Gerrymandering_Compactness}.

Thus, the focus on metrics to detect gerrymandering has shifted towards metrics  which do not depend on district shapes.   There has been a flurry in the past decade to produce such metrics, which include the  Efficiency Gap \cite{PartisanGerrymanderingEfficiencyGap},  Declination \cite{WarringtonDeclinationELJ}, and GEO metric \cite{geo_paper}.  The Mean-Median Difference and Partisan Bias are two of the older (and thus arguably more well-established) metrics\footnote{All of these non-shape metrics just named (the Efficiency Gap, Declination, GEO metric, Partisan Bias, and Mean-Median Difference) are designed to be used for a two-party system.  In our theoretical discussions, we will thus refer to ``party $A$'' and ``party $B$.'' In our empirical results we consider two-party vote shares for  the Democratic party and the Republican party.},
falling under the family of partisan symmetry metrics. 
Symmetry metrics hinge on the idea that gerrymandering happens when the two parties are not treated equally under shifts in the underlying vote data.  For example, suppose a map  gives party $A$ seat share $S$ when the vote share for party $A$ is $V$.  Symmetry metrics have values far from 0 (indicating gerrymandering) when this same map would award party $B$ a seat share very different from $S$ when party $B$'s vote share is $V$\footnote{The sign of the partisan symmetry metric is intended to indicate which party is being advantaged; presumably this sign should indicate that party $A$ is advantaged if party $B$ would win a seat share \emph{less} than $S$ when party $B$'s vote share is $V$, and the sign should indicate that party $B$ is advantaged if party $B$ would win a seat share \emph{more} than $S$ when party $B$'s vote share is $V$.}.  The desire for partisan symmetry was highlighted by Niemi and Deegan in \cite{TheoryRedistricting}.  Grofman, King,  Gelman, and others \cite{ GelmanKingEvaluatingElectoralSystems, GrofmanBias, PBGrofmanKing, KKR, KingBrowningPB} have studied and promoted the Partisan Bias metric; the Mean-Median Difference was explicitly discussed as a metric by McDonald and Best in \cite{MeanMedian} and by Nagle in \cite{Nagle}. See the extensive information presented in \cite{ImplementingPartisanSymmetry} for additional historical and practical discussion of these metrics.

In \cite{ImplementingPartisanSymmetry} these metrics were analyzed from a theoretical perspective in terms of the `gaps' between the values in the sorted vote vector and related to the Partisan Gini score analyzed in \cite{KKR}. The main theoretical result of \cite{ImplementingPartisanSymmetry} focused on characterizing the cases where where the Partisan Gini (and hence MM and PB) metric is zero and then gave a series of paradoxes and empirical examples showing that the results of these metrics can be unintuitive or uncorrelated with seats outcomes. This paper extends the analysis and results of \cite{ImplementingPartisanSymmetry} by considering the behavior of the metrics across their full range of  potential values. 

As evidenced above, the Mean-Median Difference and Partisan Bias have been part of the academic study of gerrymandering for many decades.  They have also been used in on-the-ground expert testimony for many years.  These metrics are embedded in publicly available websites to evaluate maps such as Dave's Redistricting App \cite{DRA} and PlanScore \cite{PlanScore}, and software such as Maptitude for Redistricting \cite{Maptitude}.  They can be calculated using built-in functions in Python libraries that are used by researchers to study gerrymandering such as GerryChain \cite{GerryDetails}.  While some excellent studies have revealed the ``problems and paradoxes'' of these kinds of metrics \cite{ImplementingPartisanSymmetry}, they remain ubiquitous as established metrics to detect gerrymandering.

\subsection{Contributions}

In this article, we describe the bounds of values that these metrics can take on, and the ``bugs'' that these metrics exhibit in detecting partisan gerrymandering as measured by the number of districts in which a party is favored.  
Specifically, the framework we apply is to evaluate a metric's ability to suggest when a districting plan has an extreme number of districts won by a particular party.  While this is often the way the general public understands gerrymandering, and is also the focus of many experts and researchers, we note that there are many other types of gerrymandering, such as incumbency protection or anti-competitive gerrymanders, that are not directly addressed by this analysis. However, we believe that understanding and modeling this relationship is important, and in Section \ref{sec:MotivatingExamples}, we point out many instances in which the Mean-Median Difference and Partisan Bias are used for that purpose, such as expert testimony, public-facing apps, recommendations to independent redistricting commissions, and public policy.

With this framework in place, we present two types of studies; a theoretical analysis of the values that can be taken on by these metrics for a fixed vote-share, seat-share pair $(V, S)$ (the ``bounds'' study), and an empirical study of the actual values that these metrics take on for potential redistricting maps, using a fixed election as partisan data (the ``bugs'' study).  In our bounds study, we  point out the ways in which the values taken on by these metrics run counter to the expectation that maps with more extreme numbers of districts won should \emph{be able to} exhibit more extreme values on these metrics.  And in our bugs study, we show that these metrics \emph{do not} take on more extreme values in maps with more extreme numbers of districts won.

We then recommend that the Mean-Median Difference and Partisan Bias not be used as the sole metric  to detect gerrymandering in settings such as public facing apps, by line drawing commissions, court expert reports, and in public policy discussions.   In particular, we recommend that they not be used as a proxy for ``extreme number of districts won.'' 

We provide the motivation for this analysis and the main assumptions we are operating under  in Section \ref{sec:framework}.  We then give the mathematical definitions of the Mean-Median Difference and Partisan Bias in Section \ref{sec:definitions}, which also summarizes other previously known facts about these metrics.  We give the mathematical results (the ``bounds'' study) in Section \ref{sec:bounds}; the proofs of those results are in Appendix \ref{sec:proofs}.  The empirical results (the ``bugs study'') are in Section \ref{sec:bugs}.

\section{Motivation and Background}\label{sec:framework}

In this section we provide motivation for this work by describing situations in which the partisan symmetry measures are used in practice to evaluate and optimize redistricting plans.  We also point out that the Mean-Median Difference and Partisan Bias are ubiquitously used in an effort to describe situations where a party wins a larger than expected number of districts.
This provides some context for our work comparing the results of partisan symmetry metrics on a fixed districting plan to the number of seats won by a party under that plan.  Specifically, throughout this paper we focus on the ability of a metric to reflect situations where an extreme number of districts is won by a particular party as motivated by the use cases described below. Much of the background material below is motivated by the thorough coverage in \cite{ImplementingPartisanSymmetry}, which discusses many of the same issues.  

\subsection{Counting Districts Won}\label{sec:districts_won}

That the focus of a partisan gerrymander is to create an extreme number of districts won by one party is a point of agreement among quantitative social scientists, political scientists, mathematicians,  and the general public.  In their recent article \cite{KKR}, Katz, King, and Rosenblatt state that ``Partisan gerrymanderers use their knowledge of voter preferences and their ability to draw favorable redistricting plans to maximize their party’s seat share.''  They go on to emphasize that a gerrymanderer cares not about issues such as turnout or electoral responsiveness ``unless it helps them win more seats.''  In his ``Defining a Gerrymander'' section of \cite{doi:10.1089/elj.2017.0459}, Grofman explains that gerrymandering results in one party winning an extreme number of seats compared to the number they might win in ``a plan drawn on the basis of neutral principles.''  In \cite{3PracticalTestsWang}, Wang describes the effects of partisan gerrymandering on the US House in 2012 by the number of seats won by each party, and states that ``In summary, the effects of partisan redistricting exceeded the amount of asymmetry caused by natural patterns of population.''  And in \cite{ImplementingPartisanSymmetry}, DeFord et al.\ state that ``the express intent of a partisan gerrymander is to secure for their party as many
seats as possible under the constraints of voter geography and the other rules of redistricting.''  While several of these authors caveat that what counts as ``extreme'' may depend on a region's political geography\footnote{We acknowledge that a state's political geography impacts what specific numbers may be considered ``extreme number of districts won.''  This is why we conduct our empirical analysis in Section \ref{sec:bugs}; in that section, we produce extreme maps using a state's political geography.}, they all agree that number of districts won is the foundation of a partisan gerrymander.

The general public agrees.  News stories on gerrymandering virtually always cite the extreme number of districts won as evidence of a gerrymander (see \cite{TNgerrymanderNews}, for example).  Of course, some who seek to draw maps to benefit their political party, and some of those who study such maps, may take other factors into consideration such as protecting incumbents or creating partisan firewalls\footnote{A partisan firewall can be created by, for example, ensuring a durable majority in a legislative body without necessarily \emph{maximizing} the number of seats won.}.  We agree that those strategies also disenfranchise voters and that a gerrymandering mapmaker may draw the map also aiming at those goals.  However, incumbency protection and partisan firewalls are not the first things that the general public and and many researchers consider when looking to identify and analyze partisan gerrymandering.  For the purposes of this article, we take the common position of the above scholars (and many others) that, in order to detect partisan gerrymandering, we must, first and foremost, be able to detect an extreme number of districts won.

We also agree with Wang \cite{3PracticalTestsWang} and DeFord \cite{ImplementingPartisanSymmetry} that the political geography of a region being redistricted has significant impact on what number of districts won should be considered ``extreme.''  For example, Duchin et al.\ showed that it was virtually impossible for Massachusetts to have a US House District with a majority of Republican voters, due to the political geography of that state \cite{MA_DemDomination}.  Thus, a result of 9 out of 9 seats for Democrats should  not be flagged as extreme at all for Massachusetts. We take into account each state's political geography in our empirical analyses in Section \ref{sec:bugs}.

For the purposes of our analysis, we need partisan data in order to have an estimate, for a potential redistricting map, as to how many seats each party could win with that map.  To do this, we use past election data as a proxy for partisan preference.  Using a single election outcome gives a clear interpretation of the partisan data, as opposed to an index which combines a variety of different data.   This is why, in our analyses, we refer to the number of districts ``won:''  we report the number of districts that would have been won with a potential map, using partisan data from a past statewide election. In this paper we address the question of when partisan symmetry metrics do or do not effectively represent the number of districts won by a party, and we are particularly interested in the situation where extreme values in the metrics do not correspond to extreme values in the number of seats won and vice versa.

\subsection{Motivating Examples}\label{sec:MotivatingExamples}

A key reason for focusing on MM and PB in this article is that these metrics have found common usage among the community of experts analyzing redistricting plans. These use cases have become quite broad and have begun to be adopted by members of the public, including through  public-facing apps and in commissions. In these settings, much of the nuance of the metrics can be lost, leading to individuals focusing on simply minimizing these values as a goal to reduce vote dilution or to increase fairness. We argue in Sections \ref{sec:bounds} and \ref{sec:bugs} that this approach is unlikely to be generically successful, but in this section we provide some examples of the ways these metrics are being used in practice to support the need for a closer analysis of their consequences.  One  area of concern here is that individuals will attempt to optimize for these metrics and claim that any map reaching the values close to zero must be fair. Given the range of seats outcomes possible with these metric scores, this is a potentially misleading conclusion. 

\subsubsection{Expert Testimony}
Both the Partisan-Bias and Mean-Median difference have become ubiquitous in expert reports analyzing redistricting plans for court cases. This includes experts writing in amicus briefs regarding gerrymandering, including LULAC v. Perry (2006) \cite{Amicus_LulacvPerry}, Whitford v. Gill (2018) \cite{Amicus_WhitfordvGill}, and Rucho v. Common Cause (2019) \cite{Amicus_RuchovCommonCause}, as well as by testifying experts in their reports 
including Session v. Perry (2003) \cite{EW_SessionvPerry}, Perez v Perry (2011) \cite{EW_PerezvPerry}, Favors v. Cuomo (2012) \cite{EW_FavorsvCuomo}, and an appeal of Clarke v. Wisconsin Elections Commission (2023) \cite{ WI_All, EW_ClarkevWisconsin}. In court cases where many parties have proposed maps to be evaluated, it is often the case that the experts representing each party compute and discuss the symmetry measures. As an example, in the Congressional and Legislative redistricting lawsuits in Pennsylvania in 2022, over half a dozen experts used these values in their analyses  \cite{PA2022}. These examples and the judicial opinions relying on them highlight just how ubiquitous these measures have become for evaluating redistricting plans. 

In addition to the amicus briefs and reports listed above, court consultants and special masters have also relied on these metrics to analyze redistricting plans.  A significant example from 2024 includes the analysis of Cervas and Grofman as court appointed consultants in Wisconsin \cite{GrofmanCervas}. In their report they analyze both the MM and PB across a variety of elections to determine whether there is bias in the proposed maps. As an example, after measuring the PB for all of the plans across 13 elections they write: \begin{quotation}
    The deviations from political neutrality in these maps are a significant reduction from the Current plan and are similar to values that other state courts have viewed as acceptable compliance with their state constitution regarding neither favoring nor disfavoring a particular party (though 0 bias is preferable).  
\end{quotation} 

This type of statement directly incentivizes line drawers to attempt to minimize PB, perhaps across a set of available election data, in order to construct maps achieving this `preferable' value. Regardless of what the creators and promoters of the Mean-Median Difference and Partisan Bias metrics intended those metrics to detect and measure, a takeaway from these reports in the court of public opinion is that they detect gerrymandering and by proxy whether an unfair number of districts would be won by one party under the plan under consideration.

\subsubsection{Public-Facing Apps}

Several apps have recently become available that make it easier than ever for the public to engage with the line drawing process, including drawing their own maps and evaluating current and proposed maps in their states. Examples of  a line drawing app include 
Dave's Redistricting App \cite{DRA}, while PlanScore has become a popular tool for scoring  enacted and proposed redistricting plans\cite{PlanScore}. The scores are also available to assess maps in the Maptitude for Redistricting software \cite{Maptitude}.   These apps offer friendly interfaces and results, and each of them reports on the partisan symmetry measures, often describing them as effective ways to detect unfairness or gerrymandering, and complementarily as optimization targets for fair maps. For example, the Dave's Redistricting App blog about Partisan Bias \cite{DRA_PB} writes that ``Because larger values mean more bias, smaller raw values of partisan bias are better.'' This perspective is often adopted by individuals that use these apps, although the app developers themselves note that there are significant complexities in comparing proportionality to partisan symmetry scores \cite{DRA_FAIR}.

Additionally, the Princeton Gerrymandering Project \cite{PGP} and the Harvard-based ALARM Project \cite{ALARM} have recently released ensemble analyses of most states in the country. Most of the ``Redistricting Report Cards'' that can be found on the Princeton Gerrymandering Project's site \cite{PGP} report both the Mean-Median Difference and Partisan Bias, and the ALARM Project reports the Partisan Bias and Efficiency Gap scores for its ensembles. The tooltip on the PGP report card for the Mean-Median Difference states that ``If the difference is large it could indicate a gerrymandered map where voters are crammed in a few districts.'' As with the apps, these comments lead the audience to the conclusion that smaller values lead to better representational outcomes.

\subsubsection{Optimizing for Mean-Median} 

Another formal setting in which the partisan symmetry metrics were analyzed was during the line drawing commissions work in Michigan. In the report, ``A Commissioners Guide to Redistricting in Michigan,'' prepared by researchers at Princeton University \cite{Com_Guide_MI}, a clear directive is given: ``Commissioners should make an effort to minimize the degree of
lopsided wins, mean-median difference, and the efficiency gap while
being consistent with other higher-ranked criteria.'' 
In this case a low MM value was sought as an ideal property of a plan, although the analysis in \cite{ImplementingPartisanSymmetry} shows that this is not an effective approach on its own. The Mean-Median Difference was used in an expert report to the Michigan Independent Citizens Redistricting Commission \cite{CommissionHandley}, and news reports suggest that the Michigan Independent Citizens Redistricting Commission was recommended to use maps having the Mean-Median Difference within symmetric bounds around 0, with the expectation that this would limit partisan gerrymandering \cite{HandleyNews}.

Additionally, partially fueled by the rise of the public-facing apps discussed above, this redistricting cycle saw a large surge in map submissions from the general public. Sometimes in response to map making competitions, these map makers often argued for the adoption of their plans based on scores computed from Dave's Redistricting App. An example of this is the proposals and amicus briefs filed by Petering in the Wisconsin redistricting litigation discussed above. In these briefs \cite{WI_All} Petering argued for the adoption of exactly the metrics reported by Dave's Redistricting App and used those metrics to advance arguments about the maps under consideration. 
 The Wisconsin court case where Petering filed his maps was particularly interesting, as the parties were required to upload copies of their maps to Dave's Redistricting App upon submission at the request of the court appointed consultants quoted above. 

\subsubsection{Public Policy} Recent legislative efforts in the United States, at both the state and federal levels, have attempted to incorporate standards based on partisan symmetry into criteria for evaluating plans and preventing gerrymandering. Consider the following text from H.R. 1 as passed by the House of Representatives in 2021: 

\begin{verbatim}
   [...] a redistricting plan shall be deemed to have the effect of unduly 
                favoring or disfavoring a political party if--
                            (i) modeling based on relevant historical 
                        voting patterns shows that the plan is 
                        statistically likely to result in a partisan 
                        bias of more than one seat in States with 20 or 
                        fewer congressional districts or a partisan 
                        bias of more than 2 seats in States with more 
                        than 20 congressional districts, as determined 
                        using quantitative measures of partisan 
                        fairness, which may include, but are not 
                        limited to, the seats-to-votes curve for an 
                        enacted plan, the efficiency gap, the 
                        declination, partisan asymmetry, and the mean-
                        median difference, [...]
\end{verbatim}

This very clearly situates the Partisan Bias and Mean-Median Difference as metrics for detecting gerrymandering and hence could incentivize line drawers to minimize these quantities in order to pass the relevant test. See further discussion in \cite{Forum} of the relationship between the test proposed in the Freedom to Vote Act and partisan symmetry measures, including in terms of their normative properties.

Legislation addressing partisan fairness measures in redistricting has also been proposed at the state level, particularly in states with a citizen initiative process. There has often been a lack of modeling in advance of the adoption of these metrics, which can lead to unintended consequences when line drawers attempt to optimize for these ostensibly neutral measures \cite{ImplementingPartisanSymmetry, Competitiveness}. These conflicts are often unintuitive to the public and can lead to partisan tension focused on language. The public discussion around the  proportionality requirement and corresponding failure of Ohio's Issue 1 provides an example of this type of conflict. 

\subsubsection{Motivating Examples: Summary}

In each of these domains the MM and PB, among other partisan symmetry metrics, are being used to support the analysis of districting plans and to detect or mitigate gerrymandering. In many of these cases, there are unstated assumptions about the relationship between idealized values of the metrics and the common public understanding  of gerrymandering in terms of the number of seats won by each party. Throughout this paper and particularly in Section \ref{sec:boundsDiscussion} we use these practical examples to motivate our results.

\section{Definitions and other Established Facts}\label{sec:definitions}

The Mean-Median Difference and Partisan Bias metrics are well-known in the community, but for completeness we give their definitions here.  We shall also share other established facts about these metrics, to help the unfamiliar reader to form an intuition and deeper understanding of these metrics.

\begin{Definition}\label{def:MM}
Consider an election with $n$ districts.  Let $V_1 \leq V_2 \leq \cdots \leq V_n$ be the vote shares for party $A$.  Then the \emph{Mean-Median Difference} for this election is
\begin{equation*}
MM = \text{median}\left\{V_1, V_2, \dots, V_n\right\} - \text{mean}\left\{V_1, V_2, \dots, V_n\right\}
\end{equation*}
\end{Definition}

\begin{Definition}\label{def:PB}
Consider an election with $n$ districts.  Let $V_1 \leq V_2 \leq \cdots \leq V_n$ be the vote shares for party $A$, and let $\overline{V} =\text{mean}\left\{V_1, V_2, \dots, V_n\right\}$ .  Then the \emph{Partisan Bias} for this election is
\begin{equation*}
PB = \frac{1}{2} \left(\text{Proportion of } V_i \text{ larger than } \overline{V} - \text{Proportion of } V_i \text{ smaller than } \overline{V}\right)
\end{equation*}
\end{Definition}

In the next couple of sections, (Sections \ref{sec:SVcurve} and \ref{sec:Relationships}) we assume that the \emph{statewide} vote share $V$ is simply the average of the vote shares in each district: $V = \frac{1}{n} \sum_{i=1}^n V_i$.  In particular, note that the creation of the seats-votes curve assumes that the statewide vote share $V$ satisfies $V = \frac{1}{n} \sum_{i=1}^n V_i$.  This is true when turnout in each district is the same (though is not true in general).  We drop this assumption in Section \ref{sec:UnequalTurnout}.  

\subsection{The Seats-Votes Curve}\label{sec:SVcurve}

Both the Mean-Median Difference and the Partisan Bias can be interpreted in terms of the seats-votes curve when drawn under the assumption of \emph{uniform partisan swing}.  It is worth noting that there are many different ways to draw a seats-votes curve; see Section 1.1.1 of \cite{ImplementingPartisanSymmetry} for a nice summary of various methods.  The construction of the seats-votes curve using uniform partisan swing starts with estimates for the party $A$ vote share $V_i$ in each district $i = 1, 2, \dots, n$.  This data corresponds to a statewide vote share $V$ (for now assumed to be the mean of the $V_i$s), and seat share $S$.  One assumes that if the statewide vote share swings towards party $A$, the vote shares $V_i$ swing towards party $A$ in a uniform manner; they all increase by the same vote share or decrease by the same vote share.  This results in a new vote share $V'$, and (depending on whether that vote share swing was sufficient to flip a district) a potentially new seat share $S'$.  All such new $(V', S')$ are drawn in the plane, which results in a curve which is a step function.  Two examples can be seen in Figure \ref{fig:SVcurve} below.  One curve corresponds to the seats-votes curve for a map whose vote shares are  0.2, 0.3, 0.55, 0.6, and 0.65; the other corresponds to real data from the Massachusetts 2011 congressional map.

\begin{figure}[h]
\centering
\includegraphics[width=2.5in]{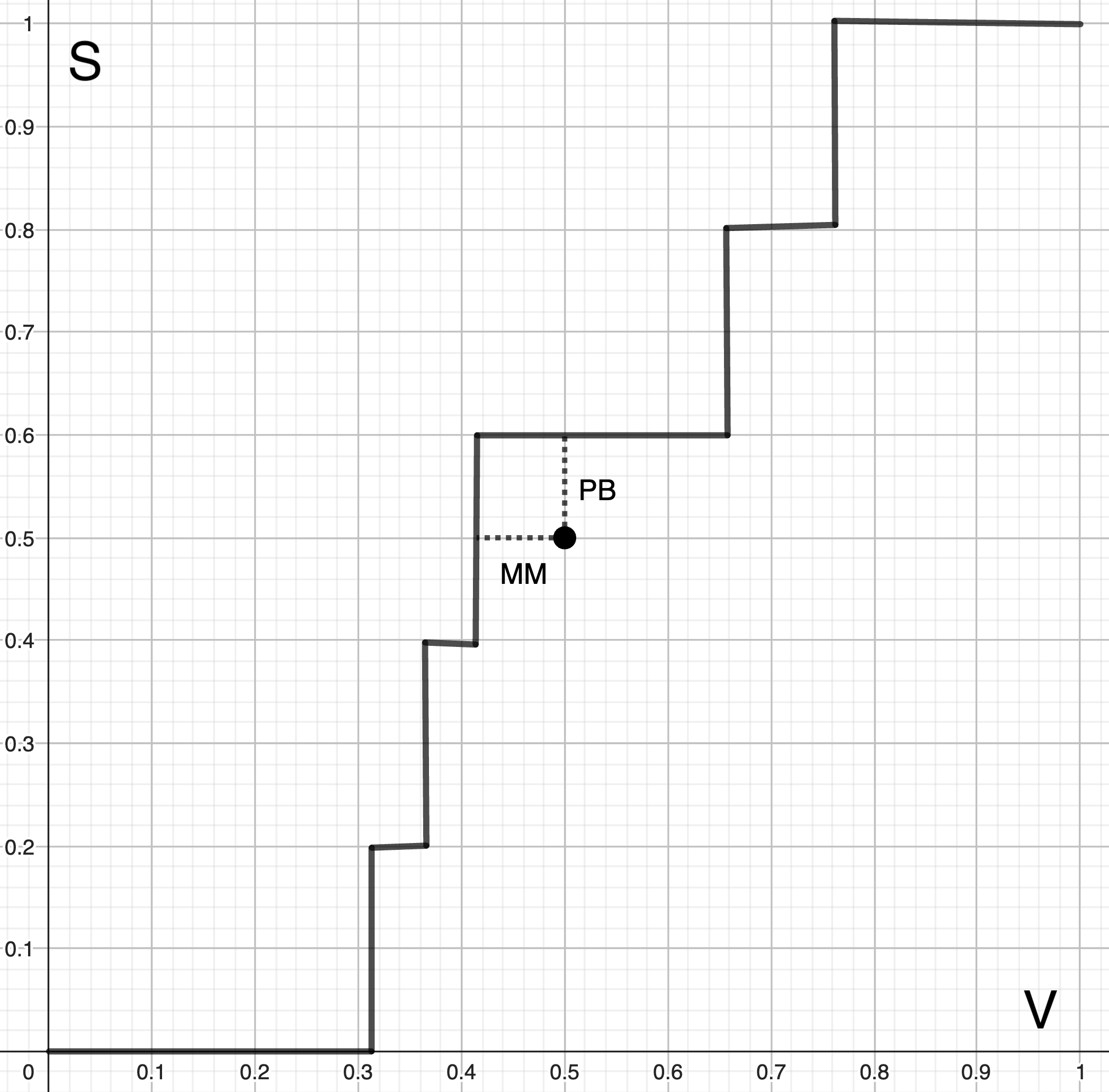} \hspace{0.5 cm}
\includegraphics[width=2.5in]{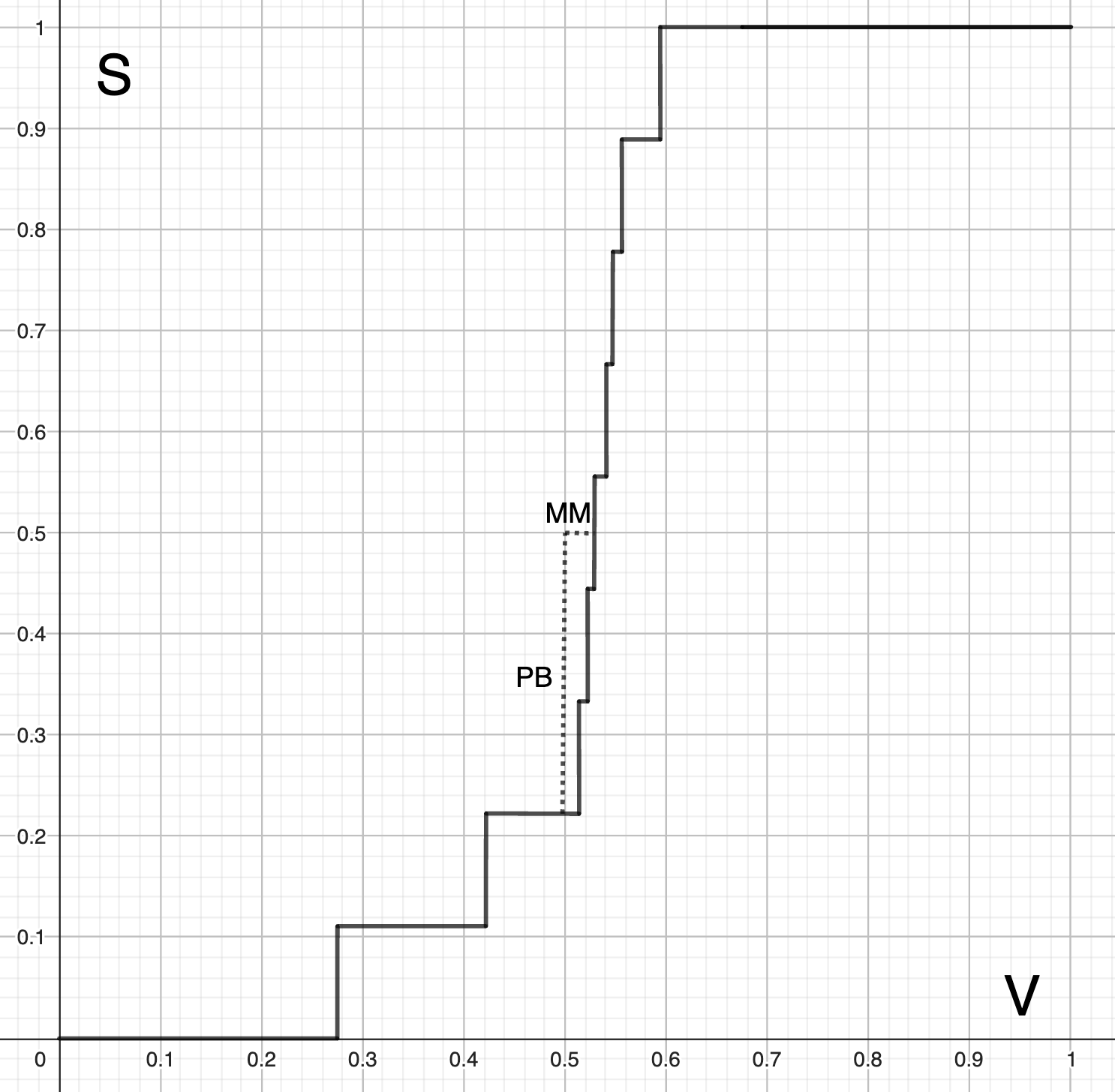}
\caption{On the left: seats-votes curve for an election with districts receiving vote shares of 0.2, 0.3, 0.55, 0.6, and 0.65.  On the right: seats-votes curve for Massachusetts 2011 congressional map and the 2016 presidential election.  Massachusetts data from \cite{MGGGstates}.}
\label{fig:SVcurve}
\end{figure}

It is well-known (see, for example, \cite{ImplementingPartisanSymmetry}) that the Mean-Median difference is the distance between the point (0.5, 0.5) and the intersection point between the seats-votes curve and the line $S = 0.5$.  Indeed, one can see from the figure that the distance between the point (0.5, 0.5) and the intersection of the seats-votes curve with the line $S = 0.5$ is a little under 0.1.  If we calculate its value using Definition \ref{def:MM}, we have $MM = 0.09$.  Because of the way it can be visualized on the seats-votes curve, the Mean-Median Difference is described as the amount the vote share can fall below (or exceed) 50\% while resulting in a seat share of 50\%.

It is also well-known (see again \cite{ImplementingPartisanSymmetry}) that the Partisan Bias is the distance between the point (0.5, 0.5) and the intersection point between the seats-votes curve and the line $V = 0.5$.   One can see from the figure that the distance between the point (0.5, 0.5) and the intersection of the seats-votes curve with the line $V = 0.5$ is approximately 0.1.  If we calculate its value using Definition \ref{def:PB}, we have $PB = 0.1$.  Because of the way it can be visualized on the seats-votes curve, the Partisan Bias is described as the amount the seat share can exceed (or fall below) 50\% when the vote share is 50\%.

The correspondence between the $MM$, $PB$, and the seats-votes curve described above is only true when the curve is drawn under the assumption of uniform partisan swing.  
Of course, uniform partisan swing is by no means guaranteed or even likely to happen in practice (see \cite{PartisanSwingModels} for other models of swing).  Grofman and King \cite{PBGrofmanKing} do not even define the Partisan Bias as we do in Definition \ref{def:PB}, and instead define it using the intersection of the seats-votes curve with the line $V = 0.5$, where the seats-votes curve is constructed not by assuming uniform partisan swing, but by using statistical methods focused on the rank order\footnote{Rank is in terms of partisan support for a particular political party.} of districts.  Grofman and King acknowledge that Partisan Bias is frequently defined as in our Definition \ref{def:PB}, but call the definition of the Partisan Bias which uses the aforementioned statistical technique of defining the seats-votes curve  the ``state of the art'' version of the metric \cite{PBGrofmanKing}.  However, as DeFord et al.\ point out,  complex definitions (such as Grofman and King's ``state of the art'' version of Partisan Bias) must be ``translated to the real-world `political thicket' of redistricting'' \cite{DeFord_Dhamankar_Duchin_Gupta_McPike_Schoenbach_Sim_2023}.  Wang agrees, by stating that a manageable standard to detect partisan gerrymandering ``should be able to be clearly stated without case-specific or mathematics-intensive assumptions'' \cite{3PracticalTestsWang}.

Indeed, our Definitions \ref{def:MM} and \ref{def:PB} are precisely the implementations of Mean-Median Difference and Partisan Bias that are currently being used in the ``real-world political thicket of redistricting.'' Dave's Redistricting App and Maptitude both calculate the Mean-Median Difference and Partisan Bias identically to Definitions \ref{def:MM} and \ref{def:PB}, subject to choice of underlying vote data.  Similarly, the Princeton Gerrymandering Project's Report Cards \cite{PGP} use the same type of calculation and  the Python library GerryChain \cite{GerryDetails}   has built-in methods for calculating the Mean-Median Difference and Partisan Bias, again using formulations identical to our Definitions \ref{def:MM} and \ref{def:PB}.  PlanScore \cite{PlanScore} also reports values of the Mean-Median Difference and Partisan Bias, contributing to the ubiquity of these metrics\footnote{The definitions, vote data, and metric scores on PlanScore's site are less clear about the exact aproach used by Planscore to calculate the MM and PB.}  
These calculations are also ubiquitous in academic research and expert testimony. 
We believe it not only useful, but imperative to study the version of the Partisan Bias (and Mean-Median Difference) that is used in software, websites analyzing maps, research, and expert testimony, which is why we define those metrics as they are being used there.

\subsection{Relationship between Mean-Median Difference and Partisan Bias}\label{sec:Relationships}

The Mean-Median Difference and Partisan Bias metrics cannot disagree in sign: if one is positive, the other is non-negative, and if one is negative the other is not positive.  This follows directly from the ways the MM and PB can be calculated on a seats-votes curve drawn using uniform partisan swing (as described in Section \ref{sec:SVcurve}), and the fact that the seats-votes curve is a non-decreasing curve going through points (0,0) and (1,1).  Note that whenever $MM = 0$, the median of the $V_i$s is equal to the mean.  This implies that the number of $V_i$ above the mean must be equal to the number of $V_i$s below; in other words, it implies that $PB = 0$.  

However, it is possible for $PB = 0$ without $MM = 0$ for a specific election.  For example, if an election had 4 districts with vote shares 0.55, 0.6, 0.7, and 0.9, $\overline{V} = 0.6875$.   Since the number of districts with vote shares larger than $\overline{V}$ is the same as the number of districts with vote shares less than $\overline{V}$, $PB = 0$.  However, the median vote share is 0.65, implying $MM = 0.65-0.6875 \not=0$.  

Nevertheless, as we shall see in Section \ref{sec:bounds}, the vote-share, seat-share pairs $(V, S)$ for which there exists constructed election data with vote share $V$, seat share $S$, and $PB = 0$ is the same as the vote-share seat-share pairs corresponding to constructed election data with $MM=0$.  When one considers the fact that seat shares $S = \frac{k}{n}$ can correspond to maps with an arbitrarily large number of districts (by scaling both the numerator and denominator by some large $M$), this fact is not difficult to believe (at least in the limit).  For example, suppose $V_1, V_2, \dots, V_n$ are district vote shares corresponding  to statewide vote share $V = \frac{1}{n} \sum_{i=1}^n V_i$, seat share $S$, and $PB = 0$.  We can then consider an election with $Mn+2$ districts, having $M$ districts of vote share $V_i$ for $i=1, 2, \dots, n$, and two districts of vote share $V$.  The mean of these districts' vote shares will still be $V$, and the seat share will be very close to $S$ (the larger $M$ is, the closer to $S$ the seat share will be).  And because the original districting map had an equal number of districts with vote share above $V$ as it had below, the median of this new election will be $V$, implying that $MM=0$ for this new election.  

The supremal difference between $PB$ and $MM$ is 0.5.  This can be argued algebraically, but it can also be seen by considering how the $MM$ and $PB$ correspond to points on the seats-votes curve.  For example, if $MM<0$, then the supremal value of $MM-PB$ occurs when $PB$ is near -0.5 and $MM$ is near 0.   An example of such a seats-votes curve can be seen in  Figure \ref{fig:MMPB_max_difference}.

\begin{figure}[h]
\centering
\includegraphics[width=3in]{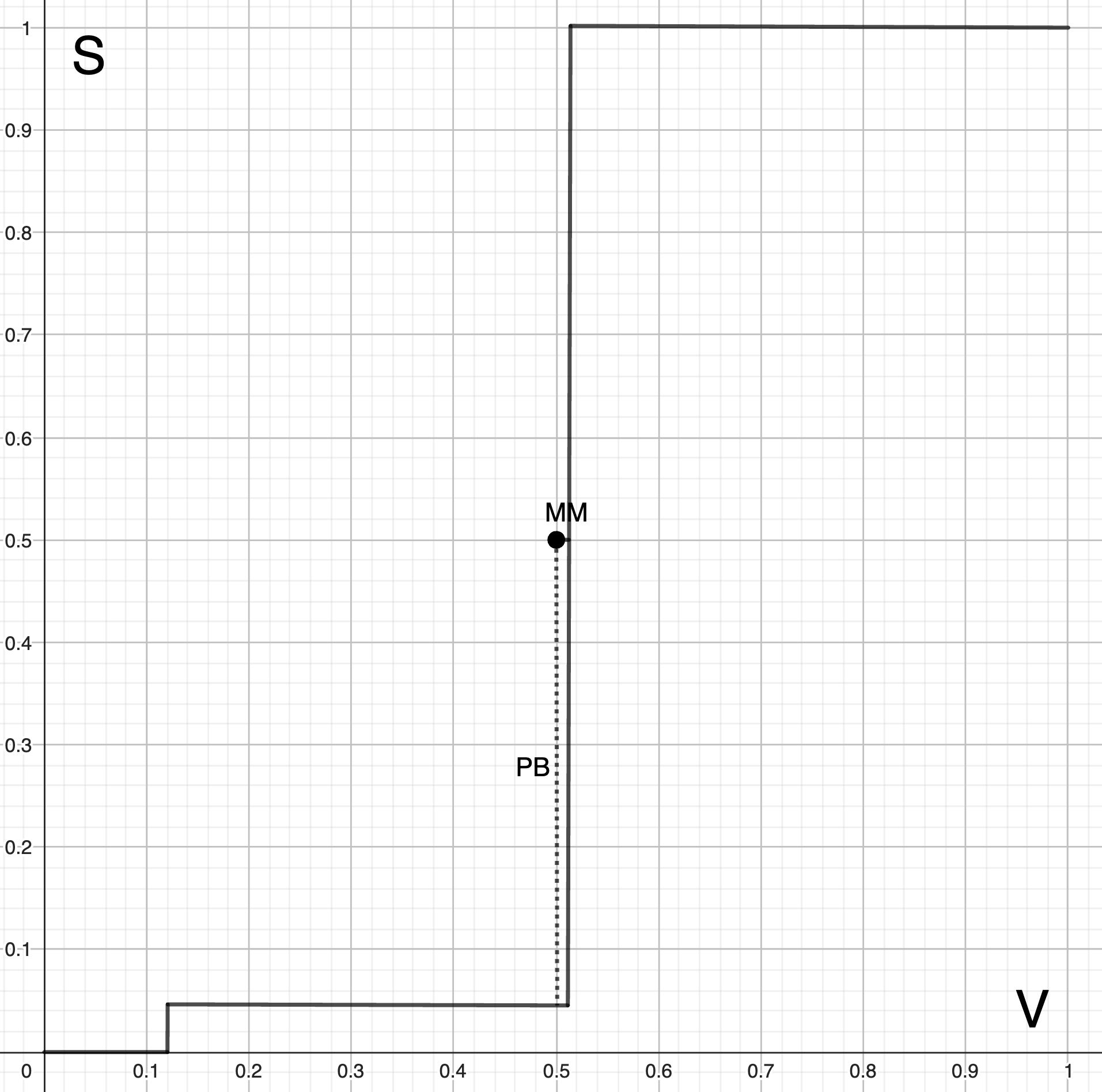}
\caption{A Seats-Votes curve whose value of $MM-PB$ is near the maximal value of 0.5.  This curve is the result of an election with 19 districts having vote share 0.49 and one district having vote share 0.89.}
\label{fig:MMPB_max_difference}
\end{figure}

\subsection{Unequal Turnout}\label{sec:UnequalTurnout}

Of course, no state has identical turnout in every district, so if the statewide vote share is $V$ and the district vote shares are $V_1, V_2, \dots, V_n$,  then in general, $V \not= \frac{1}{n}\sum_{i=1}^n V_i$.  Indeed, $V = \sum_{i=1}^n \alpha_i V_i$ where $\sum_{i=1}^n \alpha_i = 1$, and the non-negative weights $\alpha_i$ are given by turnout:
\begin{equation*}
\alpha_i = \frac{\text{turnout in district } i}{\text{turnout throughout the state}}
\end{equation*}

As we shall see in Section \ref{sec:bounds}, if we consider election data corresponding to uneven turnout between districts, we can get a wider array of $(V,S)$ pairs giving $MM=0$ and $PB=0$.  The impact of turnout differentials on measures of partisan fairness can have significant impacts on their interpretation and robustness, as discussed for the Efficiency Gap in \cite{2018arXiv180105301V}.

\section{Bounds}\label{sec:bounds}

Here, we give our mathematical results showing which values the Mean-Median Difference and Partisan Bias can take on for a fixed vote-share, seat-share pair $(V, S)$.   As we shall see, a fixed $(V, S)$ pair can correspond to an array of different data, which give vastly different values for a single metric (either the Mean-Median Difference or Partisan Bias).  Thus, our results will show that a map could have a fixed vote share, a large seat share, and the metric could indicate that the party with large seat share is gerrymandered \emph{against}.  And the same map could have the same fixed vote share, low seat share, and the metric could indicate that the map is gerrymandered in favor of the party with low seat share.  In other words, these results will indicate, at least on constructed data, that these metrics do not indicate when a party wins an extreme number of districts, as suggested by the use cases discussed in Section \ref{sec:MotivatingExamples}.

Of course, real maps and real data may be different from the space of all possible constructed data.  This is why we have the empirical arguments in forthcoming Section \ref{sec:bugs}.  But for now, we begin the setup of our theory argument.  The rigorous mathematical statements of our results are in Sections \ref{sec:boundsPB} and \ref{sec:boundsMM}.  We present a discussion of the implications of these results in Section \ref{sec:boundsDiscussion}.

We first note that, if an election has vote-share $V$ and seat-share $S$, and if we assume turnout in all districts is equal (as we do in Sections \ref{sec:boundsPB}, \ref{sec:boundsMM}, and \ref{sec:boundsDiscussion})  then we must have $0 \leq V \leq 1, 0 \leq S \leq 1$, $S \leq 2V$, and $S\geq2V-1$ (see, for example, \cite{TappPreprint}).  

Please note that Theorems \ref{thm:PB} and \ref{thm:MM} have no restrictions on the number of districts.  The corresponding Theorems for a fixed number of districts are Theorems \ref{thm:PBfixed_districts} and \ref{thm:MMfixed_districts}.

\subsection{Bounds: Partisan Bias}\label{sec:boundsPB}

\begin{Theorem}\label{thm:PB}
Suppose $(V, S)$ is a pair of rational numbers with $\frac{1}{4} \leq V \leq1$, $ \frac{1}{2} \leq S \leq 1$, $S \leq 2V$, and $S\geq2V-1$.  

If $S = 1$ and either $V = \frac{1}{2}$ or $V = 1$, then we must have $PB = 0$.

Otherwise, there exists constructed election data with vote share $V$, seat share $S$, turnout is equal in all districts, and Partisan Bias $PB$ for any $PB$ in the following ranges:
\begin{align*}
S-\frac{1}{2}  &\leq PB < \frac{1}{2}-S\left(\frac{1}{2V}-1\right) &  \text{ if }  \quad \frac{1}{4} \leq V < \frac{1}{2} \\
\frac{S-1}{2} &< PB < \frac{S}{2} & \text{ if } V = \frac{1}{2} \\
\frac{(1-S)\left(V-\frac{1}{2}\right)}{1-V} - \frac{1}{2}  &<  PB \leq S-\frac{1}{2} &  \text{ if } \quad \frac{1}{2} < V < 1\\
\end{align*}

\end{Theorem}

Note that if party $A$'s seat share is $S$, then  $1-S$ is party $B$'s seat share, and that party $A$'s PB value  is the negation of party $B$'s PB value.  Thus, Theorem \ref{thm:PB} gives all possible range of values for PB for any possible $(V, S)$.    These minimum and maximum values can be visualized in Figure \ref{fig:PB_min_max}.   Theorem \ref{thm:PB} is proved in Appendix \ref{sec:proofs}.

\begin{figure}[h]
\centering
\includegraphics[width = 3in]{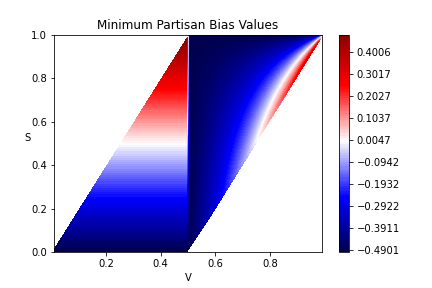}
\includegraphics[width = 3in]{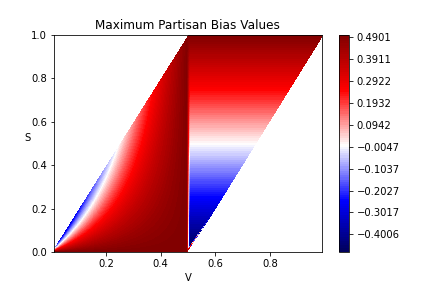}
\caption{Range of values for Partisan Bias (using constructed election data) for each vote-share, seat-share pair $(V, S)$, when turnout in each district is equal.  The hue at $(V,S)$ in the image on the left corresponds to the minimum possible value of PB.  The hue at $(V, S)$ in the image on the right corresponds to the maximum possible value of PB.  All values in between are also achievable on some constructed election data.}
\label{fig:PB_min_max}
\end{figure}

In Figure \ref{fig:PB_min_max}, we can visually see a discontinuity at $V = \frac{1}{2}$.  We can understand this intuitively, by using the following general argument: fix $V \not= \frac{1}{2}$ and $S \geq \frac{1}{2}$.  Then we can construct election data with the correct number of winning and losing districts such that the number of districts whose vote share is equal to $V$ is maximized.  If we barely raise the vote shares of nearly all of the districts whose vote share is $V$ (and lower a single seat share), this maximizes PB, and keeps the correct vote share and seat share.  If we barely lower the vote shares of nearly all of the districts whose vote share is $V$ (and raise a single seat share), this minimizes PB and keeps the correct vote share and seat share.  

However, this argument breaks down when $V = \frac{1}{2}$ because barely raising or lowering the vote share can change it from winning to losing or vice versa.  This results in  the discontinuity at that line.

Theorem \ref{thm:PB} gives bounds on the value of the Partisan Bias when the number of districts is unlimited.  We have the following bounds for a fixed number of districts:

\begin{Theorem}\label{thm:PBfixed_districts}
Suppose $(V, S)$ is a pair of rational numbers with $\frac{1}{4} \leq V \leq 1$, $ \frac{1}{2} \leq S \leq 1$, $S \leq 2V$, and $S\geq 2V-1$.  Suppose that $n$ is the number of districts; that is, $S$ can be written as a rational number with denominator $n$.  

If $S = 1$ and either $V = \frac{1}{2}$ or $V = 1$, then we must have $PB = 0$.

Otherwise, there exists constructed election data with vote share $V$, seat share $S$, $n$ districts, turnout is equal in all districts, and Partisan Bias $PB$ for any $PB$ in the following ranges:
\begin{align*}
S-\frac{1}{2}  &\leq PB \leq \frac{1}{2}\left(2S-1 +\frac{2\left\lceil n \left(1-\frac{S}{2V}\right)-1\right\rceil}{n}\right)&  \text{ if }  \quad \frac{1}{4} \leq V < \frac{1}{2} \\
\frac{1}{2}\left(S-1+\frac{1}{n}\right)&\leq PB \leq  \frac{1}{2} \left(S-\frac{1}{n}\right) & \text{ if } V = \frac{1}{2} \\
\frac{1}{2}\left(2S-1-\frac{2 \left\lceil n \left(1-\frac{1-S}{2(1-V)}\right)-1\right\rceil}{n}\right)  &\leq  PB \leq S-\frac{1}{2} &  \text{ if } \quad \frac{1}{2} < V < 1, S \not= 1\\
\frac{1}{n}-\frac{1}{2} &\leq PB \leq \frac{1}{2}-\frac{1}{n}&  \text{ if } \quad \frac{1}{2} < V < 1, S = 1\\
\end{align*}
\end{Theorem}

Again, by the symmetry of the parties in their corresponding PB values, Theorem \ref{thm:PBfixed_districts} gives all possible range of values for PB for any possible $(V, S)$.  Of course, if we take the limit as $n \to \infty$ in Theorem \ref{thm:PBfixed_districts}, we obtain the result in Theorem \ref{thm:PB}.

\subsection{Bounds: Mean-Median Difference}\label{sec:boundsMM}

We also have corresponding Theorems for the Mean-Median Difference:

\begin{Theorem}\label{thm:MM}
Suppose $(V, S)$ is a pair of rational numbers with $0 \leq V \leq 1$, $ \frac{1}{2} \leq S \leq 1$, $S \leq 2V$, and $S\geq2V-1$.  Then there exists constructed election data with vote share $V$, seat share $S$, turnout is equal in all districts, and Mean-Median Difference $MM$ for any $MM$ in the following ranges:
\begin{align*}
\frac{1}{4}-V &\leq MM < \frac{1}{4}  & \text{ if } S = \frac{1}{2}, \quad  \frac{1}{4} \leq V < \frac{1}{2} \\
-\frac{1}{4} &< MM \leq \frac{3}{4}-V  & \text{ if } S = \frac{1}{2}, \quad  \frac{1}{2} \leq V \leq \frac{3}{4} \\
\frac{1}{2}-V  &\leq MM <  \min\left\{1-V, V-S+\frac{1}{2}\right\} &  \text{ if } S > \frac{1}{2}, \quad \frac{1}{4} \leq V < \frac{3}{4} \\
\frac{2V+S-2}{2S-1}-V  & < MM \leq 1-V&  \text{ if } S > \frac{1}{2}, \quad \frac{3}{4} \leq V \leq 1
\end{align*}

\end{Theorem}

Again, note that if party $A$'s seat share is $S$, then  $1-S$ is party $B$'s seat share, and that party $A$'s MM value is the negation of party $B$'s MM value  so that Theorem \ref{thm:MM} gives all possible range of values for MM for any possible $(V, S)$.   

These minimum and maximum values can be visualized in Figure \ref{fig:MM_min_max}.  Theorem \ref{thm:MM} is also proved in Appendix \ref{sec:proofs}.

\begin{figure}[h]
\centering
\includegraphics[width = 3in]{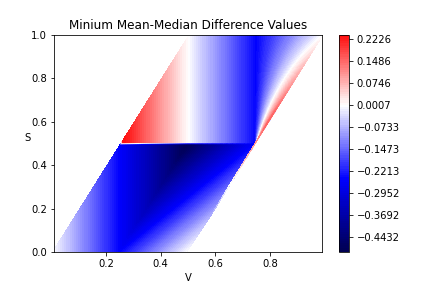}
\includegraphics[width = 3in]{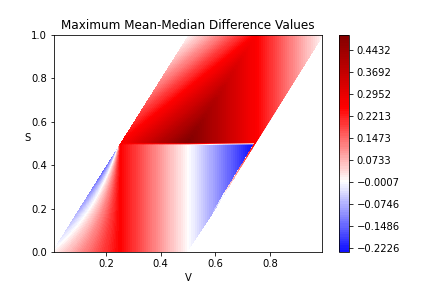}
\caption{Range of values for Mean-Median Difference (using constructed election data) for each vote-share, seat-share pair $(V, S)$, when turnout in each district is equal.  The hue at $(V,S)$ in the image on the left corresponds to the minimum possible value of MM.  The hue at $(V, S)$ in the image on the right corresponds to the maximum possible value of MM.  All values in between are also achievable on some constructed election data.}
\label{fig:MM_min_max}
\end{figure}

In Figure \ref{fig:MM_min_max}, we can visually see a discontinuity at $S = \frac{1}{2}$.  This discontinuity can be understood intuitively as follows:  fix $\frac{1}{2} \leq V  \leq \frac{3}{4}$.  When $S > \frac{1}{2}$ gets close enough to $\frac{1}{2}$,  the largest that the median could possibly be is 1 and smallest the median could possibly be is $\frac{1}{2}$.  Whereas when $S < \frac{1}{2}$ the largest the median could possibly be is $\frac{1}{2}$ and smallest it could possibly be is 0.   However, the mean is fixed at $V$.  This creates the discontinuity at $S = \frac{1}{2}$.

Theorem \ref{thm:MM} gives bounds on the value of the Mean-Median Difference when the number of districts is unlimited.  We have the following bounds for a fixed number of districts:

\begin{Theorem}\label{thm:MMfixed_districts}
Suppose $(V, S)$ is a pair of rational numbers with $\frac{1}{4} \leq V \leq 1$, $ \frac{1}{2} \leq S \leq 1$, $S \leq 2V$, and $S\geq2V-1$.  Suppose that $n \geq 3$ is the number of districts and that  $\ell$ is the number of districts lost by the party in consideration, so that $S = \frac{n-\ell}{n}$.    

If $S = \frac{1}{2}$, then there exists constructed election data with vote share $V$, seat share $S$, turnout is equal in all districts, and Mean-Median Difference $MM$ for any $MM$ in the following ranges:
\begin{align*}
\frac{1}{4}-V &\leq MM \leq \frac{1}{4}-\frac{1}{2n}  & \text{ if } S = \frac{1}{2}, \quad  \frac{1}{4} \leq V \leq \frac{n-1}{2n} \\
\frac{1}{2n}-\frac{1}{4}&\leq MM \leq \frac{1}{4}-\frac{1}{2n}  & \text{ if } S = \frac{1}{2}, \quad  \frac{n-1}{2n} < V < \frac{n+1}{2n} \\
\frac{1}{2n}-\frac{1}{4}&\leq MM \leq \frac{3}{4}-V  & \text{ if } S = \frac{1}{2}, \quad V \geq \frac{n+1}{2n}\\
\end{align*}

If $S > \frac{1}{2}$ and $n$ is even, then there exists constructed election data with vote share $V$, seat share $S$, turnout is equal in all districts, and Mean-Median Difference $MM$ for any $MM$ in the following ranges:
\begin{align*}
\frac{1}{2}-V  &\leq MM \leq  \min\left\{1-V,  \frac{2V+\frac{1}{2}-S+\frac{1}{n}}{1+\frac{2}{n}}-V \right\} &  \text{ if } S > \frac{1}{2}, \quad \frac{1}{4} \leq V \leq \frac{3n-2}{4n} \\
\frac{2V+S-2+\frac{2}{n}}{2S-1+\frac{2}{n}}-V  & \leq MM \leq \min\left\{1-V,  \frac{2V+\frac{1}{2}-S+\frac{1}{n}}{1+\frac{2}{n}} - V \right\} &  \text{ if } S > \frac{1}{2}, \quad \frac{3n-2}{4n} < V \leq 1
\end{align*}

And if $S > \frac{1}{2}$ and $n$ is odd, then there exists constructed election data with vote share $V$, seat share $S$, turnout is equal in all districts, and Mean-Median Difference $MM$ for any $MM$ in the following ranges:
\begin{align*}
\frac{1}{2}-V  &\leq MM \leq   \min\left\{1-V,  \frac{2V+\frac{1}{2}-S+\frac{1}{2n}}{1+\frac{1}{n}}-V \right\}  &  \text{ if } S > \frac{1}{2}, \quad \frac{1}{4} \leq V < \frac{3n-1}{4n} \\
\frac{2V+S-2+\frac{1}{n}}{2S-1+\frac{1}{n}}-V  & \leq MM \leq \min\left\{1-V,  \frac{2V+\frac{1}{2}-S+\frac{1}{2n}}{1+\frac{1}{n}}-V \right\}  &  \text{ if } S > \frac{1}{2}, \quad \frac{3n-1}{4n} < V \leq 1
\end{align*}
\end{Theorem}

Again, by the symmetry of the parties in their corresponding MM values, Theorem \ref{thm:MMfixed_districts} gives all possible range of values for MM for any possible $(V, S)$.

\subsection{Bounds: Discussion}\label{sec:boundsDiscussion}

From Figures \ref{fig:PB_min_max} and \ref{fig:MM_min_max}, we can see that the values that MM and PB can take respond differently to changing $V$ and $S$.  In this section, we illustrate how the values that Partisan Bias and Mean-Median Difference are \emph{able} to take on do not generally become more extreme as the number of seats won becomes more extreme (in Sections \ref{sec:PBability} and \ref{sec:MMability}).  In sections \ref{sec:discussionMM} and \ref{sec:discussionPB}, we also give some specific election examples (and their corresponding $V$, $S$, and metric values) that run counter to expectations of MM and PB tracking extreme seats outcomes, based on the use cases outlined in Section \ref{sec:MotivatingExamples}.  

\subsubsection{Partisan Bias: Decoupling of Extreme Values with Extreme Number of Districts Won}\label{sec:PBability}

For a fixed vote share $V$, the Partisan Bias does not in general allow for more extreme values when the seat share $S$ becomes small.  Consider a vote share $V> 0.5$ for party $A$, and the potential values that the Partisan Bias can take on with different seat shares $S$.  From Theorem \ref{thm:PB} (and Figure \ref{fig:PB_min_max}), we can see that as $S$ \emph{decreases}, the minimum values that the Partisan Bias function can take on \emph{increase}.  In other words, for a fixed $V > 0.5$, as the number of seats that party $A$ wins \emph{decreases} the Partisan Bias function is \emph{less} able to indicate that party $A$ is winning extremely few seats.  In particular, if $V > 0.75$, once the number of seats won decreases far enough, the Partisan Bias can only take on positive values.  This can be seen in the ``red crescent'' in Figure \ref{fig:PB_min_max}.  That is, once the number of seats won by party $A$ is \emph{low} enough, the Partisan Bias can only indicate a gerrymander \emph{in favor} of party $A$.  The Partisan Bias is \emph{not able} to indicate that party $A$ is winning extremely few seats\footnote{By symmetry, one can similarly show that for $V < 0.5$, the Partisan Bias is not able to indicate that party $A$ is winning an extreme number of seats.}.

Similarly, for a fixed seat share $S$, the Partisan Bias does not allow for more extreme values indicating a less favorable position for party $A$ as party $A$'s vote share $V$ increases, nor does it allow for more extreme values indicating a more favorable position for party $A$ as $V$ decreases.  For example, consider a fixed $S$ for party $A$.  From Theorem \ref{thm:PB} (and Figure \ref{fig:PB_min_max}), we can see that as $V$ increases beyond 0.5, the minimum values that the Partisan Bias can take on \emph{increase}.  That is, for a fixed seat share, if the vote share increases, the Partisan Bias is \emph{less able} to indicate a gerrymander \emph{against} party $A$. Indeed, if $S > 0.5$ and the vote share increases large enough, the Partisan Bias is only able to take on positive values (the ``red crescent'' in Figure \ref{fig:PB_min_max}), indicating a gerrymander \emph{favoring} party $A$. Similarly, for a fixed $S$ for party $A$, as $V$ decreases below 0.5, the maximum values that the Partisan Bias can take on  \emph{decrease}.  That is, for a fixed seat share, if the vote share decreases, the Partisan Bias is \emph{less able} to indicate a gerrymander \emph{favoring} party $A$.

\subsubsection{Partisan Bias: Elections with Values Counter to Public Expectations}\label{sec:discussionPB}
In Section \ref{sec:MotivatingExamples}, we mentioned many instances in which the Partisan Bias has been used (or its usage has been promoted) in an effort to avoid maps that allow for one party to win an extreme number of districts.  In this section, we give sample elections along with their Partisan Bias calculations.  We believe these give convincing evidence that the Partisan Bias does not act in the ways suggested by many of the testimonies, apps, reports, and policies mentioned in Section \ref{sec:MotivatingExamples}.

From Figure \ref{fig:PB_min_max}, we can see that the same coordinate position $(V, S)$ can have extremely different hues in the two images.  This corresponds to the fact that, for fixed $V$ and $S$ values, the Partisan Bias can vary wildly.  For example, consider the election data in Table \ref{tab:PBv60s90}

\begin{table}[h!]
    \centering
    \begin{tabular}{|c|c|c|c|c|c|c|c|c|c|c|}
        \hline
        District & 1 & 2 & 3 & 4 & 5 & 6 & 7 & 8 & 9 & 10\\ \hline\hline
        \multirow{4}{*}{Vote share} & \multicolumn{10}{c|}{Election 1} \\ \cline{2-11}
        & 37\% & 61\% & 61\% & 61\% & 61\% & 61\% & 61\% & 61\% & 61\% & 75\% \\ \cline{2-11}
        & \multicolumn{10}{c|}{Election 2} \\ \cline{2-11}
        & 49\% & 59\% & 59\% & 59\% & 59\% & 59\% & 59\% & 59\% & 59\% & 79\% \\ \hline
    \end{tabular}
    \caption{Two different elections with $S = 90\%$ and $V = 60\%$.  Election 1 has $\text{PB} = 0.4$, while Election 2 has $\text{PB} = -0.4$.}
    \label{tab:PBv60s90}
\end{table}

Although both Elections in Table \ref{tab:PBv60s90} have a vote share of 60\% and seat share of 90\%, the Partisan Bias metric considers Election 1 to be an extreme gerrymander in favor of the party in question, while Election 2 is considered to be an extreme gerrymander \emph{against} the party in question.  The different values here indicate that the Partisan Bias believes that 9 out of 10 seats is extremely large in the first election, but not nearly large enough in the second.  

Additionally, note in Figure \ref{fig:PB_min_max} that for $V$ just barely less than $\frac{1}{2}$, PB can achieve nearly the maximal value of $\frac{1}{2}$, regardless of $S$.  (And symmetrically, for $V$ just barely above $\frac{1}{2}$, PB can achieve nearly the minimal value of $-\frac{1}{2}$).  For example, consider the elections in Table \ref{tab:v48pb4}.

\begin{table}[h!]
    \centering
    \begin{tabular}{|c|c|c|c|c|c|c|c|c|c|c|}
        \hline
        District & 1 & 2 & 3 & 4 & 5 & 6 & 7 & 8 & 9 & 10\\ \hline\hline
        \multirow{6}{*}{Vote share} & \multicolumn{10}{c|}{Election 1} \\ \cline{2-11}
        & 35\% & 49\% & 49\% & 49\% & 49\% & 49\% & 49\% & 49\% & 49\% & 53\% \\ \cline{2-11}
        & \multicolumn{10}{c|}{Election 2} \\ \cline{2-11}
        & 19\% & 49\% & 49\% & 49\% & 49\% & 53\% & 53\% & 53\% & 53\% & 53\% \\ \cline{2-11}
        & \multicolumn{10}{c|}{Election 3} \\ \cline{2-11}
        & 11\% & 51\% & 51\% & 51\% & 51\% & 53\% & 53\% & 53\% & 53\% & 53\% \\ \hline
    \end{tabular}
    \caption{Three different elections with $V = 48\%$ and $\text{PB} = 0.4$.  Election 1 has $S = 10\%$, Election 2 has $S = 50\%$, and Election 3 has $S = 90\%$.}
    \label{tab:v48pb4}
\end{table}

In each of the elections in Table \ref{tab:v48pb4}, the vote share is 48\% and the Partisan Bias is 0.4, indicating that each is a very strong gerrymander in favor of the party in question, even though the seat shares vary wildly (from 10\% to 50\% to 90\%).  The Partisan Bias is indicating that, for each of the elections in Table \ref{tab:v48pb4} (which have the party in question winning 1, 5, and 9 seats out of 10 respectively), an extremely large number of districts is won.  Indeed, since $\text{PB} = 0.4$ for each of those elections, the Partisan Bias is indicating that 1, 5, and 9 seats won out of 10 are \emph{equally extremely large}.

The last observation we make in this section is that for a fixed $S$ in the feasible region with $\frac{1}{2} \leq S \leq 1$, $\frac{1}{2} \leq V \leq 1$, the maximum possible value for $PB$ doesn't depend on $V$ at all, as it is $S-\frac{1}{2}$.  For example, consider the election data in Table \ref{tab:SPBfixed}:
\begin{table}[h!]
    \centering
    \begin{tabular}{|c|c|c|c|c|c|c|c|c|c|c|}
        \hline
        District & 1 & 2 & 3 & 4 & 5 & 6 & 7 & 8 & 9 & 10\\ \hline\hline
        \multirow{4}{*}{Vote share} & \multicolumn{10}{c|}{Election 1} \\ \cline{2-11}
        & 49\% & 49\% & 49\% & 49\% & 51\% & 51\% & 51\% & 51\% & 51\% & 51\% \\ \cline{2-11}
        & \multicolumn{10}{c|}{Election 2} \\ \cline{2-11}
        & 49\% & 49\% & 49\% & 49\% & 95\% & 95\% & 95\% & 95\% & 95\% &95\% \\ \hline
    \end{tabular}
    \caption{Two different elections with $S = 60\%$ and $\text{PB} = 0.1$.  Election 1 has vote share 50.2\%, while Election 2 has vote share 76.6\%.}
    \label{tab:SPBfixed}
\end{table}

We see that, according to the Partisan Bias, these elections are both biased towards the party in question.  Moreover, since PB is equal for these two elections, they are \emph{equally} biased towards the party receiving 6 seats in both elections, even though in Election 1 the party receives less than proportionally many seats and in Election 2 the party receives more than proportionally many seats.

\subsubsection{Mean-Median Difference: Decoupling of Extreme Values with Extreme Number of Districts Won}\label{sec:MMability}

For a fixed vote share $V$ with $0.5 < V < 0.75$, the Mean-Median Difference does not in general allow for more extreme values when the seat share $S$ becomes small.  Fix a vote share $0.5 < V < 0.75$ for party $A$.  From Theorem \ref{thm:MM} (and Figure \ref{fig:MM_min_max}), we can see that, as party $A$'s number of seats won decreases, the minimum value that the Mean-Median Difference can take on does not change until $S$ crosses from just above 0.5 to just below.  Once the seat share $S$ crosses below 0.5, the minimum value that the Mean-Median Difference can take on \emph{increases}.  That is, for a fixed $V$ with $0.5 < V < 0.75$ for party $A$ and $S <0.5$, if party $A$ wins fewer seats, the Mean-Median Difference is \emph{less} able to detect that party $A$ is winning extremely few districts\footnote{By symmetry, if $0.25 < V < 0.5$, if party $A$ wins more seats, the Mean-Median Difference is less able to detect that party $A$ is winning an extreme number of districts.}.  

Additionally, for a fixed vote share $V$ with $0.5 < V < 0.75$, the Mean-Median Difference does not in general allow for more extreme values when the seat share $S$ becomes large.  If we fix $V$ with $0.5 < V < 0.75$ and now consider the maximum values that the Mean-Median difference can take on, from Theorem \ref{thm:MM} and Figure \ref{fig:MM_min_max} we can see that as party $A$ wins more seats, this maximum value increases to a point and then decreases.  In other words, as Party $A$ wins more seats, the Mean-Median Difference can detect that party $A$ is winning a more extreme number of seats . . . until party $A$ wins even more seats, at which point the Mean-Median Difference is \emph{less} able to detect that party $A$ is winning an extreme number of seats\footnote{By symmetry, if $0.25 < V < 0.5$, as Party $A$ wins fewer seats, the Mean-Median Difference can detect that party $A$ is winning extremely few seats . . . until party $A$ wins even fewer seats, at which point the Mean-Median Difference is \emph{less} able to detect that party $A$ is winning extremely few seats.}.  

Finally, for a fixed vote share $V$ with $V > 0.75$, the Mean-Median Difference does not in general allow for more extreme values when the seat share $S$ becomes small.  Indeed, suppose party $A$ has a fixed vote share $V > 0.75$.  Then from Theorem \ref{thm:MM} (and Figure \ref{fig:MM_min_max}), as party $A$'s number of seats won decreases, the minimum value that the Mean-Median Difference can take on \emph{increases}.  In other words, as party $A$ wins fewer seats, the Mean-Median Difference is less able to show that party $A$ is winning an extremely small number of districts.  In fact, once the seat share decreases to small enough values, the Mean-Median Difference is \emph{only} able to take on values indicating a gerrymander \emph{for} party $A$ (this can be seen in the ``red crescent'' in Figure \ref{fig:MM_min_max})\footnote{By symmetry, for a fixed vote share $V < 0.25$, once the seat share increases large enough, the Mean-Median Difference is only able to take on values indicating a gerrymander against party $A$.}.

 Similarly, for a fixed seat share $S$, the Mean-Median Difference does not allow for more extreme values indicating a less favorable position for party $A$ as party $A$'s vote share $V$ increases, nor does it allow for more extreme values indicating a more favorable position for party $A$ as $V$ decreases. For example, consider a fixed $S > 0.5$.  From Theorem \ref{thm:MM} (and Figure \ref{fig:MM_min_max}), we see that as party $A$'s vote share $V$ increases, the Mean-Median Difference's minimum value decreases, but only to a point.  Then the Mean-Median Difference's minimum value increases.  That is, for large enough vote share $V$, if the seat share stays the same but the vote share increases, the Mean-Median Difference is \emph{less} able to detect that party $A$ is winning extremely few districts.  If the vote share increases large enough, the Mean-Median Difference can only indicate a gerrymander \emph{favoring} party $A$ (as indicated by the red crescent in Figure \ref{fig:MM_min_max})\footnote{By symmetry, if $S<0.5$, for small enough vote share, if the seat share stays the same but the vote share decreases, the Mean-Median Difference is less able to detect that party $A$ is winning an extreme number of districts. 
 And if the vote share decreases to be sufficiently small, the Mean-Median Difference can only indicate a gerrymander \emph{against} party $A$ (as indicated by the blue crescent in Figure \ref{fig:MM_min_max}).}.If we again fix $S > 0.5$,  as $V$ decreases, the Mean-Median Difference's maximum value increases to a point, but then decreases.  That is, even though party $A$ is winning the same number of seats, as the vote share decreases, the Mean-Median Difference is eventually \emph{less able} to show that party $A$ is winning an extreme number of districts\footnote{By symmetry, for $S < 0.5$, as the vote share increases, the Mean-Median Difference is eventually less able to show that party $A$ is winning extremely few districts.}.

\subsubsection{Mean-Median Difference: Elections with Values Counter to Public Expectations}\label{sec:discussionMM}

In Section \ref{sec:MotivatingExamples}, we mentioned many instances in which the Mean-Median Difference has been used (or its usage has been promoted) in an effort to avoid maps that allow for one party to win an extreme number of districts.  In this section, we give sample elections along with their Mean-Median Difference calculations.  We believe these give convincing evidence that the Mean-Median Difference does not always act in the ways suggested by many of the testimonies, apps, reports, and policies mentioned in Section \ref{sec:MotivatingExamples}.

From Figure \ref{fig:MM_min_max}, similar to the observations we saw in Figure \ref{fig:PB_min_max} in Section \ref{sec:discussionPB}, we can see that the same coordinate position $(V, S)$ can have extremely different hues in the two images.  This corresponds to the fact that, for fixed $V$ and $S$ values, the Mean-Median Difference can vary wildly.  For example, consider the election data in Table \ref{tab:MMv70s90}

\begin{table}[h!]
    \centering
    \begin{tabular}{|c|c|c|c|c|c|c|c|c|c|c|}
        \hline
        District & 1 & 2 & 3 & 4 & 5 & 6 & 7 & 8 & 9 & 10\\ \hline\hline
        \multirow{4}{*}{Vote share} & \multicolumn{10}{c|}{Election 1} \\ \cline{2-11}
                & 1\% & 51\% & 51\% & 51\% & 91\% & 91\% & 91\% & 91\% & 91\% & 91\% \\ \cline{2-11}
        & \multicolumn{10}{c|}{Election 2} \\ \cline{2-11}
                & 49\% & 59\% & 59\% & 59\% & 59\% & 59\% & 59\% & 99\% & 99\% & 99\% \\ \hline
    \end{tabular}
    \caption{Two different elections with $S = 90\%$ and $V = 70\%$.  Election 1 has $\text{MM} = 0.21$, while Election 2 has $\text{MM} = -0.11$.}
    \label{tab:MMv70s90}
\end{table}

Although both Elections in Table \ref{tab:MMv70s90} have a vote share of 70\% and seat share of 90\%, the Mean-Median Difference considers Election 1 to be an extreme gerrymander in favor of the party in question, while Election 2 is considered to be an extreme gerrymander \emph{against} the party in question.  The different values here indicate that the Mean-Median Difference believes that 9 out of 10 seats is extremely large in the first election, but not nearly large enough in the second.  

Additionally, note in Figure \ref{fig:MM_min_max} that for $V>\frac{1}{2}$, MM has the same minimal value, regardless of $S$.  For example, consider the elections in Table \ref{tab:v60mmm9}.

\begin{table}[h!]
    \centering
    \begin{tabular}{|c|c|c|c|c|c|c|c|c|c|c|c|}
        \hline
        District & 1 & 2 & 3 & 4 & 5 & 6 & 7 & 8 & 9 & 10 & 11\\ \hline\hline
        \multirow{6}{*}{Vote share} & \multicolumn{11}{c|}{Election 1} \\ \cline{2-12}
        & 47\% & 47\% & 47\% & 47\% & 47\% & 51\% & 74.8\% & 74.8\% & 74.8\% &74.8\% & 74.8\% \\ \cline{2-12}
        & \multicolumn{11}{c|}{Election 2} \\ \cline{2-12}
        & 47\% & 51\% & 51\% & 51\% & 51\% & 51\% & 71.6\% & 71.6\% &71.6\% & 71.6\% & 71.6\%\\ \hline
    \end{tabular}
    \caption{Two different elections with $V = 60\%$ and $\text{MM} = -0.09$.  Election 1 has $S = \frac{6}{11}\approx 54.5\%$, and Election 2 has $S = \frac{10}{11}\approx 90.9\%$.}
    \label{tab:v60mmm9}
\end{table}

In each of the elections in Table \ref{tab:v60mmm9}, the vote share is 60\% and the Mean-Median Difference is -0.09, indicating that each is a gerrymander against the party in question, even though the seat shares vary wildly (from 54.5\% to 90.9\%).  The Mean-Median Difference is indicating that, for each of the elections in Table \ref{tab:v60mmm9} (which have the party in question winning 6 and 10 seats out of 11 respectively), an extremely small number of districts is won.  Indeed, since $\text{MM} = -0.09$ for each of those elections, the Partisan Bias is indicating that 6 and 10 seats won out of 11 are \emph{equally extremely small}.

The last observation we make in this section is that for a fixed $V$ in the feasible region with $\frac{1}{2} \leq S \leq 1$, $\frac{1}{2} \leq V \leq 1$,  the largest maximum values occur near the line $S = 2V-\frac{1}{2}$.  For example, consider the election data in Table \ref{tab:MMnearline}:
\begin{table}[h!]
    \centering
    \begin{tabular}{|c|c|c|c|c|c|c|c|c|c|c|}
        \hline
        District & 1 & 2 & 3 & 4 & 5 & 6 & 7 & 8 & 9 & 10\\ \hline\hline
        \multirow{4}{*}{Vote share} & \multicolumn{10}{c|}{Election 1} \\ \cline{2-11}
        & 1\% & 1\% & 1\% & 1\% & 82.67\% & 82.67\% & 82.67\% & 82.67\% & 82.67\% & 82.67\% \\ \cline{2-11}
        & \multicolumn{10}{c|}{Election 2} \\ \cline{2-11}
        & 1\% & 55.44\% &  55.44\% &  55.44\% &  55.44\% & 55.44\% &  55.44\% & 55.44\% &  55.44\% & 55.44\% \\ \hline
    \end{tabular}
    \caption{Two different elections with $V = 50\%$.  Election 1 has seat share 60\% and $\text{MM} \approx 0.33$, while Election 2 has seat share 90\% and $\text{MM} \approx0.05$.}
    \label{tab:MMnearline}
\end{table}

In both of these elections the vote share is 50\%.  In Election 1, the seat share is lower (6 seats out of 10) and the Mean-Median Difference is approximately 0.33, while in Election 2 the seat share is higher (9 seats out of 10) and the Mean-Median Difference is approximately 0.05.  Hence, according to the Mean-Median Difference, these elections are both biased towards the party in question.  However, since MM is larger  for Election 1 (when 6 seats are won), 6 out of 10 seats won is considered a more extreme gerrymander for the party in question than in Election 2 (when 9 out of 10 seats are won).

\subsection{The Special Case of Metric Value 0} \label{sec:bounds_metrics0}

From Theorems \ref{thm:PB} and \ref{thm:MM}, we deduce that when $S \geq \frac{1}{2}$, both the Mean-Median Difference and Partisan Bias cannot be zero for $V < \frac{1}{2}$, and that both the Mean-Median Difference and Partisan Bias can be zero for any $V$ with $\frac{1}{2} \leq V < \frac{3}{4}$.  We can deduce the remainder of Corollary \ref{cor:pb_equal_turnout} by setting the lower bound equal to 0 in the region $\frac{3}{4} \leq V \leq 1$.

\begin{Corollary}\label{cor:pb_equal_turnout}
Suppose $V$ is a rational number with $\frac{1}{2} \leq V  < \frac{3}{4}$ and  $S$ is any rational number with $\frac{1}{2} \leq S \leq 1$.  Alternatively, suppose that $V$ is a rational number with $\frac{3}{4} \leq V \leq 1$ and $S$ is any rational number with $ \frac{3V-2}{2V-1} < S \leq 1$.  Then one can construct election data with seat share $S$, vote share $V$, turnout in each district is equal, and $PB = 0$.  Additionally, one can construct election data with seat share $S$, vote share $V$, turnout in each district is equal, and $MM = 0$.
\end{Corollary}

By symmetry, we have all $(V,S)$ pairs for which there exists constructed election data with vote share $V$, seat share $S$, and Partisan Bias 0.  This result shows that, when turnout in each district is equal, a very wide array of $(V,S)$ pairs can result in $MM=0$ and $PB=0$.  The region consisting of all $(V,S)$ for which there is constructed election data with $MM=0$ and $PB=0$ can be seen in Figure \ref{fig:SV_mm_pb}. 

\begin{figure}[h]
\centering
\includegraphics[width=2in]{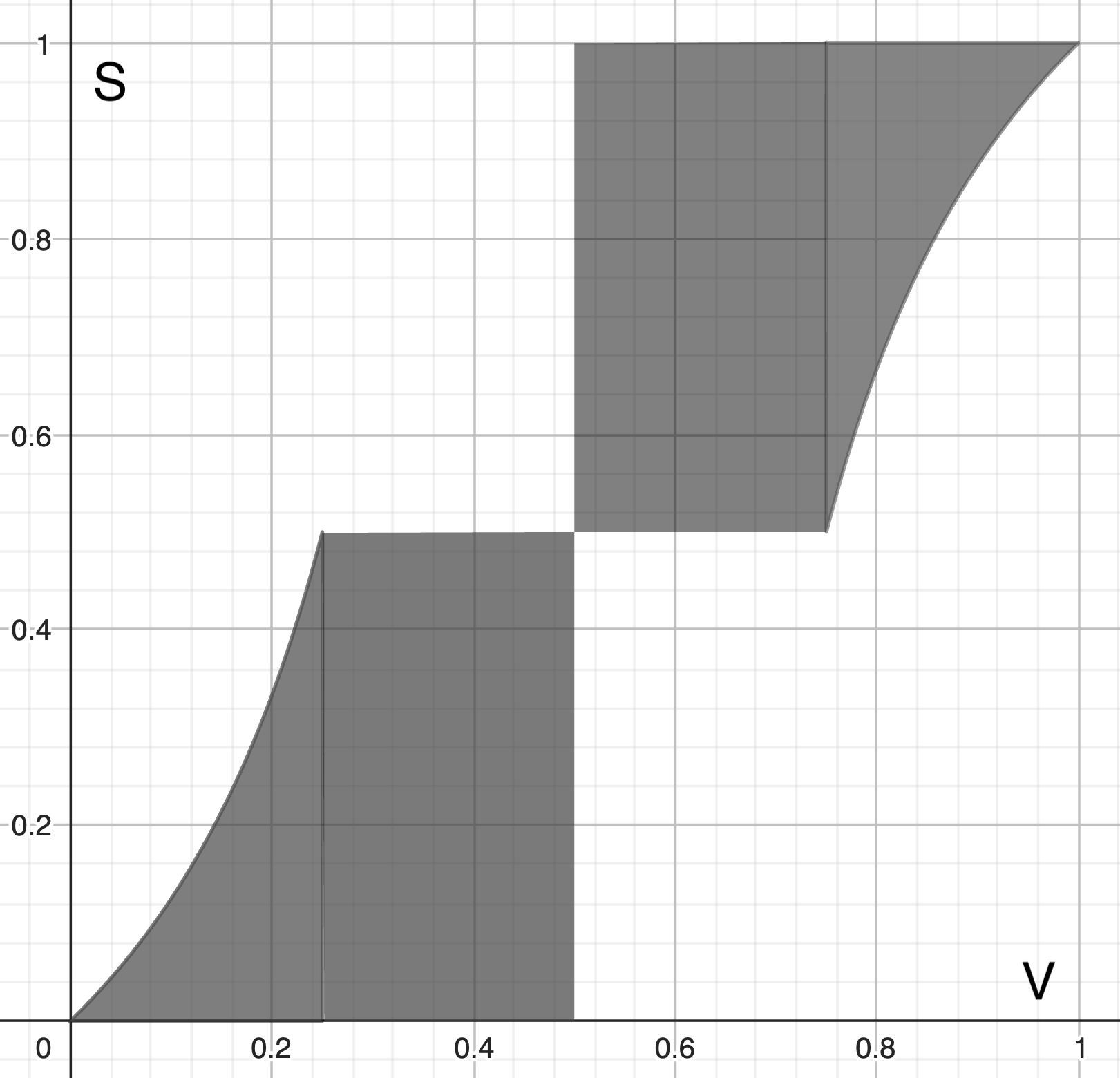}
\caption{Region of all $(V,S)$ for which there is constructed election data with turnout in all districts equal and $MM=0$ (which is the same as the region of all $(V,S)$ for which there is constructed election data with turnout in all districts equal and $PB=0$).}
\label{fig:SV_mm_pb}
\end{figure}

This is different from the set of $(V,S)$ for which there is constructed election data with turnout equal in all districts and $\text{Declination} = 0$ \cite{DeclinationAsMetric} (see Figure \ref{fig:SV_dec}), and also quite different from  the set of $(V,S)$ for which there is constructed election data with turnout equal in all districts and $EG = 0$ \cite{2018arXiv180105301V} (see Figure \ref{fig:SV_eg}).  From these Figures, we can deduce that the Partisan Bias and Mean-Median Difference are the least restrictive metrics in terms of which $(V,S)$ pairs are allowed to potentially have a metric value which indicates that no gerrymandering has occurred.

\begin{figure}[h]
\centering
\includegraphics[width=2in]{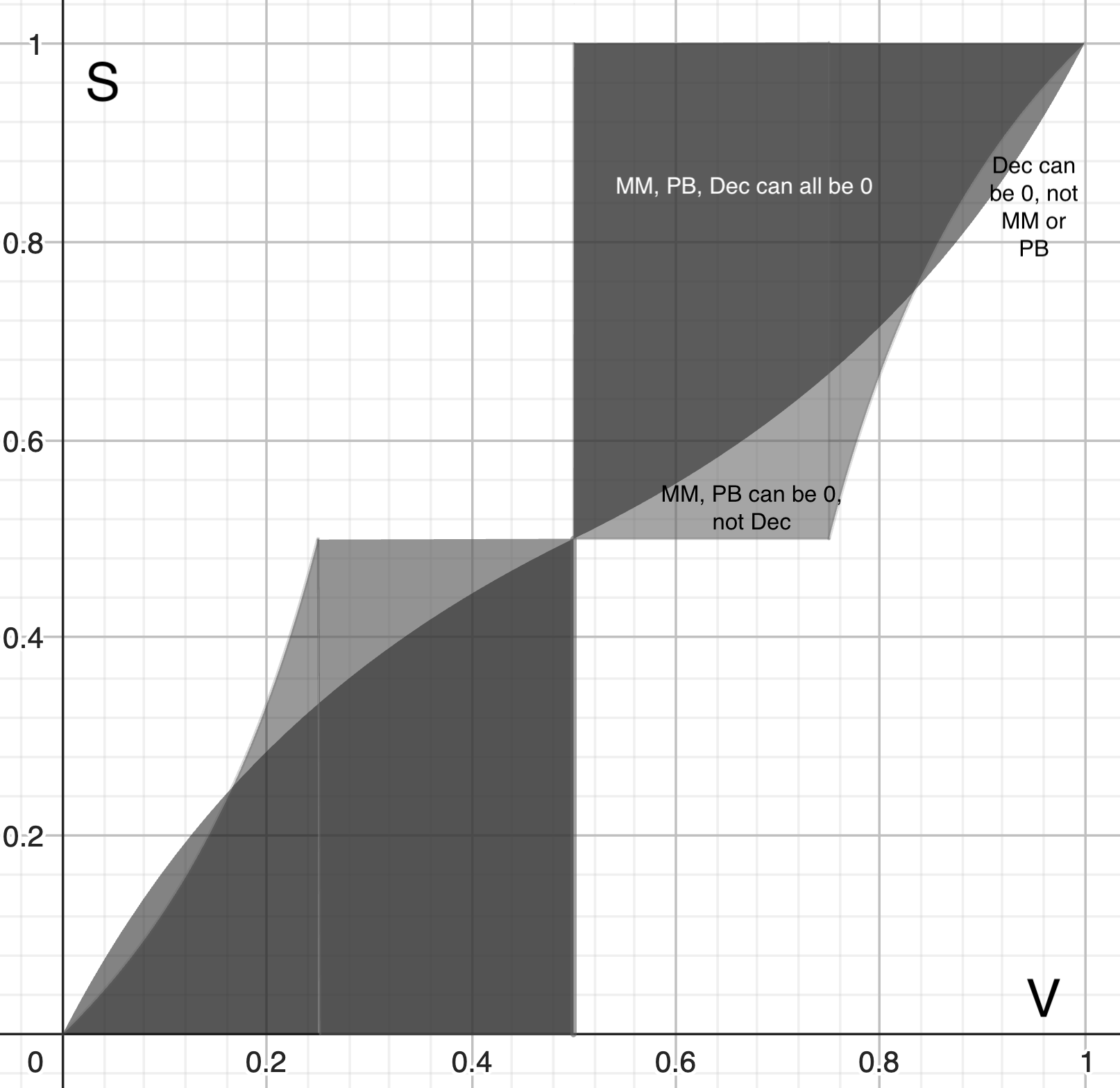}
\caption{Region of all $(V,S)$ for which there is constructed election data with turnout in all districts equal and $\text{Declination}=0$ (dark grey), overlayed with the region of all $(V,S)$ for which there is constructed election data with turnout in all districts equal and $PB/MM=0$ (light grey).}
\label{fig:SV_dec}
\end{figure}

\begin{figure}[h]
\centering
\includegraphics[width=2in]{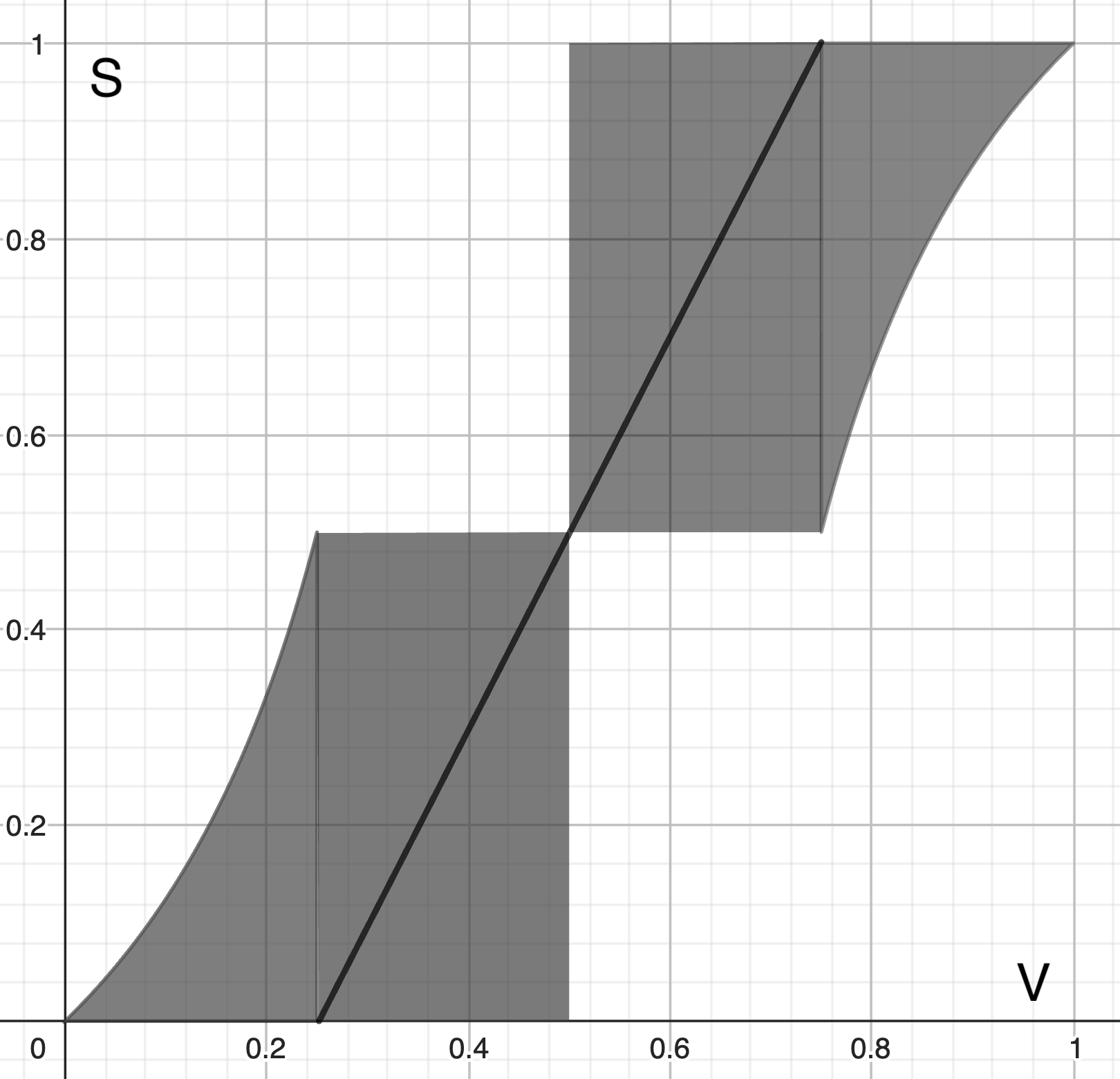}
\caption{Region of all $(V,S)$ for which there is constructed election data with turnout in all districts equal and $EG=0$ (dark grey), overlayed with the region of all $(V,S)$ for which there is election data with turnout in all districts equal and $PB/MM=0$ (light grey).}
\label{fig:SV_eg}
\end{figure}

Corollary \ref{cor:pb_equal_turnout} assumes that turnout in each district is equal, which is never the case in real elections.  This raises the question: if turnout is unequal, say the maximum turnout divided by minimum turnout is $C$, what $(V,S)$ pairs allow for Partisan Bias and Mean-Median Difference 0?  The following result answers this question.

\begin{Theorem}\label{thm:turnout_different}
Fix the maximum turnout to minimum turnout ratio at $C\geq 1$.  Suppose that  $S \geq \frac{1}{2}$ for rational seat share $S$.  Then for any rational vote share $V^*$ with 
\begin{equation*}
\frac{1}{2(S+C(1-S))} \leq V^* < \frac{1+C(3-2S)}{(C+1)(3-2S)}
\end{equation*}
there is constructed election data with vote share $V^*$, seat share $S$,  and $PB= 0$.  Additionally, for any such $V^*$  there is constructed election data with vote share $V^*$, seat share $S$,  with $MM=0$.
\end{Theorem}

Theorem \ref{thm:turnout_different} gives the lower bound of $\frac{1}{2}$  when $C= 1$ (when turnout in all districts is equal), which is what we already know to be true from Corollary \ref{cor:pb_equal_turnout}.  For the upper bound, when $C = 1$ we again get the result of Corollary \ref{cor:pb_equal_turnout}.  Specifically, setting $C=1$, Theorem \ref{thm:turnout_different} gives $V^* \leq \frac{2-S}{3-2S}$, which is the result of solving the equation $S =\frac{3V^*-2}{2V^*-1}$ (from Corollary \ref{cor:pb_equal_turnout}) for $V^*$.   When turnout is unequal, Theorem \ref{thm:turnout_different} expands the pairs $(V,S)$ for which there is election data with vote share $V$, seat share $S$, and Partisan Bias/Mean-Median Difference values of 0.  For example, Figure \ref{fig:SV_pb_mm_turnout_unequal} shows how this region expands for different values of $C$.  As $C \to \infty$, the region expands to fill the entire $[0,1] \times [0,1]$ square!  That is, for \emph{any} $(V, S)$ with $0<V<1$ and $0<S<1$, there exists constructed election data\footnote{with potentially extremely varied turnout between districts} with $MM=0$, and the same can be said for $PB$.

 The authors of \cite{2018arXiv180105301V} and \cite{DeclinationAsMetric} calculated similar regions for the Efficiency Gap and the Declination.  For both of those metrics, the set of $(V, S)$ pairs for which the Efficiency Gap (Declination) can be 0 does \emph{not} fill the entire $[0,1] \times [0,1]$ square, even when there are no restrictions on turnout.  Again, this indicates that the Partisan Bias and Mean-Median Difference are the least restrictive metricss in terms of which $(V,S)$ pairs are allowed to potentially have a metric value which indicates that no gerrymandering has occurred, including when there are no restrictions on turnout.

\begin{figure}[h]
\centering
\includegraphics[width=2in]{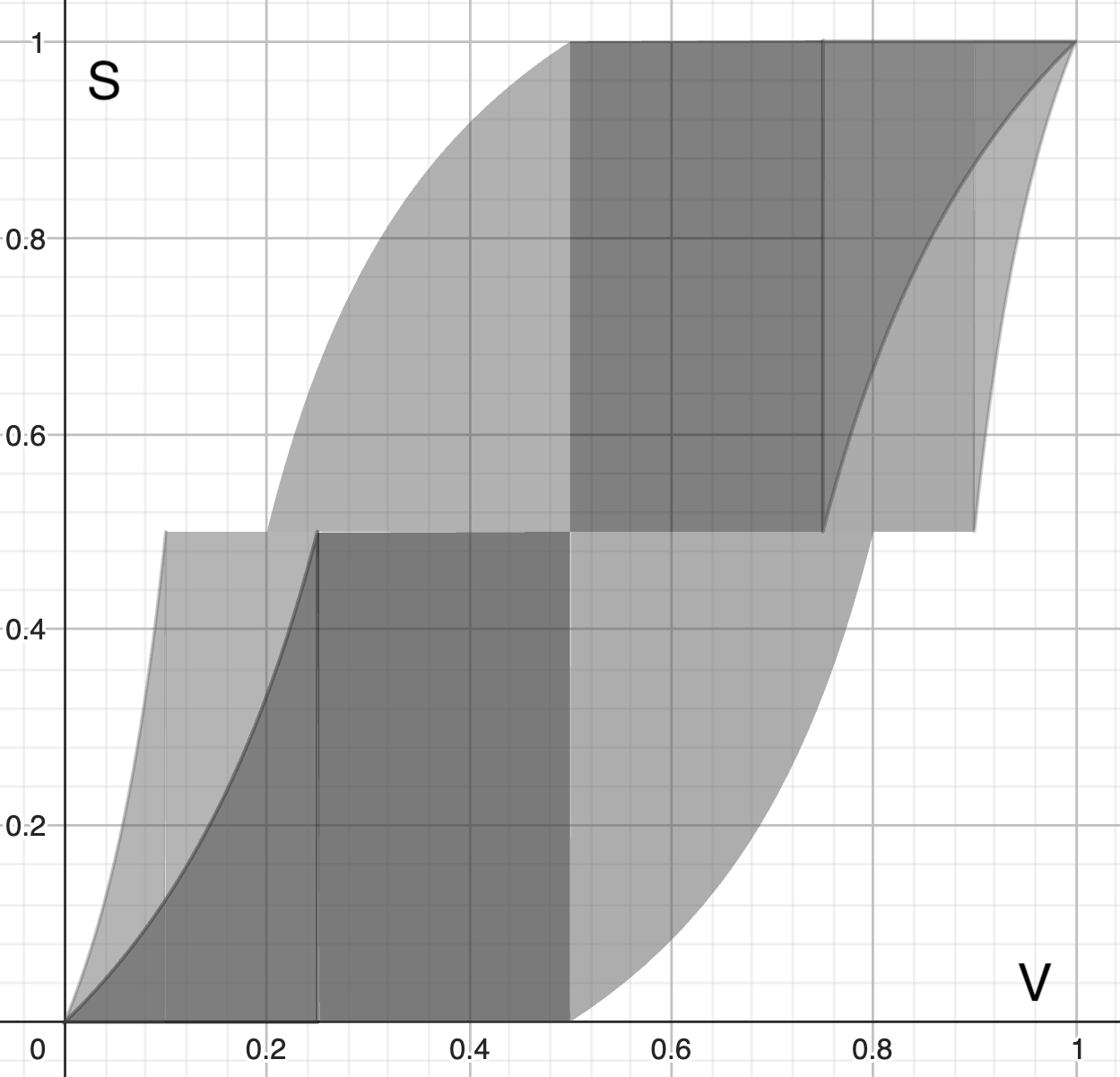} \hspace{0.5 cm}
\includegraphics[width=2in]{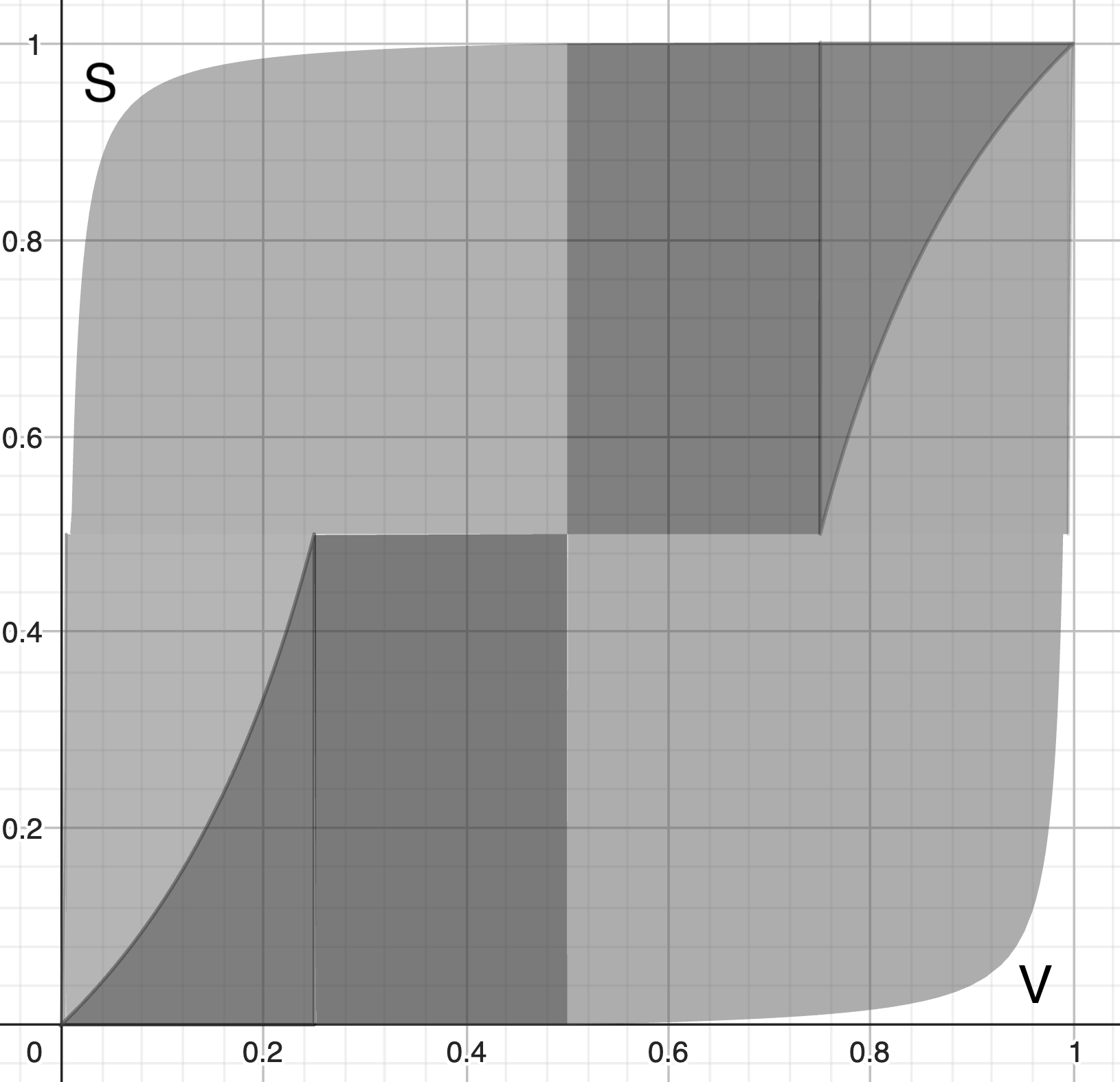}
\caption{Region of all $(V,S)$ for which there is constructed election data with $PB/MM=0$.  The dark grey portion corresponds to turnout ratio $C = 1$ and the light grey portion corresponds to turnout ratio $C = 4$ in the left image, and $C = 100$ in the right image. }
\label{fig:SV_pb_mm_turnout_unequal}
\end{figure}

We note that in \cite{2018arXiv180105301V}, the author evaluated election turnout ratios for all states with at least eight congressional districts, using data from 2014 and 2016, and showed that the ratio of maximum turnout to minimum turnout can be as high as 4.4 (in California).  The 2022 general congressional district elections in California had a turnout ratio of 3.29, suggesting that high turnout ratios remain something that is seen in practice, even in general elections.  

Theorem \ref{thm:turnout_different} is also proven in  Appendix \ref{sec:proofs}.

\section{Bugs}\label{sec:bugs}

In Section \ref{sec:bounds}, we stated Theorems (proved in Appendix \ref{sec:proofs}) that gave bounds on the values that the Mean-Median Difference and Partisan Bias could take on for each possible vote-share, seat-share pair $(V, S)$.  We showed how one could use those results to see that the MM and PB do not allow for more extreme metric values when a more extreme number of districts is won.  We also constructed election examples in which the metric values run counter to how these metrics are expected (and advertised) to act.  Of course, the space of constructible examples may be quite different from the space of implemented redistricting maps.  Here we give clear evidence that, on real maps with real data, the Mean-Median Difference and Partisan Bias cannot detect maps with an extreme number of districts won.

In order to test whether the Mean-Median Difference and the Partisan Bias have extreme values on maps with an extreme number of districts won, we need to produce maps with an extreme number of districts won.  We do this with a recently developed hill-climbing method called the ``Short Burst Method,'' developed by Cannon et al.\ \cite{cannonShortBursts}.  This hill-climbing process uses GerryChain, which is a Python library that implements a well-researched and widely accepted Markov Chain Monte Carlo (MCMC) process to create an ensemble of potential redistricting maps \cite{RecomMGGG, GerryDetails}.  In \cite{cannonShortBursts}, Cannon et al.\ showed that this hill-climbing method worked better than a biased random walk\footnote{A biased random walk is the most common method of using a Markov Chain to find local or global maxima; in this case, we are searching for redistricting maps with a near-maximum number of districts won by some party.} at finding maps with an extreme number of majority-minority districts.  Our application is a direct parallel; while Cannon et al.\ were looking for maps with an extreme number of districts which were majority-minority, here we are looking for maps with an extreme number of districts which are majority-Democrat or majority-Republican.

\subsection{Methods}\label{sec:bugs_methods}

As stated above, we use the Short Burst Method to construct maps with an extreme number of districts won by a single party.  The Short Burst Method takes a small number of steps in a Markov Chain (say 10 steps).  At the end of that ``short burst'', the map with the highest number of districts won (for the party in question) is then the seed for the next short burst of size 10.  (If there is more than one map with the same highest number of districts won among those 10, the last such produced map is the seed for the next short burst.)  This process is repeated thousands of times in an effort to find the maximum number of majority-minority districts achievable.  Note that the Short Burst method can be thought of as a non-deterministic version of beam search on the metagraph\footnote{The ``metagraph'' here is the graph whose nodes are maps.  Two maps are adjacent if one can be obtained by the other through a single step in the Markov Chain constructing the ensemble of maps.  In this case, that is through a single ReCom move in GerryChain.} of redistricting maps, using depth-first search instead of breadth-first, and in which a random selection of maps near to the current map are explored.  This process is repeated for each state/map under consideration, and each party.  We considered 6 states, 3 maps per state (State House, State Senate, Congressional districting maps), and 2 parties (Democratic and Republican).  This made for a total of 36 runs.  We first describe the states chosen and then describe the parameters for each Short Burst run.

Our choice of which states to evaluate was based on two main factors.  Firstly, we needed to use states for which there was reliable data on which we could run GerryChain's ReCom MCMC process.  The Metric Geometry Gerrymandering group \cite{MGGGstates} has an excellent GitHub page with this data for some states, and we chose our states from among those that they had available, which all used maps from the 2011 redistricting cycle.  Our studies are not intended to focus on any particular map, or on the specific election for that state (which is used as a proxy for party preference), but rather to focus on the metrics themselves.  Thus, we are unconcerned that these states' maps and partisan data are 6-10 years old.

Secondly, we wanted to choose states with different types of political geographies on which to test the Mean-Median Difference and Partisan Bias.  Some states are fairly homogeneous, with most parts of the state having the same partisan breakdown.  Other states have pockets that more heavily lean Republican and other pockets that more heavily lean Democratic; such states may have an overall clear partisan lean (Democratic or Republican) or they may be more ``purple.''   The states that we chose are Massachusetts, Michigan, Oklahoma, Oregon, Pennsylvania,  and Texas, which we believe represents a variety of political geographies\footnote{These are the same states chosen for a similar analysis, in \cite{GameabilityStudy}.}.  We include in Figure \ref{image:choropleths} a choropleth of each state (using the election data that we used for our analysis) for the reader's reference.  

\begin{figure}[h]
\includegraphics[width = 2in]{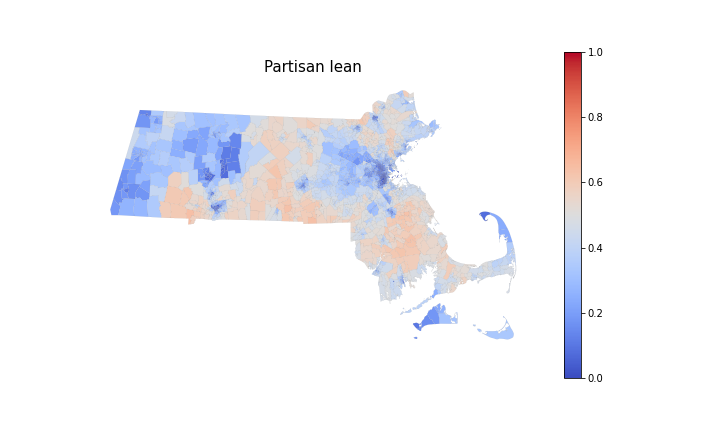}
\includegraphics[width = 2in]{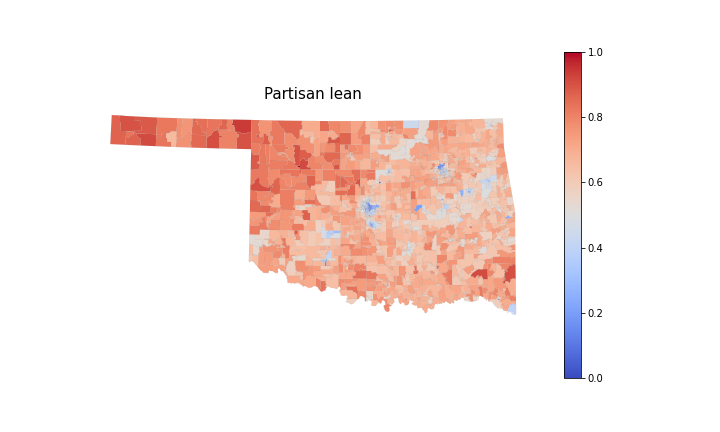}
\includegraphics[width = 2in]{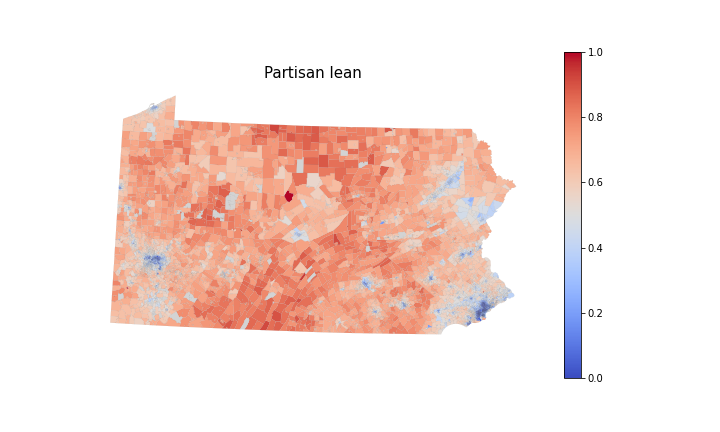}
\includegraphics[width = 2in]{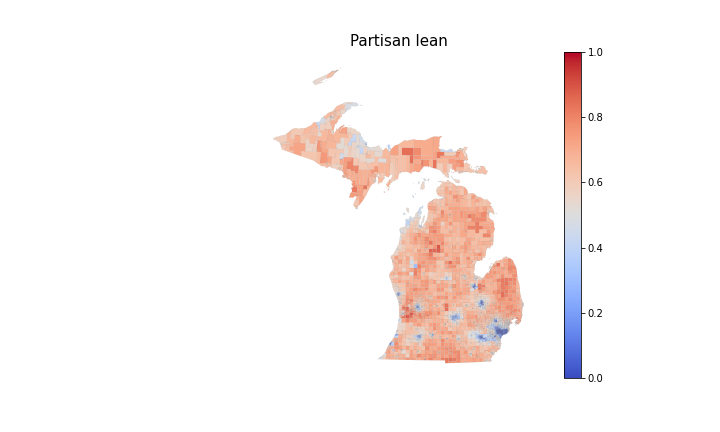}
\includegraphics[width = 2in]{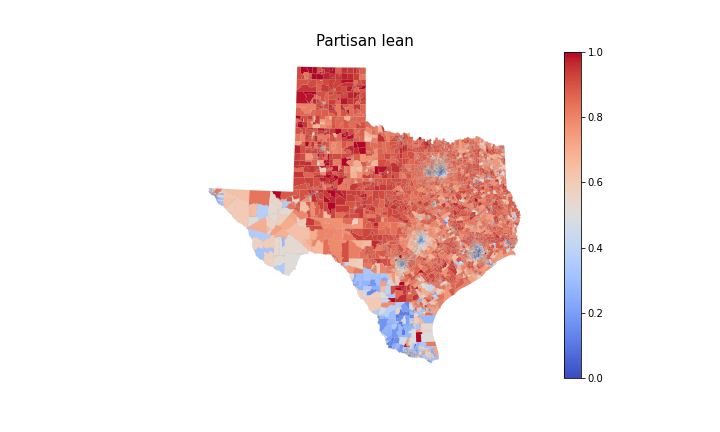}
\includegraphics[width = 2in]{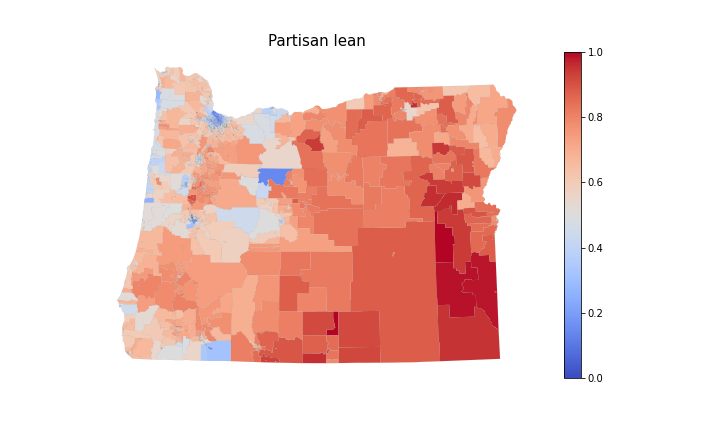}
\caption{Choropleths of the states in our analysis, using the election results from Table \ref{tab:burst_table}.  They are, in clockwise order from the top left, Massachusetts, Oklahoma, Pennsylvania, Oregon, Texas, and Michigan.}
\label{image:choropleths}
\end{figure}

\begin{table}
\begin{tabular}{|p{1.8cm}||p{1.9cm}|p{1.6cm}|p{1.9cm}|p{2cm}|p{1.9cm}|p{1.6cm}|}\hline
State & MA & OK & PA & MI & TX & OR \\ \hline
Election Used & 2018  \newline US Senate & 2018\newline Governor & 2016 \newline US Senate & 2016 \newline Presidential &  2014\newline US Senate & 2018\newline Governor \\ \hline
\end{tabular}
\caption{Elections chosen as a proxy for partisan support for each state.}
  \label{tab:burst_table}
\end{table}

For congressional maps, we allowed up to 5\% deviation from ideal population in each district.  For the lower and upper house maps, we allowed up to 11\% deviation from ideal population.   As noted in \cite{RecomMGGG}, ``Excessively tight requirements for population balance can spike the rejection rate of the Markov chain and impede its efficiency or even disconnect the search space entirely.'' 
 In \cite{RecomMGGG}, they use a deviation of 5\% for their congressional map of Arkansas, which is why we chose 5\% for our congresssional maps here.  State legislatures have significantly looser restrictions on population deviation; Brown v Thomson \cite{BrownvThomson} states that state legislative districts can deviate by up to 10\% in population, and many states do have population deviations of that magnitude.  We allowed a deviation of up to 11\% because this is sufficiently close to 10\% and is needed in order for the Markov Chain process to not become quickly stuck and be unable to produce additional sample maps.  

 As described in \cite{RecomMGGG}, the ReCom process in GerryChain naturally samples plans with more compact districts with higher probability.  Outside of this built-in preference for compact maps, we did not impose additional compactness constraints.

For each state, each party, and each districting map,  we ran 10 short burst samples, each of size 5,000, and each having burst length of size 10.  The burst length of size 10 was theorized to be the ideal burst length in \cite{cannonShortBursts}, which also found that the exact value of the burst length did not drastically affect the results.  The step size of 5,000 was chosen due to the fact that this length corresponded to a drastic slowdown in the increase of the metric value in the study from section 5.2.3 of \cite{cannonShortBursts}, and thus reaches close to the maximum possible number of districts won.    We ran these bursts of 5,000 steps 10 times as was done by Cannon et al.\ in their study showing the efficacy of the Short Burst method \cite{cannonShortBursts}, in order to validate the consistency of the results.  This entire process produced potentially 45,000 different maps\footnote{Since the last maximal map of each burst is the first map in the next burst, but the seed at the beginning of each sample is the same.}
for all of the state, party, and map combinations. These maps provide a set of examples of `locally optimized' plans that can serve as a proxy for maps that are expected to return a large number of seats for a particular party, using the actual political geography of each state. 

We note that the authors of \cite{GameabilityStudy} also used this same kind of Short Burst analysis to test the ``gameability'' of some metrics intended to detect gerrymandering (specifically, the GEO Metric, Mean-Median Difference, Declination, Efficiency Gap).  They define a metric as ``gameable'' if the metric value can stay within acceptable bounds\footnote{For a metric $m$, they considered ``acceptable bounds'' to be $0.16\inf(m) \leq m \leq 0.16 \sup(m)$, based on the bounds recommended by the creators of the Efficiency Gap\cite{PartisanGerrymanderingEfficiencyGap}, and the apples-to-apples conversion of those bounds to the other metrics.}, while allowing a map with the most extreme number of Democratic-won (or Republican-won) districts that can be achieved by that map.  The authors of \cite{GameabilityStudy} allowed us to use the same code to produce the Figures in Section \ref{sec:bugs_results}

\subsection{Results of Empirical Study}\label{sec:bugs_results}

As mentioned in Section \ref{sec:bugs_methods}, we have results for both the Democratic party and the Republican Party for three different maps on each of six states.  Note that, as explained in Section \ref{sec:bugs_methods}, the method we use to produce our maps is a hill-climbing method recently developed by Cannon et al.\ \cite{cannonShortBursts} which utilizes the GerryChain library. 
 Thus, the maps under consideration are not constructed in the same way as those constructed for a typical ensemble analysis (see Section \ref{sec:bugs_methods} for details) and are not intended to provide representative samples from a specified distribution. Instead, they represent a large set of examples that have been optimized for extreme seats outcomes. 
 
 Recall that, if a metric is to be able to detect maps with an extreme number of districts won, then it should have more extreme values on extreme maps.  In other words, if we see the metric values for maps with $k$ districts won by a particular party, that range of metric values should generally go up as $k$ goes up\footnote{Or the metric value should go down as $k$ goes down; the direction depends on whether $V$ and $S$ are measured for the Democratic party of Republican party.  For our calculations, we always let $V$ represent the Democratic party's vote share and $S$ the Democratic party's seat share.  Thus, for the Democratic party we would expect the metric value to go up as $k$ (the number of seats won by Democrats) goes up, while for the Republican party we would expect the metric value to go down as $k$ (the number of seats won by Republicans) goes down.}.

We see sample results in Figures \ref{fig:short_bursts_MAlowerD} and \ref{fig:short_bursts_MAlowerR} for Massachusetts' State House.  Just to compare the Mean-Median Difference and Partisan Bias to other gerrymandering metrics, we include the Efficiency Gap \cite{PartisanGerrymanderingEfficiencyGap} and the GEO Metric ratio\cite{geo_paper}\footnote{Since these metrics are used purely for comparison purposes, we do not define them here.  Their definitions are not our concern; our concern is simply that other metrics intended to detect gerrymandering do have extreme values on more extreme maps, as one would expect.  The Efficiency Gap is defined in \cite{PartisanGerrymanderingEfficiencyGap}, the GEO metric is defined in \cite{geo_paper}, and the GEO metric ratio is defined in \cite{GameabilityStudy} }.  

\begin{figure}[h]
\centering
\includegraphics[width=1.5in]{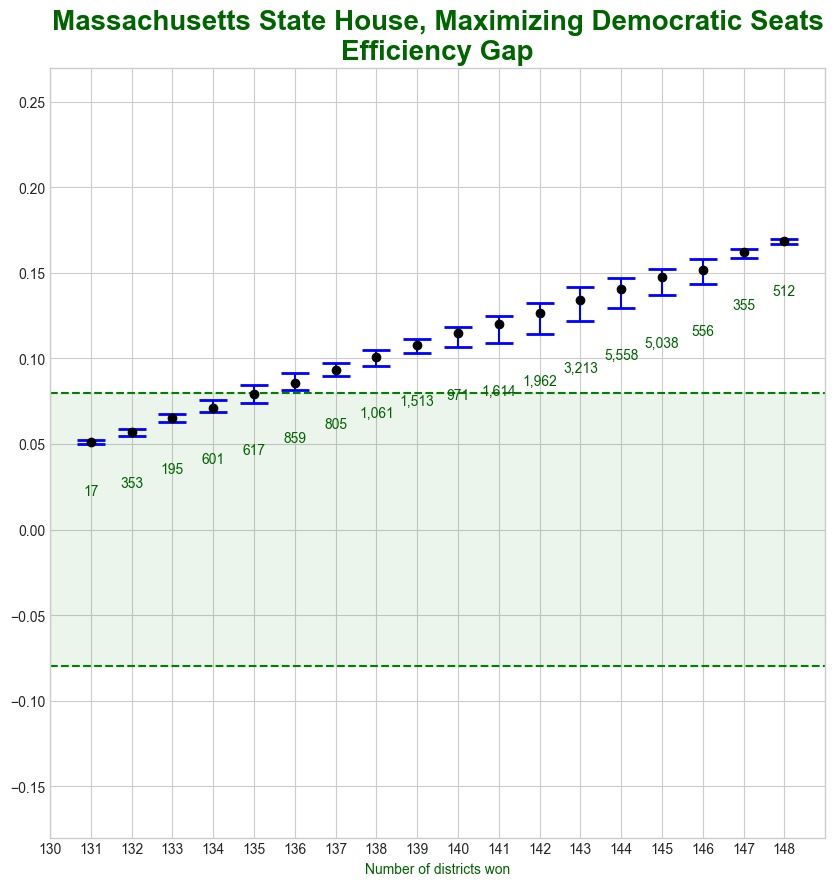}
\includegraphics[width=1.5in]{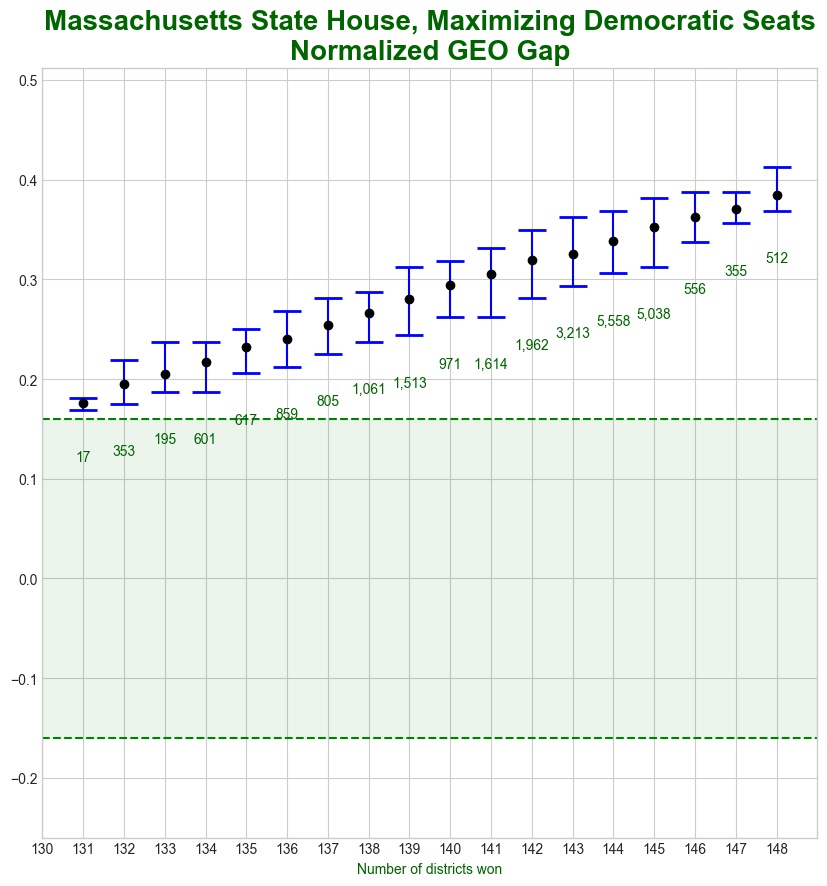}
\includegraphics[width=1.5in]{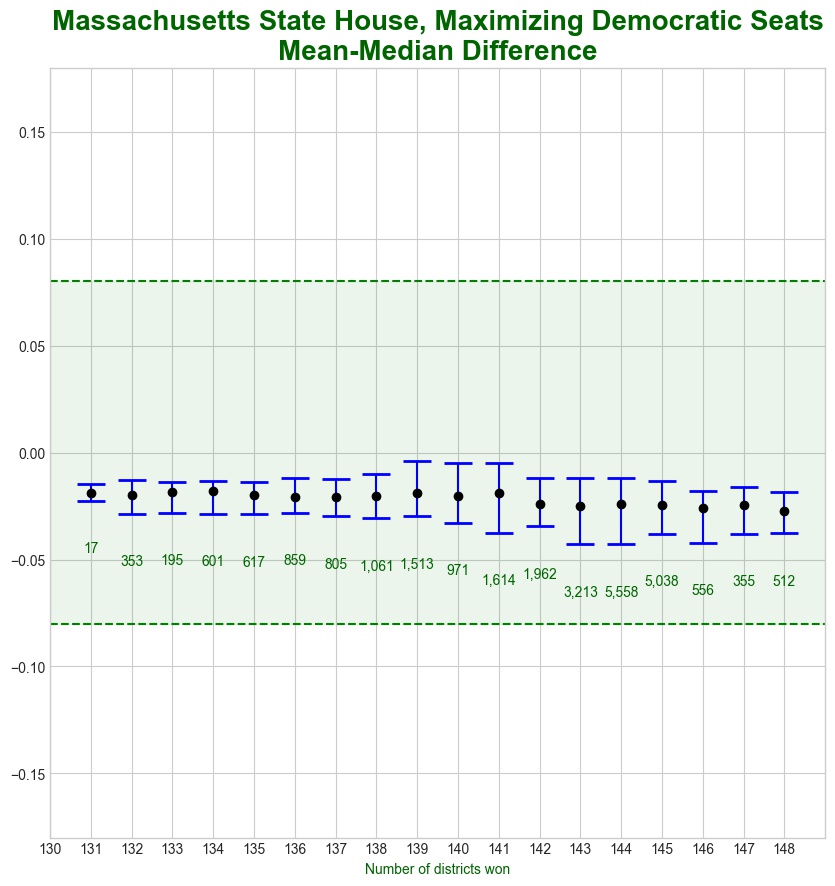}
\includegraphics[width=1.5in]{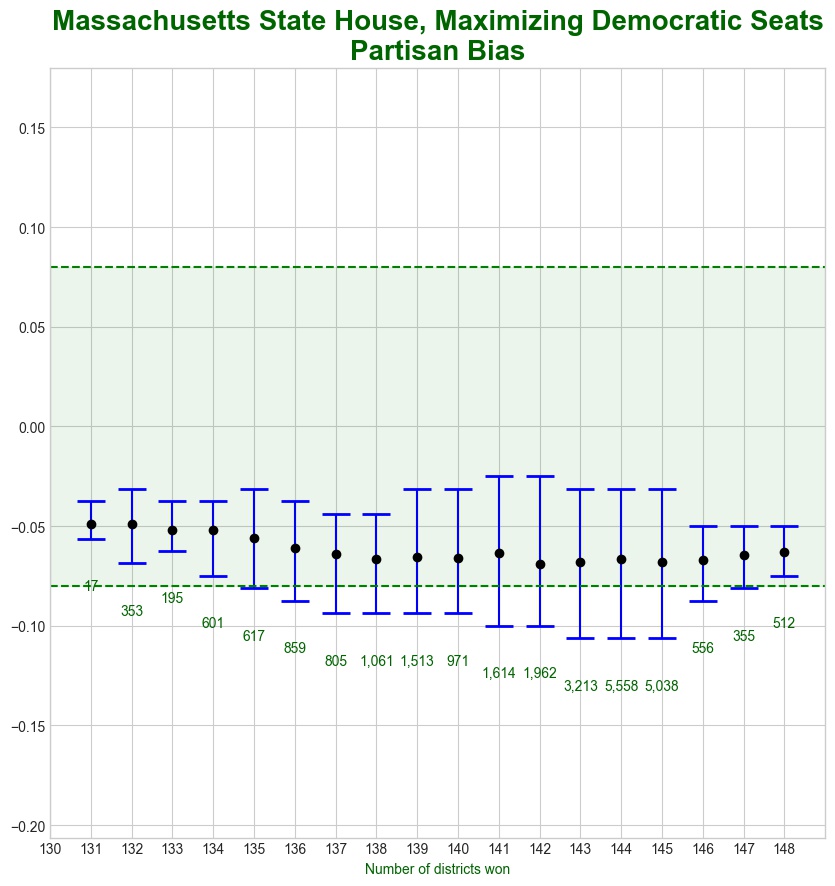}
\caption{Empirical results for Massachusetts's State House map and 2018 US Senate election data, searching for maps with as many Democratic-won districts as possible.  Horizontal axis is number of districts won, vertical axis is metric value ranges.  The green region is from $0.16\inf(m)$ to $0.16\sup(m)$ for each metric $m$.  The small number below each metric value range is the number of maps produced that had the corresponding number of districts won.  The dot within each vertical bar is the mean value of that metric on all produced maps with the corresponding number of districts won.}
\label{fig:short_bursts_MAlowerD}
\end{figure}

\begin{figure}[h]
\centering
\includegraphics[width=1.5in]{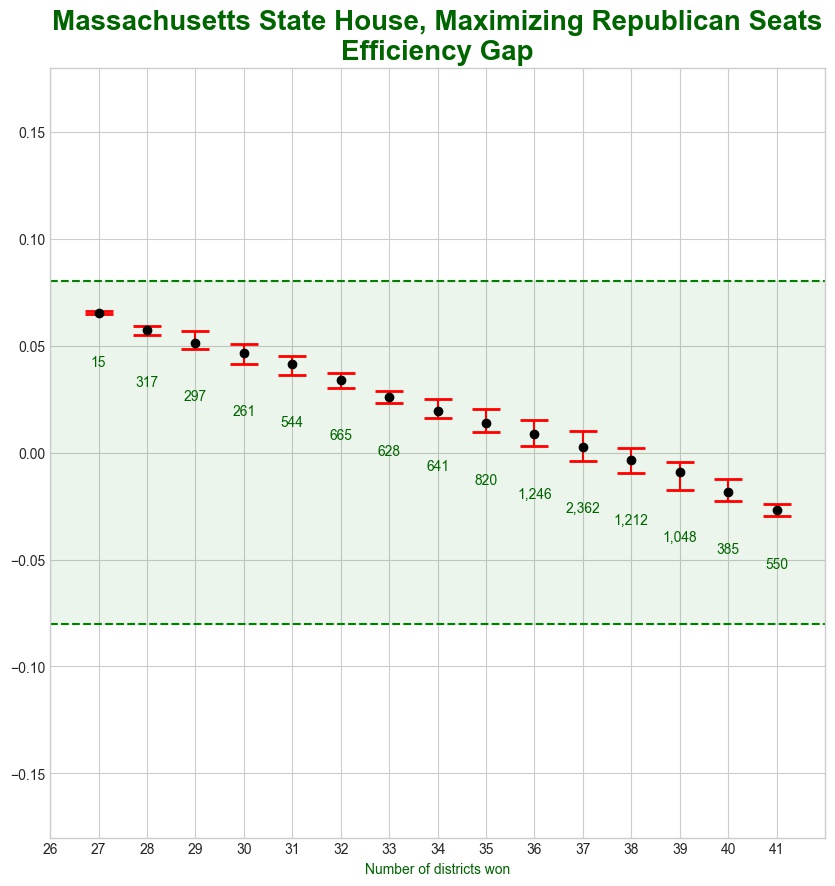}
\includegraphics[width=1.5in]{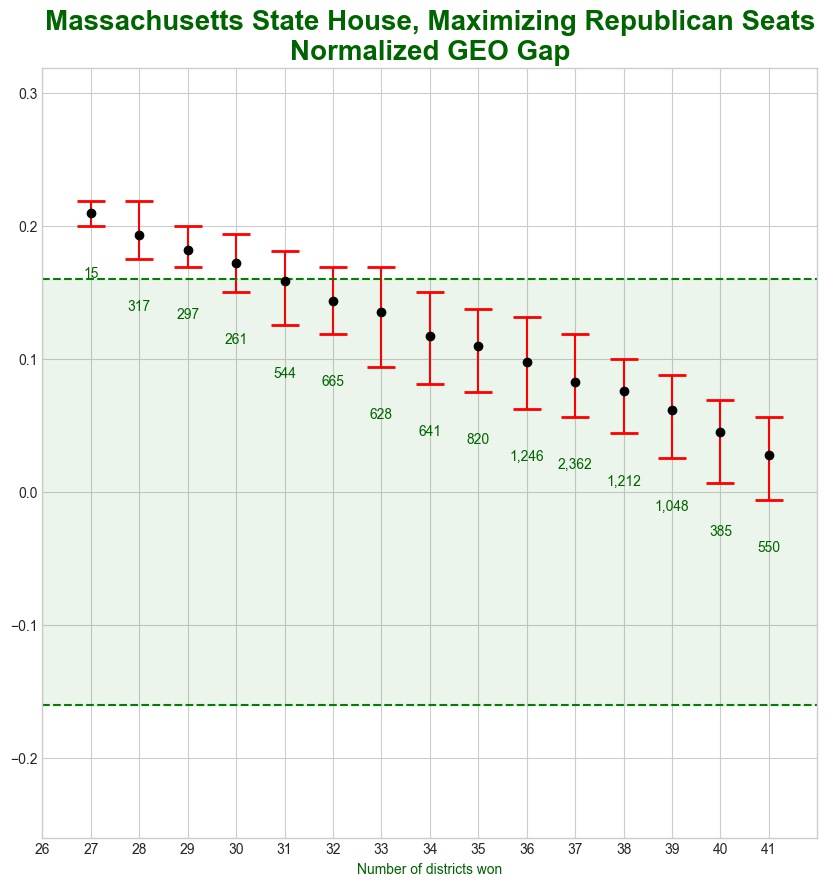}
\includegraphics[width=1.5in]{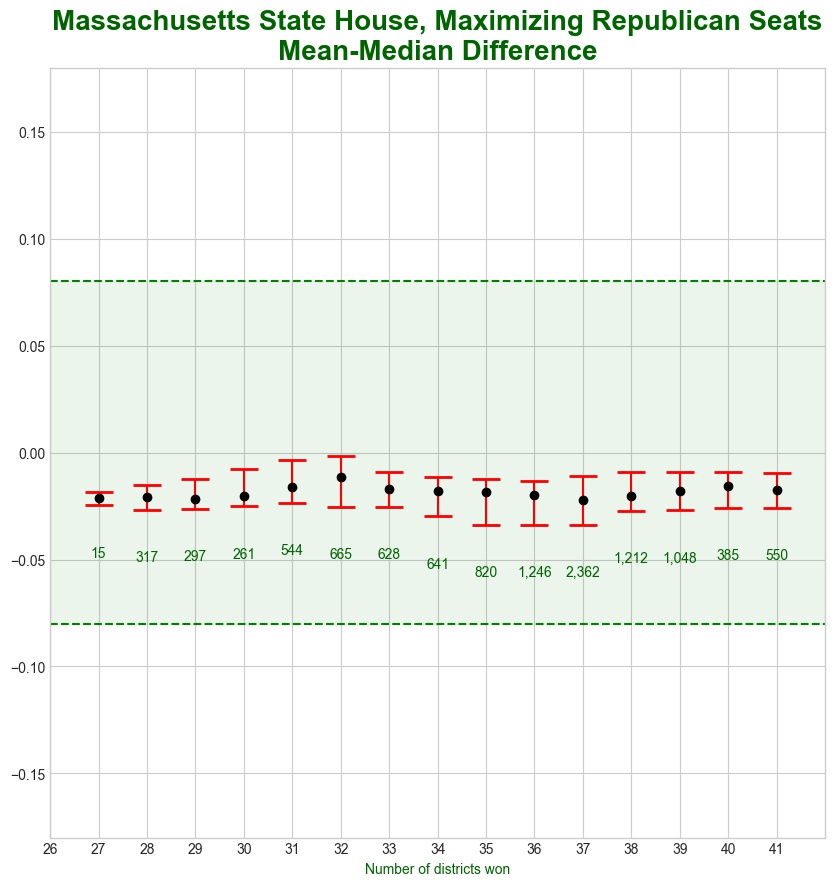}
\includegraphics[width=1.5in]{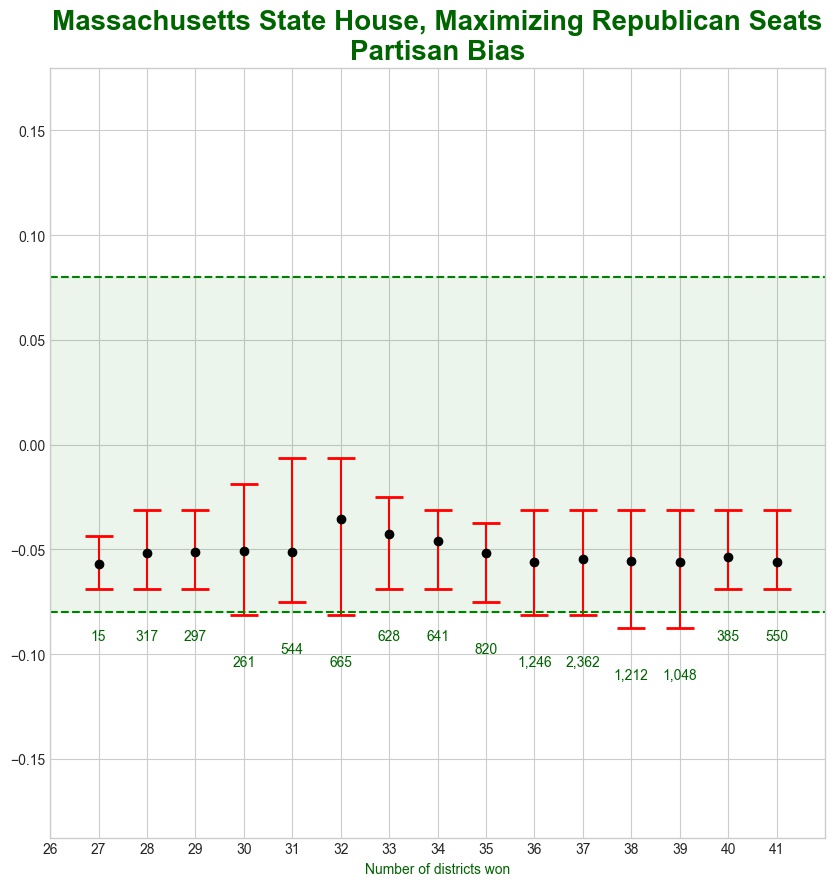}
\caption{Empirical results for Massachusetts's State House map and 2018 US Senate election data, searching for maps with as many Republican-won districts as possible. Declination not included because it is frequently undefined.   Horizontal axis is number of districts won, vertical axis is metric value ranges.  The green region is from $0.16\inf(m)$ to $0.16\sup(m)$ for each metric $m$.  The small number below each metric value range is the number of maps produced that had the corresponding number of districts won.  The dot within each vertical bar is the mean value of that metric on all produced maps with the corresponding number of districts won.}
\label{fig:short_bursts_MAlowerR}
\end{figure}

In Figures \ref{fig:short_bursts_MAlowerD} and \ref{fig:short_bursts_MAlowerR}, we can see that  the range of values for the Efficiency Gap and GEO Ratio is higher when the number of Democratic-won districts is higher.  And conversely, the range of values for the Efficiency Gap and GEO Ratio is lower when the number of Republic-won districts is higher.  To reiterate, this is what would be expected for a metric which detects an extreme number of districts won: the metric has a range of values that is more extreme on maps with a more extreme number of districts won.  In contrast, for the Mean-Median Difference and Partisan Bias, the range of values of those metrics is similar, regardless of the number of districts won by a particular party\footnote{These are the results found in \cite{GameabilityStudy}, except that  \cite{GameabilityStudy} did not consider  the Partisan Bias, which is included here.}.  Additionally, these ranges of values does not seem to have an upward trend on the maps with an extreme number of Democratic districts nor a downward trend on the maps with an extreme number of Republican districts.  

These results, that extreme values of the Mean-Median Difference and Partisan Bias do \emph{not} correlate with extreme number of districts won, are seen on every map we tested.  We will explore a few more examples; the reader interested in seeing the others can find them in Appendix \ref{appendix:sb_images}.

While the previous example (Massachusets State House, in Figures  \ref{fig:short_bursts_MAlowerD} and \ref{fig:short_bursts_MAlowerR}) corresponds to a very Democratic-leaning state and a map with many districts (160), the next example is for a very Republican-leaning state with few districts (5): Oklahoma's congressional map.  The results for that map are in Figures \ref{fig:short_bursts_OKcongD} and \ref{fig:short_bursts_OKcongR}.

\begin{figure}[h]
\centering
\includegraphics[width=1.5in]{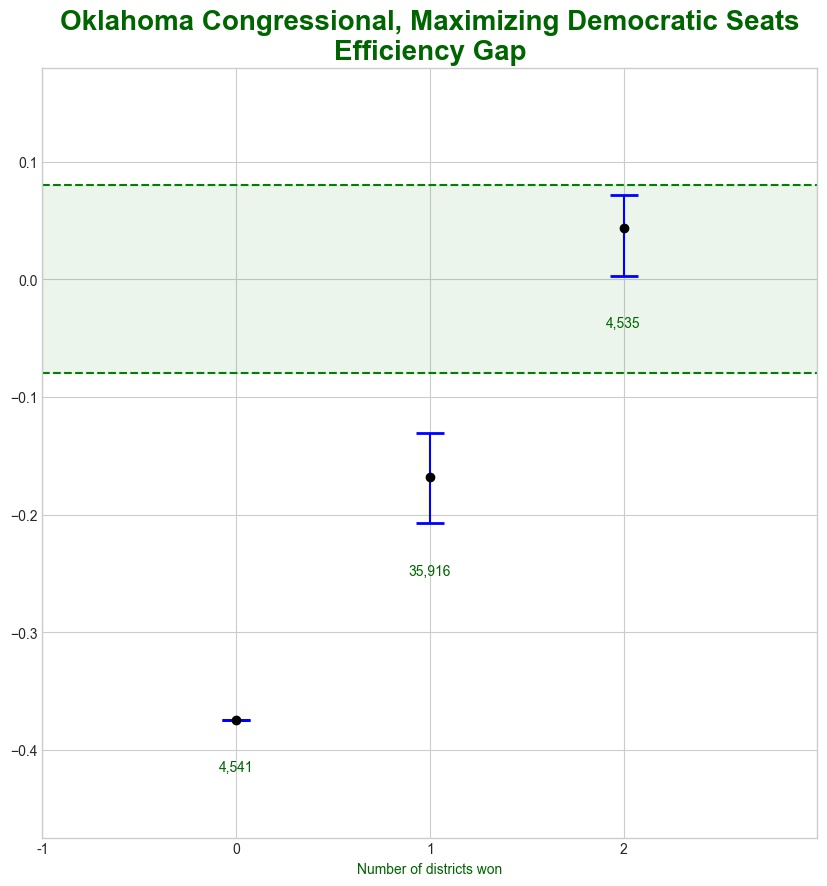}
\includegraphics[width=1.5in]{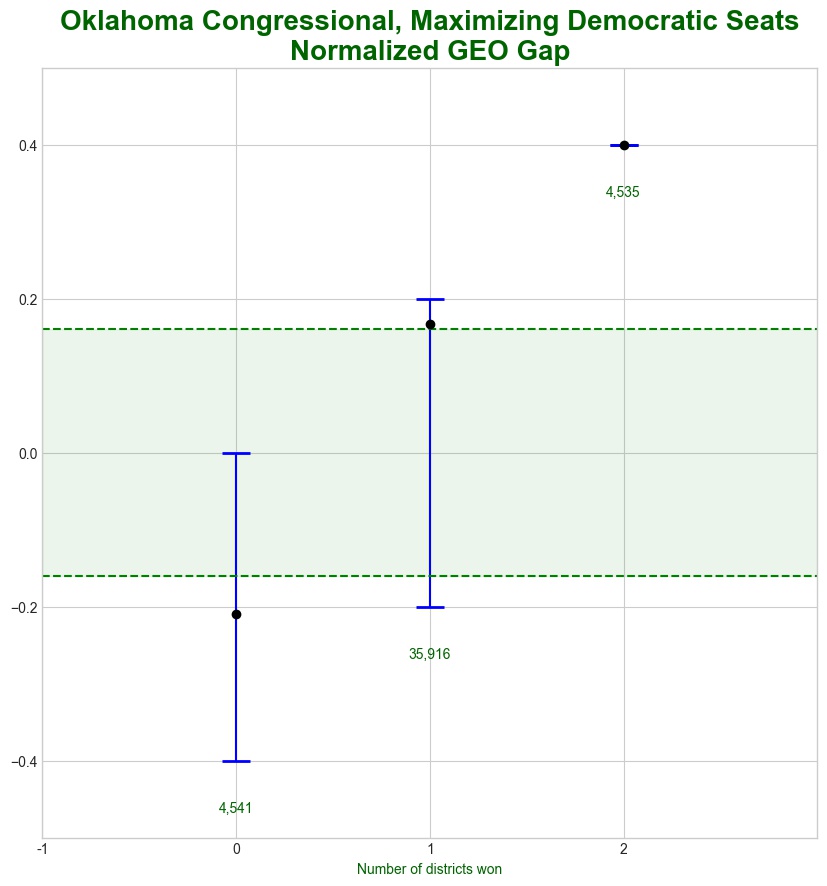}
\includegraphics[width=1.5in]{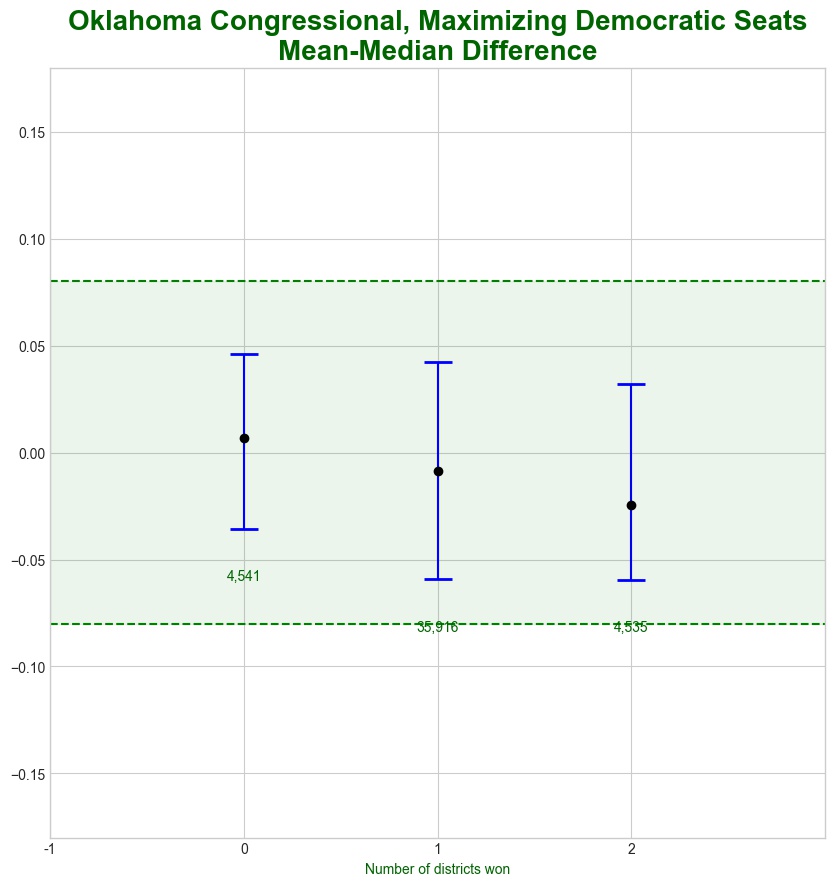}
\includegraphics[width=1.5in]{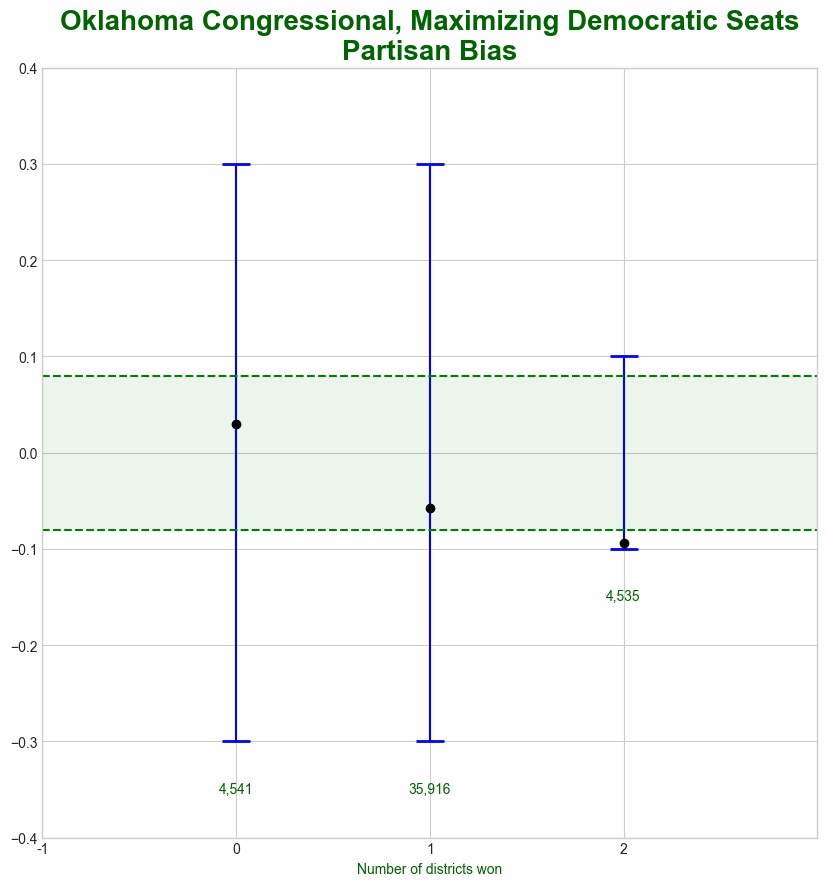}
\caption{Empirical results for Oklahoma's Congressional map and 2018 Gubenatorial election data, searching for maps with as many Democratic-won districts as possible.  Horizontal axis is number of districts won, vertical axis is metric value ranges.  The green region is from $0.16\inf(m)$ to $0.16\sup(m)$ for each metric $m$.  The small number below each metric value range is the number of maps produced that had the corresponding number of districts won.  The dot within each vertical bar is the mean value of that metric on all produced maps with the corresponding number of districts won.}
\label{fig:short_bursts_OKcongD}
\end{figure}

\begin{figure}[h]
\centering
\includegraphics[width=1.5in]{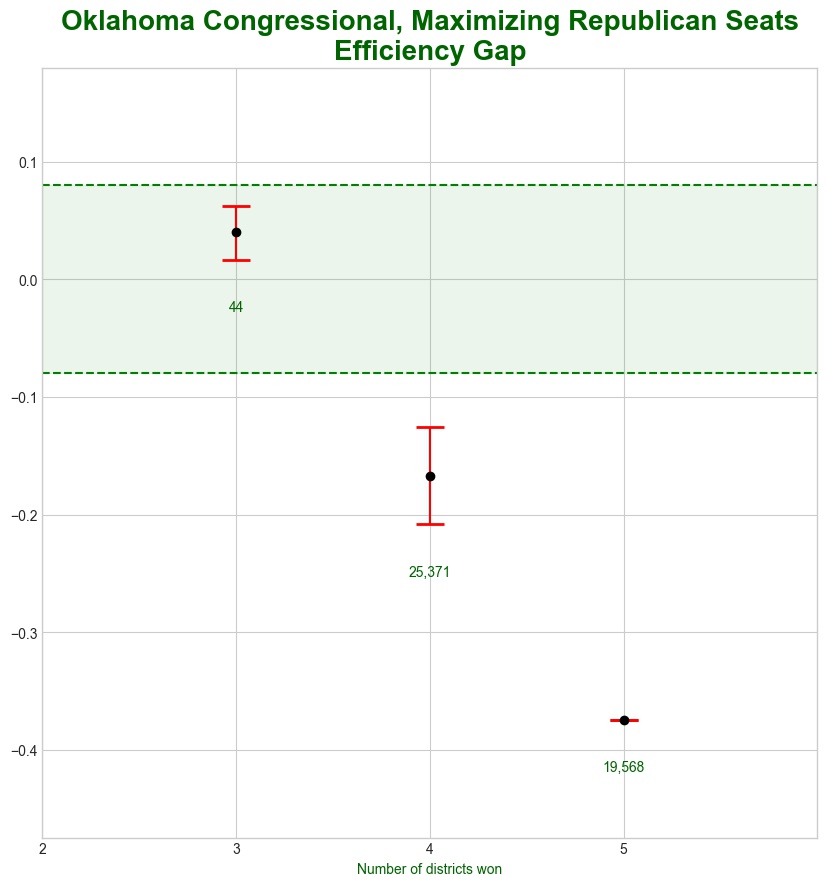}
\includegraphics[width=1.5in]{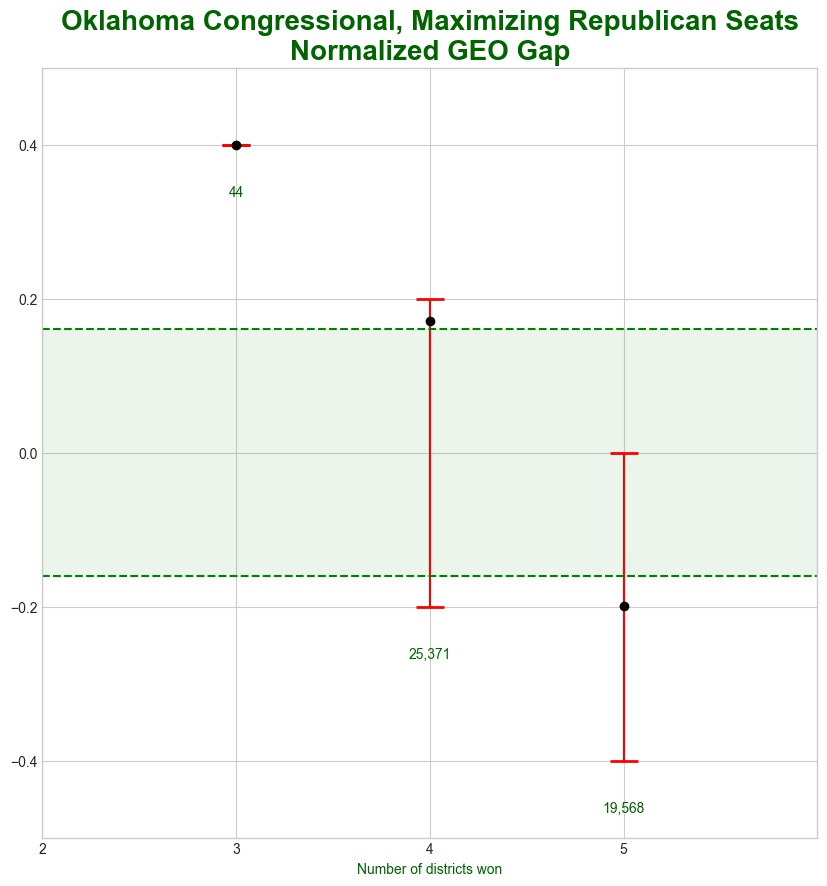}
\includegraphics[width=1.5in]{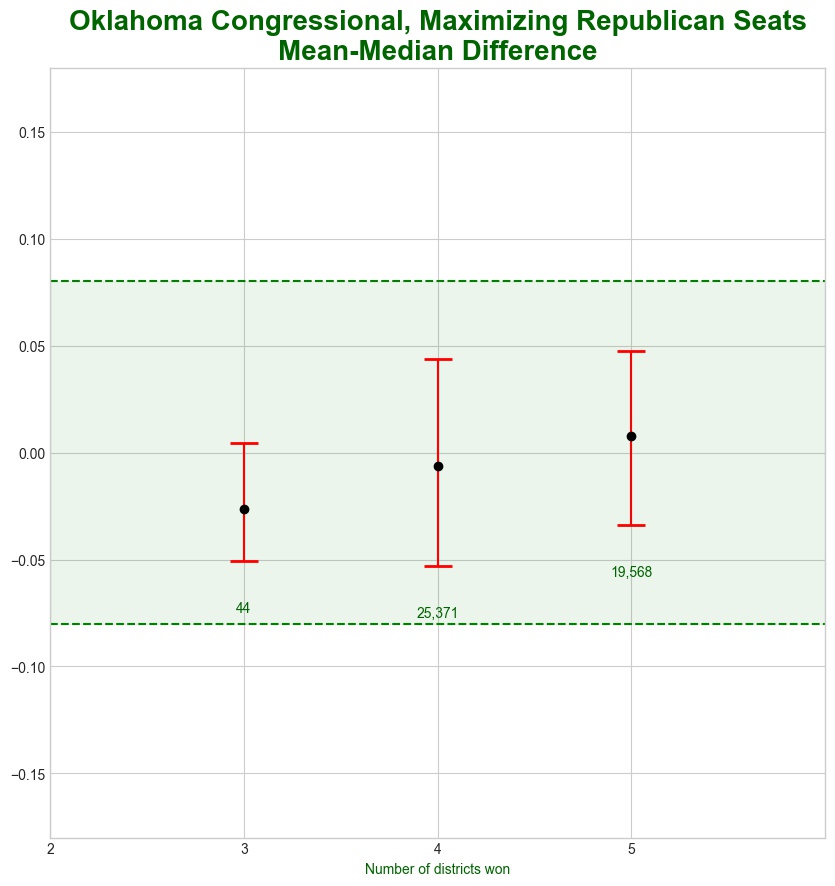}
\includegraphics[width=1.5in]{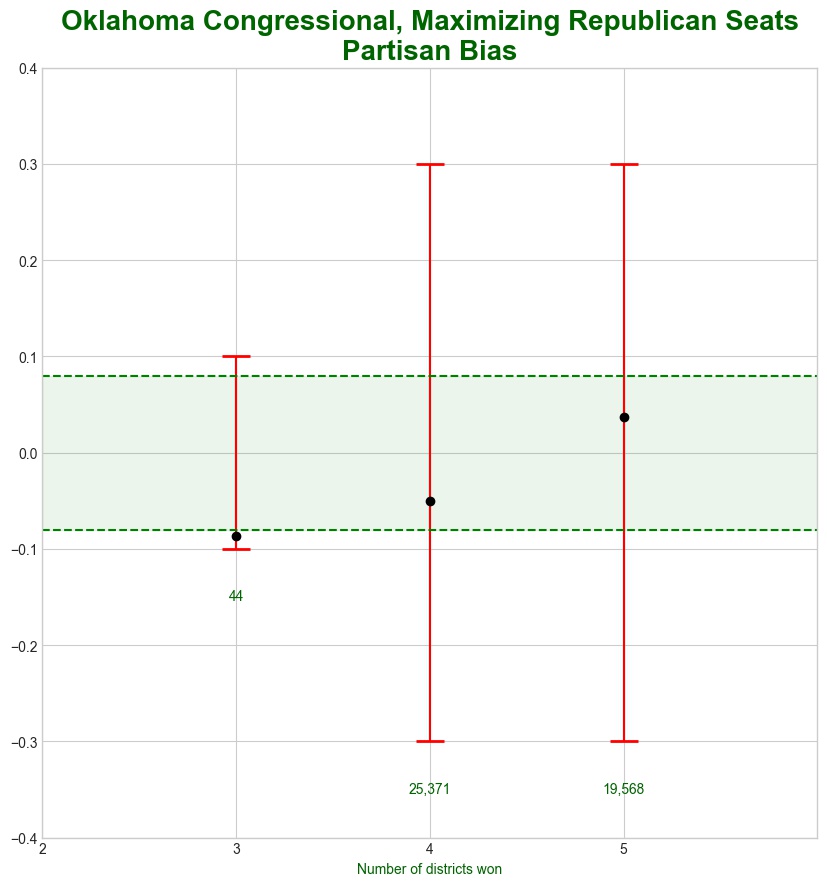}
\caption{Empirical results for Oklahoma's Congressional map and 2018 Gubenatorial election data, searching for maps with as many Republican-won districts as possible.  Horizontal axis is number of districts won, vertical axis is metric value ranges.  The green region is from $0.16\inf(m)$ to $0.16\sup(m)$ for each metric $m$.  The small number below each metric value range is the number of maps produced that had the corresponding number of districts won.  The dot within each vertical bar is the mean value of that metric on all produced maps with the corresponding number of districts won.}
\label{fig:short_bursts_OKcongR}
\end{figure}

Again, in Figures  \ref{fig:short_bursts_OKcongD} and \ref{fig:short_bursts_OKcongR}, we see that the range of values for the Mean-Median Difference and Partisan Bias are \emph{not} more extreme on maps with a more extreme number of districts won by the party in question.  In fact, while the mean values of the Efficiency Gap and GEO ratio increase as the number of districts won by Democrats increases (and the mean values decrease as the number of districts won by Republicans increases), as expected, those mean values do the opposite for the Mean-Median Difference and Partisan Bias.  That is, the ``average'' value of the Mean-Median Difference (Partisan Bias) on maps with an extreme number of Democratic districts suggests that they are \emph{less gerrymandered towards Democrats} than maps with a fewer number of Democratic districts won\footnote{And the ``average'' value of the Mean-Median Difference (Partisan Bias) on maps with an extreme number of Republican districts suggests that they are \emph{less gerrymandered towards Republicans} than maps with a fewer number of Republican districts won.}.  

The third and final example result that we give in this section is for a ``purple'' state (one not clearly Democratic or Republican leaning): Pennsylvania.  In Figures \ref{fig:short_bursts_PAcongD} and \ref{fig:short_bursts_PAcongR}, we see the results for Pennsylvania's congressional map.  Again, we see that the most extreme values of the Mean-Median Difference and Partisan Bias are achieved on many maps, not just those with the most extreme number of districts won by a particular party.  

\begin{figure}[h]
\centering
\includegraphics[width=1.5in]{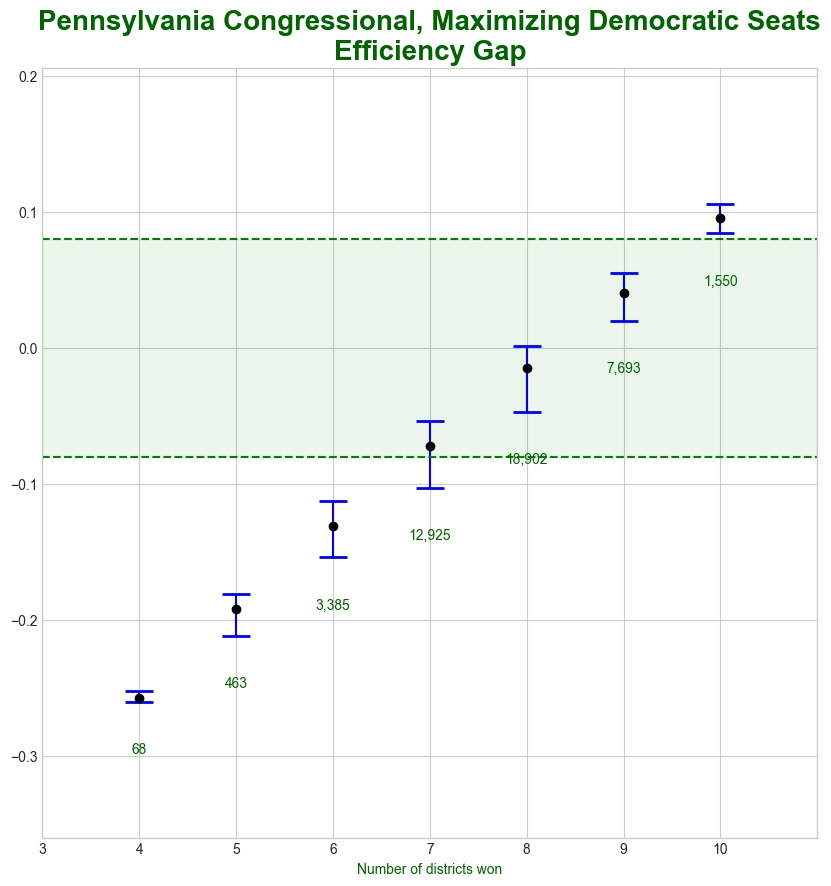}
\includegraphics[width=1.5in]{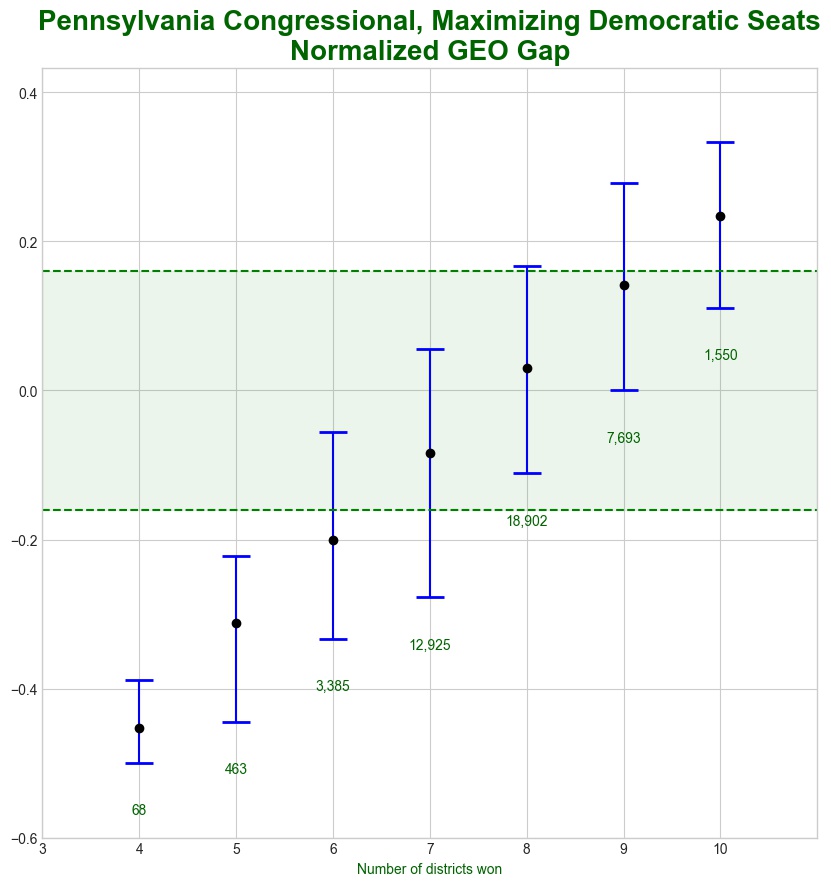}
\includegraphics[width=1.5in]{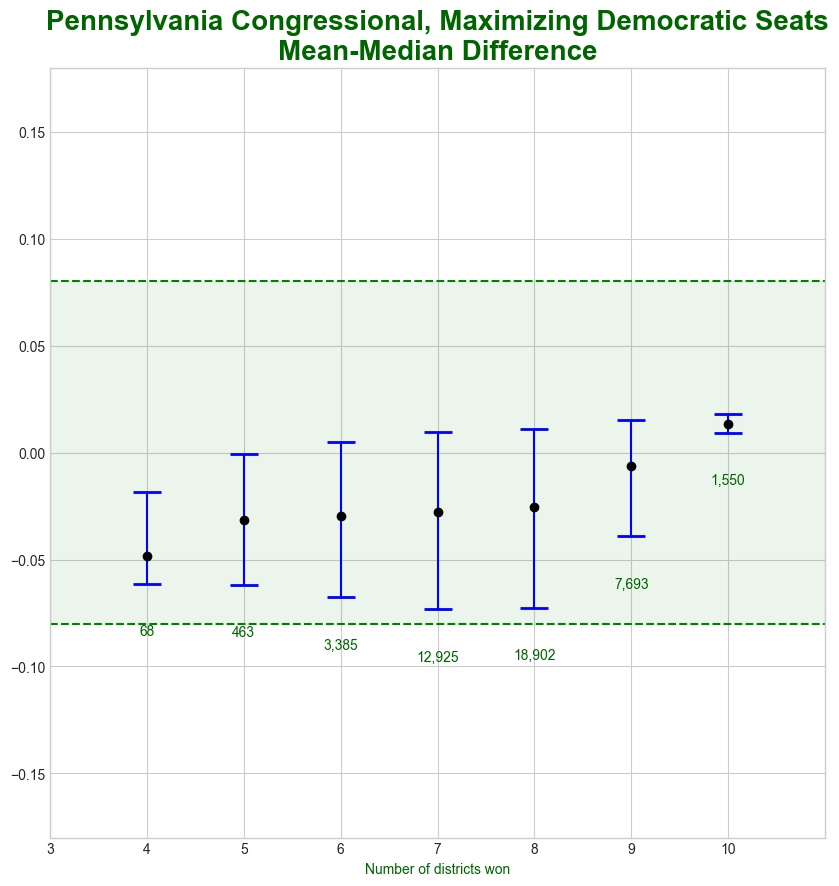}
\includegraphics[width=1.5in]{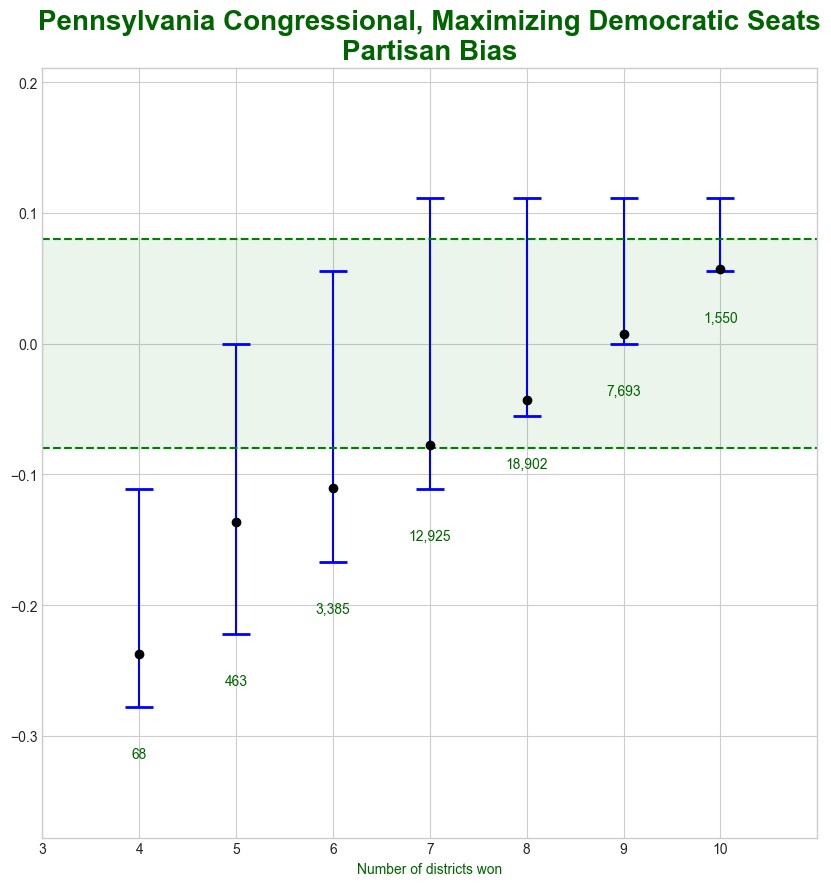}
\caption{Empirical results for Pennsylvania's Congressional map and 2016 US Senate election data, searching for maps with as many Democratic-won districts as possible.  Horizontal axis is number of districts won, vertical axis is metric value ranges.  The green region is from $0.16\inf(m)$ to $0.16\sup(m)$ for each metric $m$.  The small number below each metric value range is the number of maps produced that had the corresponding number of districts won.  The dot within each vertical bar is the mean value of that metric on all produced maps with the corresponding number of districts won.}
\label{fig:short_bursts_PAcongD}
\end{figure}

\begin{figure}[h]
\centering
\includegraphics[width=1.5in]{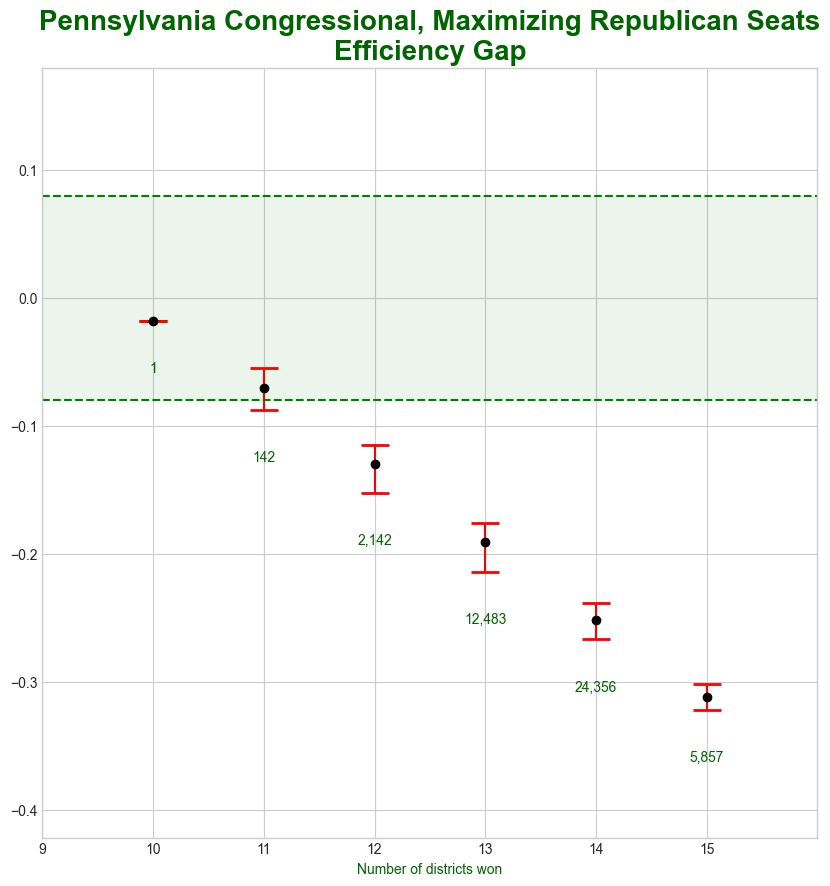}
\includegraphics[width=1.5in]{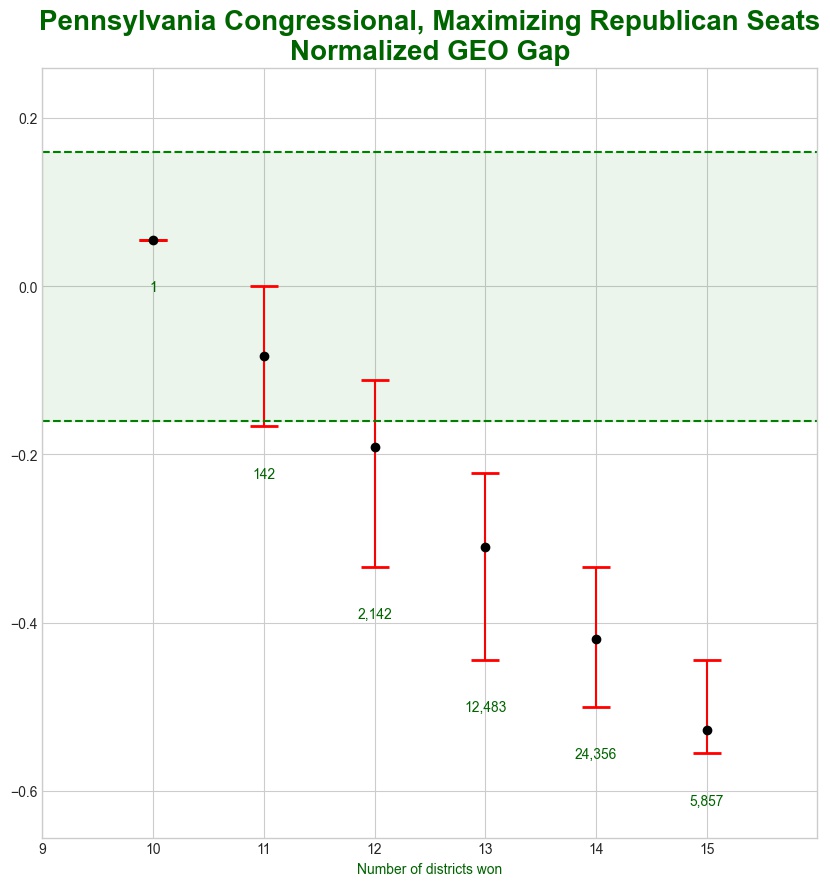}
\includegraphics[width=1.5in]{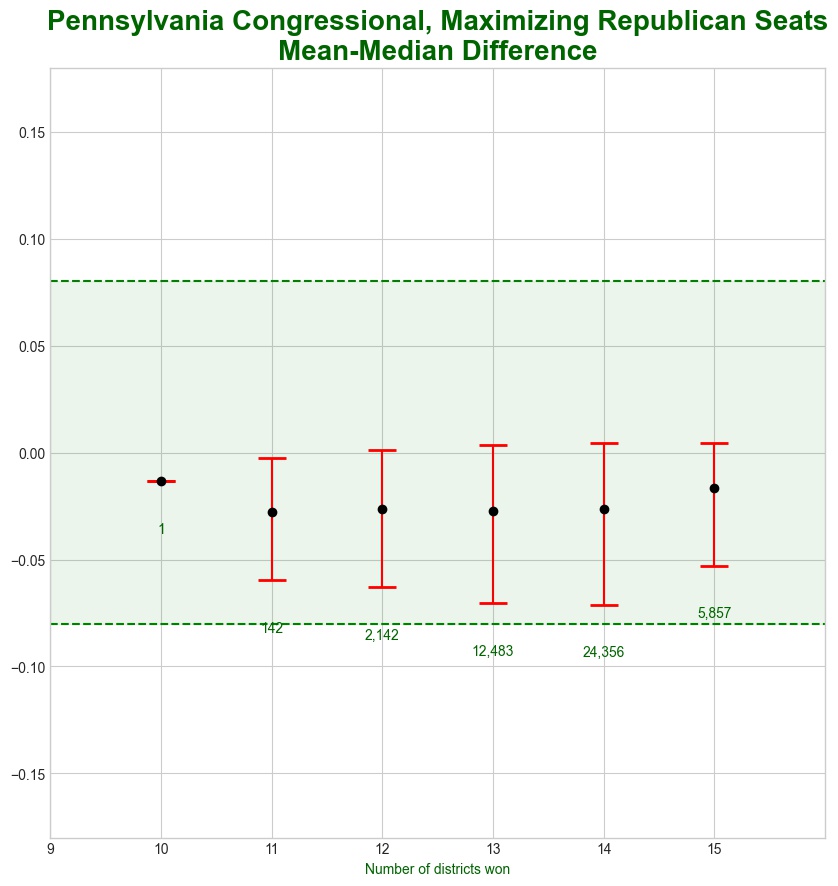}
\includegraphics[width=1.5in]{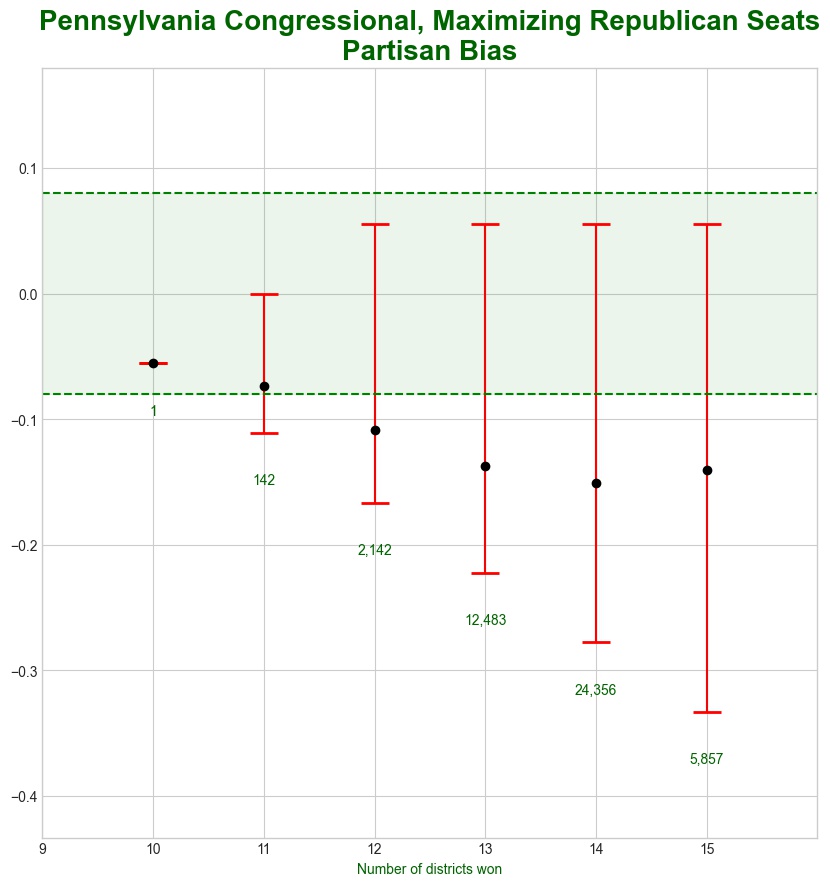}
\caption{Empirical results for Pennsylvania's Congressional map and 2016 US Senate election data, searching for maps with as many Republican-won districts as possible.  Horizontal axis is number of districts won, vertical axis is metric value ranges.  The green region is from $0.16\inf(m)$ to $0.16\sup(m)$ for each metric $m$.  The small number below each metric value range is the number of maps produced that had the corresponding number of districts won.  The dot within each vertical bar is the mean value of that metric on all produced maps with the corresponding number of districts won.}
\label{fig:short_bursts_PAcongR}
\end{figure}

Figures  \ref{fig:short_bursts_MAlowerD},  \ref{fig:short_bursts_MAlowerR}, \ref{fig:short_bursts_OKcongD}, \ref{fig:short_bursts_OKcongR}, \ref{fig:short_bursts_PAcongD}, \ref{fig:short_bursts_PAcongR}, and the other Figures in Appendix \ref{appendix:sb_images} give clear evidence that extreme values of the Mean-Median Difference and Partisan Bias do \emph{not} correspond to maps with an extreme number of districts won.  That is, while we saw in Section \ref{sec:bounds} that the values these two metrics \emph{are able to take on} do not become more extreme when a more extreme number of districts is won, here we see that, even on  real maps/data, the values that these two metrics \emph{do take on} cannot detect extreme number of districts won.

We note that these results (indicating that the symmetry metricss of Partisan Bias and Mean-Median Difference values are not useful in flagging extreme maps) are similar to those in \cite{ImplementingPartisanSymmetry}, although here we see those results on 18 maps, rather than the three maps studied in \cite{ImplementingPartisanSymmetry}.  Thus, while the results of \cite{ImplementingPartisanSymmetry} could potentially have been dismissed as ``flukes'' and maps specifically chosen because of some kind of unique property, our results robustly show that on \emph{every single map} of the 18 maps we considered, the Mean-Median Difference and Partisan Bias do not detect an extreme number of districts won.  We encourage the reader to view the results for the other maps, which can be found in Appendix \ref{appendix:sb_images}.

\section{Conclusion}

In this article, we considered the relationship between possible values of the ubiquitous Mean-Median Difference and Partisan Bias metrics and the number of seats won under a fixed districting plan as a proxy for detecting gerrymandering. 
The results are contained in two studies: a theoretical study (the ``bounds'' study) and an empirical study (the ``bugs'' study).  In our theoretical study, we saw how the values the MM and PB could take on were different for different vote-share, seat-share pairs.  We saw that the values these metrics could take on are not, in general, more extreme when a more extreme number of districts is won.  We also saw how these metrics can take on values that run counter to the public's expectations, and to the ways they are frequently used and described.
In our empirical study, we considered 18 different maps: 3 maps for each of 6 states with varying political geographies.  These 18 maps also had a wildly varying number of districts (ranging from 5 to 203).  In each of these maps, we saw that more extreme values of the Mean-Median Difference and Partisan Bias did not correspond to maps with more extreme number of districts won.

Given the evidence we present, we hope that experts and the general public take extreme care when attempting to use these metrics in the future.  We also hope that 
public-facing apps, experts, and academics will provide clear 
descriptions of these metrics and their limitations in order to more accurately describe what they can and cannot do. 
We are not the first to give this warning \cite{ImplementingPartisanSymmetry}, though  our paper provides significant additional evidence that this warning should be heeded. 
Extreme values of these metrics simply do not correspond to an extreme number of districts won, and thus these metrics should not be solely used for that purpose.  

\section{Acknowledgements}

The authors would like to thank Thomas Ratliff for sharing the code he wrote to make the beautiful images for the Empirical analysis results, most of which can be seen in Appendix \ref{appendix:sb_images}.  The authors would also like to thank Moon Duchin, whose comments, questions, and suggestions have greatly improved this article.  Finally, the authors would like to thank the reviewers for their careful reading and thoughtful comments.

This material is based upon work supported by the National Science Foundation under Grant No. DMS-1928930, while the authors were in residence at the Mathematical Sciences Research Institute in Berkeley, California, during the semester/year of Fall 2023.

\bibliographystyle{plain}
\bibliography{MMandPartisanBias}

\appendix

\section{Proofs}\label{sec:proofs}

For ease of exposition, we use the notation that a district has vote share $\frac{1}{2}^{+}$ if party $A$'s vote share is $\frac{1}{2}$ and party $A$ wins that district.  Similarly, a district has vote share $\frac{1}{2}^-$ if party $A$'s vote share is $\frac{1}{2}$ and party $A$ loses that district.

\begin{proof}[Proof of Theorem \ref{thm:PBfixed_districts}]
Suppose $(V, S)$ is a pair of rational numbers with $\frac{1}{4} \leq V \leq 1$, $ \frac{1}{2} \leq S  \leq 1 $, $S \leq 2V$, and $S\geq2V-1$. Let $n$ be the number of districts; that is, $S$ can be written as a rational number with denominator $n$.

\noindent\underline{Case 1:}  $V = \frac{1}{2}$

Suppose first that $S = 1$.  In this case, all districts must have vote share $\frac{1}{2}^+$, which means that the PB of this election is 0.

Otherwise, say  $S = \frac{n - \ell}{n}$ for $\ell \geq 1$.  In this case, choose $\epsilon >0$, $\epsilon \in \Q$ small enough so that  $\frac{1}{2}-(n-\ell)\epsilon >0$.  Set $n-\ell$ districts to have vote share $V+\epsilon = \frac{1}{2}+\epsilon$; these are winning districts.  Set one losing district to have vote share $V-(n-\ell)\epsilon = \frac{1}{2}-(n-\ell)\epsilon$, and the remaining losing districts to have vote share $V = \frac{1}{2}^-$.  This constructed election has vote share $V$, seat share $S$, a maximum number of districts with vote share above $V$, and a minimum number of districts with vote share below $V$.  The Partisan Bias is 
\begin{equation*}
\frac{1}{2} \left(\frac{n-\ell}{n}-\frac{1}{n}\right) = \frac{1}{2} \left(S-\frac{1}{n}\right)
\end{equation*}
as claimed in the Theorem statement.  

Now choose $\epsilon >0$, $\epsilon \in \Q$ small enough so that  $\frac{1}{2}+\ell\epsilon <1$.  Set $\ell$ districts to have vote share $V-\epsilon = \frac{1}{2}-\epsilon$; these are losing districts.  Set one winning district to have vote share $V+\ell\epsilon = \frac{1}{2}+\ell\epsilon$, and the remaining winning districts to have vote share $V = \frac{1}{2}^+$.  This constructed election has vote share $V$, seat share $S$,  a minimum number of districts with vote share above $V$, and a maximum number of districts with vote share below $V$.  The Partisan Bias is 
\begin{equation*}
\frac{1}{2} \left(\frac{1}{n}-\frac{\ell}{n}\right) = \frac{1}{2}\left(S-1+\frac{1}{n}\right)
\end{equation*}
as claimed in the Theorem statement.

\noindent\underline{Case 2:} $\frac{1}{4} \leq V < \frac{1}{2}$

Let $k = \lceil n \left(1- \frac{S}{2V}\right)-1\rceil$.  Note that $k$ is chosen as the largest integer such that

\begin{equation*}\label{eq:k_boundPB}
V \cdot \frac{k}{n} + S\cdot \frac{1}{2} < V
\end{equation*}
 
To maximize PB on an election with seat share $S$ and vote share $V$, set proportion $S$ districts to have vote share $\frac{1}{2} + \delta$, and set $k$  districts to have vote share $V+\epsilon$ for $\delta, \epsilon \in \Q$, $\delta, \epsilon >0$ small enough that  $(V+\epsilon) \cdot \frac{k}{n} + S\cdot \left(\frac{1}{2} +\delta\right)= V$ and $V + \epsilon \leq \frac{1}{2}$, $\frac{1}{2}+\delta \leq 1$.  (Such $\epsilon$ and $\delta$ are guaranteed by the choice of $k$, along with the fact that $V < \frac{1}{2} \leq S$.) The remaining districts will have vote share 0.  This maximal PB value is $\frac{1}{2}\left( S+ \frac{k}{n}-\left(1-S-\frac{k}{n}\right)\right) = \frac{1}{2}\left(2S-1+\frac{2k}{n}\right)$.

If we want to make PB as small as possible on an election with seat share $S$ and vote share  $V$, we recall that  $S \leq 2V$, so that the $p\in \Q$ satisfying 
\begin{equation*}
p(1-S) + S \cdot \frac{1}{2} = V
\end{equation*}
must have $p\geq0$.  Since $V < \frac{1}{2}$, we must have $p < V$.  Thus, the minimal PB achievable is $\frac{1}{2} \left(S - (1-S)\right) = S-\frac{1}{2}$.

Using these two constructions, we have for $(V, S)$ the following bounds:
\begin{align*}
S-\frac{1}{2}  &\leq PB \leq \frac{1}{2}\left(2S-1+\frac{2k}{n}\right) \\
S-\frac{1}{2} & \leq PB \leq  \frac{1}{2}\left(2S-1 +\frac{2\left\lceil n \left(1-\frac{S}{2V}\right)-1\right\rceil}{n}\right)
\end{align*}

\noindent\underline{Case 3:} $\frac{1}{2} < V < 1$, $S = 1$

In this case, we can choose $\epsilon > 0, \epsilon \in \Q$ small enough that $V + (n-1)\epsilon \leq 1$ and $V - (n-1)\epsilon \geq \frac{1}{2}$.  The largest PB is achieved with $n-1$ districts having vote share $V+\epsilon$ and one district having vote share $V-(n-1)\epsilon$, resulting in a PB value of $\frac{1}{2}\left(\frac{n-1}{n}-\frac{1}{n}\right) = \frac{n-2}{2n}$.  And the smallest PB is achieved with $n-1$ districts having vote share $V-\epsilon$ and one district having vote share $V+(n-1)\epsilon$, resulting in a PB value of $\frac{1}{2}\left(\frac{1}{n} - \frac{n-1}{n}\right) = \frac{2-n}{2n}$.

\noindent\underline{Case 4:}. $\frac{1}{2} < V < 1$, $S < 1$

Here, if we want to make PB as large as possible for an election winning seat share $S$ and vote share  $V$, we construct election data with proportion $S$ districts having vote share $V+\epsilon$ for $\epsilon \in \Q, \epsilon > 0$, and proportion $1-S$ districts having vote share $\frac{1}{2}-\delta$ for $\delta \in \Q, \delta \geq 0$ so that $V+\epsilon \leq 1$, $\frac{1}{2}-\delta \geq 0$, and 
\begin{equation*}
\left(1-S\right)\left(\frac{1}{2}-\delta\right) + S\left(V+\epsilon\right) = V.
\end{equation*}
Such $\epsilon$ and $\delta$ are guaranteed by the facts that $S < 1$ and $S \geq 2V-1$.  This results in a maximal PB of $\frac{1}{2}\left(S-(1-S)\right) = S-\frac{1}{2}$.

If we want to make PB as small as possible  for an election winning $S$ districts and vote share  $V$, we set $k = \lceil n\left(1-\frac{1-S}{2(1-V)}\right)-1 \rceil$.  Note that $k$ is chosen as the largest integer such that

\begin{equation*}
(1-S) \cdot \frac{1}{2} + \frac{k}{n}V + \left(S-\frac{k}{n}\right) >V
\end{equation*}

Choose $\epsilon \in \Q, \delta \in \Q$, and $\gamma \in \Q$ with $\epsilon > 0, \delta, \gamma \geq 0$ so that $(1-S) \cdot \left(\frac{1}{2}-\delta\right) + \frac{k}{n}(V-\epsilon) + \left(S-\frac{k}{n}\right)\left(1-\gamma\right) = V$, $\frac{1}{2} - \delta \geq 0$, $V-\epsilon \geq \frac{1}{2}$ and $1-\gamma > V$.  Such $\epsilon, \delta$, and $\gamma$ exist because of the choice of $k$ and the facts that $V > \frac{1}{2}$ and $S < 1$.  Thus, to minimize PB, we construct election data with a proportion $S- \frac{k}{n}$ of districts having vote share $1-\gamma$, a proportion of $1-S$ districts having vote share $\frac{1}{2}-\delta$, and $k$ districts having vote share $V - \epsilon$.  This results in a PB value of $\frac{1}{2}\left(S-\frac{k}{n} - \left(1-S+\frac{k}{n}\right)\right) = \frac{1}{2}\left(2S-1-\frac{2k}{n}\right)$.

Using these two constructions, we have for $(V, S)$ the following bounds:
\begin{align*}
\frac{1}{2}\left(2S-1-\frac{2k}{n}\right)&\leq PB\leq  S-\frac{1}{2} \\
\frac{1}{2}\left(2S-1-\frac{2 \left\lceil n \left(1-\frac{1-S}{2(1-V)}\right)-1\right\rceil}{n}\right)   &\leq  PB \leq S-\frac{1}{2}
\end{align*}

\noindent\underline{Case 5:} $V = 1$.  In this case, $S = 1$ and all districts have vote share 1.  Thus, the only value for PB is 1.  

\end{proof}

\begin{proof}[Proof of Theorem \ref{thm:PB}]
Take the limit as $n \to \infty$ in Theorem \ref{thm:PBfixed_districts}.

\end{proof}

\begin{proof}[Proof of Theorem \ref{thm:MMfixed_districts}]

Suppose $(V, S)$ is a pair of rational numbers with $\frac{1}{4} \leq V \leq 1$, $ \frac{1}{2} \leq S \leq 1$, $S \leq 2V$, and $S\geq2V-1$.  Suppose that $n \geq 3$ is the number of districts and that  $\ell$ is the number of losing districts, so that $S = \frac{n-\ell}{n}$.

\noindent\underline{Case 1:} $S = \frac{1}{2}$.

Since $S = \frac{1}{2}$, the smallest that the median value could be is $\frac{1}{4}$, which would occur when the middle vote share values are 0 (a losing district) and $\frac{1}{2}^+$ (a winning district).  For $V \geq \frac{1}{4} $, this median value is achievable by letting the other winning vote shares range from $\frac{1}{2}$ (when $V = \frac{1}{4}$)  to 1 (for $V$ approaching $\frac{1}{2}$).  This gives a minimum MM of $\frac{1}{4}-V$, which is achievable for $V \geq \frac{1}{4}$, so long as $V \leq \frac{0+\frac{1}{2}+\frac{n}{2}-1}{n} = \frac{n-1}{2n}$.

For $\frac{n-1}{2n} \leq V  \leq \frac{3}{4}$, let $p\in \Q$ be the value such that
\begin{equation*}
\frac{1}{n} \left(\frac{n}{2}p + \frac{1}{2} + \left(\frac{n}{2}-1\right)\cdot 1 \right) = V
\end{equation*}

Note that this implies
\begin{equation*}
p = 2\left(V + \frac{1}{2n} - \frac{1}{2}\right)
\end{equation*}

Then for $\frac{n-1}{2n} \leq V  \leq \frac{3}{4}$ the smallest that the median value can be for a constructed election with seat share $S= \frac{1}{2}$ and vote share $V$ occurs when there are $\frac{n}{2}-1$ districts with vote share 1,  1 district with vote share $\frac{1}{2}^+$, and $\frac{n}{2}$ districts with vote share $p$.  This results in a MM value of
\begin{equation*}
\frac{p+\frac{1}{2}}{2} - V= \frac{2\left(V + \frac{1}{2n} - \frac{1}{2}\right)+\frac{1}{2}}{2} -V = \frac{1}{2n}-\frac{1}{4}.
\end{equation*}
 Using symmetry, we get
\begin{align*}
\frac{1}{4}-V & \leq MM \leq \frac{1}{4}-\frac{1}{2n}  & \text{ if } S = \frac{1}{2}, \quad  \frac{1}{4} \leq V \leq \frac{n-1}{2n} \\
\frac{1}{2n}-\frac{1}{4}&\leq MM \leq \frac{1}{4}-\frac{1}{2n}  & \text{ if } S = \frac{1}{2}, \quad  \frac{n-1}{2n} < V < \frac{n+1}{2n} \\
\frac{1}{2n}-\frac{1}{4}&\leq MM \leq \frac{3}{4}-V  & \text{ if } S = \frac{1}{2}, \quad V \geq \frac{n+1}{2n}\\
\end{align*}

\noindent\underline{Case 2:} $S> \frac{1}{2}$, $\frac{1}{4} < V \leq 1$, lower bounds for MM

Firstly note that, since we assume that $S > \frac{1}{2}$, the median vote share must be at least as large as $\frac{1}{2}$.  For $V \leq \frac{1}{2}$, we can achieve a median vote share of $\frac{1}{2}$ by constructing $S$ of the districts to have vote share $\frac{1}{2}^+$  and choosing the losing districts to have the vote share $p\in \Q$ satisfying $p(1-S)+\frac{1}{2} \cdot S = V$.  This $p$ will have $p \leq V \leq \frac{1}{2}$, since $V \leq \frac{1}{2}$.  This would give vote share $V$, seat share $S$, and the smallest possible median of $\frac{1}{2}$.  

For $\frac{1}{2} < V $, first suppose that $n$ is even (so that $n \geq 4$).  Then we can choose $\ell$ districts to have vote share $\frac{1}{2}^-$ and  $\frac{n}{2}+1-\ell$ districts to have vote share $\frac{1}{2}^+$, giving the smallest possible median value of $\frac{1}{2}$.  The remaining districts have vote share $p \in \Q$ which results in a statewide vote share of $V$.  However this is only possible if the equation
\begin{equation*}
\frac{\left(\frac{n}{2}+1\right) \cdot \frac{1}{2} + p \left(\frac{n}{2}-1\right)}{n} = V
\end{equation*}
has a solution with $p \leq 1$.  Thus, we have shown that if $V \leq \frac{3n-2}{4n}$, we can use this construction to obtain election data with MM value of $\frac{1}{2}-V$.  

If $V > \frac{3n-2}{4n}$, then the smallest possible median value $p$ would be achieved when $\ell$ districts have vote share $\frac{1}{2}^-$, $\frac{n}{2}-1$ districts have vote share 1, and the remaining districts have vote share $p \in \Q$ where $p$ satisfies the equation
\begin{equation*}
\frac{1}{n}\left(\frac{1}{2} \cdot \ell + \left(\frac{n}{2}-\ell+1\right)p + \left(\frac{n}{2}-1\right) \cdot 1 \right) = V.
\end{equation*}
This results in
\begin{equation*}
p = \frac{2V+S-2+\frac{2}{n}}{2S-1+\frac{2}{n}}
\end{equation*}
and an MM value of $p - V$.

Now suppose that $n$ is odd (so that $n \geq 3$).    Then we can choose $\ell$ districts to have vote share $\frac{1}{2}^-$ and  $\frac{n-1}{2}+1-\ell$ districts to have vote share $\frac{1}{2}^+$, giving the smallest possible median value of $\frac{1}{2}$.  The remaining districts have vote share $p \in \Q$ which results in a statewide vote share of $V$.  However this is only possible if the equation
\begin{equation*}
\frac{\left(\frac{n-1}{2}+1\right) \cdot \frac{1}{2} + p \cdot \frac{n-1}{2}}{n} = V
\end{equation*}
has a solution with $p \leq 1$.  Thus, we have shown that if $V \leq \frac{3n-1}{4n}$, we can use this construction to obtain election data with MM value of $\frac{1}{2}-V$.  

If $V > \frac{3n-1}{4n}$, then the smallest possible median value $p$ would be achieved when $\ell$ districts have vote share $\frac{1}{2}^-$, $\frac{n-1}{2}$ districts have vote share 1, and the remaining districts have vote share $p \in \Q$ where $p$ satisfies the equation
\begin{equation*}
\frac{1}{n}\left(\frac{1}{2} \cdot \ell + \left(\frac{n-1}{2}-\ell+1\right)p + \frac{n-1}{2} \cdot 1 \right) = V.
\end{equation*}
This results in
\begin{equation*}
p = \frac{2V+S-2+\frac{1}{n}}{2S-1+\frac{1}{n}}
\end{equation*}
and a mean-median value of $p - V$.

\noindent\underline{Case 3:} $S> \frac{1}{2}$, $\frac{1}{4} < V \leq 1$, upper bounds for MM

Suppose first that $n$ is even.  The maximum possible median is achieved by setting $\ell$ districts to have vote share 0, $\frac{n}{2}-\ell-1$ to have vote share $\frac{1}{2}^+$,  and the remaining $\frac{n}{2}+1$ districts to have vote share $p \in \Q$ where $p$ satisfies the equation

\begin{equation*}
\frac{1}{n}\left(\left(\frac{n}{2}+1\right)p+\left(\frac{n}{2}-\ell-1\right) \cdot \frac{1}{2}\right) = V.
\end{equation*}

This results in 
\begin{equation*}
p = \frac{2V+\frac{1}{2}-S+\frac{1}{n}}{1+\frac{2}{n}}
\end{equation*}
and a mean-median value of $p-V$.  However, this is only achievable if $p \leq 1$; otherwise the maximum possible median value is 1, resulting in a mean-median value of $1-V$.

Now suppose that $n$ is odd.  The maximum possible median is achieved by setting $\ell$ districts to have vote share 0, $\frac{n-1}{2}-\ell$ to have vote share $\frac{1}{2}^+$,  and the remaining $\frac{n-1}{2}+1$ districts to have vote share $p \in \Q$ where $p$ satisfies the equation

\begin{equation*}
\frac{1}{n}\left(\left(\frac{n-1}{2}+1\right)p+\left(\frac{n-1}{2}-\ell \right) \cdot \frac{1}{2}\right) = V.
\end{equation*}

This results in 
\begin{equation*}
p = \frac{2V+\frac{1}{2}-S+\frac{1}{n}}{1+\frac{2}{n}}
\end{equation*}
and a mean-median value of $p-V$.  However, this is only achievable if $p \leq 1$; otherwise the maximum possible median value is 1, resulting in a mean-median value of $1-V$.

\end{proof}

\begin{proof}[Proof of Theorem \ref{thm:MM}]
Take the limit as $n \to \infty$ in Theorem \ref{thm:MMfixed_districts}.
\end{proof}

\begin{proof}[Proof of Theorem \ref{thm:turnout_different}]
For ease of exposition, we state the following in terms of the Partisan Bias, but note that the same examples will work for the Mean-Median Difference as well.

Suppose $S = \frac{n-\ell}{n} \geq \frac{1}{2}$ and vote shares $V_1, V_2, \dots, V_n$ give a Partisan Bias value of 0.  In order to achieve the smallest possible vote share $V^*$ with turnout ratio $C$ corresponding to seat share $S$, we need to find $V_i$s which allow for the lowest possible vote shares in losing districts, and we need to weight those the most heavily.  This corresponds to:
\begin{align*}
V_1 = V_2 = \cdots = V_\ell &=0 \\
V_\ell = V_{\ell+1} = \cdots = V_{n-\ell-1} &= \frac{1}{2} \\
V_{n-\ell} = V_{n-\ell+1} = \cdots = V_n &= 1
\end{align*}
(Here districts with vote shares $\frac{1}{2}$ are winning).

The districts numbered 1 through $\ell$ will have $C$ times as much weight as the others, which implies they are each weighted $\frac{C}{n-\ell+C\ell}$, while districts numbered $\ell+1$ through $n$ are weighted $\frac{1}{n-\ell+C\ell}$.  Now we can calculate the vote share $V^*$:
\begin{align*}
V^* &= 0 \cdot \ell \cdot \frac{C}{n-\ell+C\ell} + (n-2\ell)\cdot \frac{1}{2} \cdot \frac{1}{n-\ell+C\ell}+\ell \cdot 1 \cdot \frac{1}{n-\ell+C\ell} \\
&= \frac{n}{2(n-\ell+C\ell)}
\end{align*}
Now we use the fact that $S = \frac{n-\ell}{n}$ to get
\begin{align*}
V^* &= \frac{n}{2(nS+C(n-nS))} \\
&= \frac{1}{2(S+C(1-S))} 
\end{align*}
By adjusting the weight $C$, we can get any $V^*$ in between $\frac{1}{2(S+C(1-S))} $ and $\frac{1}{2}$.

Now we prove the upper bound by similar means.  Suppose $S = \frac{n-\ell}{n} \geq \frac{1}{2}$ and vote shares $V_1, V_2, \dots, V_n$ give a Partisan Bias value of 0.  In order to achieve the largest possible vote share $V^*$ with turnout ratio $C$ corresponding to seat share $S$, we need to find $V_i$s which allow for the largest number of vote shares at value 1, along with the largest possible vote shares in losing districts.  Using the results from Corollary \ref{cor:pb_equal_turnout}, we know this corresponds to:
\begin{align*}
V_1 = V_2 = \cdots = V_\ell &=\frac{1}{2} \\
V_{\ell+1} = V_{\ell+2} = \cdots = V_{n/2} &=V \\
V_{n/2+1} = V_{n-\ell+1} = \cdots = V_n &= 1
\end{align*}
 where $V = \frac{2-S}{3-2S}$, which we get by solving the equation $S = \frac{3V-2}{2V-1}$ for $V$.  (Here districts with vote shares $\frac{1}{2}$ are losing).  

The districts numbered $n/2+1$ through $n$ will have $C$ times as much weight as the others, which implies they are each weighted $\frac{2C}{n(C+1)}$, while districts numbered $1$ through $n/2$ are weighted $\frac{2}{n(C+1)}$.  Now we can calculate the vote share $V^*$:
\begin{align*}
V^* &= \frac{1}{2} \cdot \ell \cdot \frac{2}{n(C+1)} +\left(\frac{n}{2}-\ell\right) \cdot \frac{2-S}{3-2S} \cdot \frac{2}{n(C+1)} + \frac{n}{2} \cdot \frac{2C}{n(C+1)} \\
\end{align*}
Again we use the fact that $S = \frac{n-\ell}{n} = 1-\frac{\ell}{n}$ to get
\begin{align*}
V^*&= \frac{1-S}{C+1} + \frac{2-S}{3-2S} \cdot \frac{1}{C+1} - \frac{2(1-S)(2-S)}{(C+1)(3-2S)} + \frac{C}{C+1} \\
&=\frac{1+C(3-2S)}{(C+1)(3-2S)}
\end{align*}
By adjusting the weight $C$, we can get any $V^*$ in between $\frac{2-S}{3-2S} $ and $\frac{1+C(3-2S)}{(C+1)(3-2S)}$.

From Corollary \ref{cor:pb_equal_turnout}, we already know that we can get any $V^*$ between $\frac{1}{2}$ and $\frac{2-S}{3-2S} $, so we are now done.

\end{proof}

\section{Images from Empirical Analysis}\label{appendix:sb_images}

\begin{figure}[h]
\centering
\includegraphics[width=1.5in]{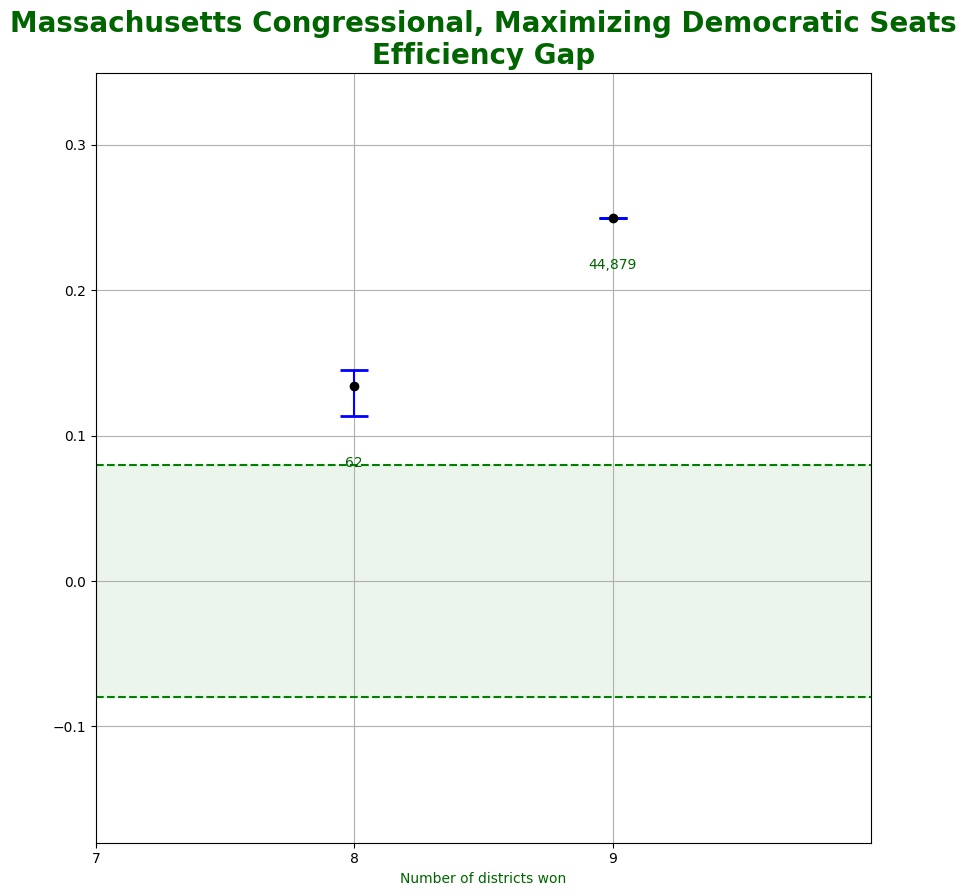}
\includegraphics[width=1.5in]{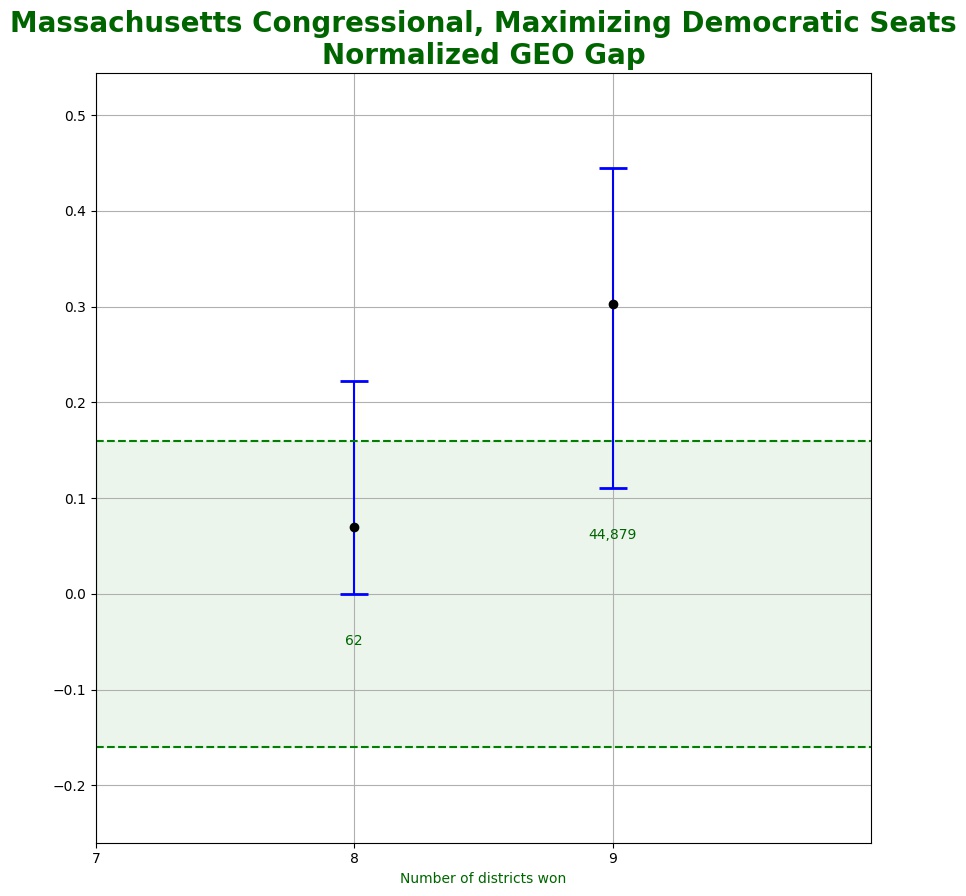}
\includegraphics[width=1.5in]{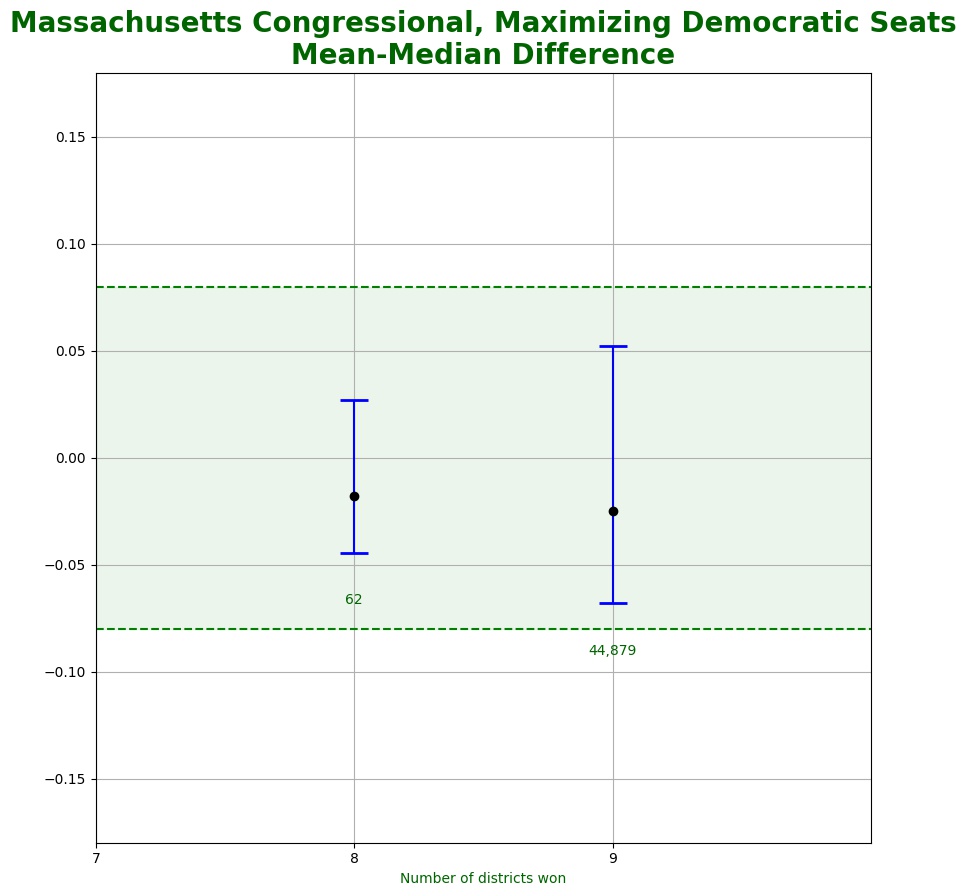}
\includegraphics[width=1.5in]{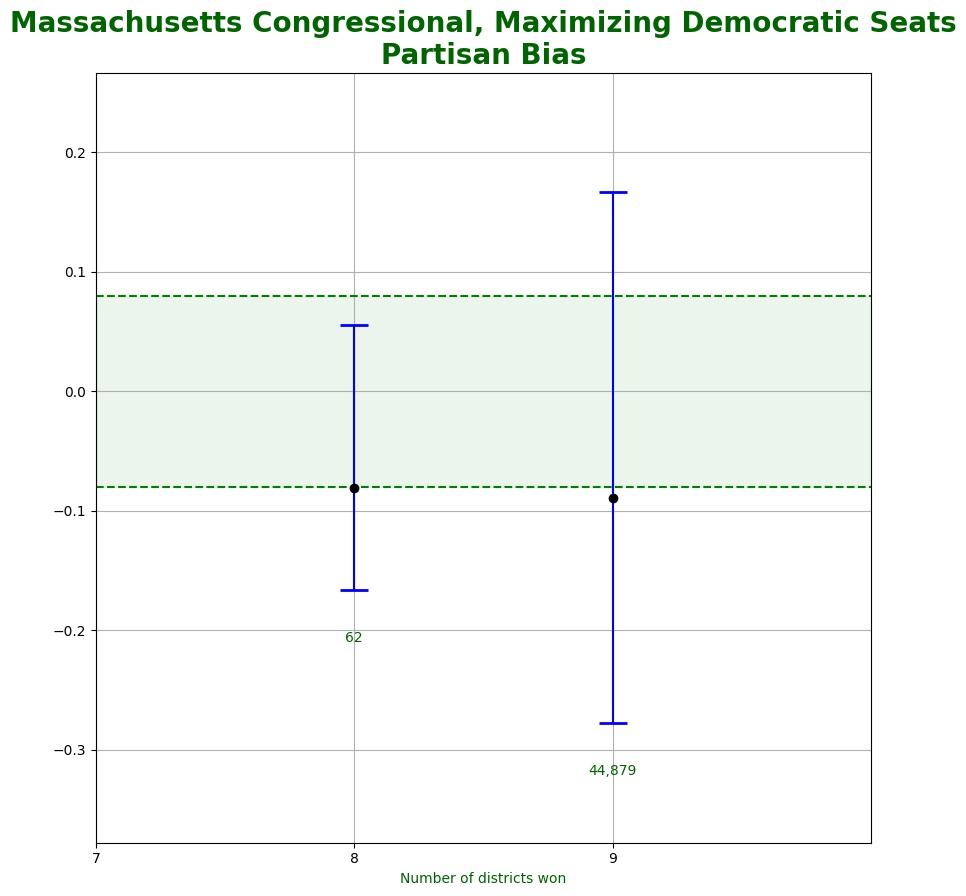}
\caption{Empirical results for Massachusetts's Congressional map and 2018 US Senate election data, searching for maps with as many Democratic-won districts as possible.  Horizontal axis is number of districts won, vertical axis is metric value ranges.  The green region is from $0.16\inf(m)$ to $0.16\sup(m)$ for each metric $m$.  The small number below each metric value range is the number of maps produced that had the corresponding number of districts won.  The dot within each vertical bar is the mean value of that metric on all produced maps with the corresponding number of districts won. }
\label{fig:short_bursts_MAcongD}
\end{figure}

\begin{figure}[h]
\centering
\includegraphics[width=1.5in]{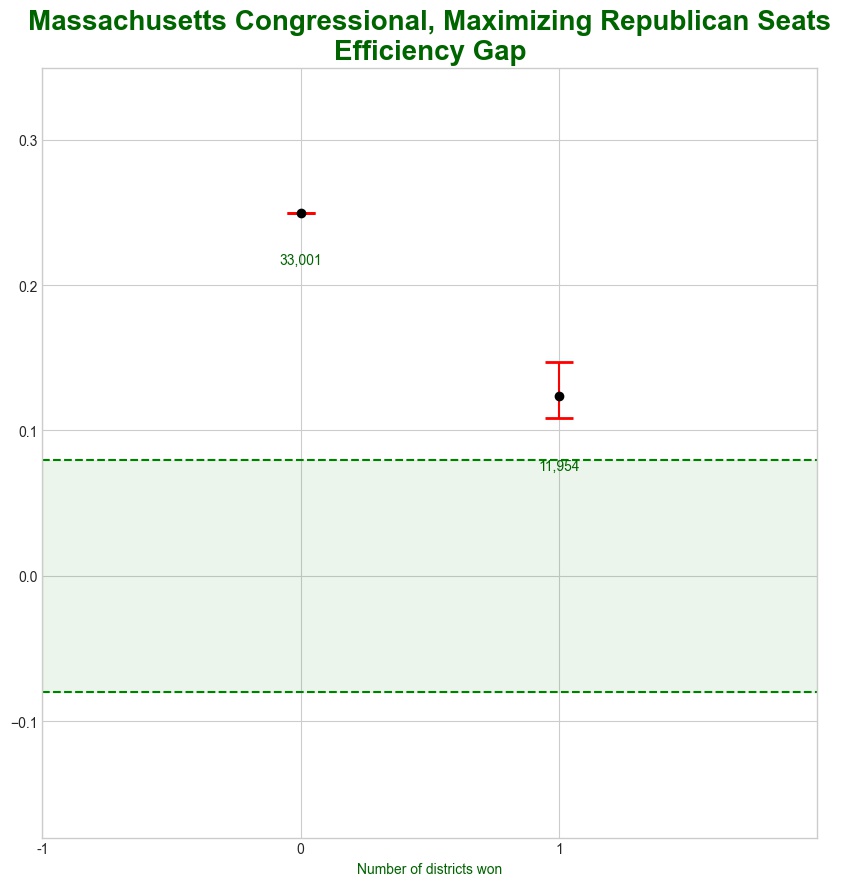}
\includegraphics[width=1.5in]{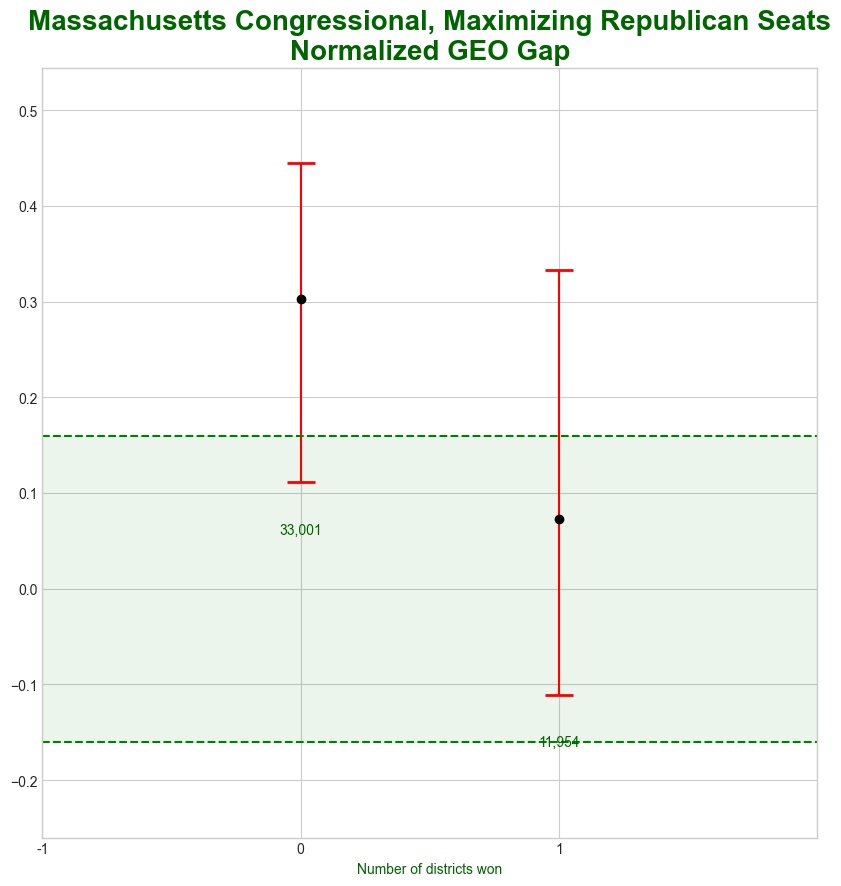}
\includegraphics[width=1.5in]{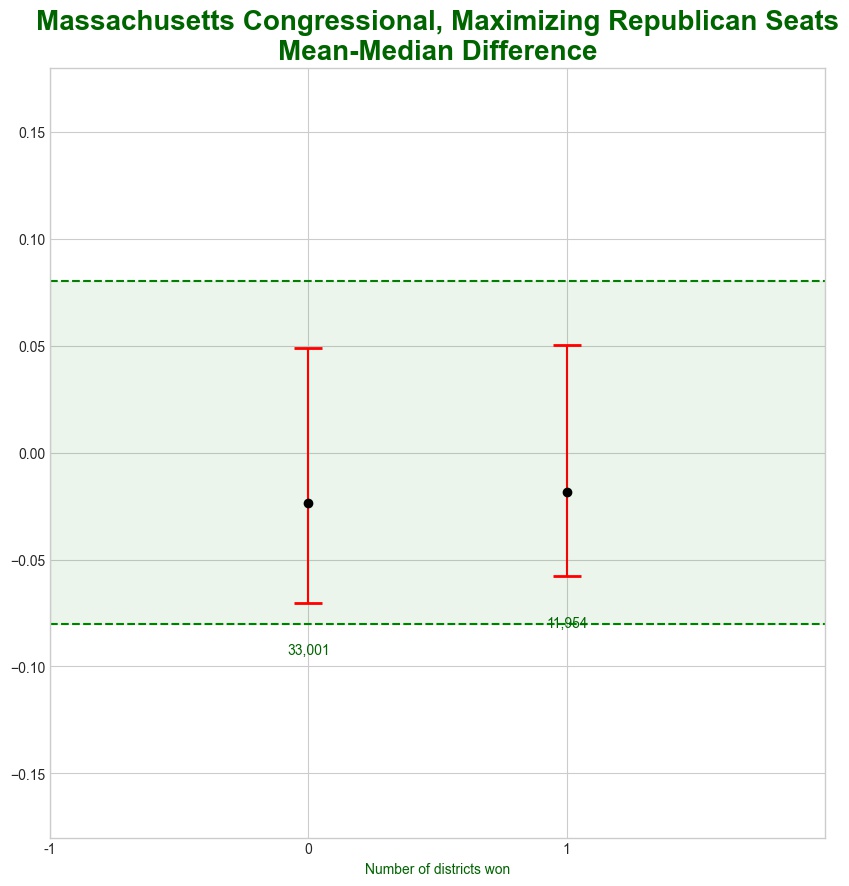}
\includegraphics[width=1.5in]{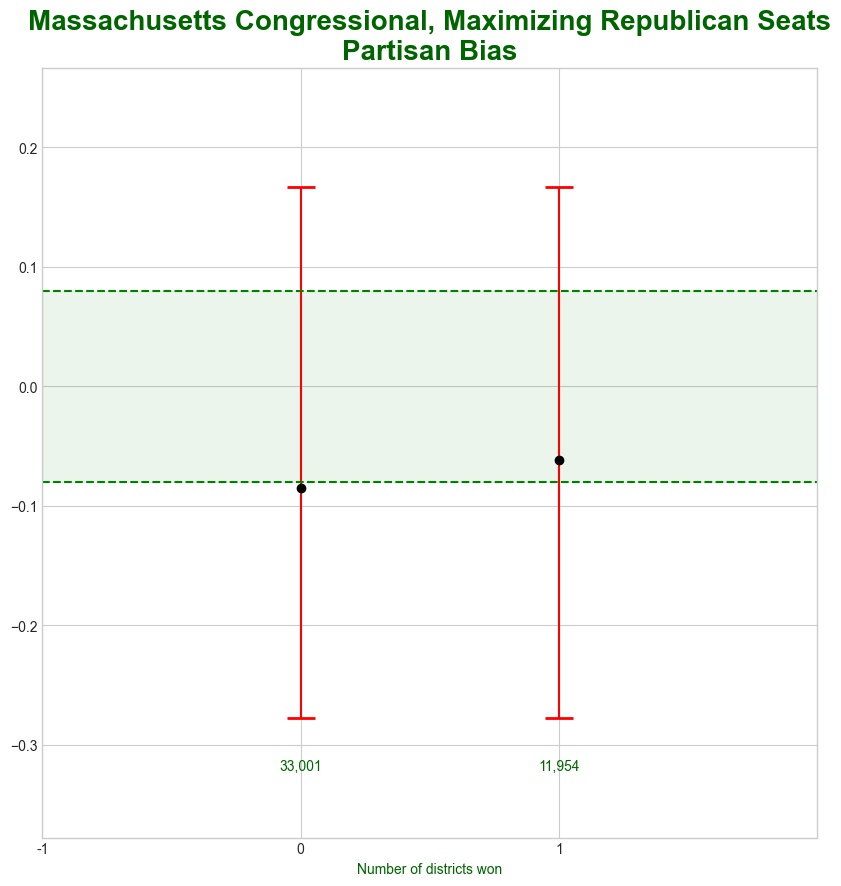}
\caption{Empirical results for Massachusetts's Congressional map and 2018 US Senate election data, searching for maps with as many Republican-won districts as possible.  Horizontal axis is number of districts won, vertical axis is metric value ranges.  The green region is from $0.16\inf(m)$ to $0.16\sup(m)$ for each metric $m$.  The small number below each metric value range is the number of maps produced that had the corresponding number of districts won.  The dot within each vertical bar is the mean value of that metric on all produced maps with the corresponding number of districts won.}
\label{fig:short_bursts_MAcongR}
\end{figure}

\begin{figure}[h]
\centering
\includegraphics[width=1.5in]{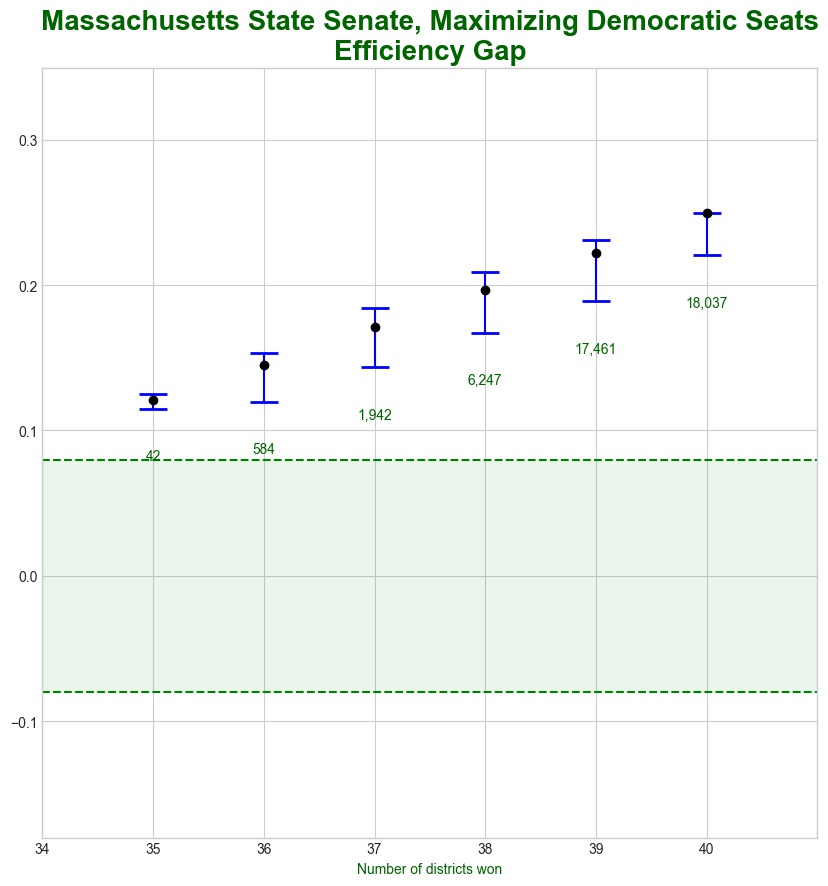}
\includegraphics[width=1.5in]{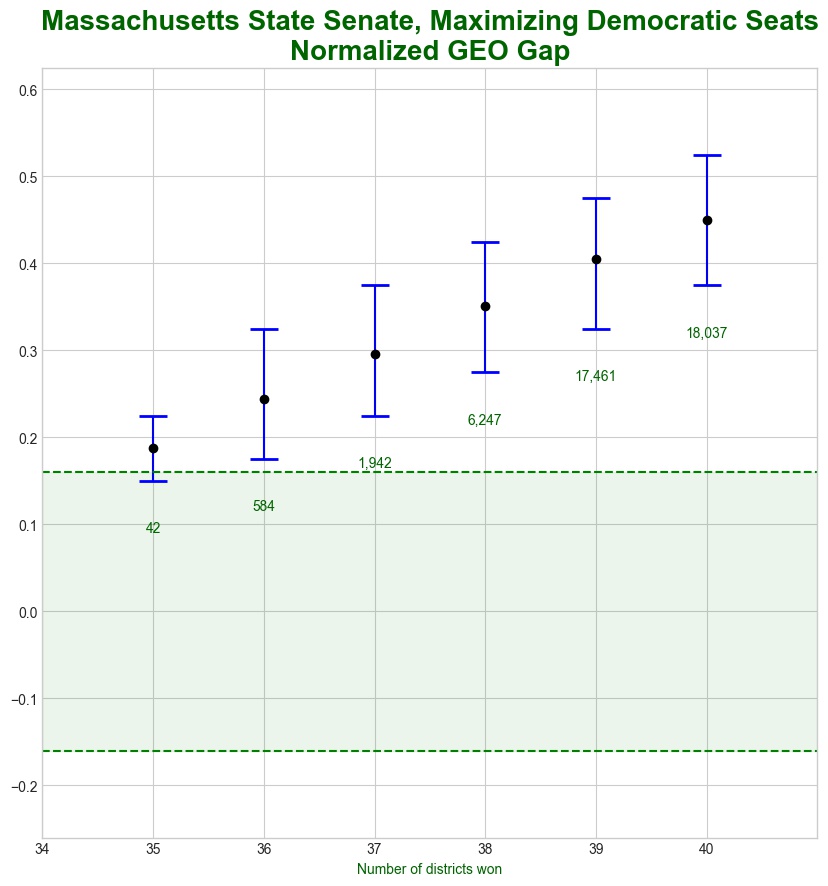}
\includegraphics[width=1.5in]{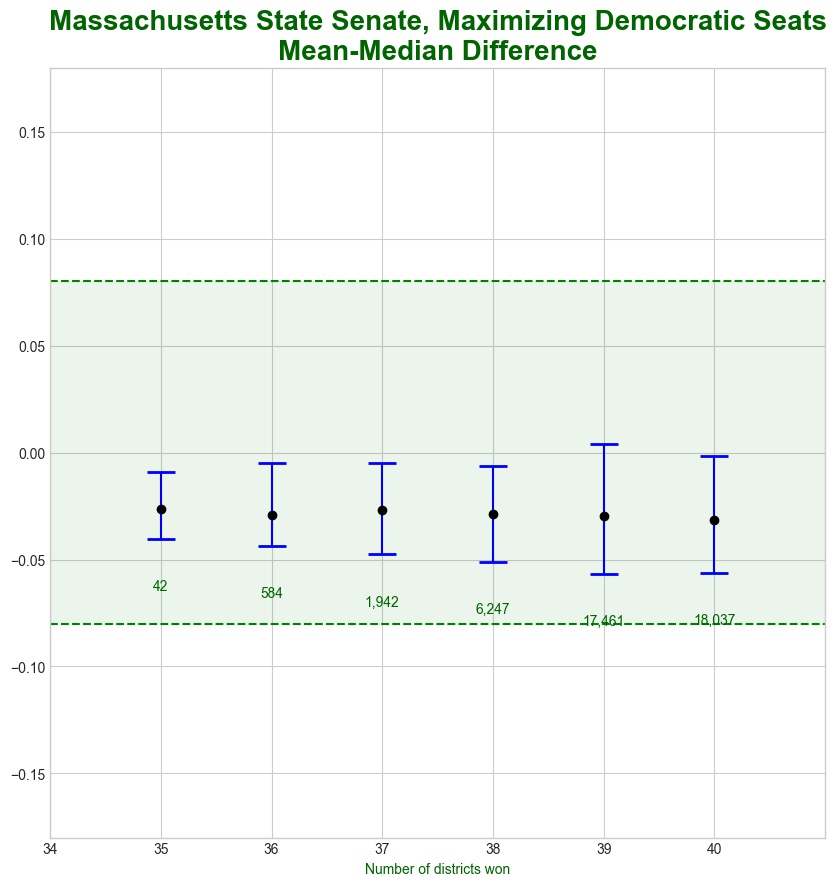}
\includegraphics[width=1.5in]{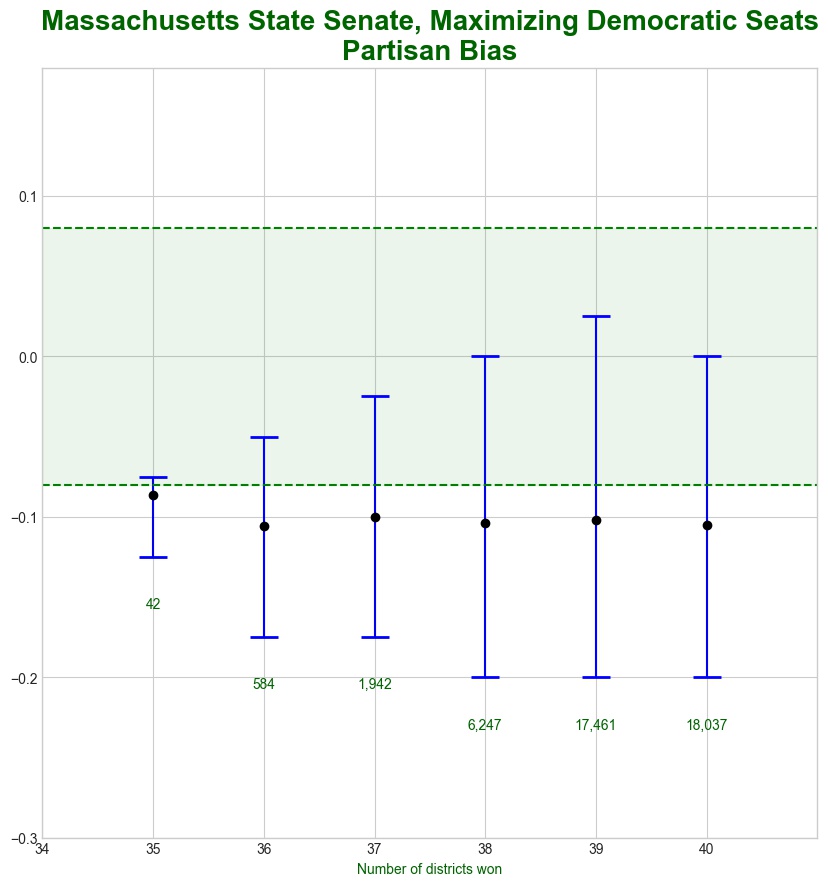}
\caption{Empirical results for Massachusetts's State Senate map and 2018 US Senate election data, searching for maps with as many Democratic-won districts as possible.  Horizontal axis is number of districts won, vertical axis is metric value ranges.  The green region is from $0.16\inf(m)$ to $0.16\sup(m)$ for each metric $m$.  The small number below each metric value range is the number of maps produced that had the corresponding number of districts won.  The dot within each vertical bar is the mean value of that metric on all produced maps with the corresponding number of districts won.}
\label{fig:short_bursts_MAupperD}
\end{figure}

\begin{figure}[h]
\centering
\includegraphics[width=1.5in]{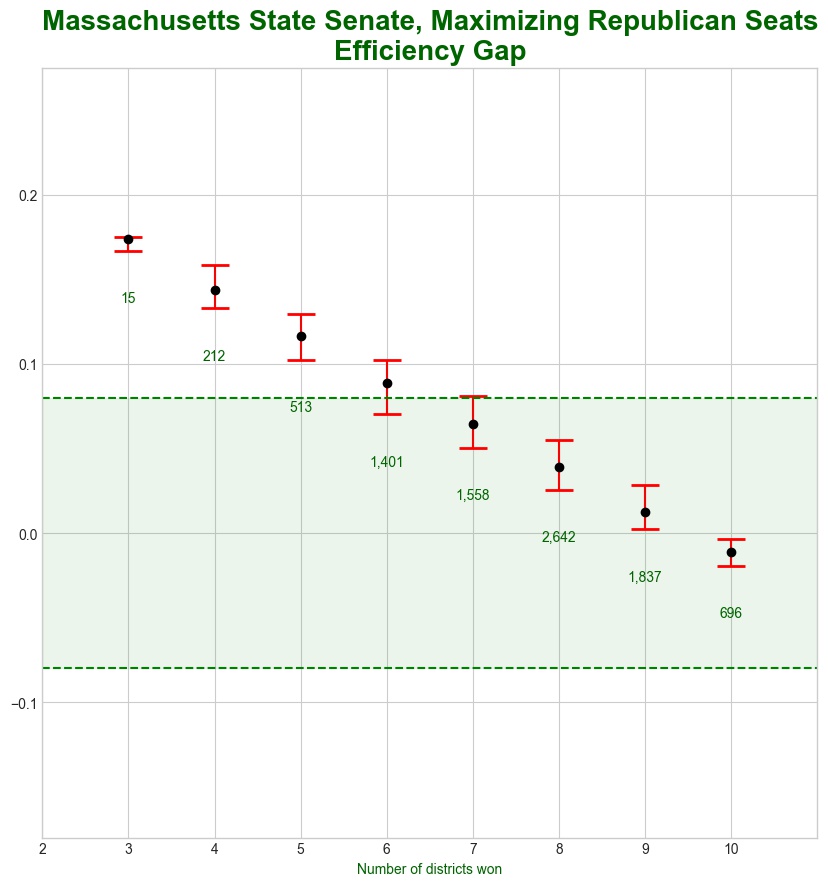}
\includegraphics[width=1.5in]{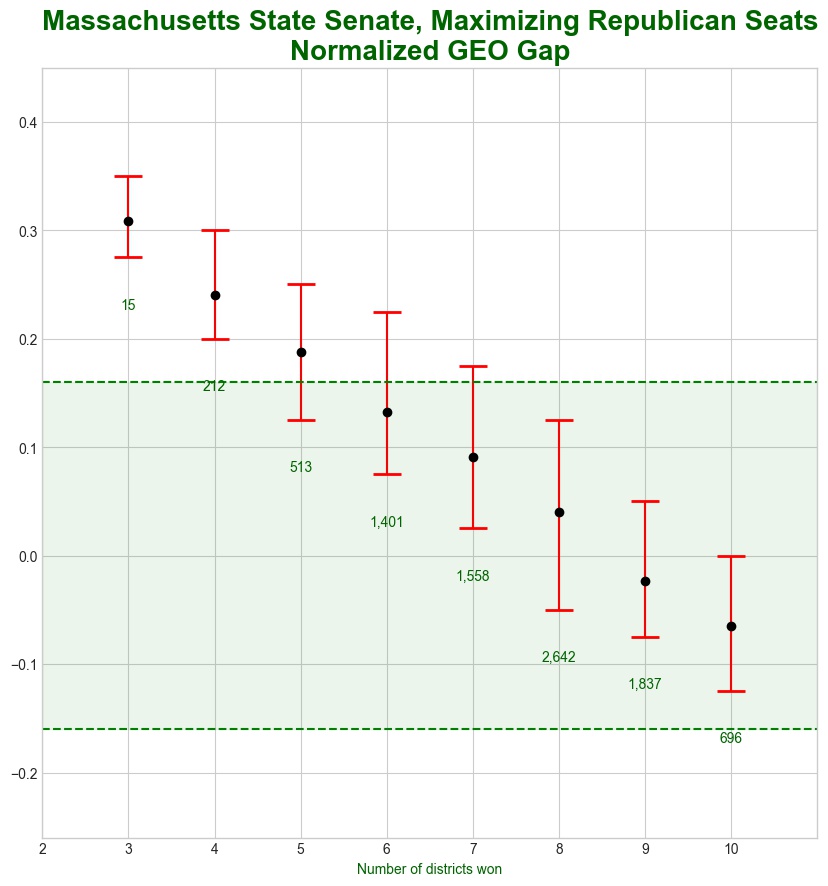}
\includegraphics[width=1.5in]{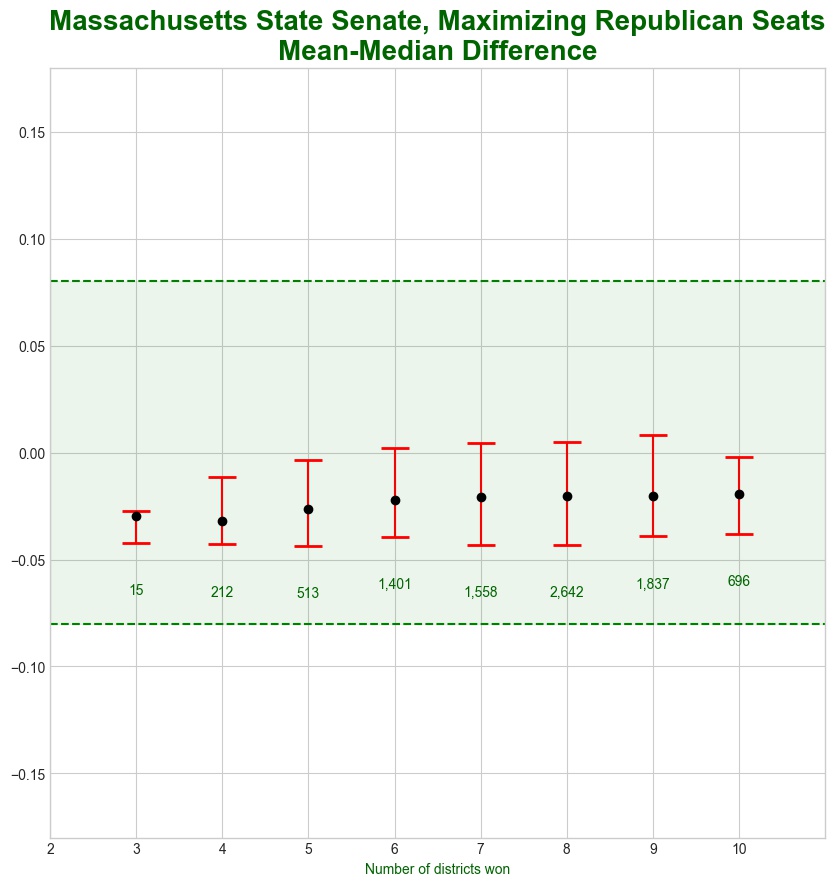}
\includegraphics[width=1.5in]{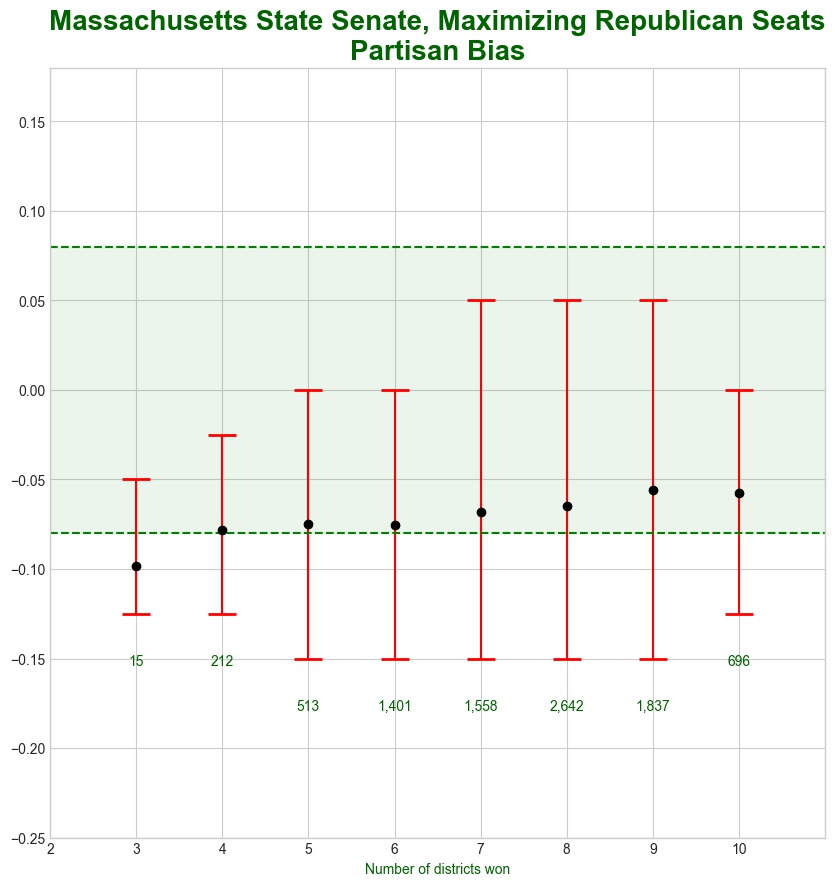}
\caption{Empirical results for Massachusetts's State Senate map and 2018 US Senate election data, searching for maps with as many Republican-won districts as possible.  Horizontal axis is number of districts won, vertical axis is metric value ranges.  The green region is from $0.16\inf(m)$ to $0.16\sup(m)$ for each metric $m$.  The small number below each metric value range is the number of maps produced that had the corresponding number of districts won.  The dot within each vertical bar is the mean value of that metric on all produced maps with the corresponding number of districts won.}
\label{fig:short_bursts_MAupperR}
\end{figure}

\begin{figure}[h]
\centering
\includegraphics[width=1.5in]{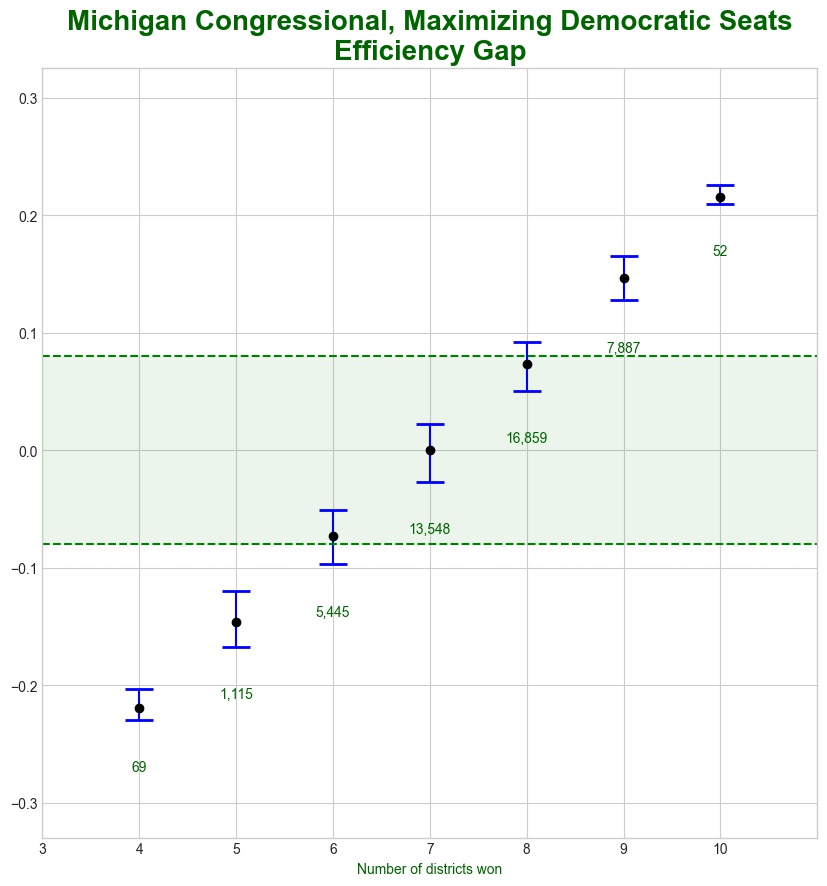}
\includegraphics[width=1.5in]{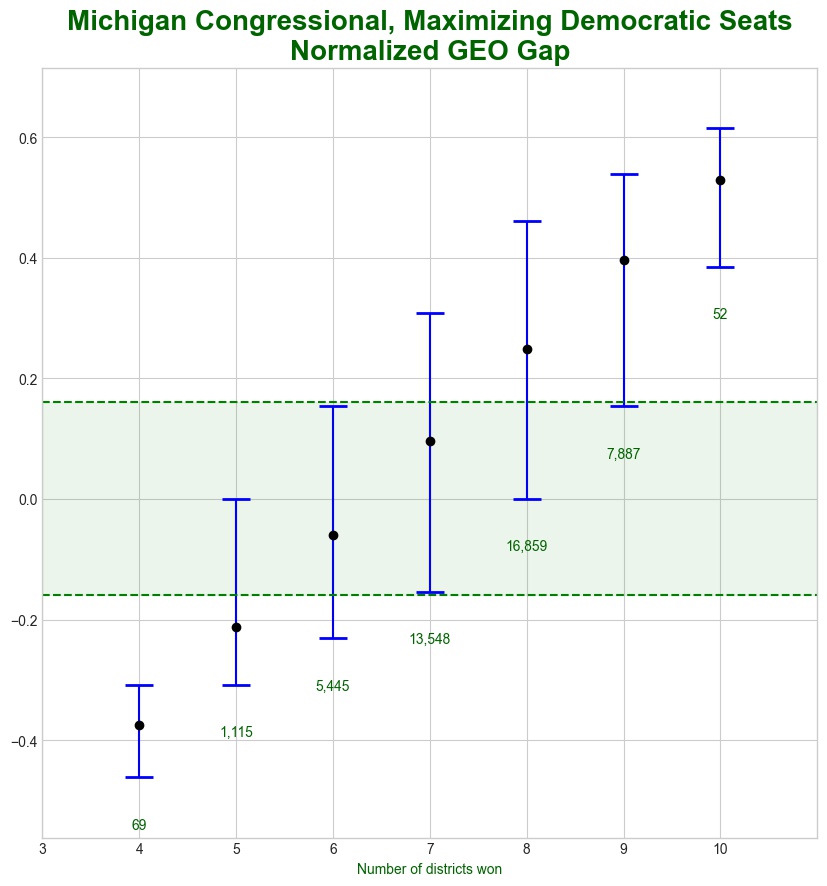}
\includegraphics[width=1.5in]{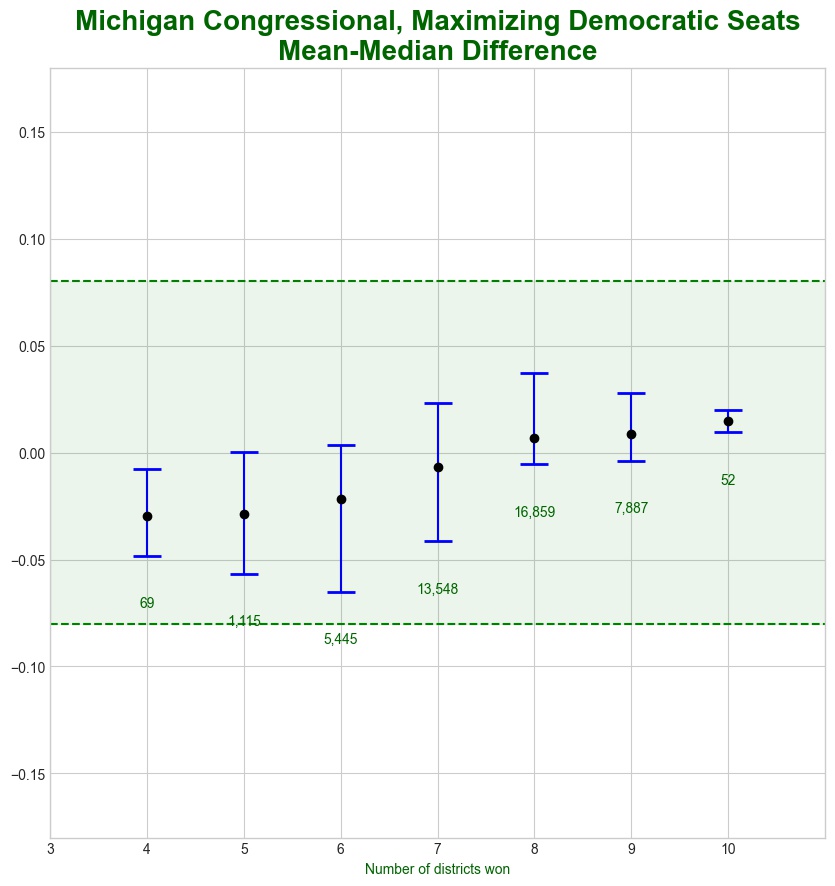}
\includegraphics[width=1.5in]{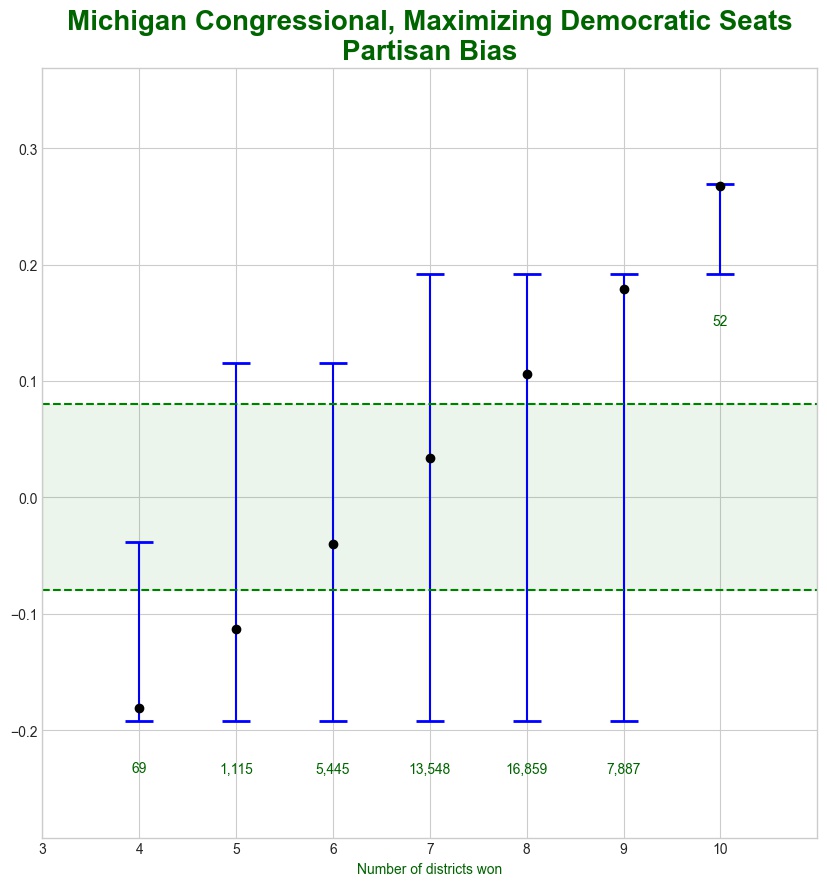}
\caption{Empirical results for Michigan's Congressional map and 2016 Presidential election data, searching for maps with as many Democratic-won districts as possible.  Horizontal axis is number of districts won, vertical axis is metric value ranges.  The green region is from $0.16\inf(m)$ to $0.16\sup(m)$ for each metric $m$.  The small number below each metric value range is the number of maps produced that had the corresponding number of districts won.  The dot within each vertical bar is the mean value of that metric on all produced maps with the corresponding number of districts won.}
\label{fig:short_bursts_MIcongD}
\end{figure}

\begin{figure}[h]
\centering
\includegraphics[width=1.5in]{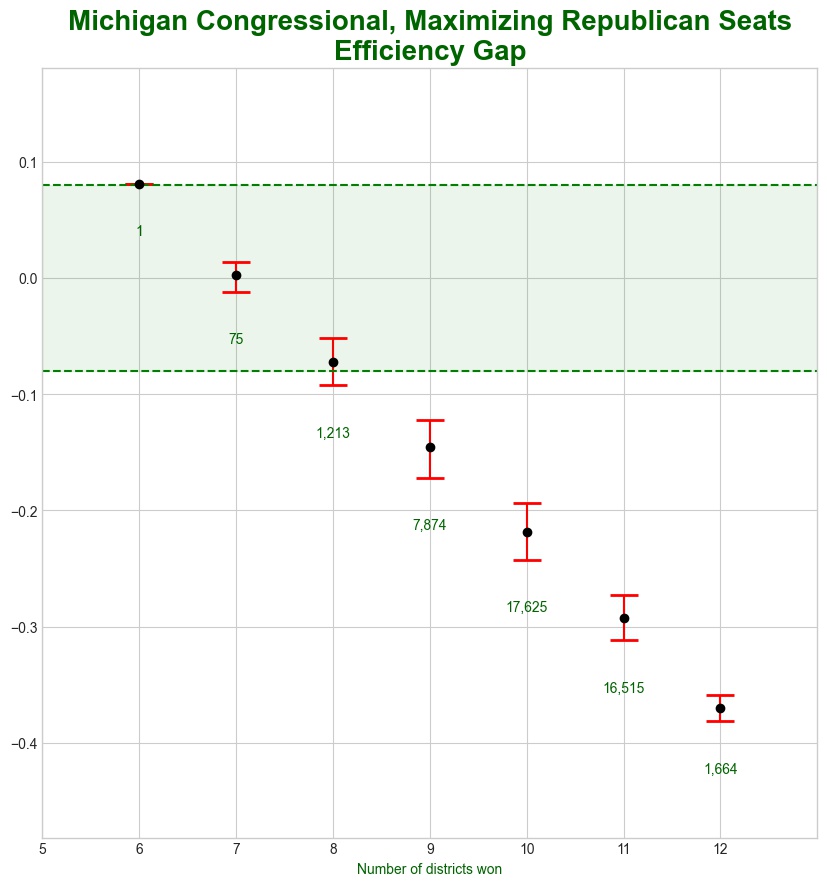}
\includegraphics[width=1.5in]{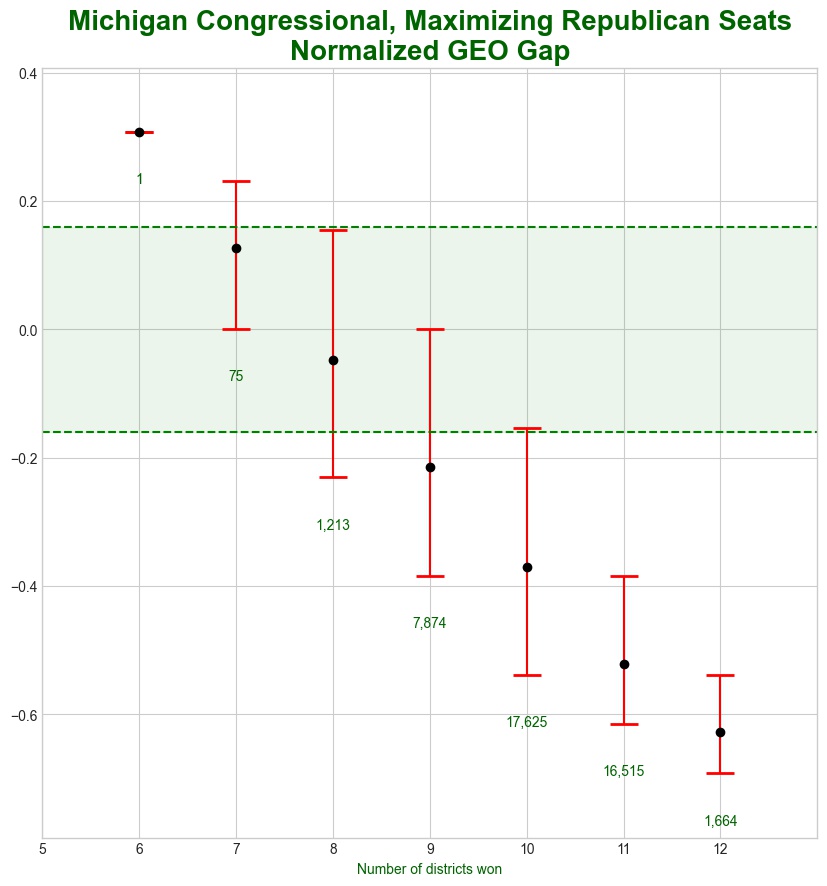}
\includegraphics[width=1.5in]{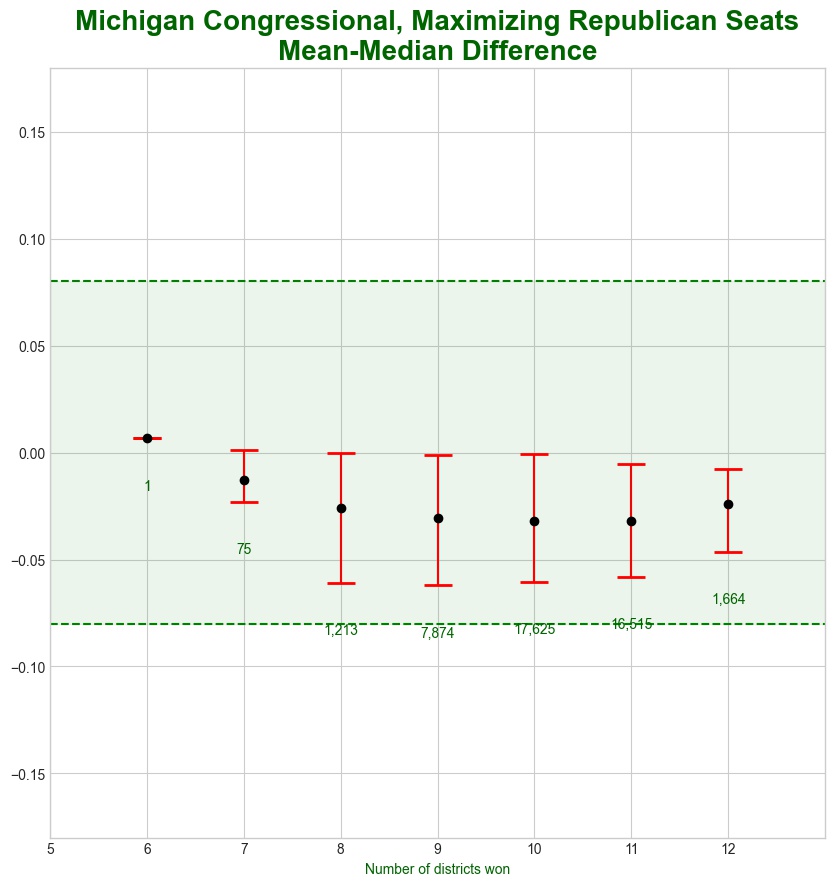}
\includegraphics[width=1.5in]{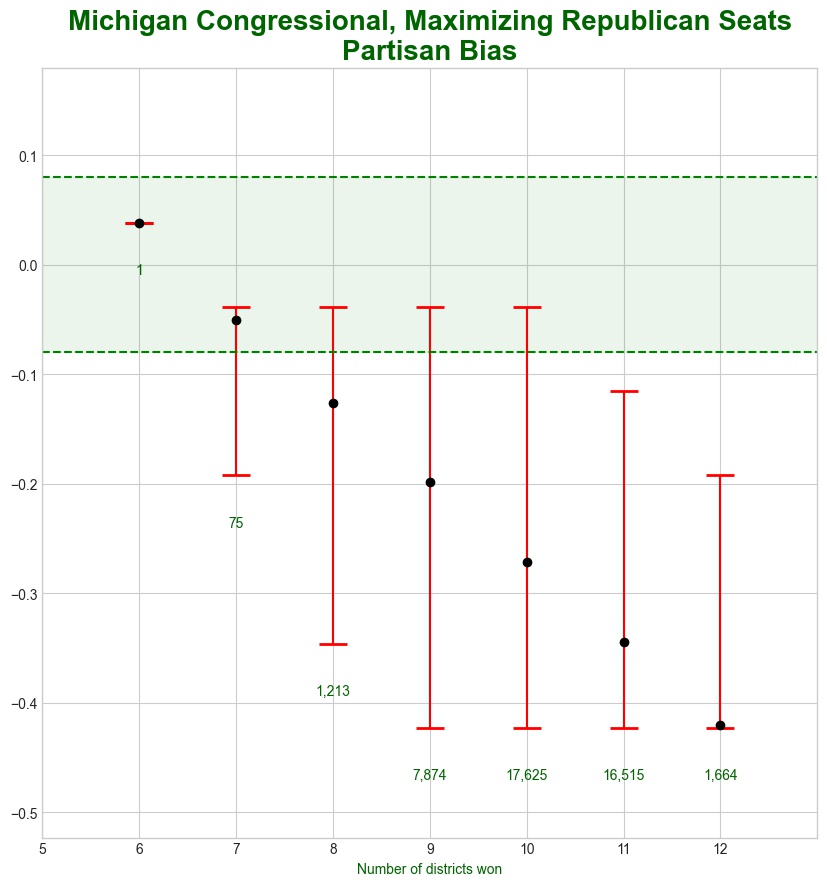}
\caption{Empirical results for Michigan's Congressional map and 2016 Presidential election data, searching for maps with as many Republican-won districts as possible.  Horizontal axis is number of districts won, vertical axis is metric value ranges.  The green region is from $0.16\inf(m)$ to $0.16\sup(m)$ for each metric $m$.  The small number below each metric value range is the number of maps produced that had the corresponding number of districts won.  The dot within each vertical bar is the mean value of that metric on all produced maps with the corresponding number of districts won.}
\label{fig:short_bursts_MIcongR}
\end{figure}

\begin{figure}[h]
\centering
\includegraphics[width=1.5in]{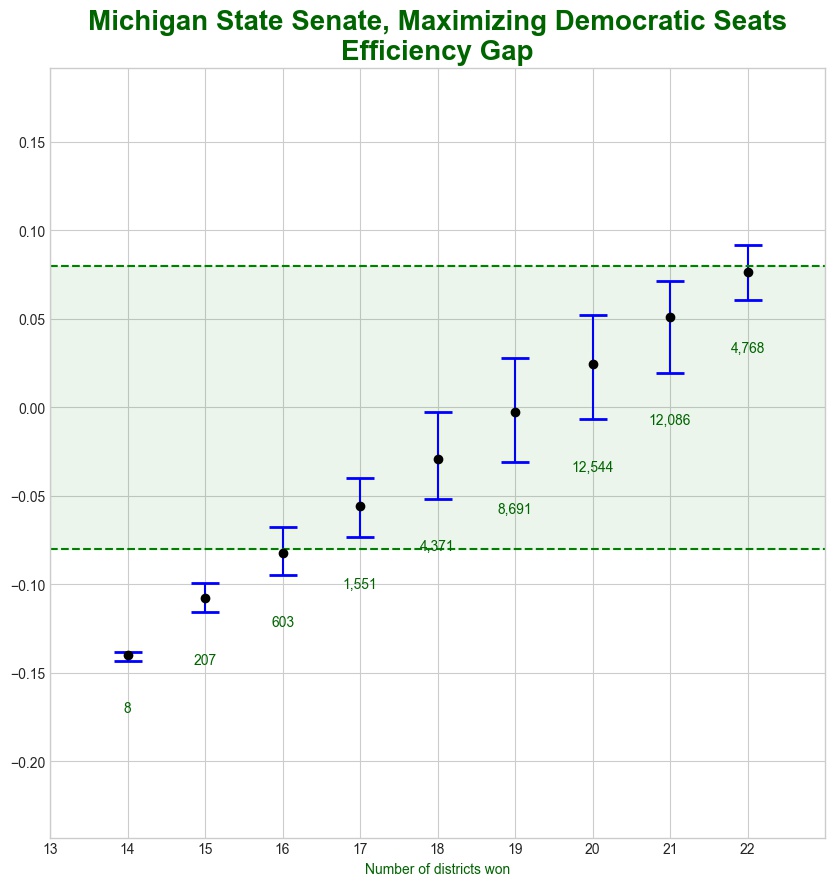}
\includegraphics[width=1.5in]{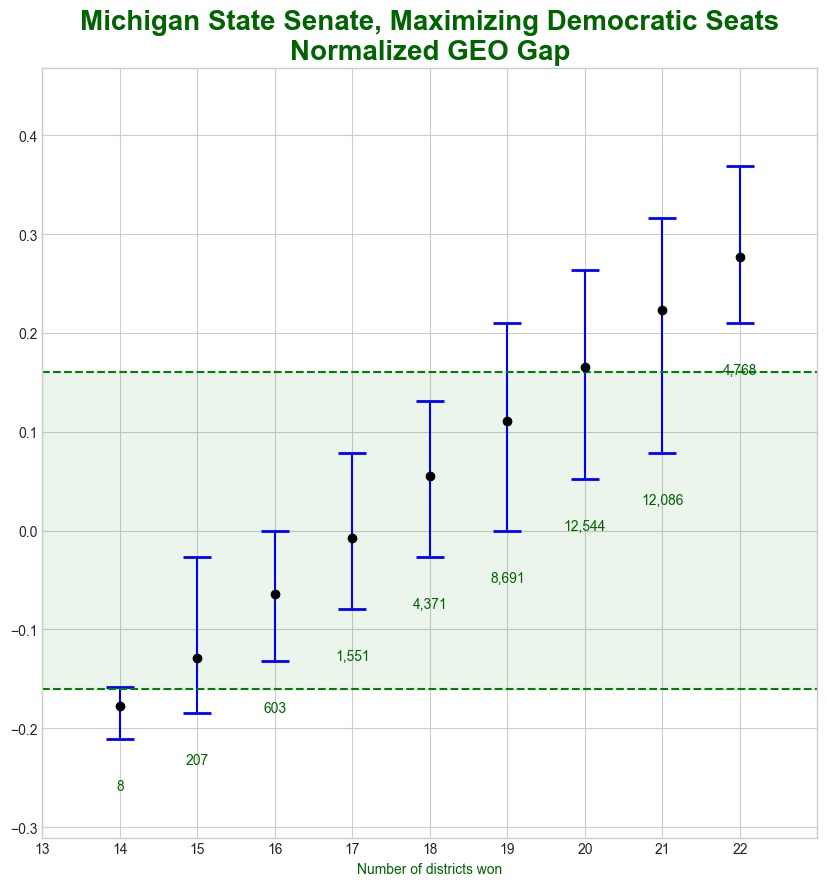}
\includegraphics[width=1.5in]{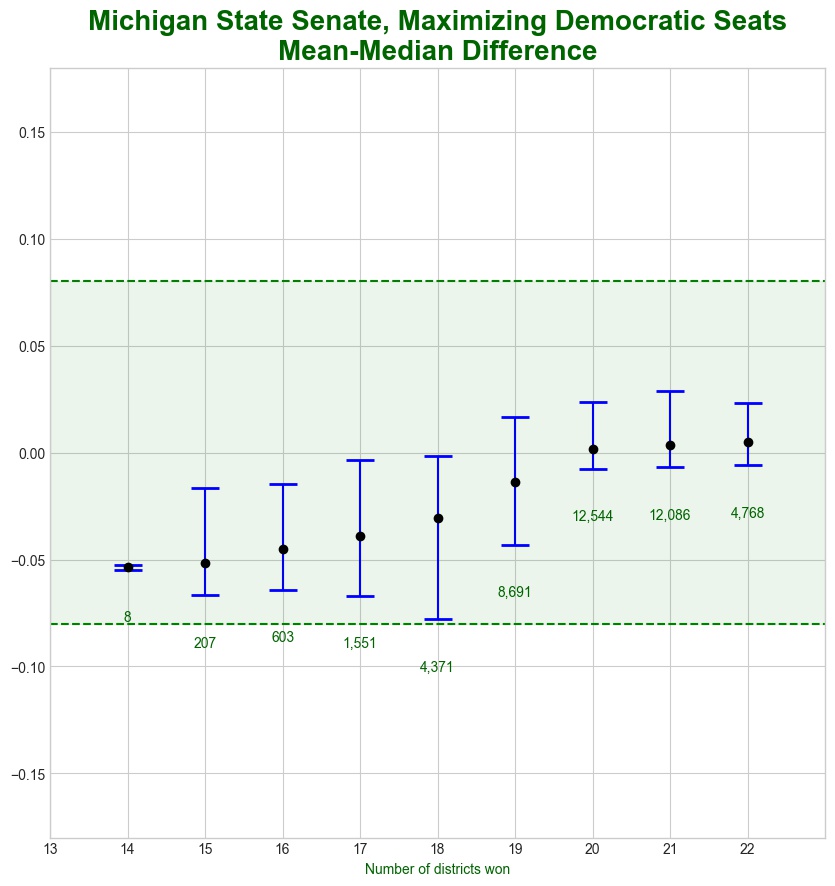}
\includegraphics[width=1.5in]{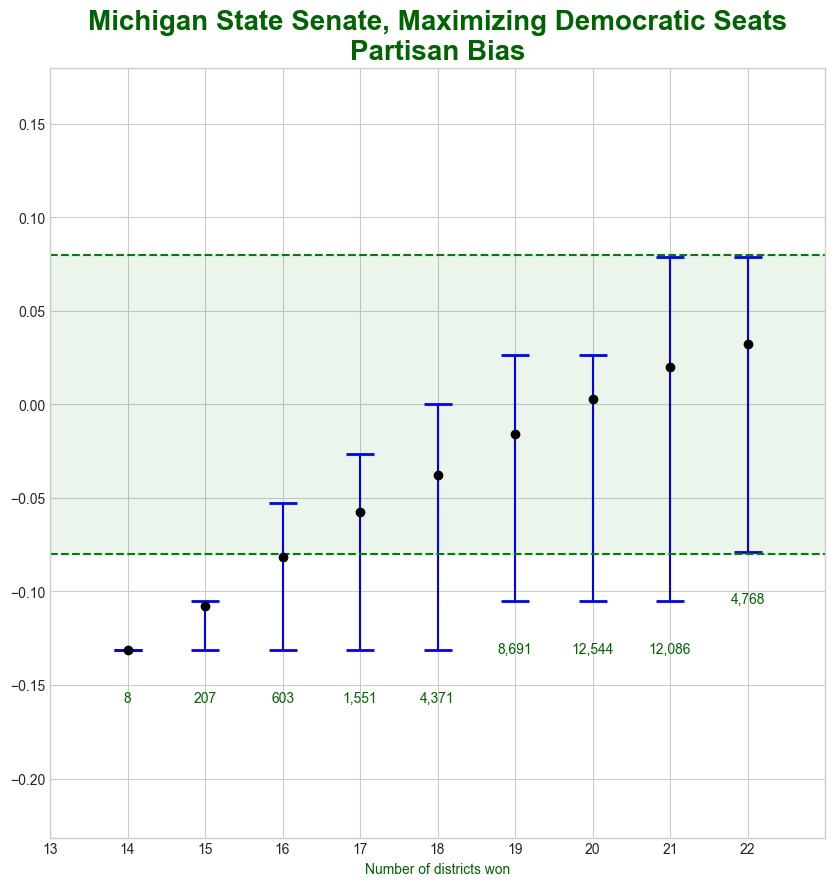}
\caption{Empirical results for Michigan's State Senate map and 2016 Presidential election data, searching for maps with as many Democratic-won districts as possible.  Horizontal axis is number of districts won, vertical axis is metric value ranges.  The green region is from $0.16\inf(m)$ to $0.16\sup(m)$ for each metric $m$.  The small number below each metric value range is the number of maps produced that had the corresponding number of districts won.  The dot within each vertical bar is the mean value of that metric on all produced maps with the corresponding number of districts won.}
\label{fig:short_bursts_MIupperD}
\end{figure}

\begin{figure}[h]
\centering
\includegraphics[width=1.5in]{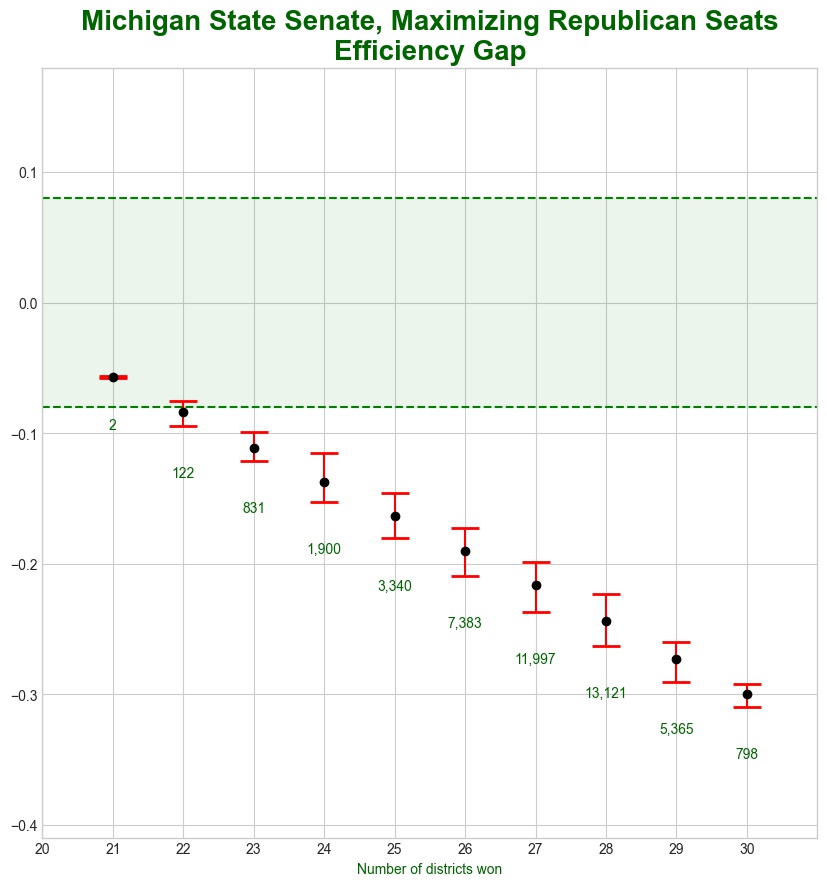}
\includegraphics[width=1.5in]{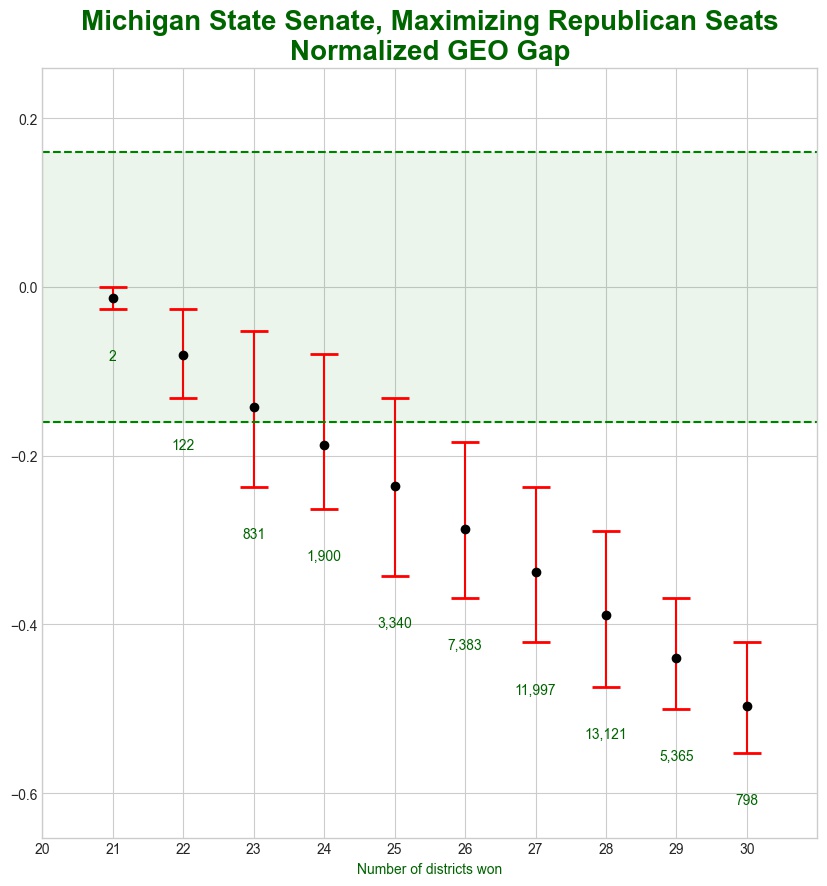}
\includegraphics[width=1.5in]{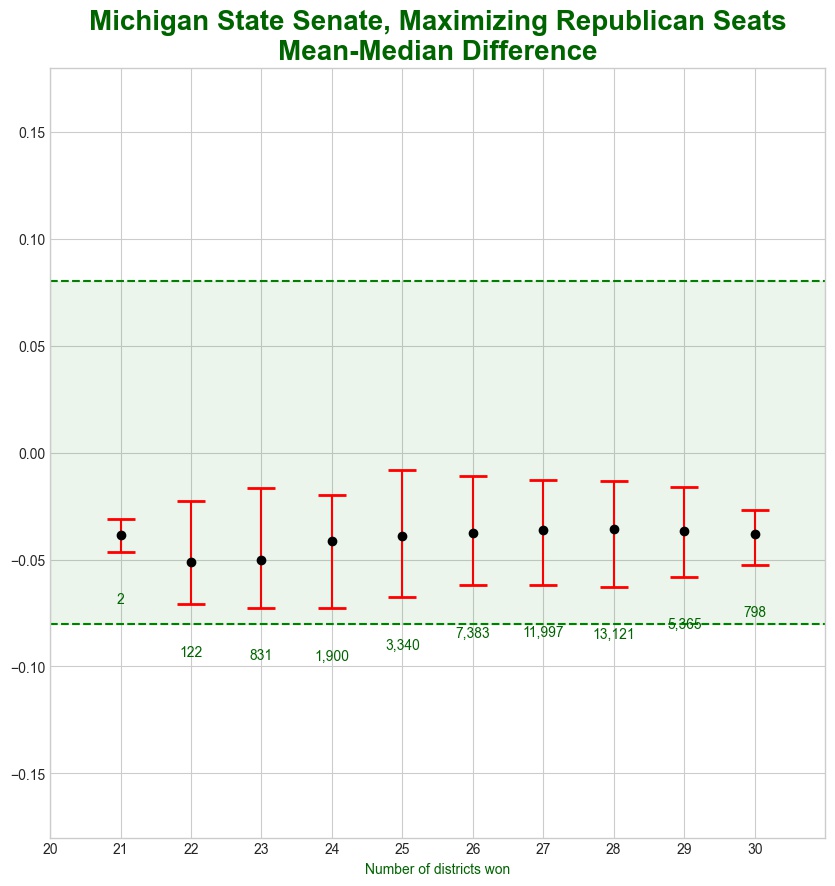}
\includegraphics[width=1.5in]{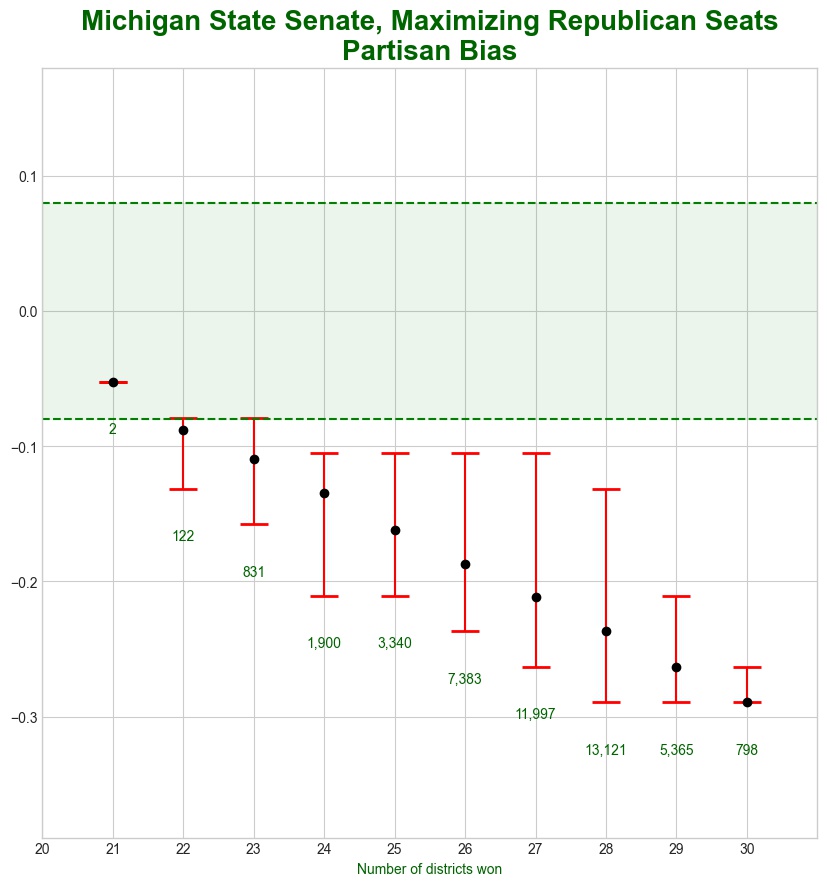}
\caption{Empirical results for Michigan's State Senate map and 2016 Presidential election data, searching for maps with as many Republican-won districts as possible.  Horizontal axis is number of districts won, vertical axis is metric value ranges.  The green region is from $0.16\inf(m)$ to $0.16\sup(m)$ for each metric $m$.  The small number below each metric value range is the number of maps produced that had the corresponding number of districts won.  The dot within each vertical bar is the mean value of that metric on all produced maps with the corresponding number of districts won.}
\label{fig:short_bursts_MIupperR}
\end{figure}

\begin{figure}[h]
\centering
\includegraphics[width=1.5in]{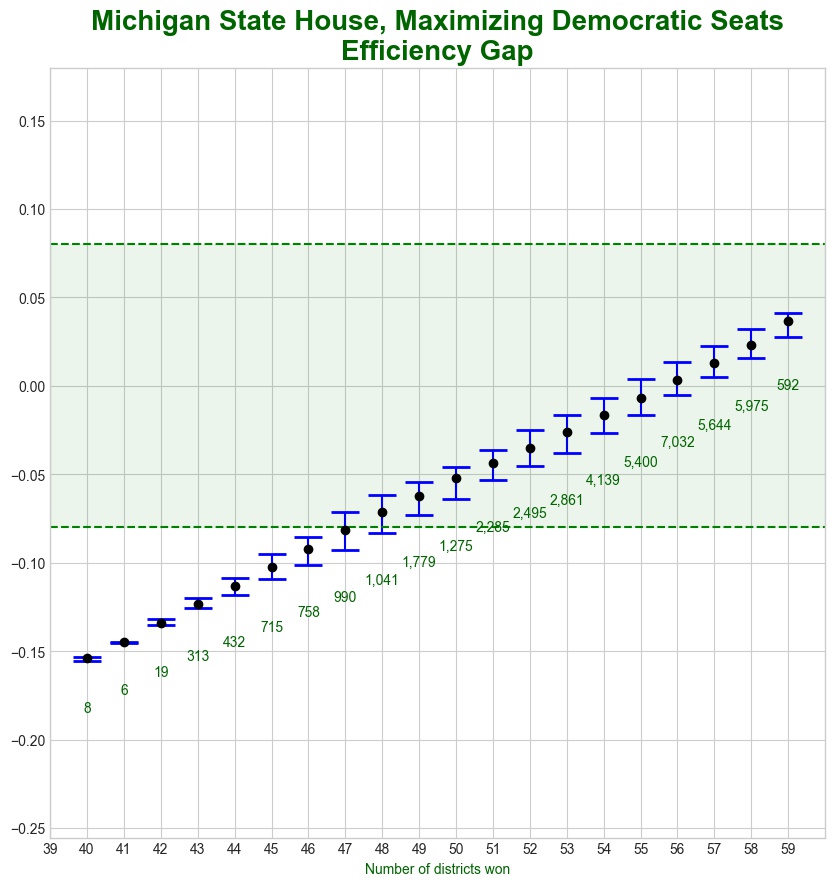}
\includegraphics[width=1.5in]{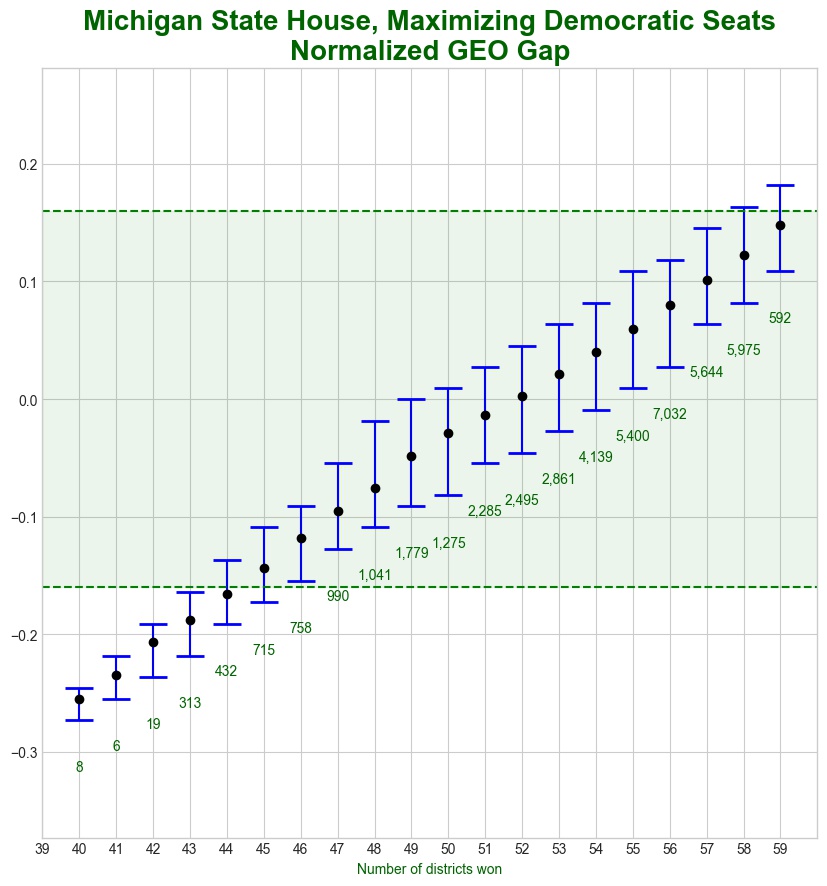}
\includegraphics[width=1.5in]{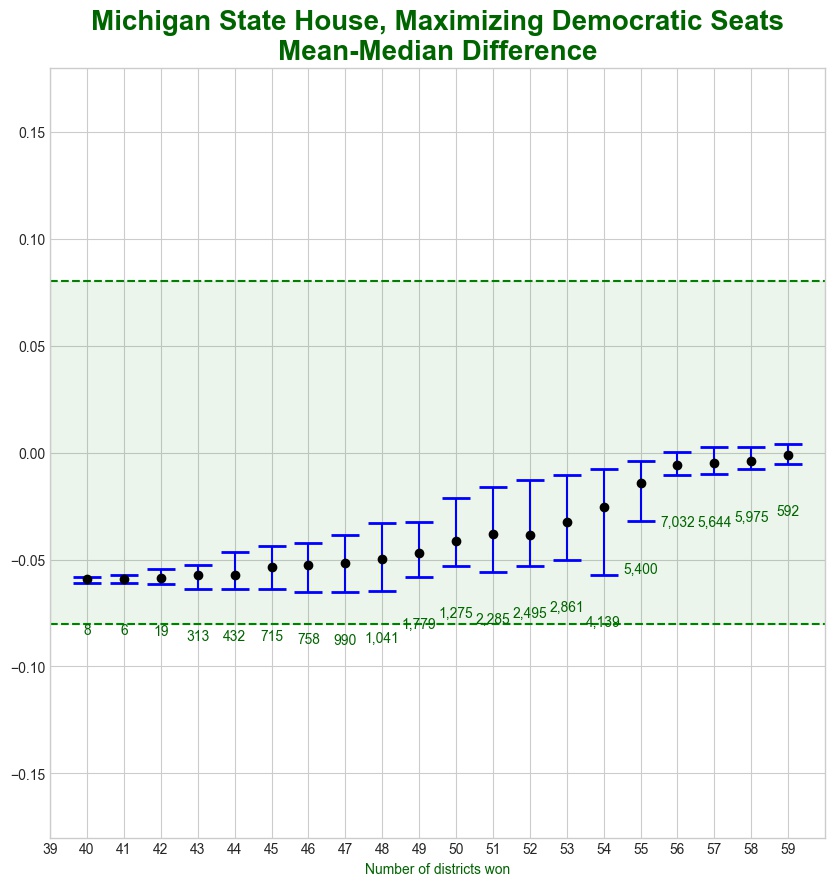}
\includegraphics[width=1.5in]{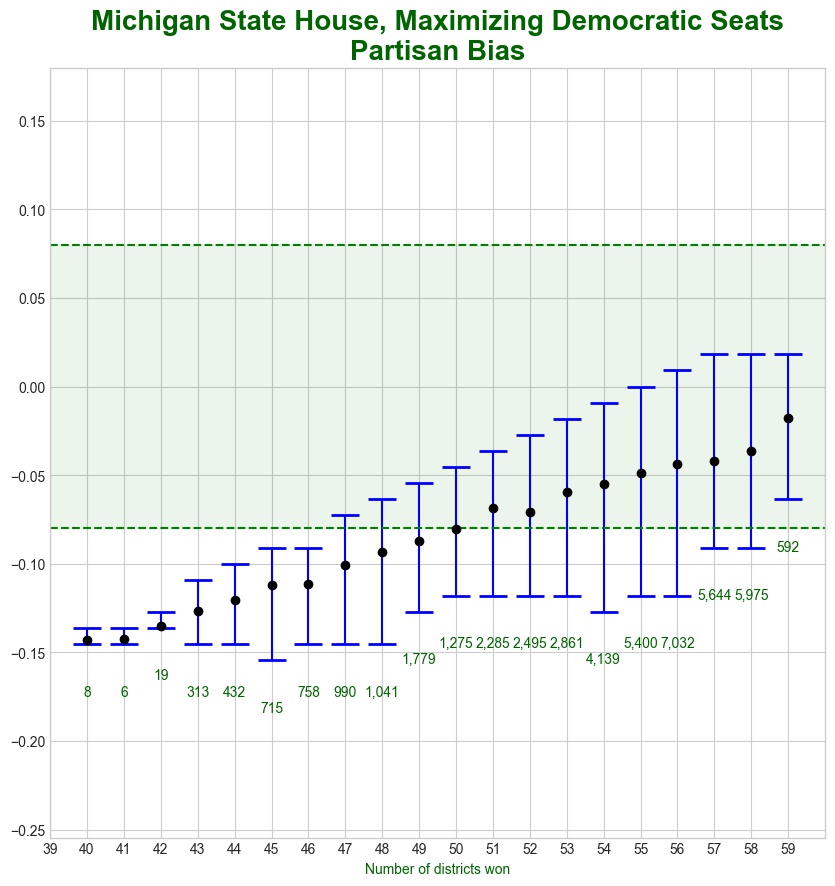}
\caption{Empirical results for Michigan's State House map and 2016 Presidential election data, searching for maps with as many Democratic-won districts as possible.  Horizontal axis is number of districts won, vertical axis is metric value ranges.  The green region is from $0.16\inf(m)$ to $0.16\sup(m)$ for each metric $m$.  The small number below each metric value range is the number of maps produced that had the corresponding number of districts won.  The dot within each vertical bar is the mean value of that metric on all produced maps with the corresponding number of districts won.}
\label{fig:short_bursts_MIlowerD}
\end{figure}

\begin{figure}[h]
\centering
\includegraphics[width=1.5in]{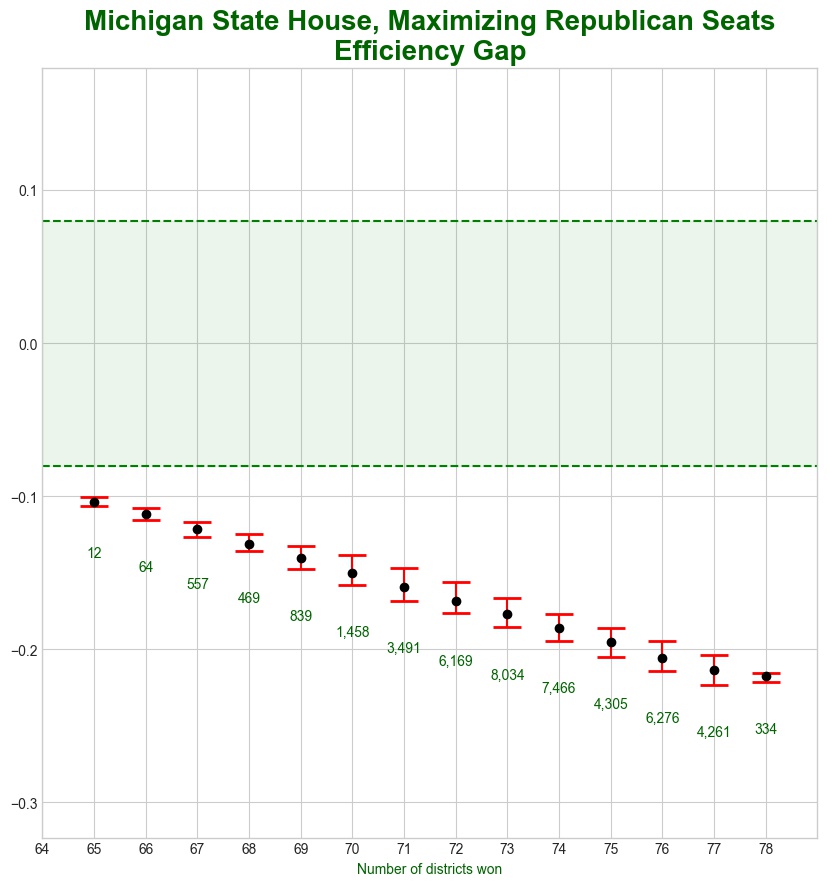}
\includegraphics[width=1.5in]{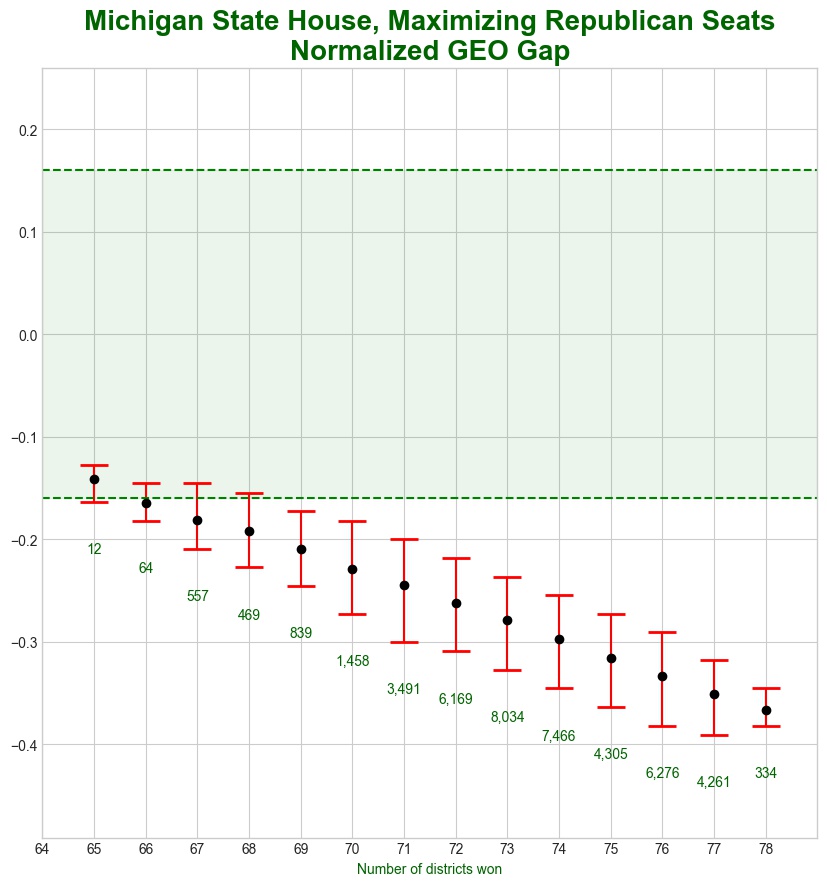}
\includegraphics[width=1.5in]{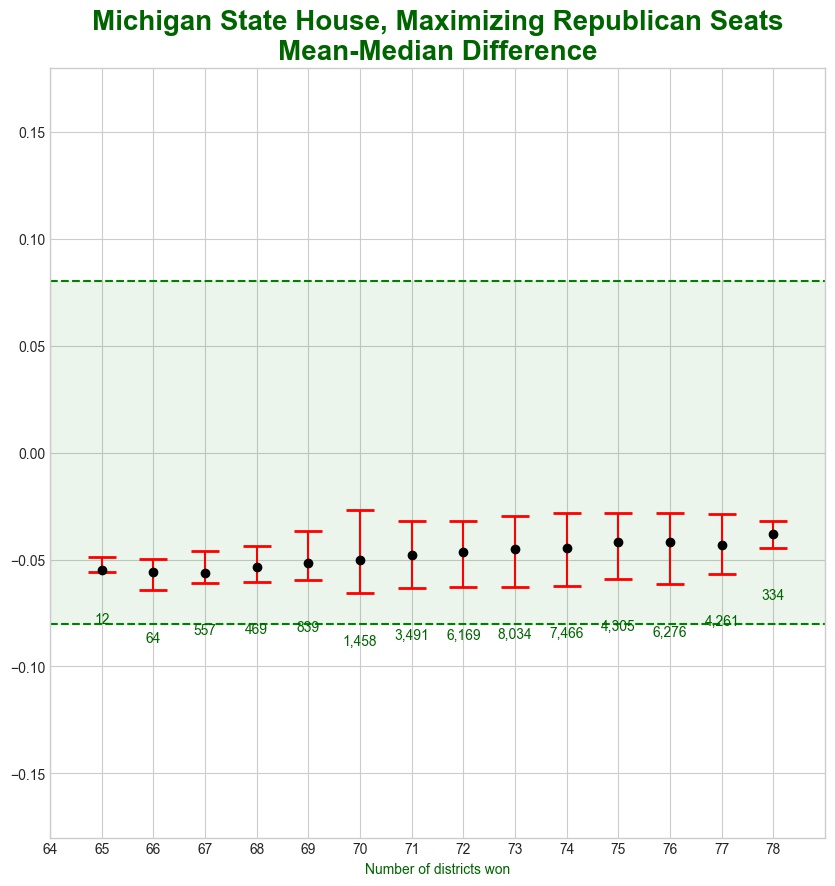}
\includegraphics[width=1.5in]{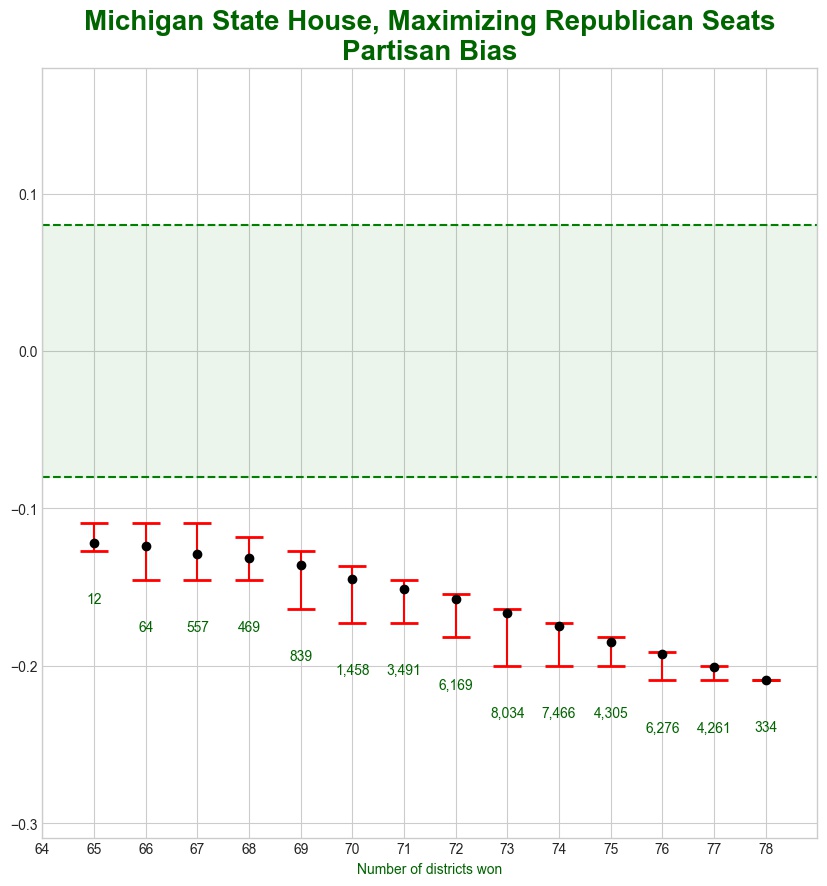}
\caption{Empirical results for Michigan's State House map and 2016 Presidential election data, searching for maps with as many Republican-won districts as possible.  Horizontal axis is number of districts won, vertical axis is metric value ranges.  The green region is from $0.16\inf(m)$ to $0.16\sup(m)$ for each metric $m$.  The small number below each metric value range is the number of maps produced that had the corresponding number of districts won.  The dot within each vertical bar is the mean value of that metric on all produced maps with the corresponding number of districts won.}
\label{fig:short_bursts_MIlowerR}
\end{figure}

\begin{figure}[h]
\centering
\includegraphics[width=1.5in]{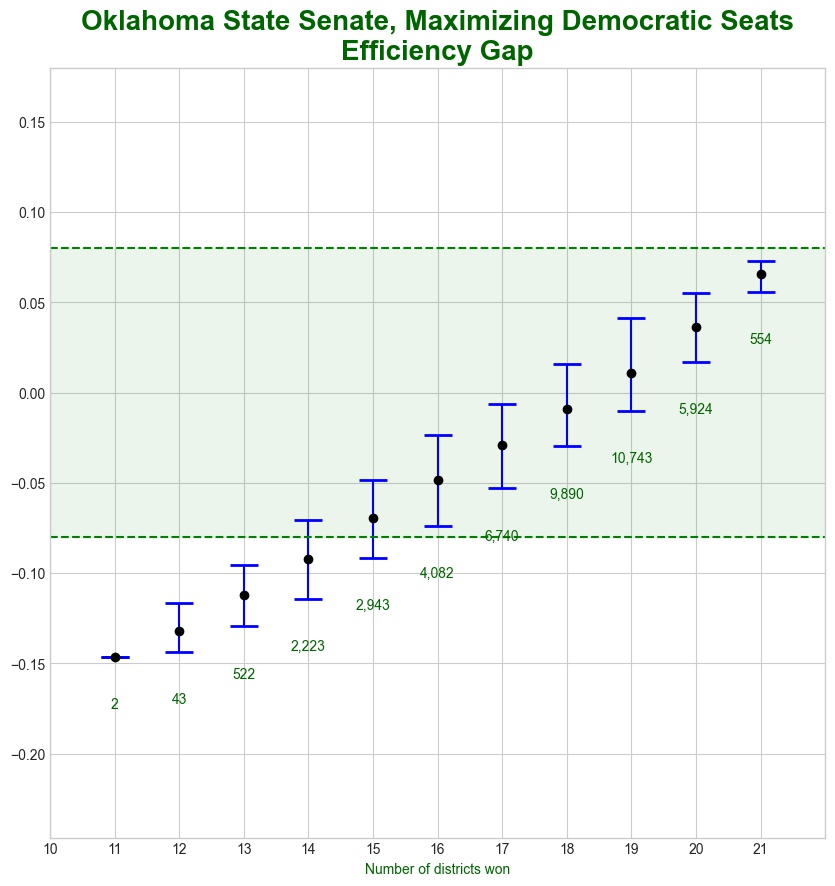}
\includegraphics[width=1.5in]{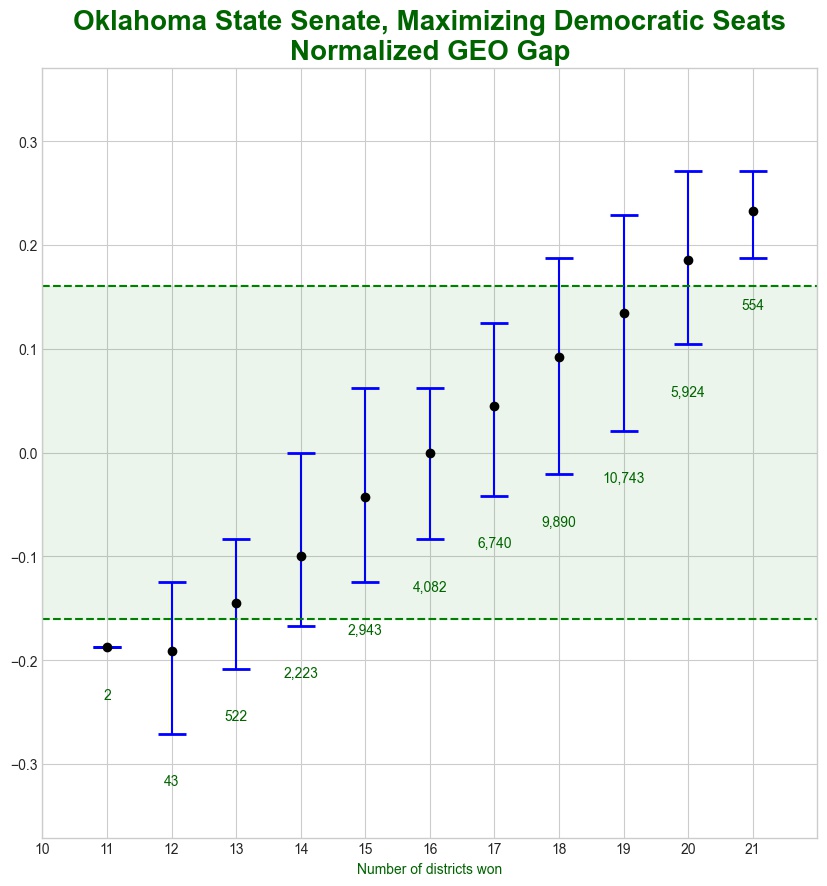}
\includegraphics[width=1.5in]{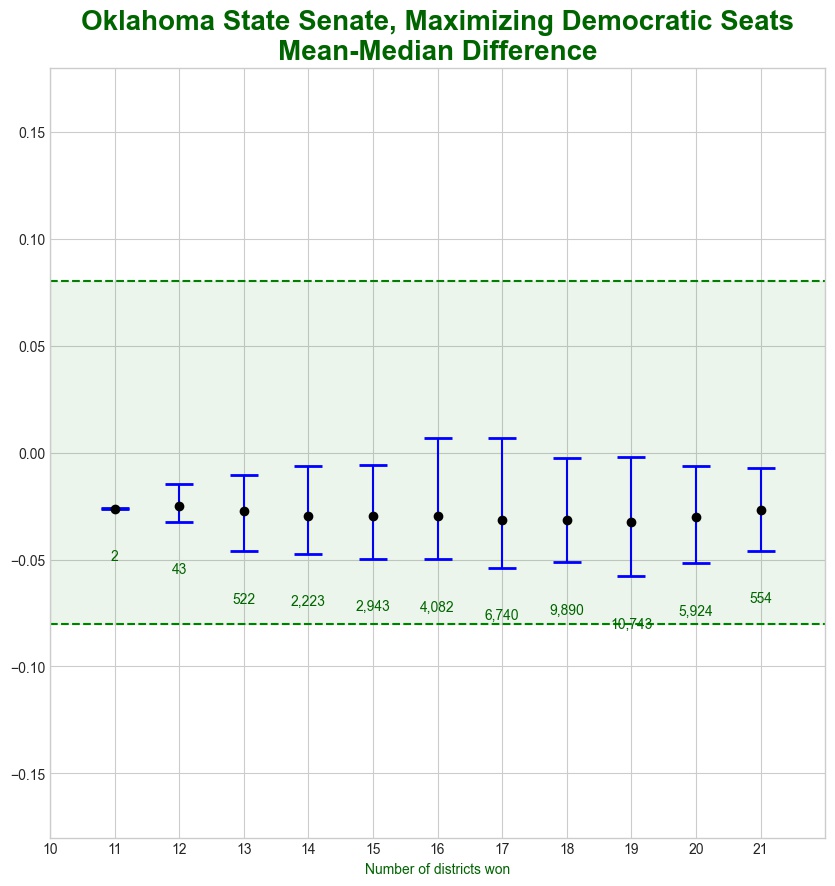}
\includegraphics[width=1.5in]{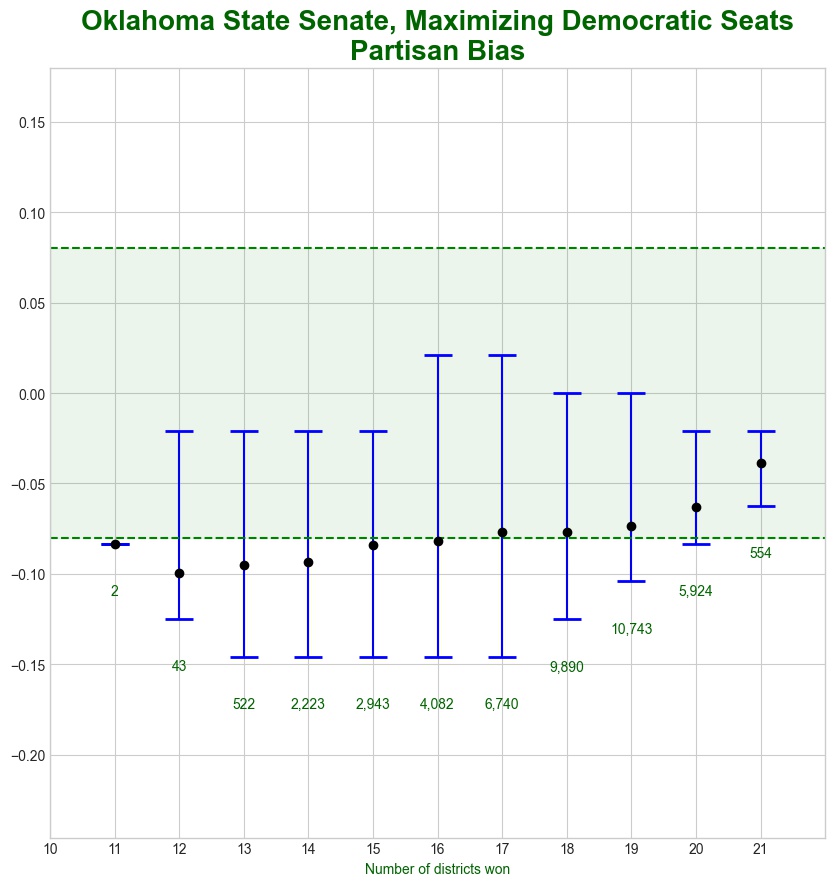}
\caption{Empirical results for Oklahoma's State Senate map and 2018 Gubenatorial election data, searching for maps with as many Democratic-won districts as possible.  Horizontal axis is number of districts won, vertical axis is metric value ranges.  The green region is from $0.16\inf(m)$ to $0.16\sup(m)$ for each metric $m$.  The small number below each metric value range is the number of maps produced that had the corresponding number of districts won.  The dot within each vertical bar is the mean value of that metric on all produced maps with the corresponding number of districts won.}
\label{fig:short_bursts_OKupperD}
\end{figure}

\begin{figure}[h]
\centering
\includegraphics[width=1.5in]{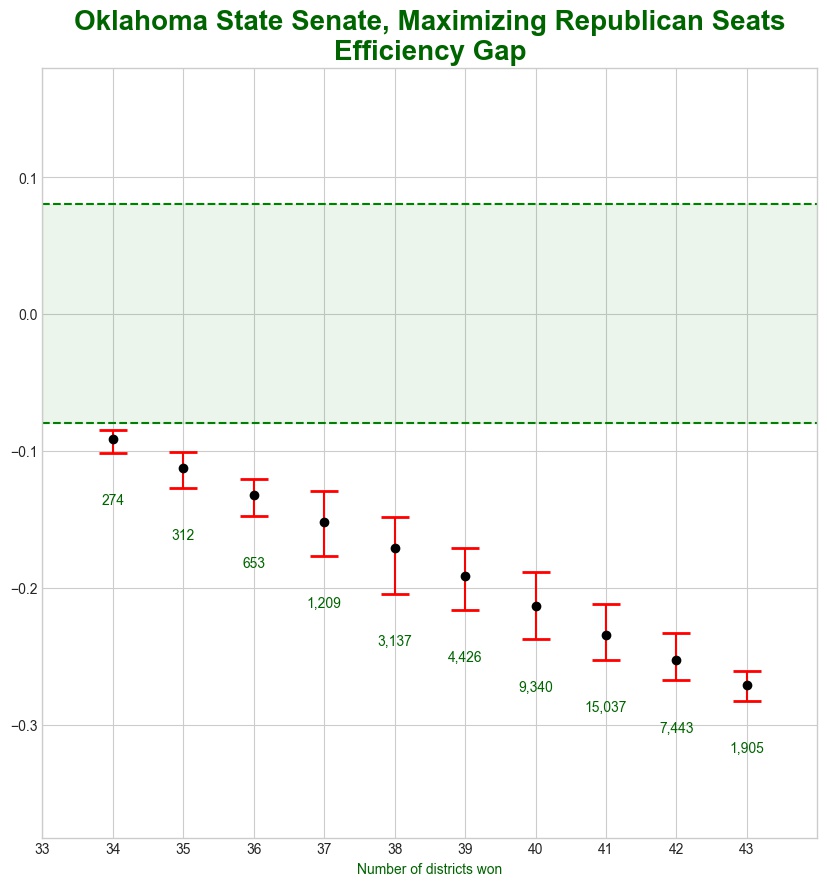}
\includegraphics[width=1.5in]{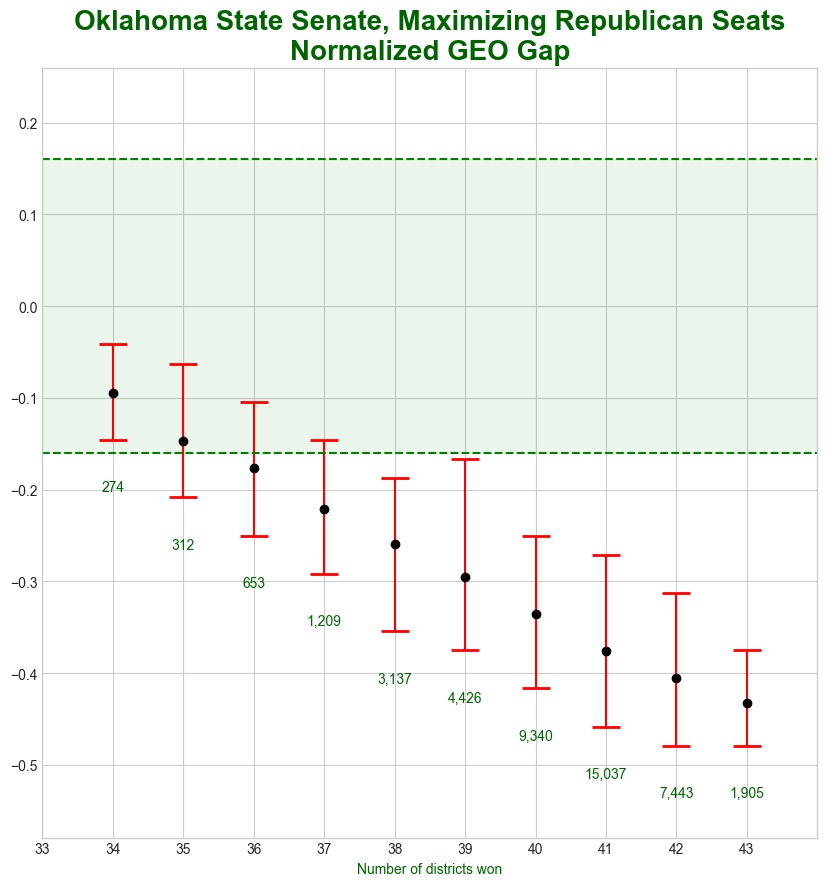}
\includegraphics[width=1.5in]{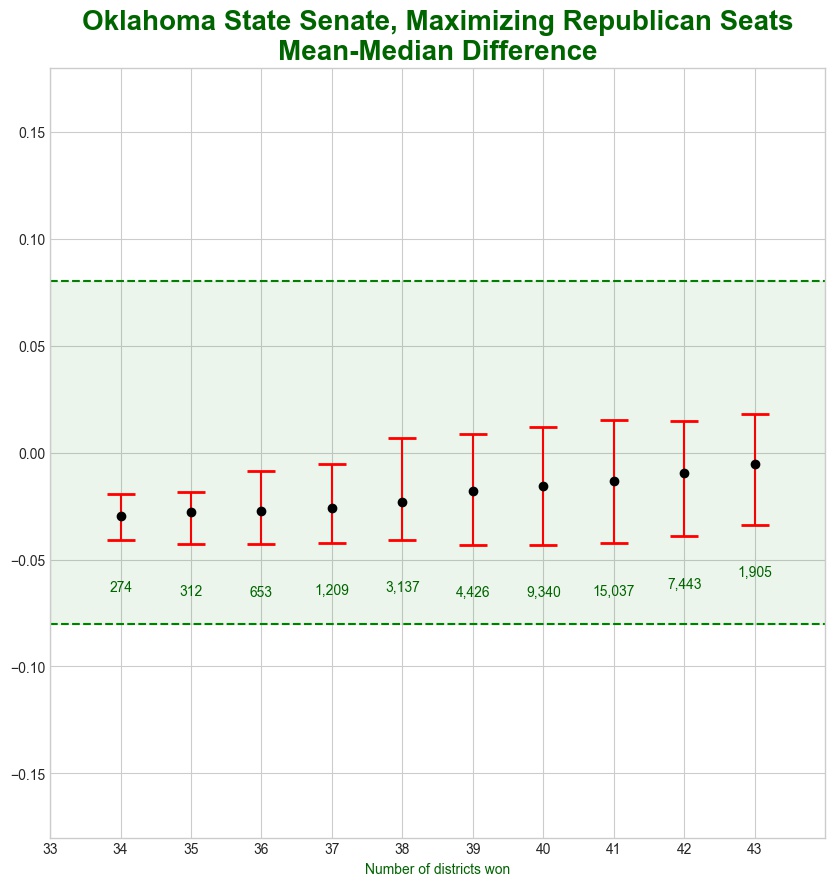}
\includegraphics[width=1.5in]{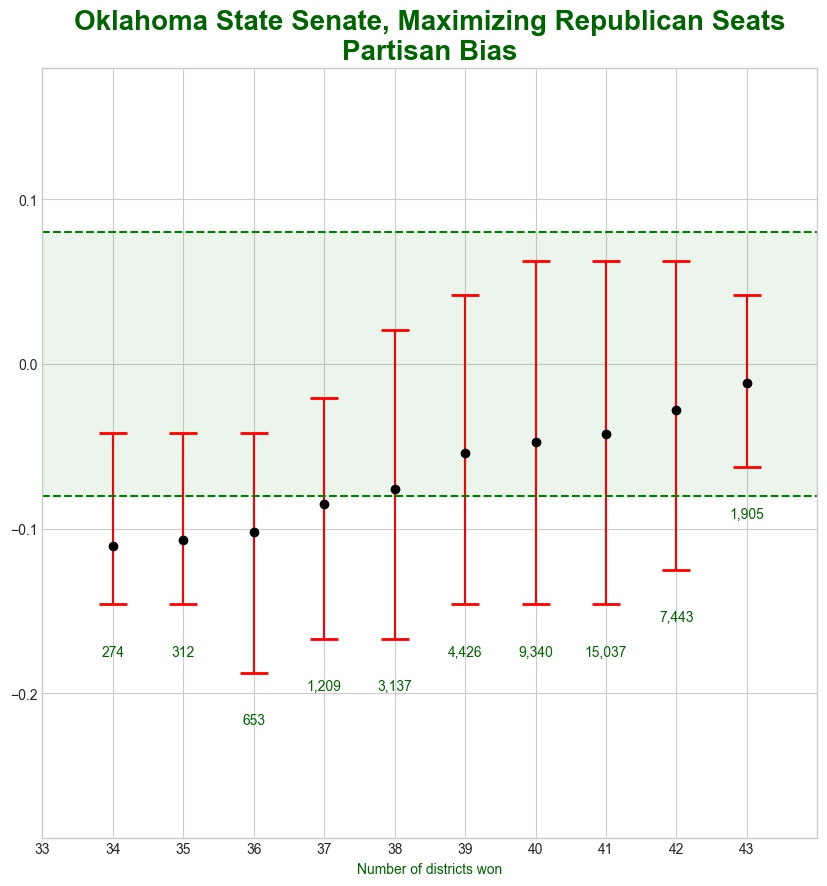}
\caption{Empirical results for Oklahoma's State Senate map and 2018 Gubenatorial election data, searching for maps with as many Republican-won districts as possible.  Horizontal axis is number of districts won, vertical axis is metric value ranges.  The green region is from $0.16\inf(m)$ to $0.16\sup(m)$ for each metric $m$.  The small number below each metric value range is the number of maps produced that had the corresponding number of districts won.  The dot within each vertical bar is the mean value of that metric on all produced maps with the corresponding number of districts won.}
\label{fig:short_bursts_OKupperR}
\end{figure}

\begin{figure}[h]
\centering
\includegraphics[width=1.5in]{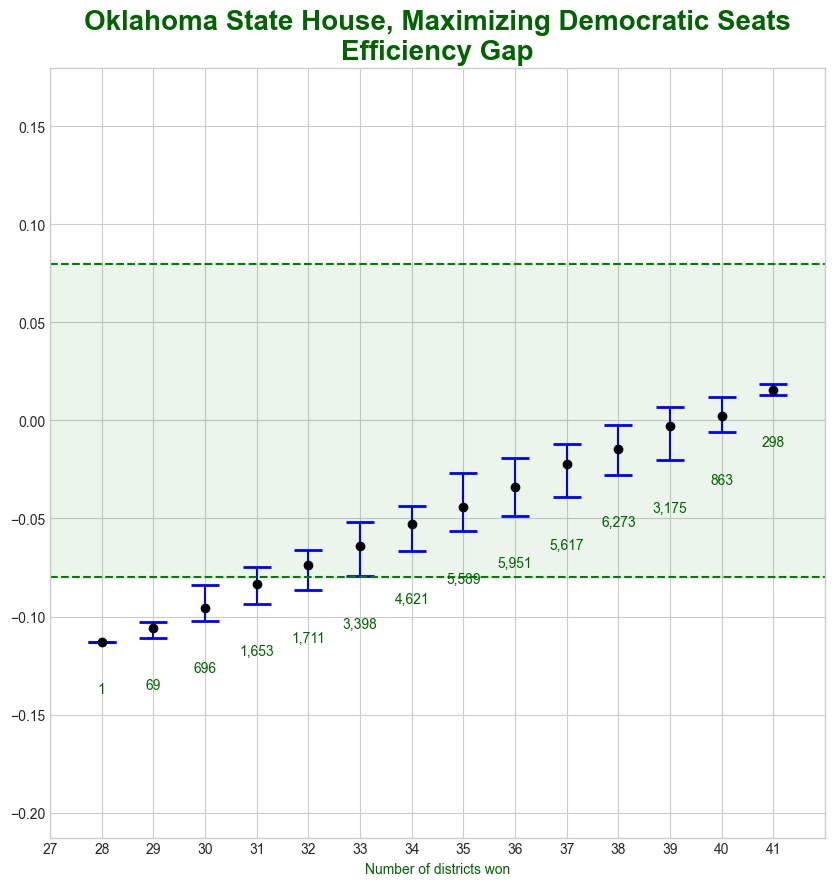}
\includegraphics[width=1.5in]{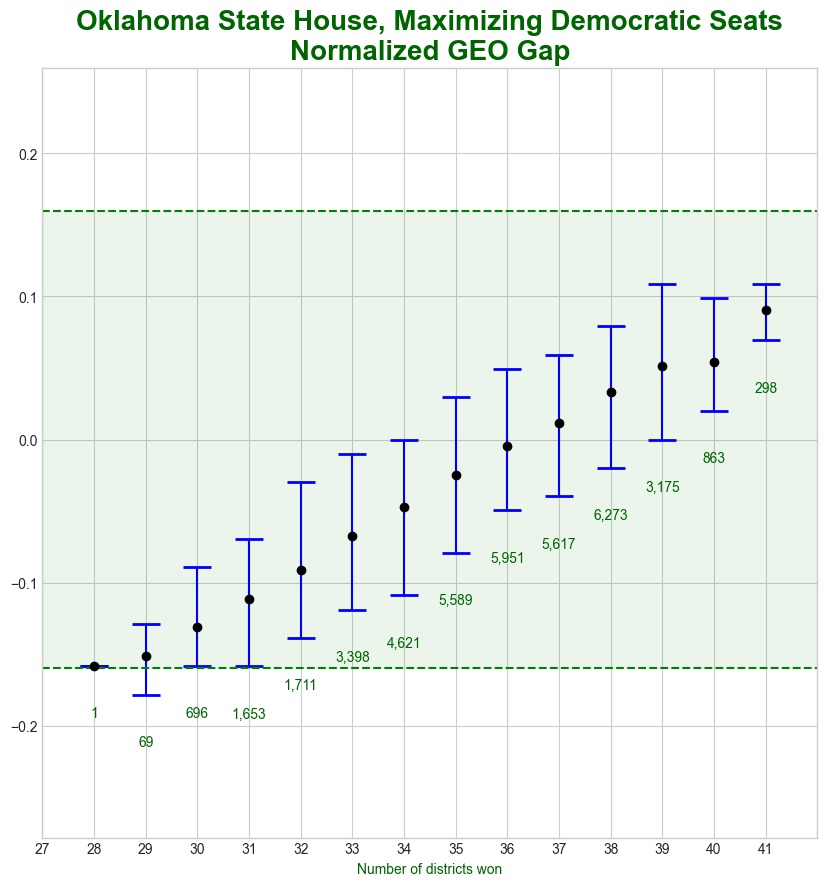}
\includegraphics[width=1.5in]{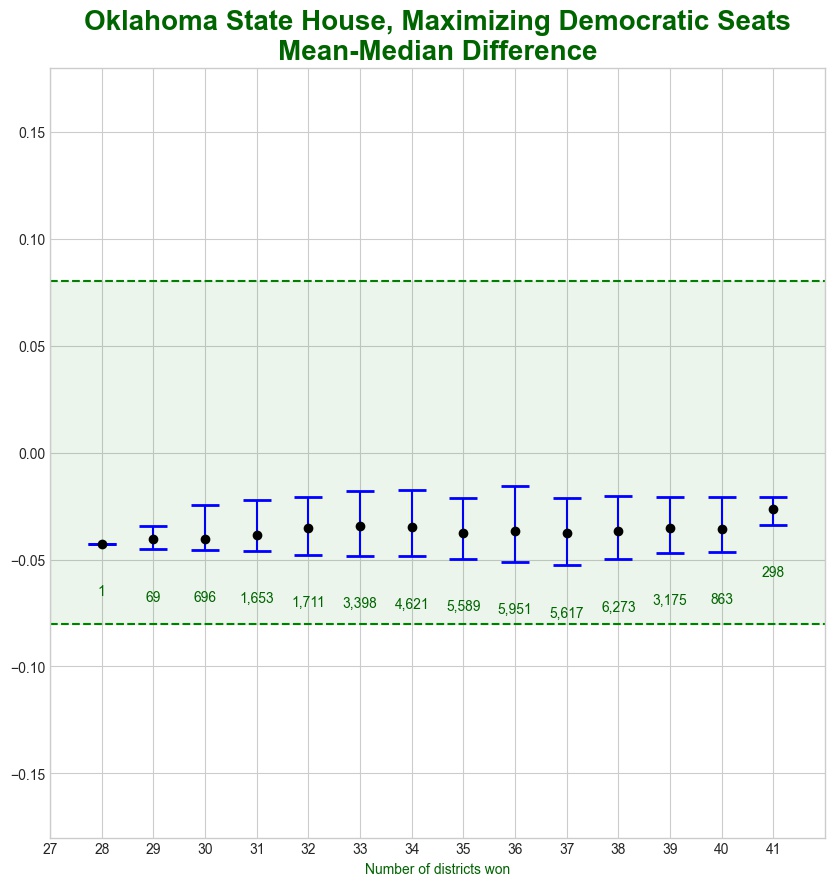}
\includegraphics[width=1.5in]{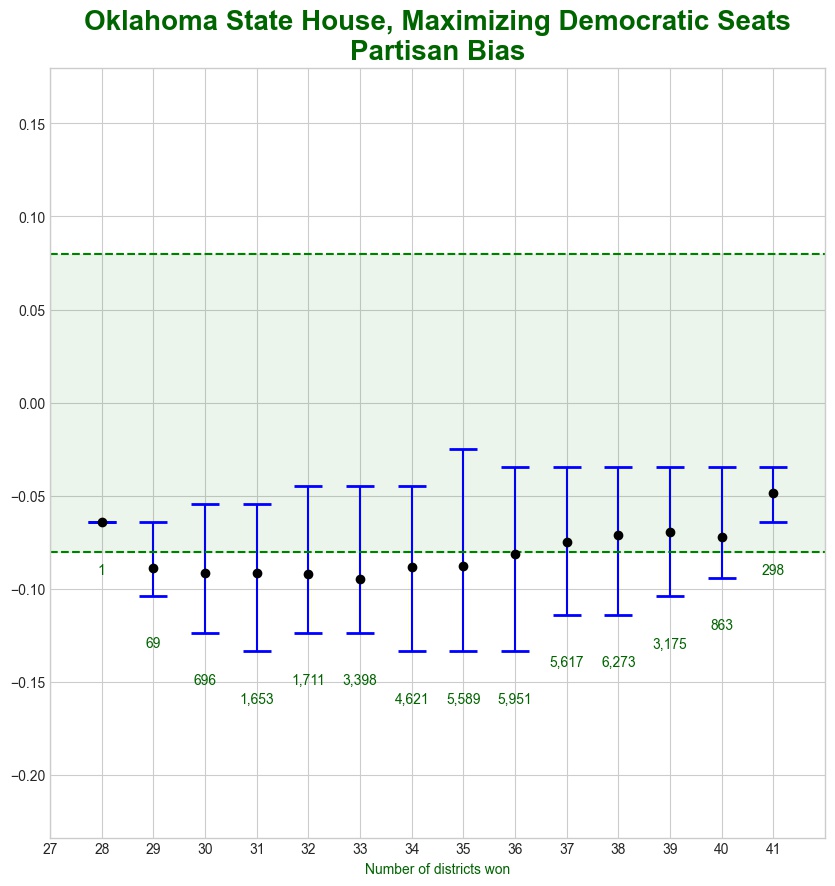}
\caption{Empirical results for Oklahoma's State House map and 2018 Gubenatorial election data, searching for maps with as many Democratic-won districts as possible.  Horizontal axis is number of districts won, vertical axis is metric value ranges.  The green region is from $0.16\inf(m)$ to $0.16\sup(m)$ for each metric $m$.  The small number below each metric value range is the number of maps produced that had the corresponding number of districts won.  The dot within each vertical bar is the mean value of that metric on all produced maps with the corresponding number of districts won.}
\label{fig:short_bursts_OKlowerD}
\end{figure}

\begin{figure}[h]
\centering
\includegraphics[width=1.5in]{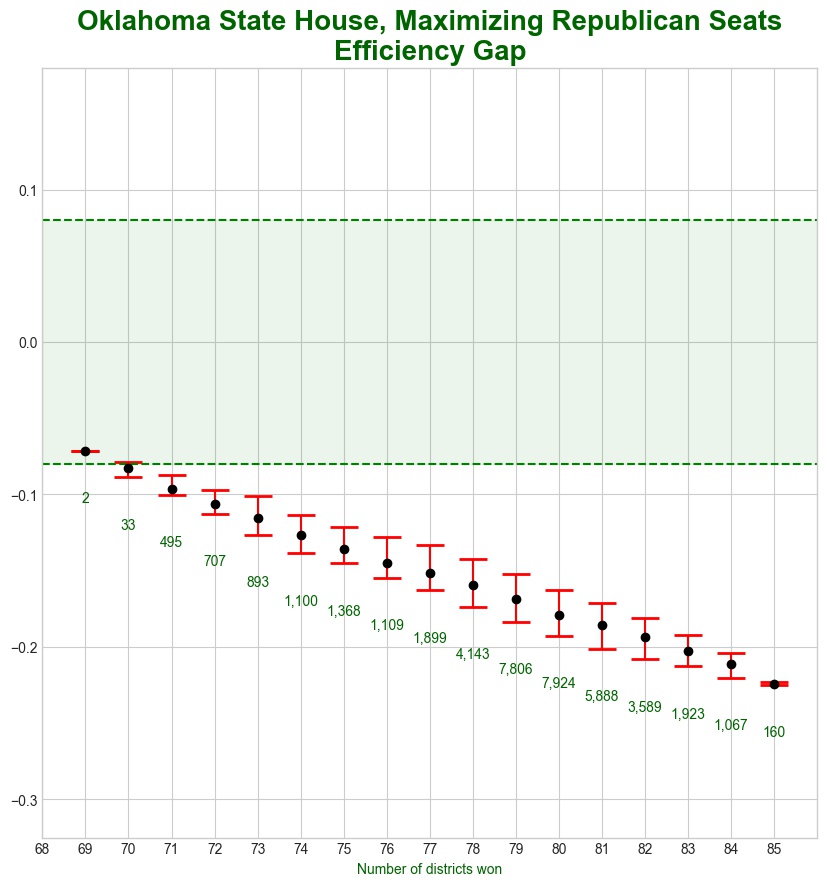}
\includegraphics[width=1.5in]{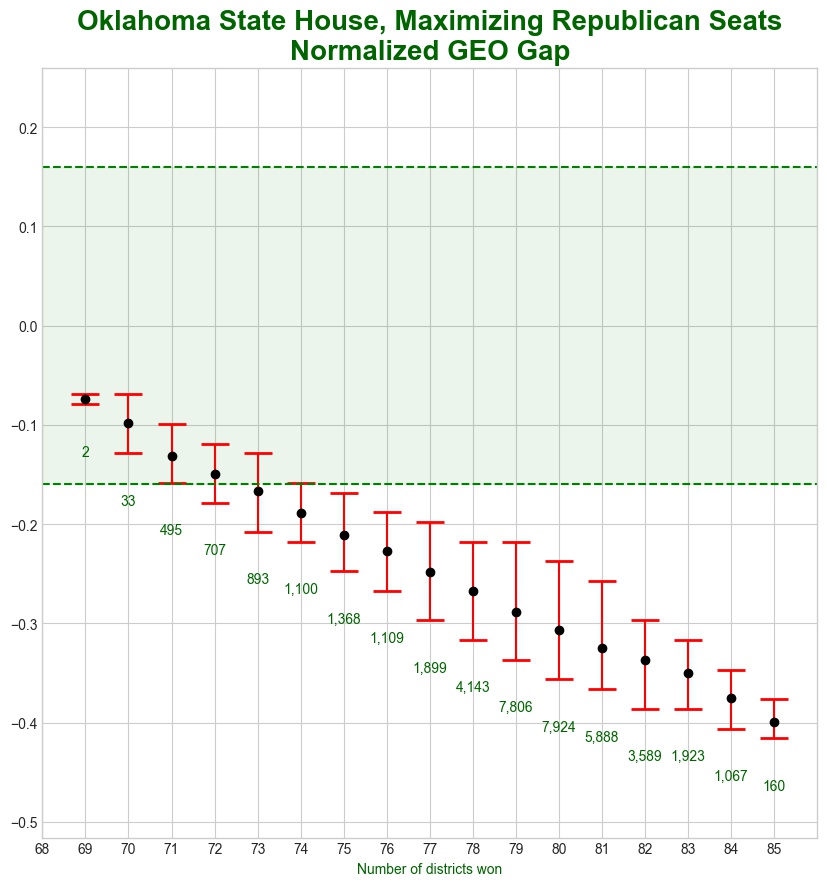}
\includegraphics[width=1.5in]{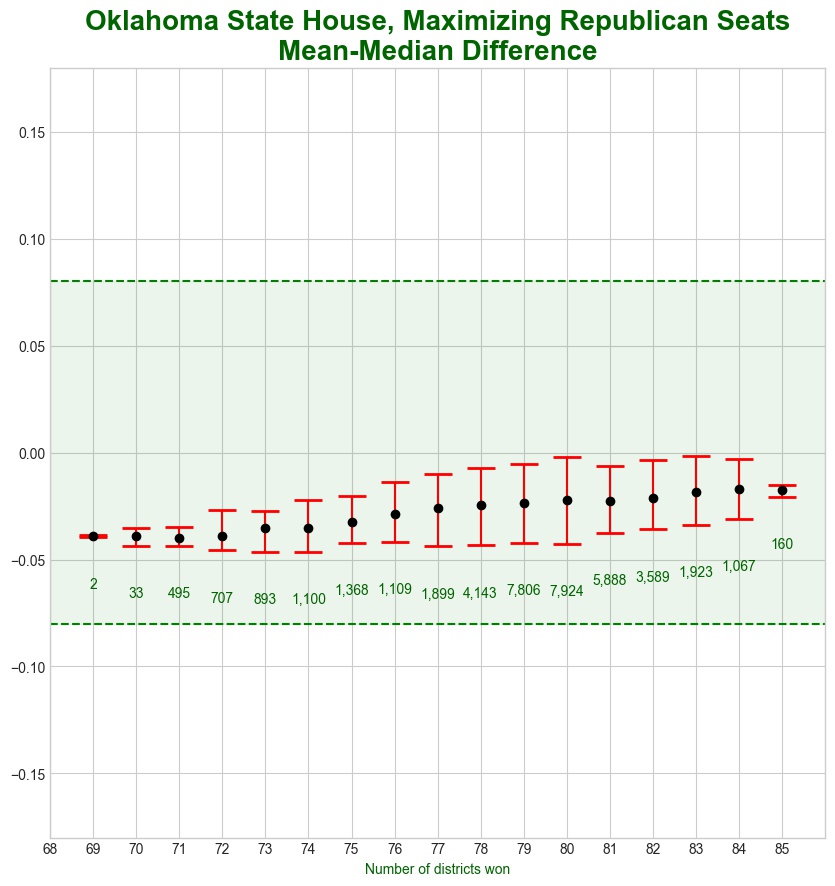}
\includegraphics[width=1.5in]{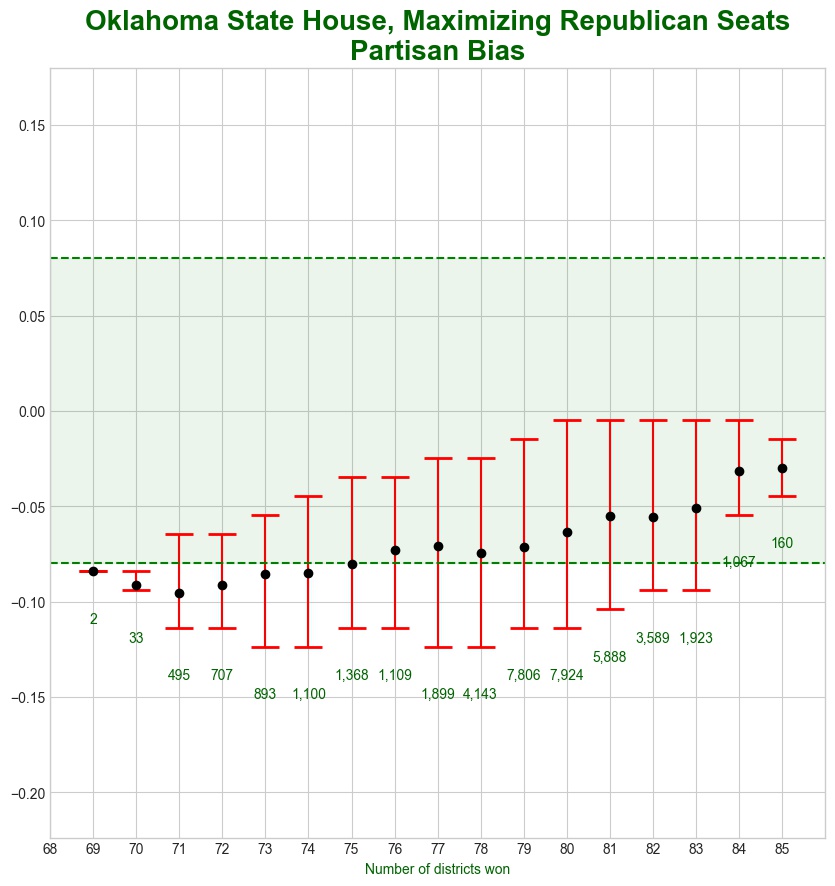}
\caption{Empirical results for Oklahoma's State House map and 2018 Gubenatorial election data, searching for maps with as many Republican-won districts as possible.  Horizontal axis is number of districts won, vertical axis is metric value ranges.  The green region is from $0.16\inf(m)$ to $0.16\sup(m)$ for each metric $m$.  The small number below each metric value range is the number of maps produced that had the corresponding number of districts won.  The dot within each vertical bar is the mean value of that metric on all produced maps with the corresponding number of districts won.}
\label{fig:short_bursts_OKlowerR}
\end{figure}

\begin{figure}[h]
\centering
\includegraphics[width=1.5in]{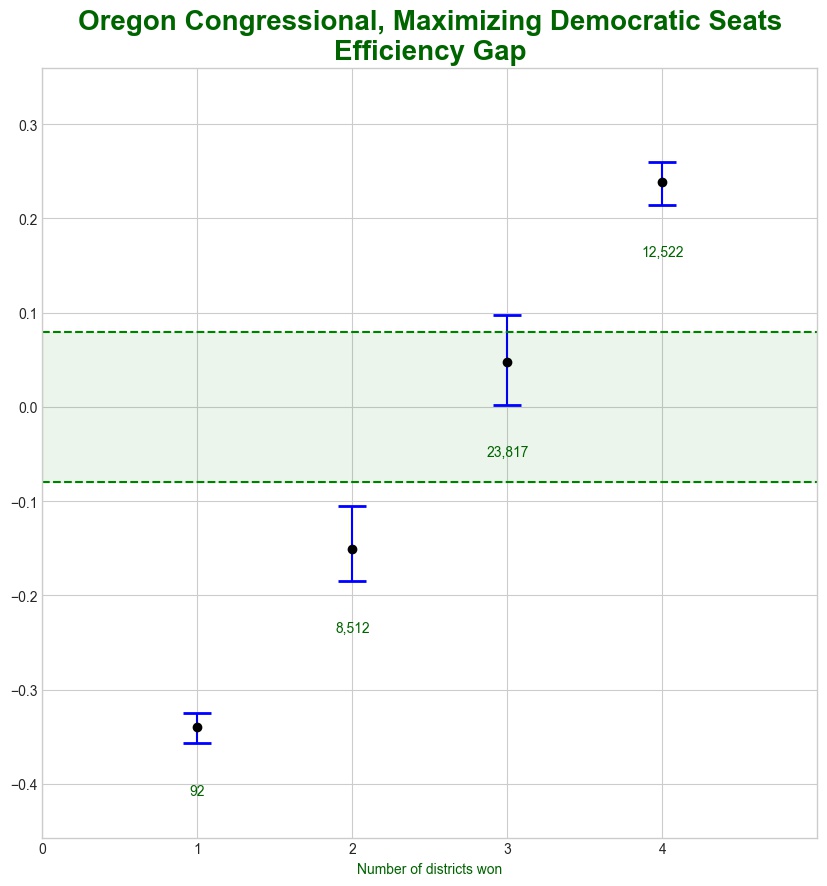}
\includegraphics[width=1.5in]{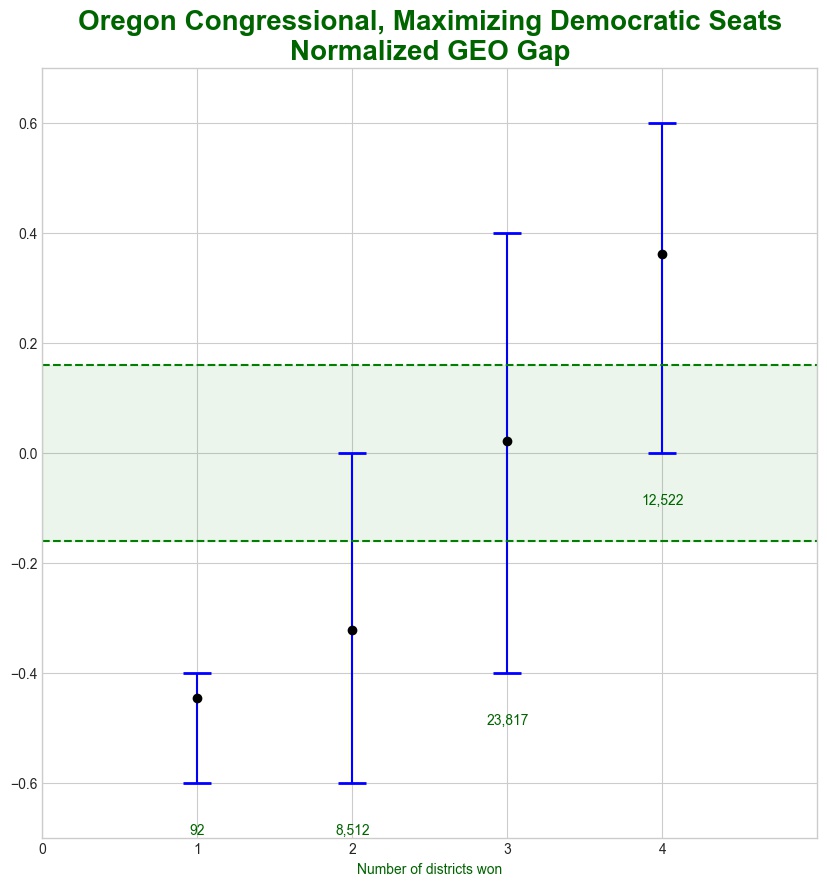}
\includegraphics[width=1.5in]{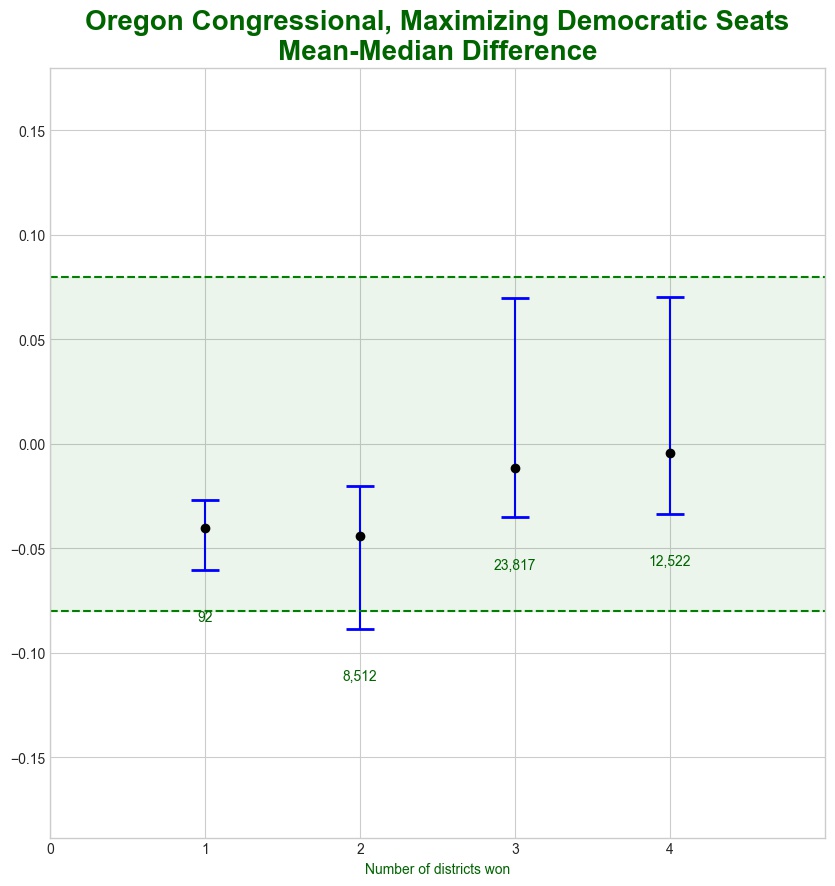}
\includegraphics[width=1.5in]{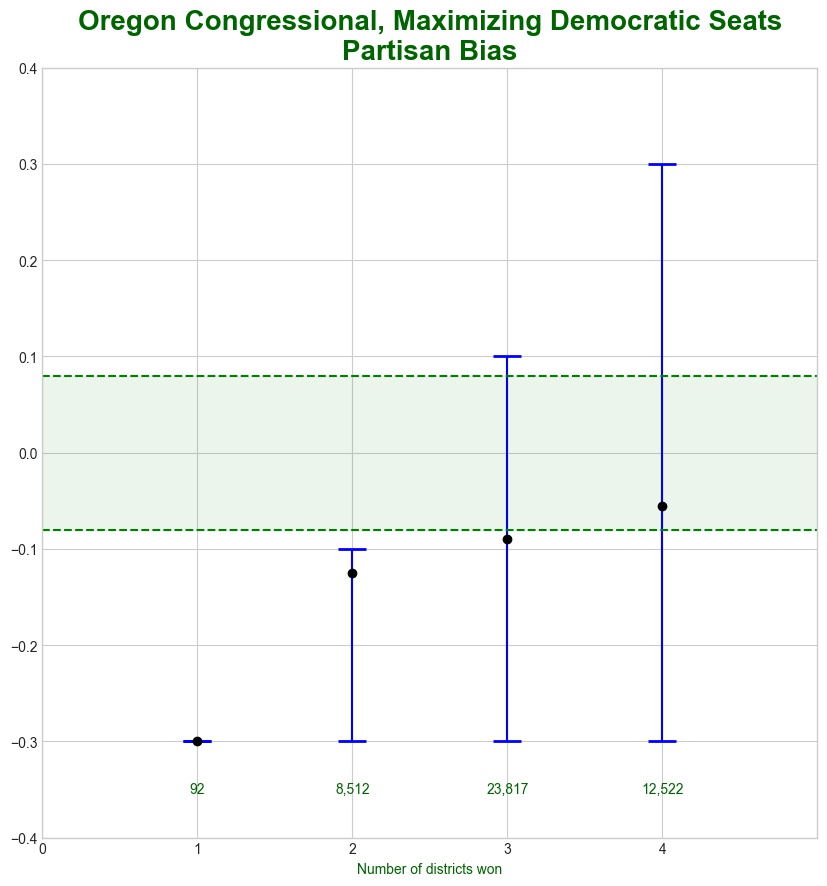}
\caption{Empirical results for Oregon's Congressional map and 2018 Gubenatorial election data, searching for maps with as many Democratic-won districts as possible.  Horizontal axis is number of districts won, vertical axis is metric value ranges.  The green region is from $0.16\inf(m)$ to $0.16\sup(m)$ for each metric $m$.  The small number below each metric value range is the number of maps produced that had the corresponding number of districts won.  The dot within each vertical bar is the mean value of that metric on all produced maps with the corresponding number of districts won.}
\label{fig:short_bursts_ORcongD}
\end{figure}

\begin{figure}[h]
\centering
\includegraphics[width=1.5in]{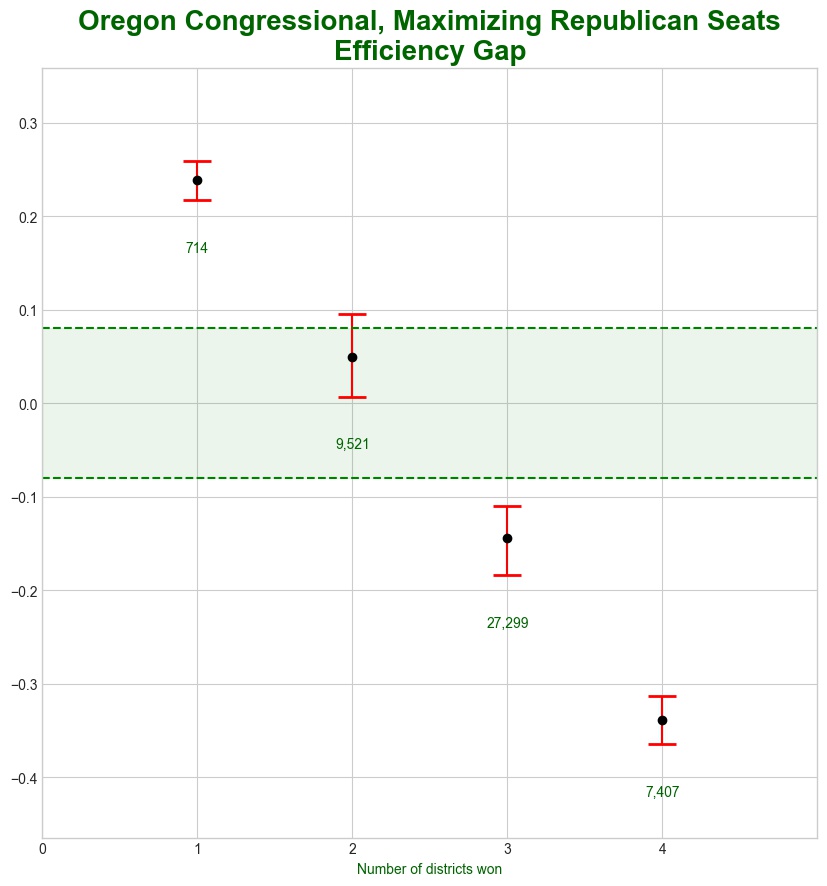}
\includegraphics[width=1.5in]{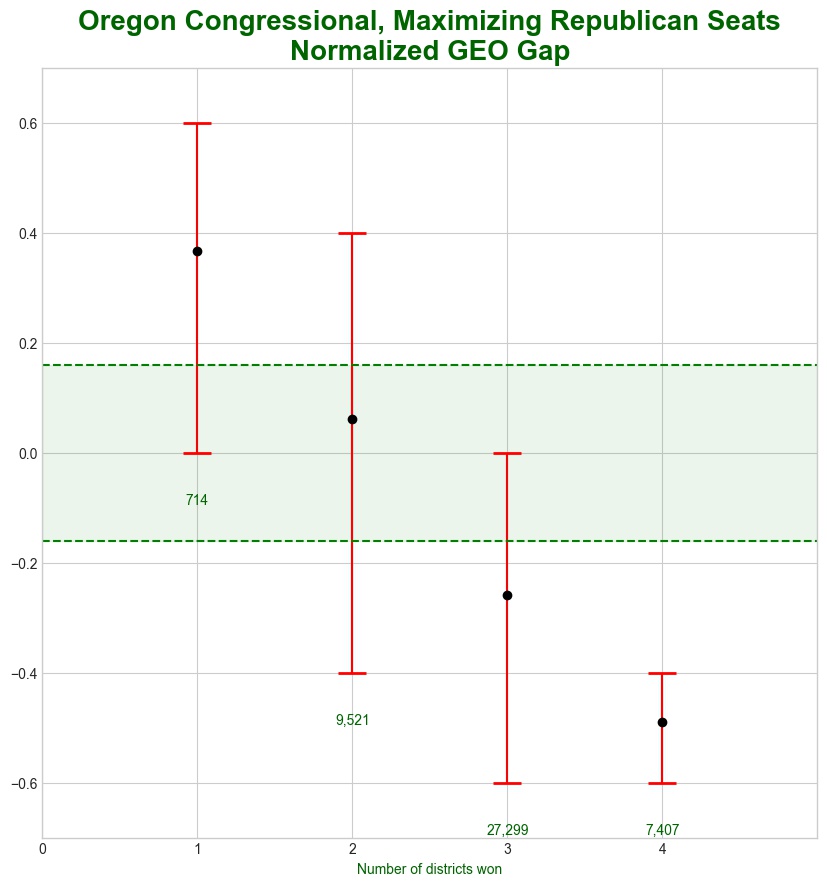}
\includegraphics[width=1.5in]{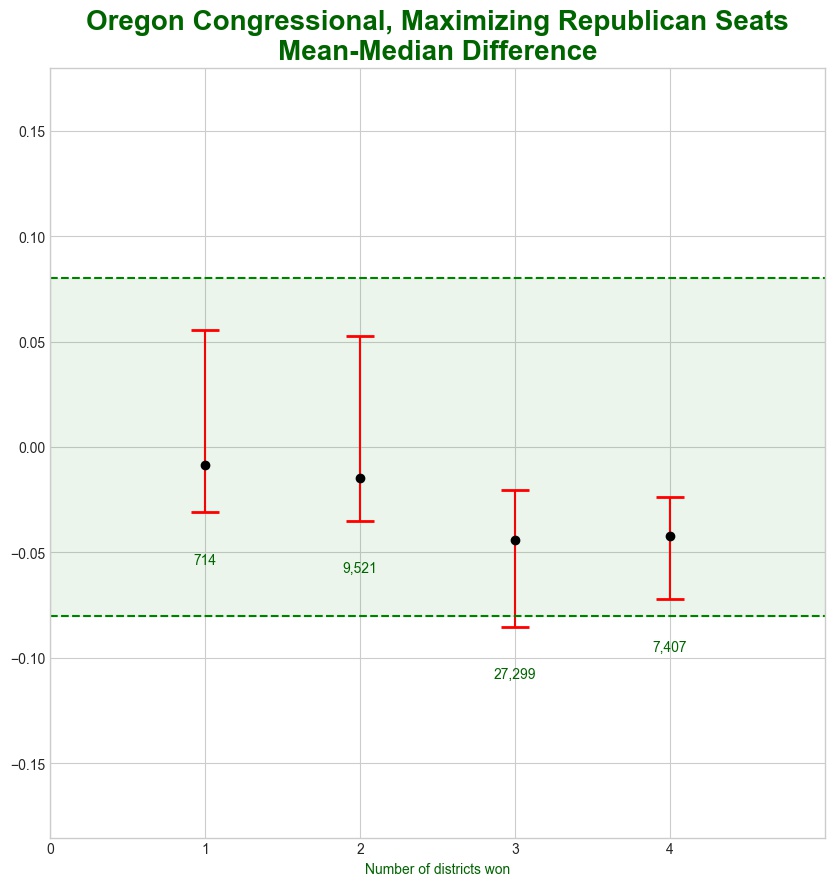}
\includegraphics[width=1.5in]{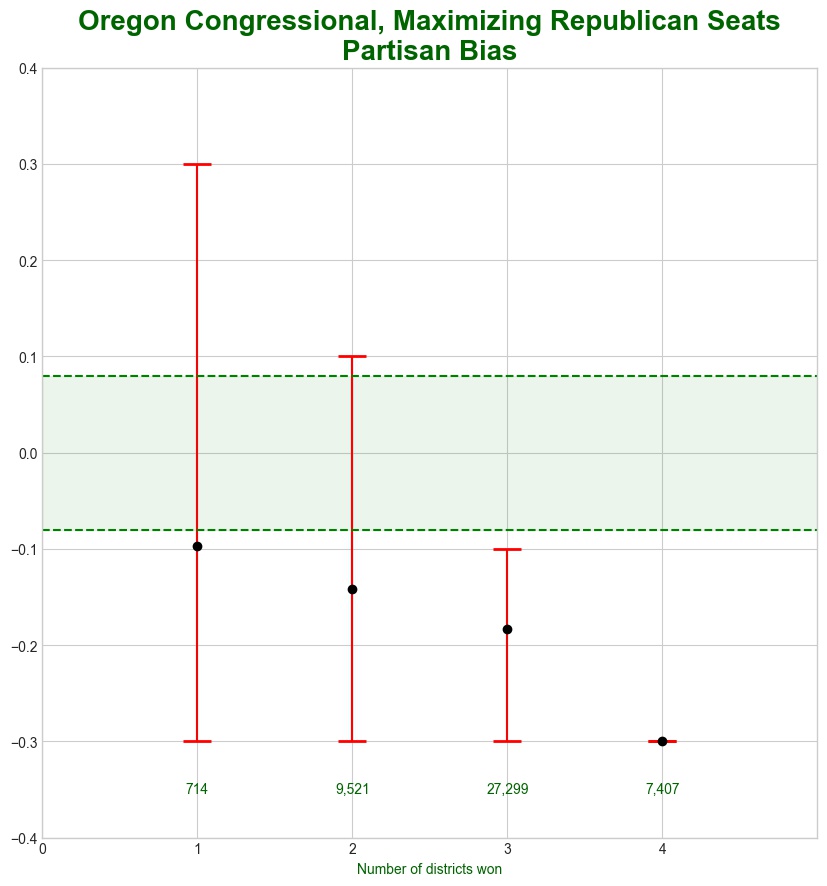}
\caption{Empirical results for Oregon's Congressional map and 2018 Gubenatorial election data, searching for maps with as many Republican-won districts as possible.  Horizontal axis is number of districts won, vertical axis is metric value ranges.  The green region is from $0.16\inf(m)$ to $0.16\sup(m)$ for each metric $m$.  The small number below each metric value range is the number of maps produced that had the corresponding number of districts won.  The dot within each vertical bar is the mean value of that metric on all produced maps with the corresponding number of districts won.}
\label{fig:short_bursts_ORcongR}
\end{figure}

\begin{figure}[h]
\centering
\includegraphics[width=1.5in]{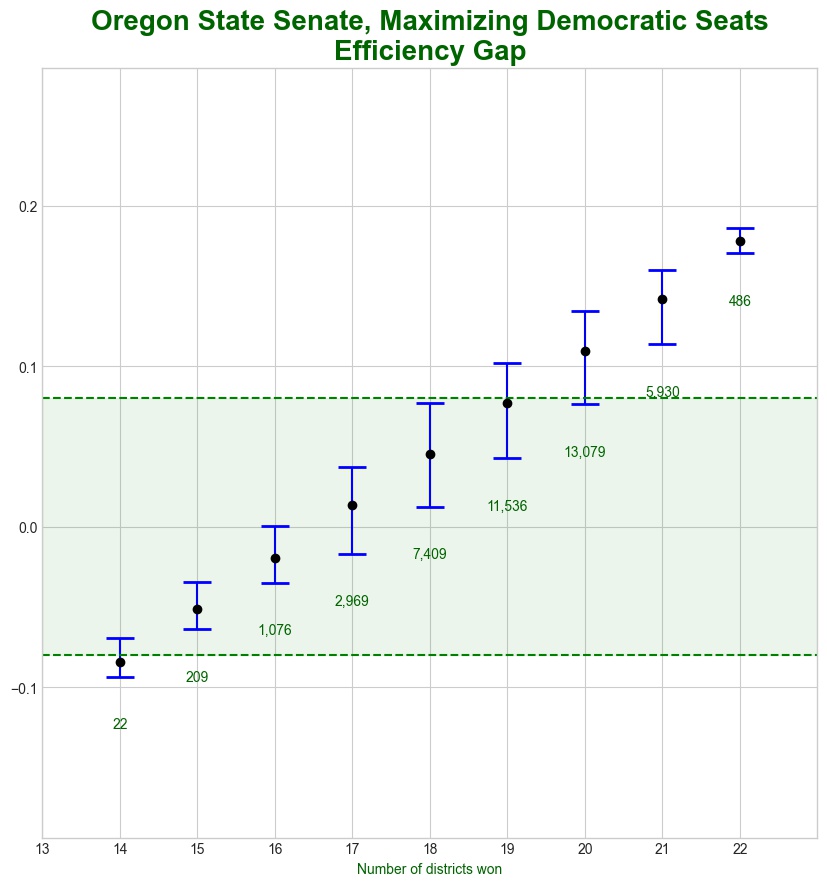}
\includegraphics[width=1.5in]{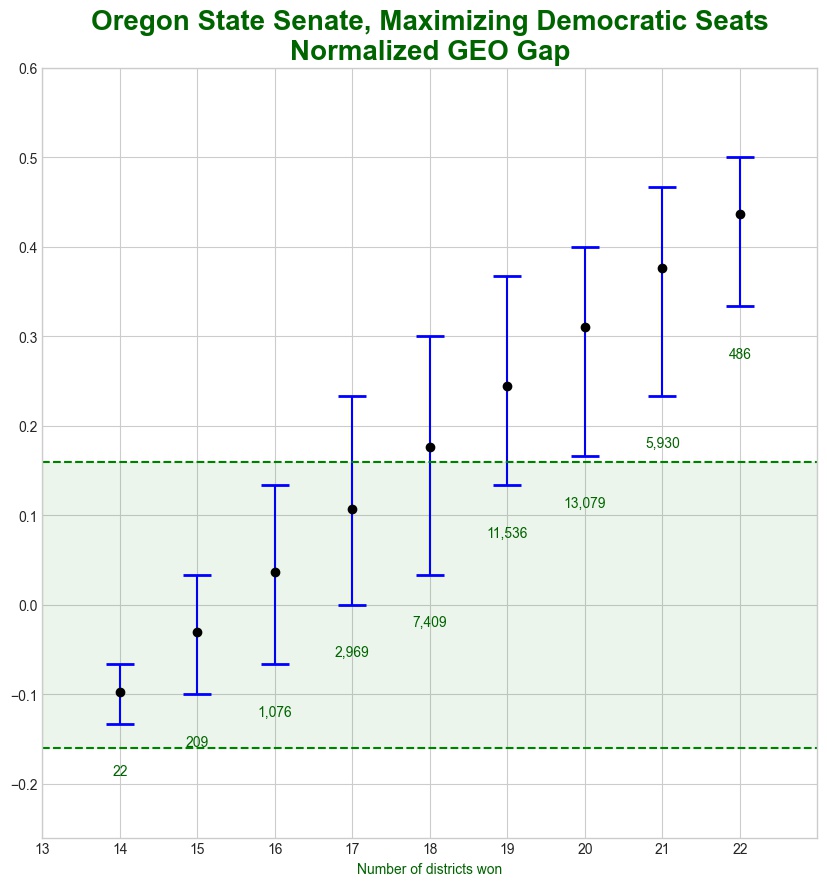}
\includegraphics[width=1.5in]{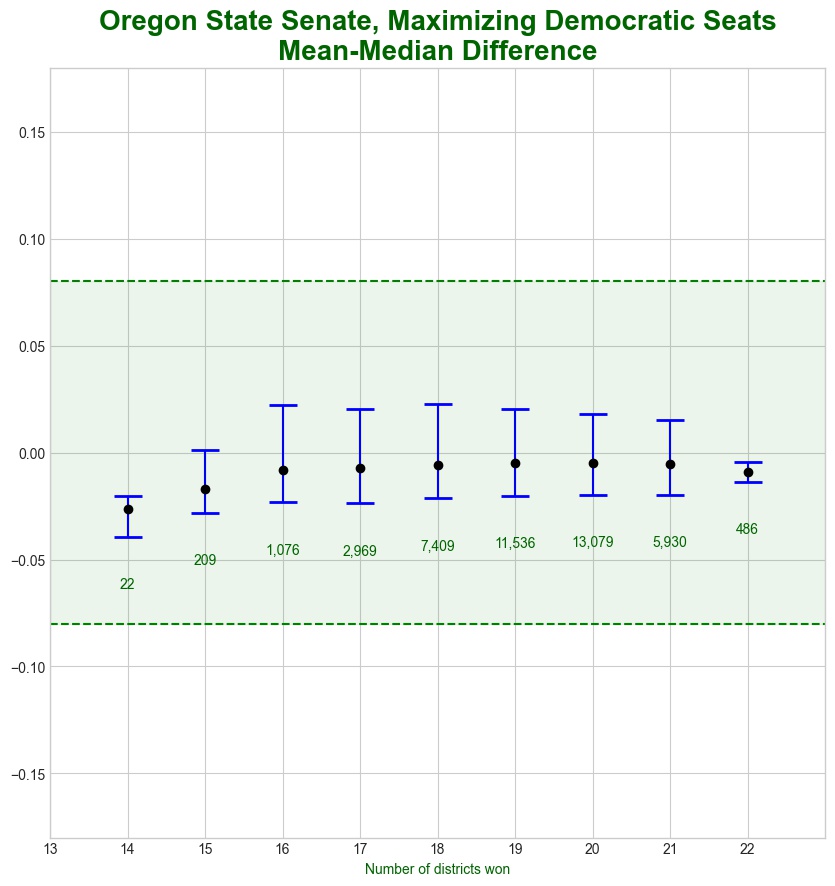}
\includegraphics[width=1.5in]{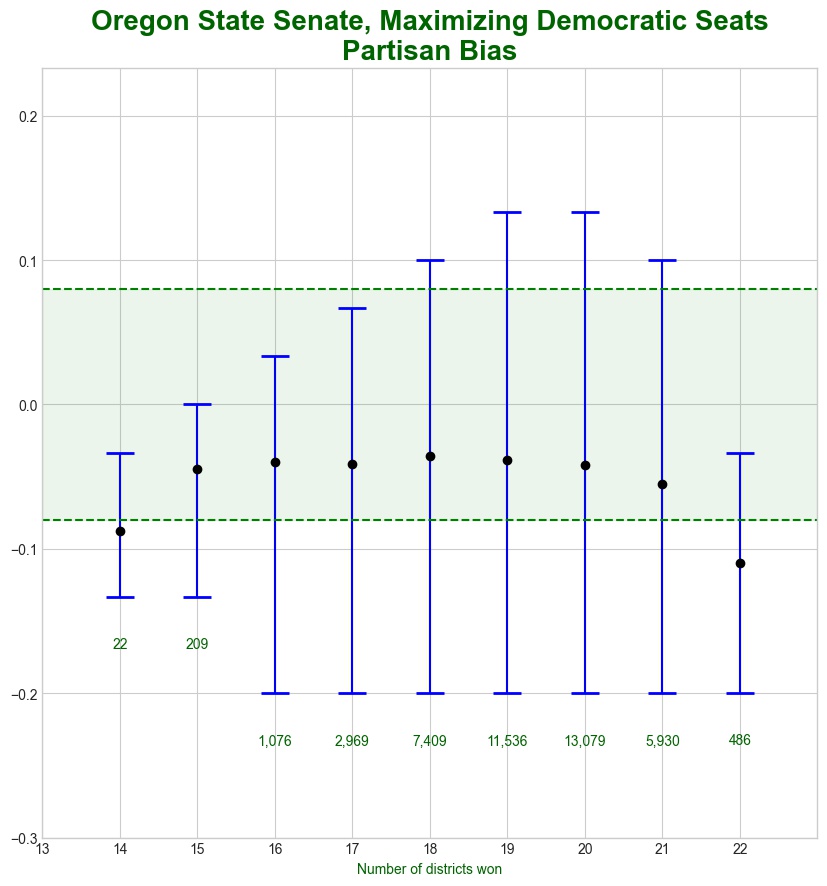}
\caption{Empirical results for Oregon's State Senate map and 2018 Gubenatorial election data, searching for maps with as many Democratic-won districts as possible.  Horizontal axis is number of districts won, vertical axis is metric value ranges.  The green region is from $0.16\inf(m)$ to $0.16\sup(m)$ for each metric $m$.  The small number below each metric value range is the number of maps produced that had the corresponding number of districts won.  The dot within each vertical bar is the mean value of that metric on all produced maps with the corresponding number of districts won.}
\label{fig:short_bursts_ORupperD}
\end{figure}

\begin{figure}[h]
\centering
\includegraphics[width=1.5in]{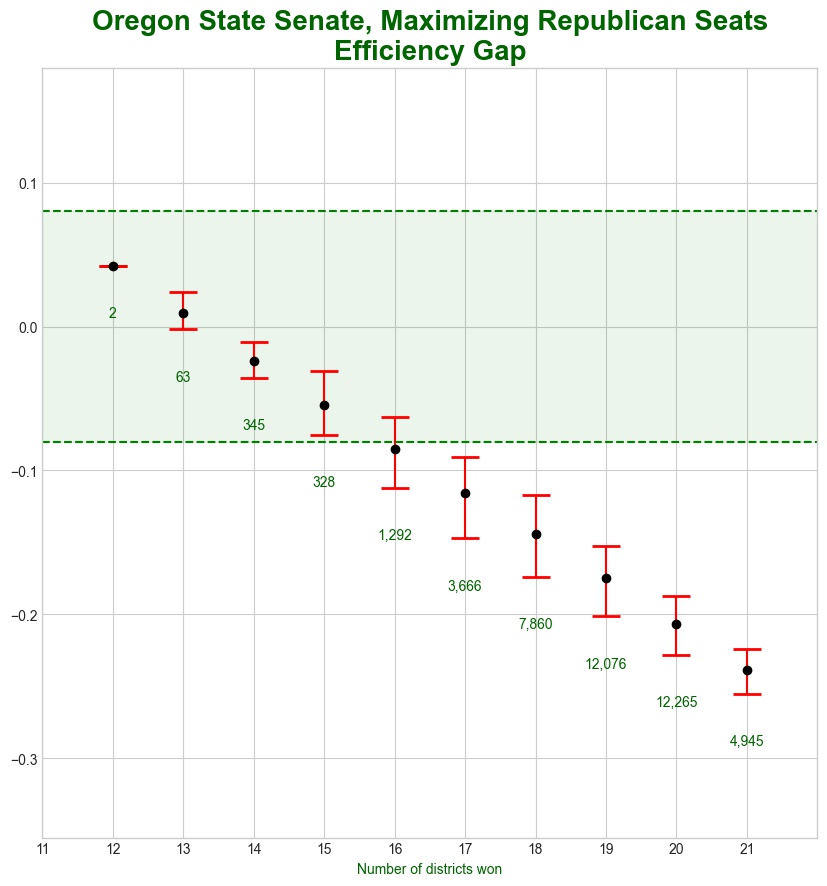}
\includegraphics[width=1.5in]{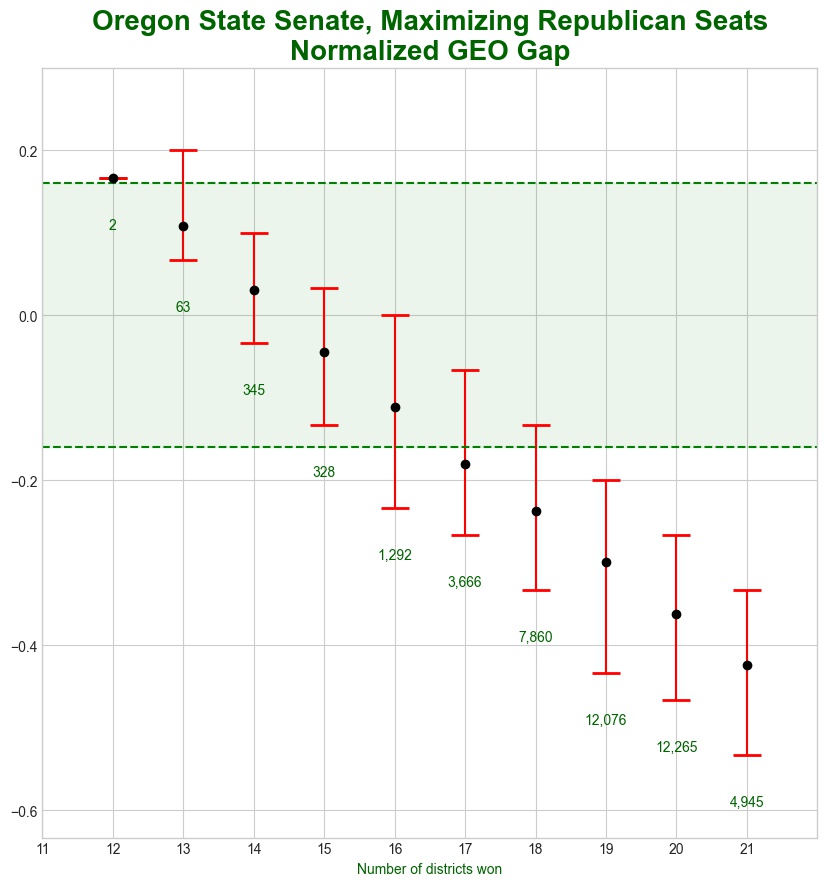}
\includegraphics[width=1.5in]{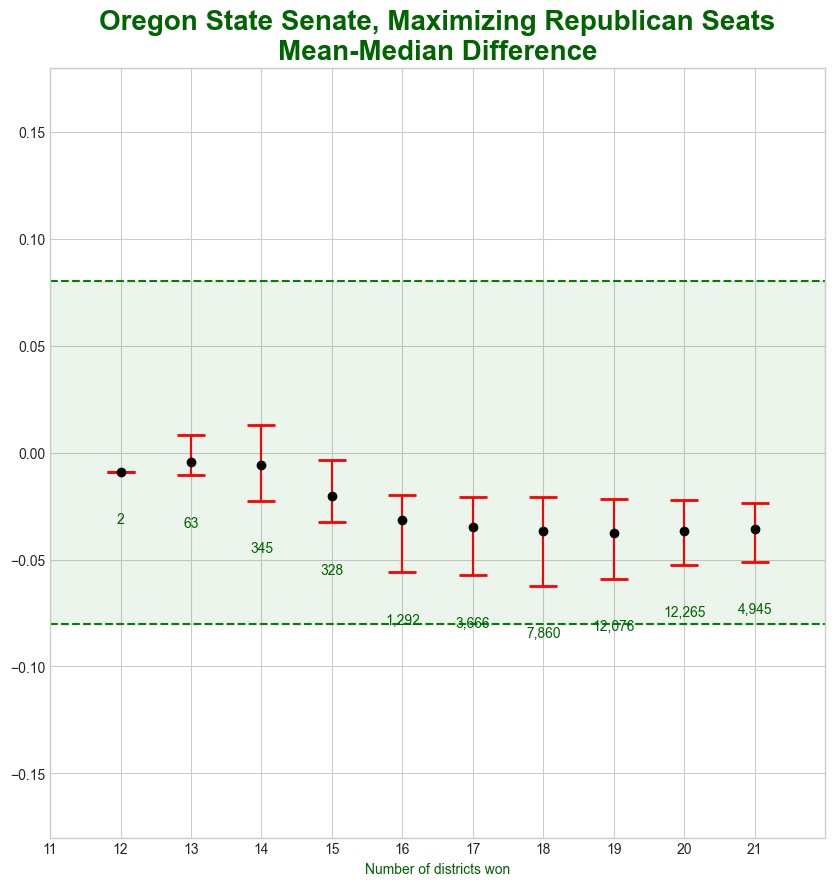}
\includegraphics[width=1.5in]{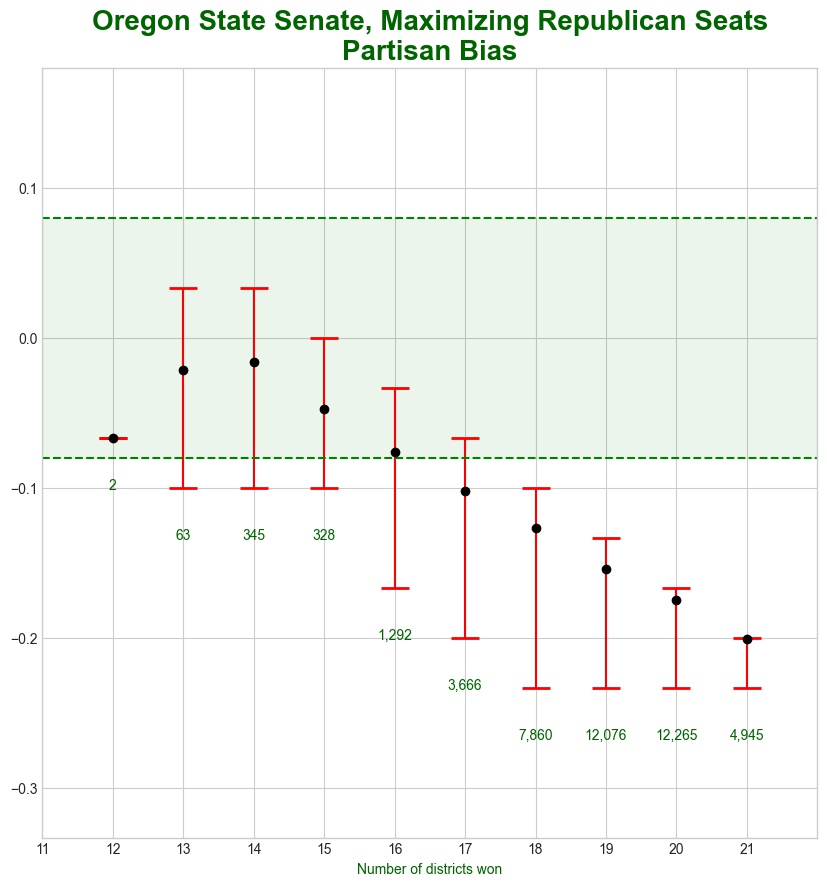}
\caption{Empirical results for Oregon's State Senate map and 2018 Gubenatorial election data, searching for maps with as many Republican-won districts as possible.  Horizontal axis is number of districts won, vertical axis is metric value ranges.  The green region is from $0.16\inf(m)$ to $0.16\sup(m)$ for each metric $m$.  The small number below each metric value range is the number of maps produced that had the corresponding number of districts won.  The dot within each vertical bar is the mean value of that metric on all produced maps with the corresponding number of districts won.}
\label{fig:short_bursts_ORupperR}
\end{figure}

\begin{figure}[h]
\centering
\includegraphics[width=1.5in]{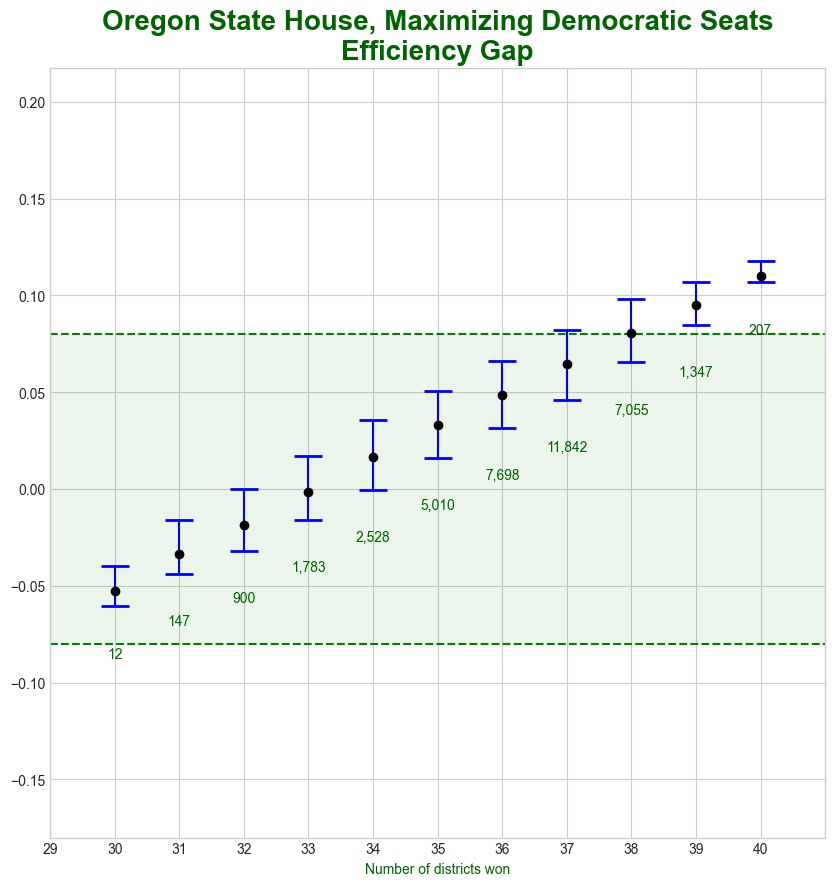}
\includegraphics[width=1.5in]{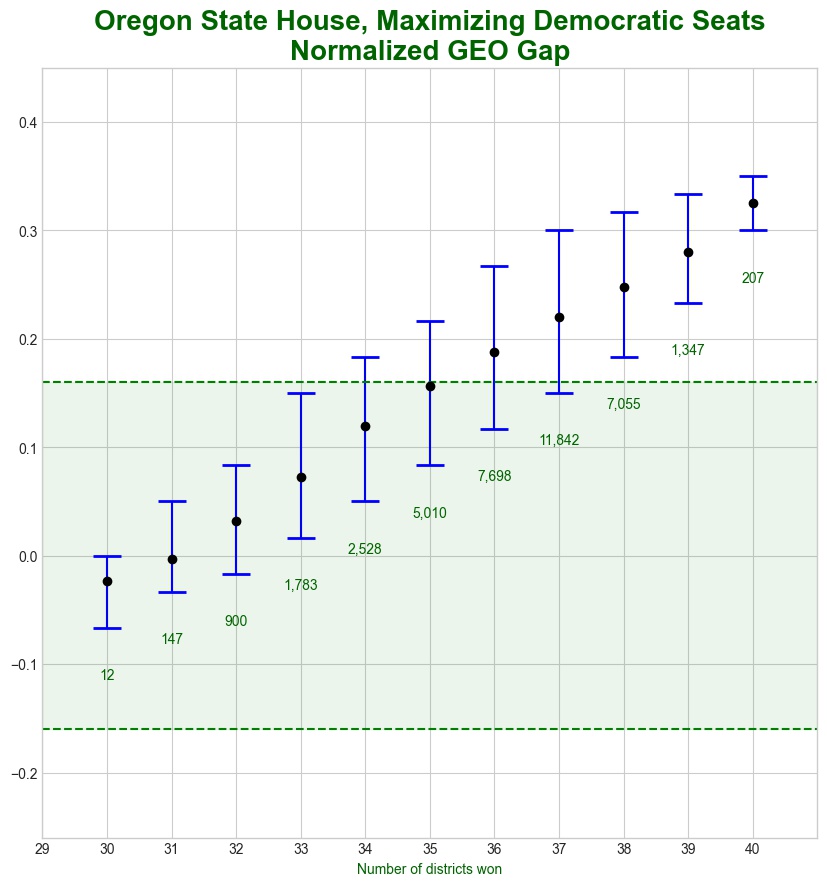}
\includegraphics[width=1.5in]{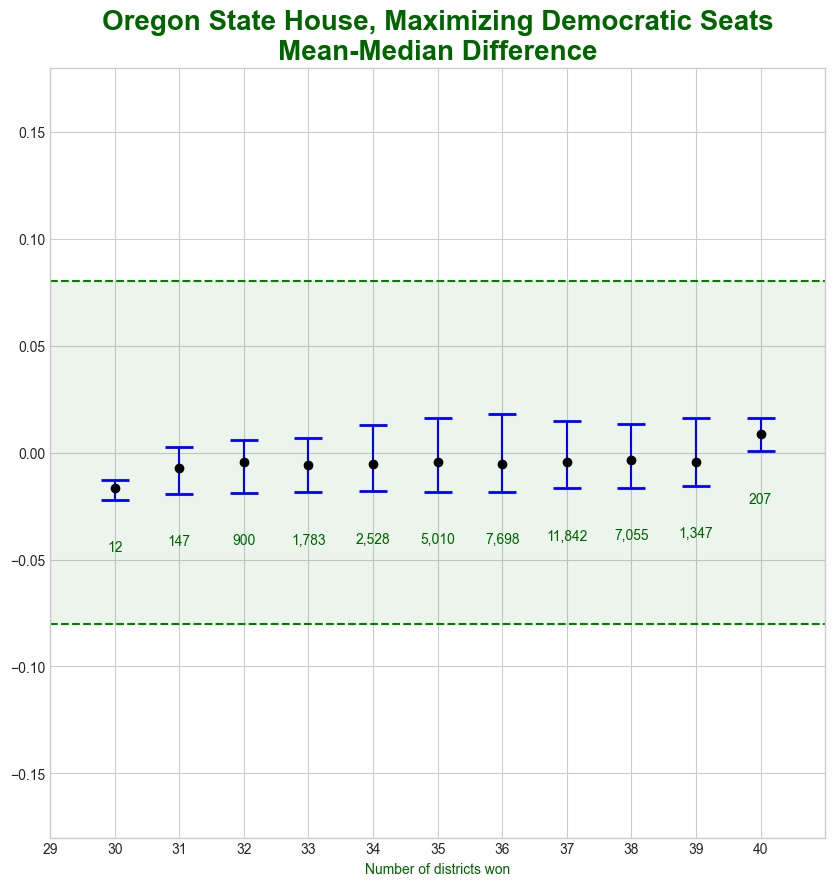}
\includegraphics[width=1.5in]{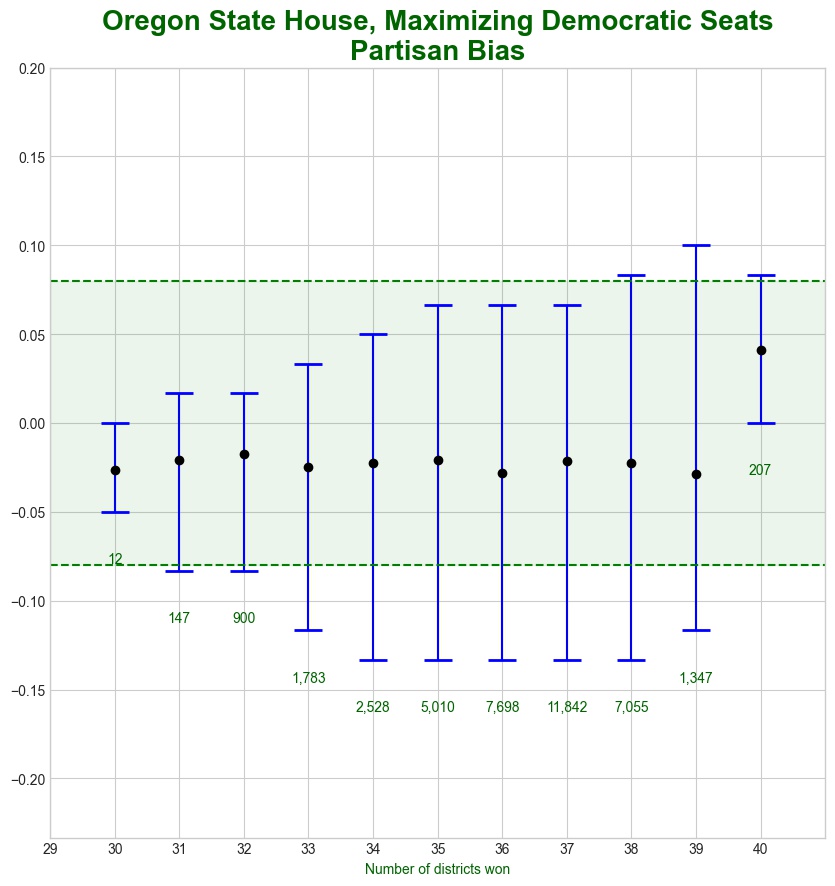}
\caption{Empirical results for Oregon's State House map and 2018 Gubenatorial election data, searching for maps with as many Democratic-won districts as possible.  Horizontal axis is number of districts won, vertical axis is metric value ranges.  The green region is from $0.16\inf(m)$ to $0.16\sup(m)$ for each metric $m$.  The small number below each metric value range is the number of maps produced that had the corresponding number of districts won.  The dot within each vertical bar is the mean value of that metric on all produced maps with the corresponding number of districts won.}
\label{fig:short_bursts_ORlowerD}
\end{figure}

\begin{figure}[h]
\centering
\includegraphics[width=1.5in]{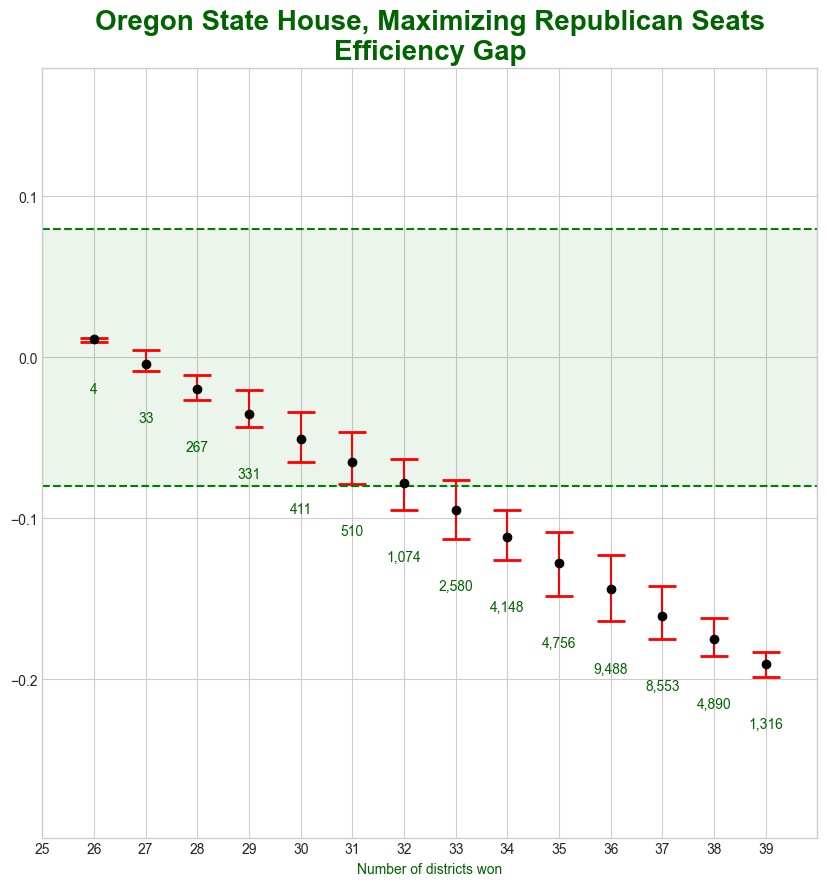}
\includegraphics[width=1.5in]{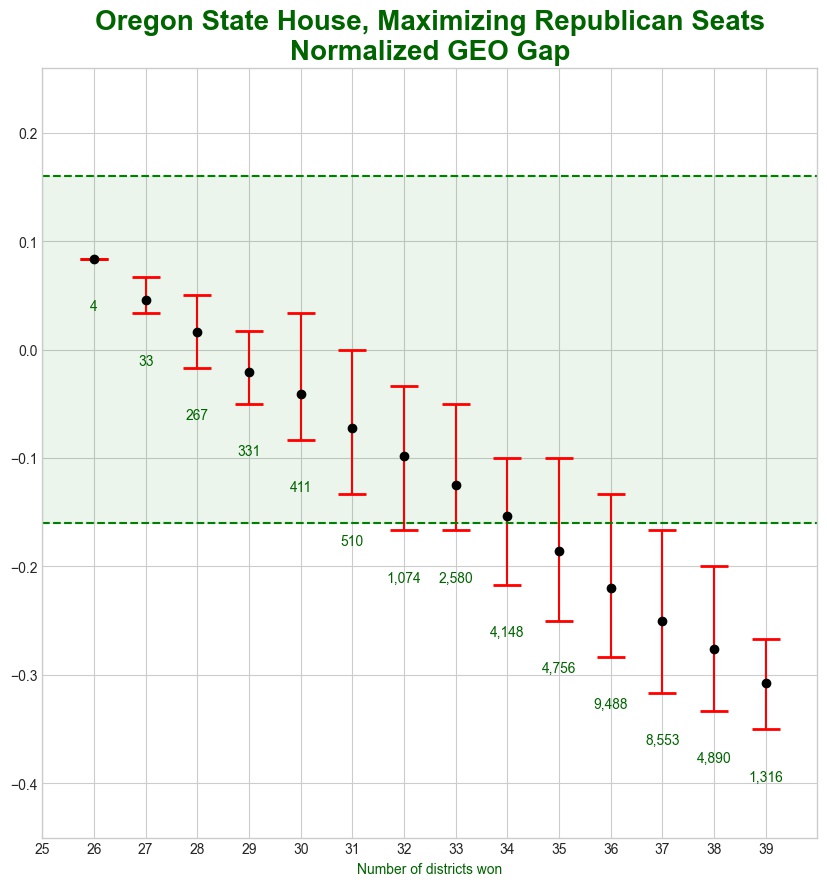}
\includegraphics[width=1.5in]{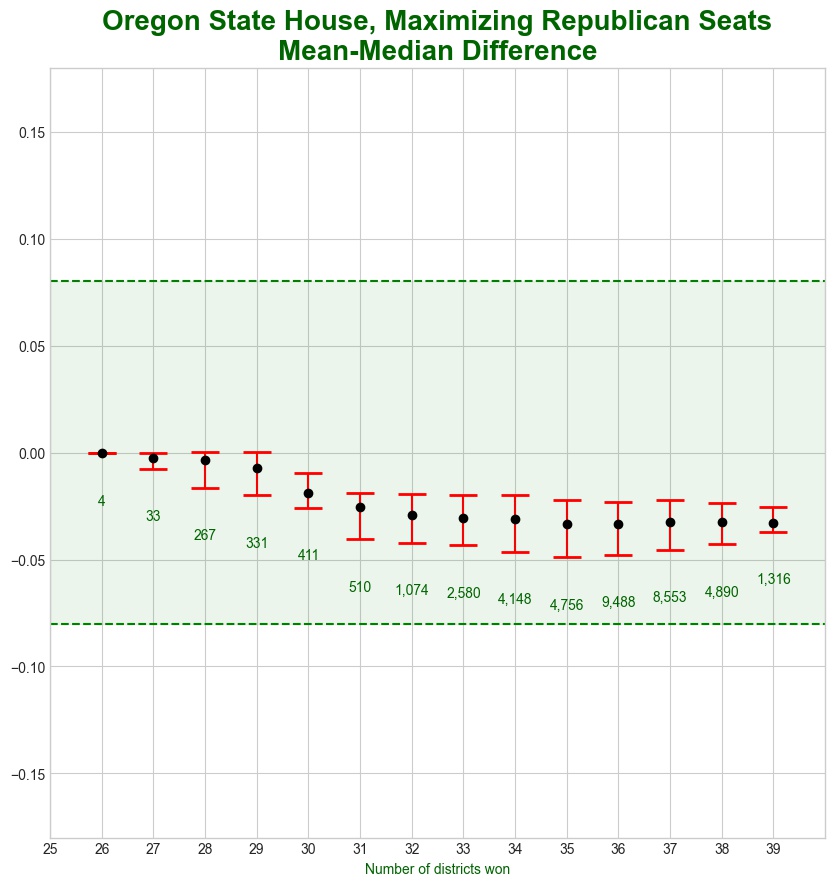}
\includegraphics[width=1.5in]{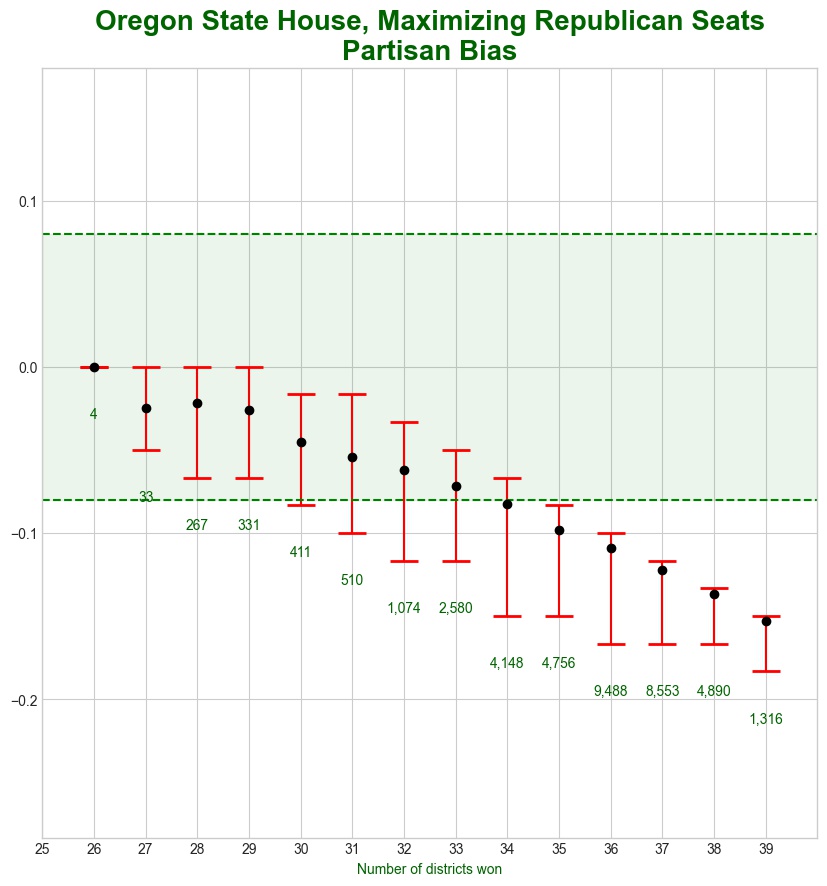}
\caption{Empirical results for Oregon's State House map and 2018 Gubenatorial election data, searching for maps with as many Republican-won districts as possible.  Horizontal axis is number of districts won, vertical axis is metric value ranges.  The green region is from $0.16\inf(m)$ to $0.16\sup(m)$ for each metric $m$.  The small number below each metric value range is the number of maps produced that had the corresponding number of districts won.  The dot within each vertical bar is the mean value of that metric on all produced maps with the corresponding number of districts won.}
\label{fig:short_bursts_ORlowerR}
\end{figure}

\begin{figure}[h]
\centering
\includegraphics[width=1.5in]{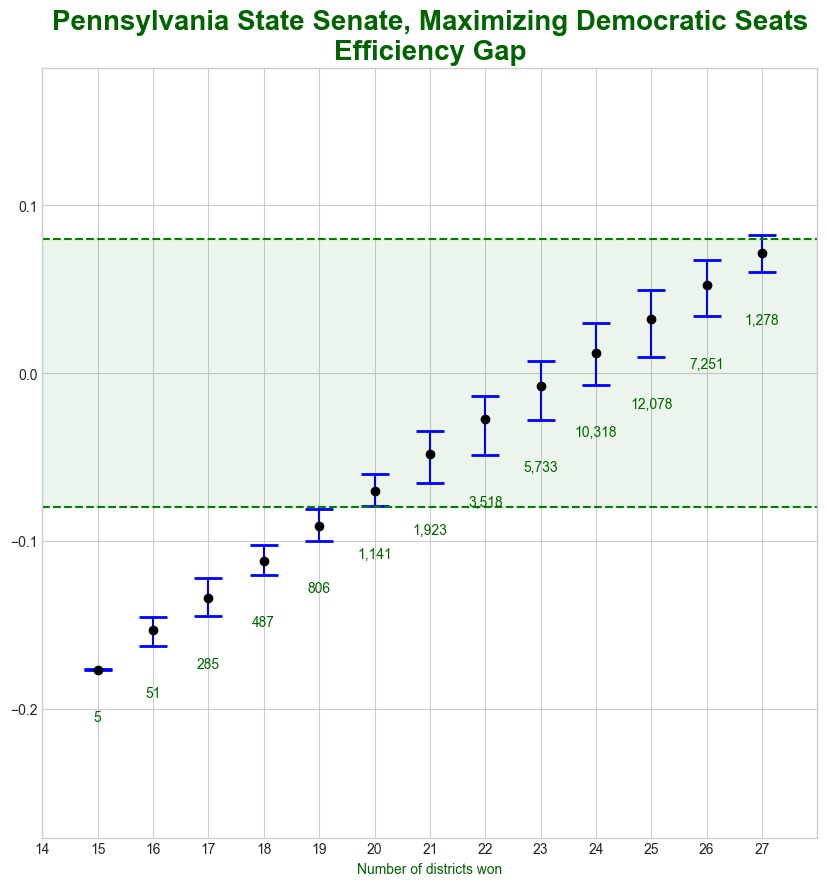}
\includegraphics[width=1.5in]{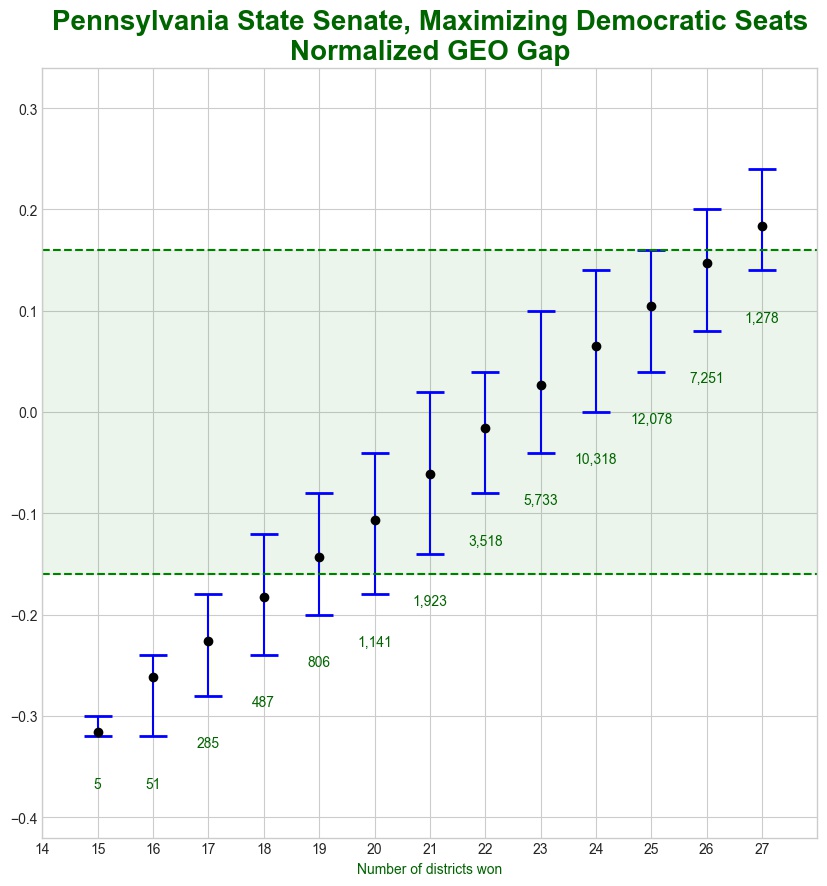}
\includegraphics[width=1.5in]{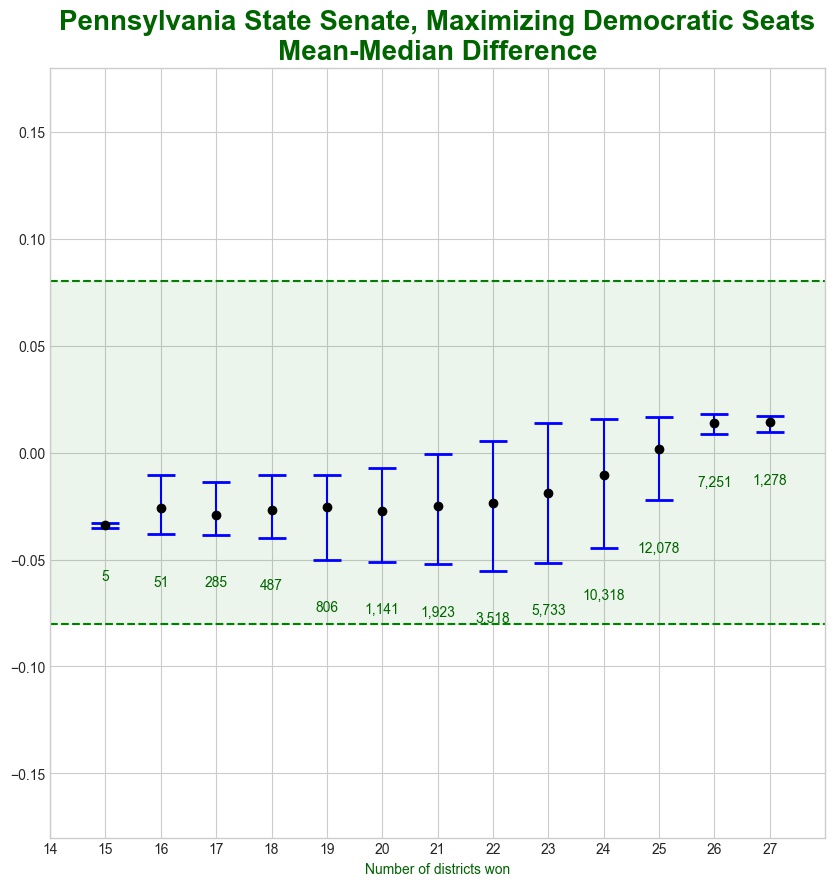}
\includegraphics[width=1.5in]{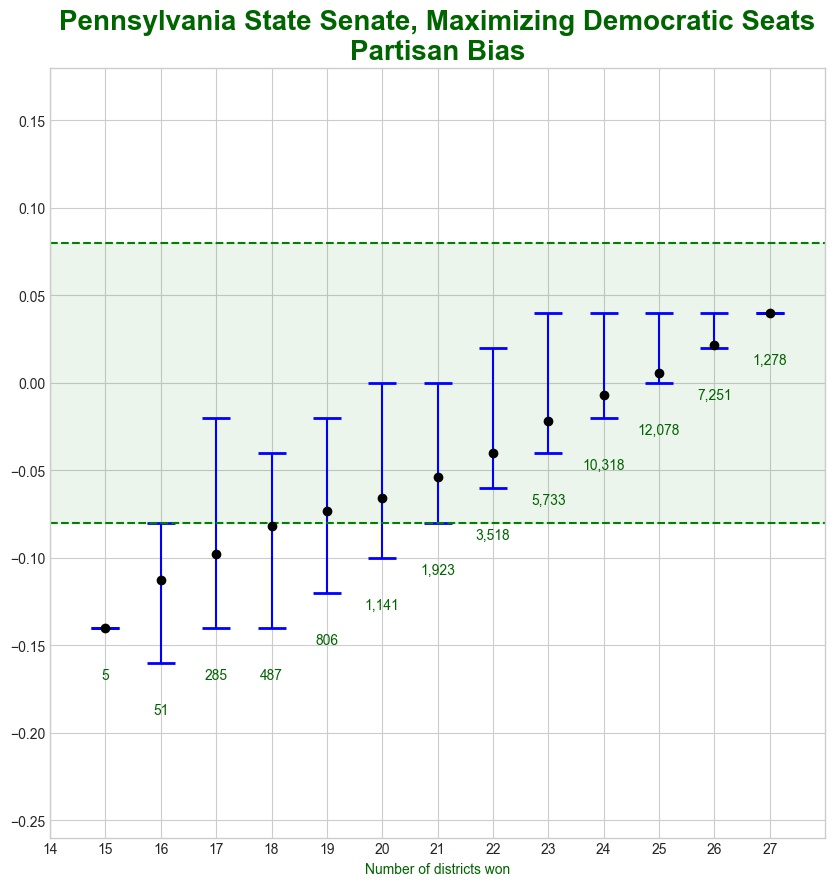}
\caption{Empirical results for Pennsylvania's State Senate map and 2016 US Senate election data, searching for maps with as many Democratic-won districts as possible.  Horizontal axis is number of districts won, vertical axis is metric value ranges.  The green region is from $0.16\inf(m)$ to $0.16\sup(m)$ for each metric $m$.  The small number below each metric value range is the number of maps produced that had the corresponding number of districts won.  The dot within each vertical bar is the mean value of that metric on all produced maps with the corresponding number of districts won.}
\label{fig:short_bursts_PAupperD}
\end{figure}

\begin{figure}[h]
\centering
\includegraphics[width=1.5in]{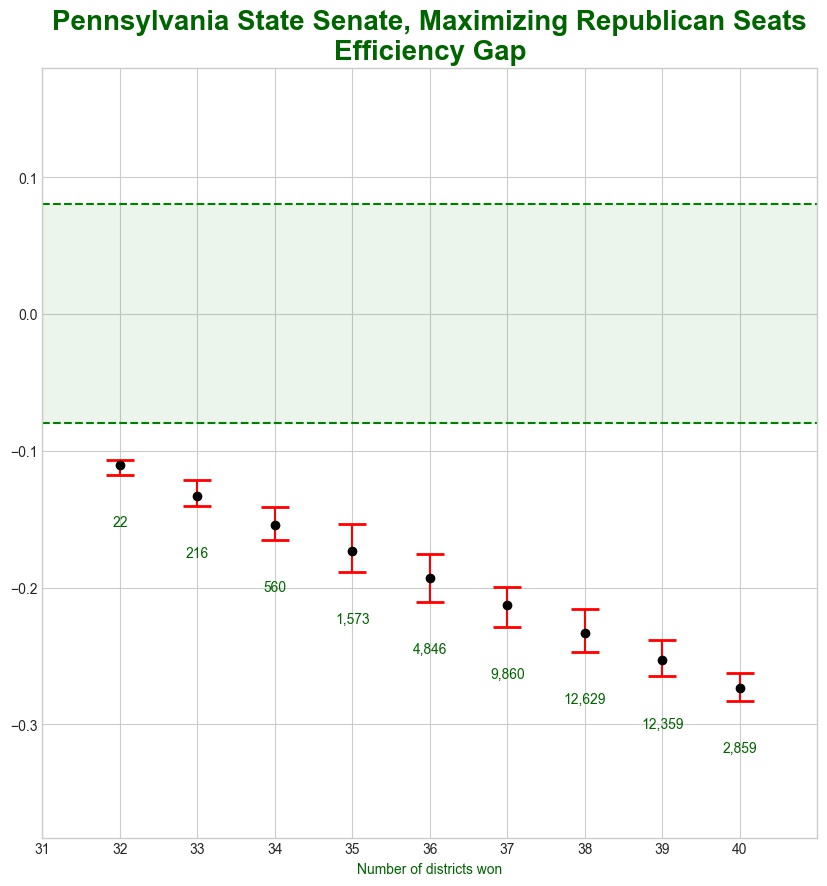}
\includegraphics[width=1.5in]{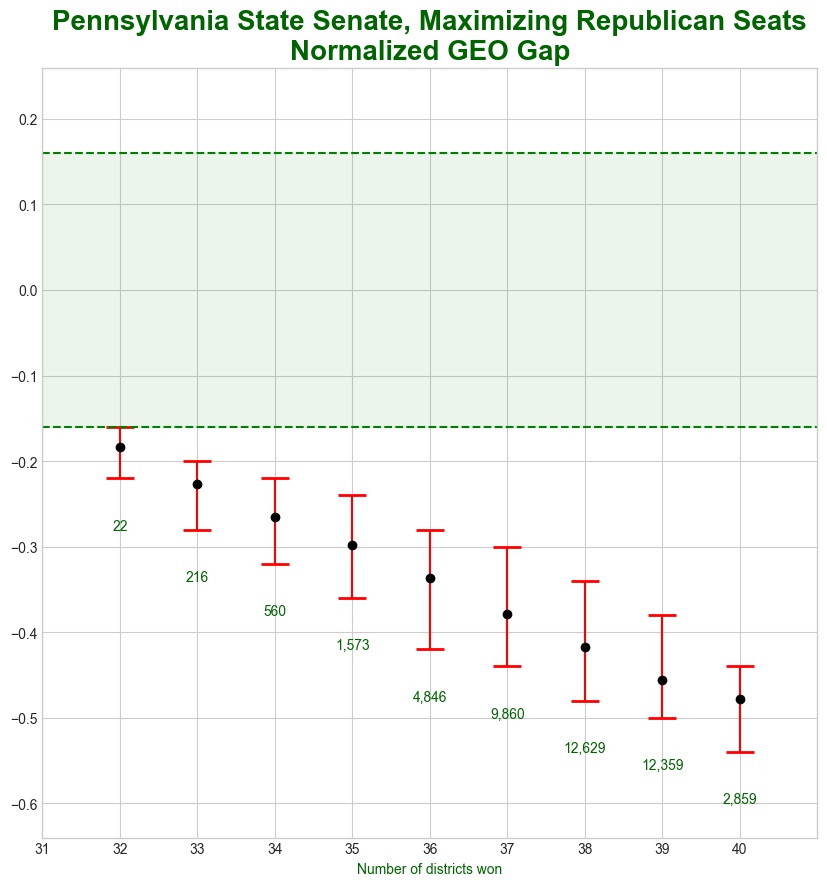}
\includegraphics[width=1.5in]{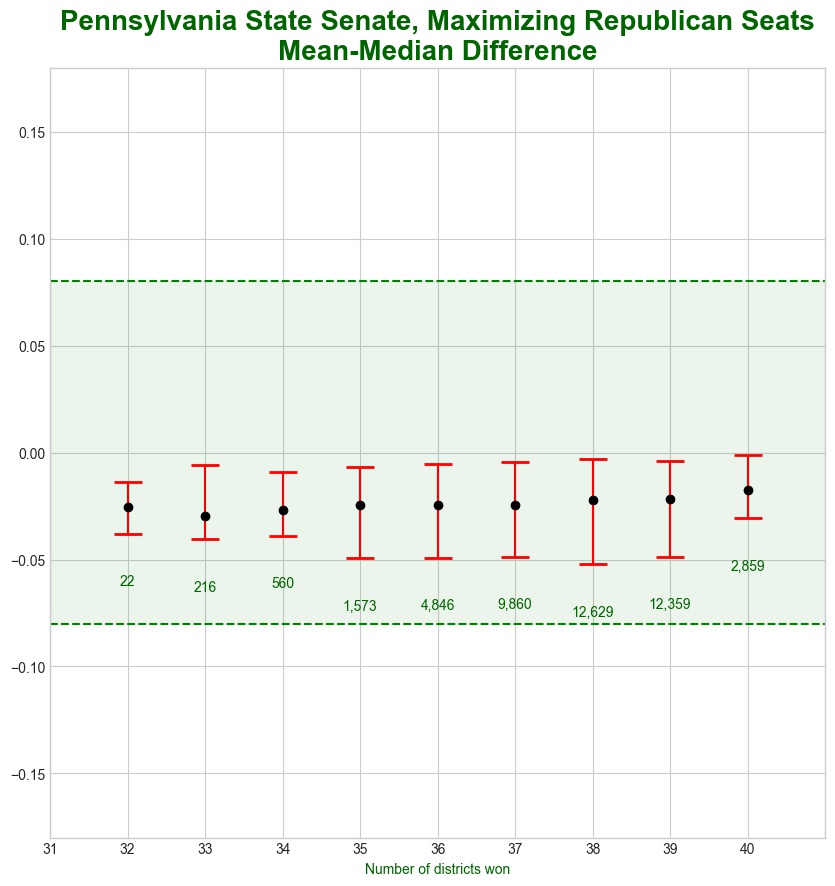}
\includegraphics[width=1.5in]{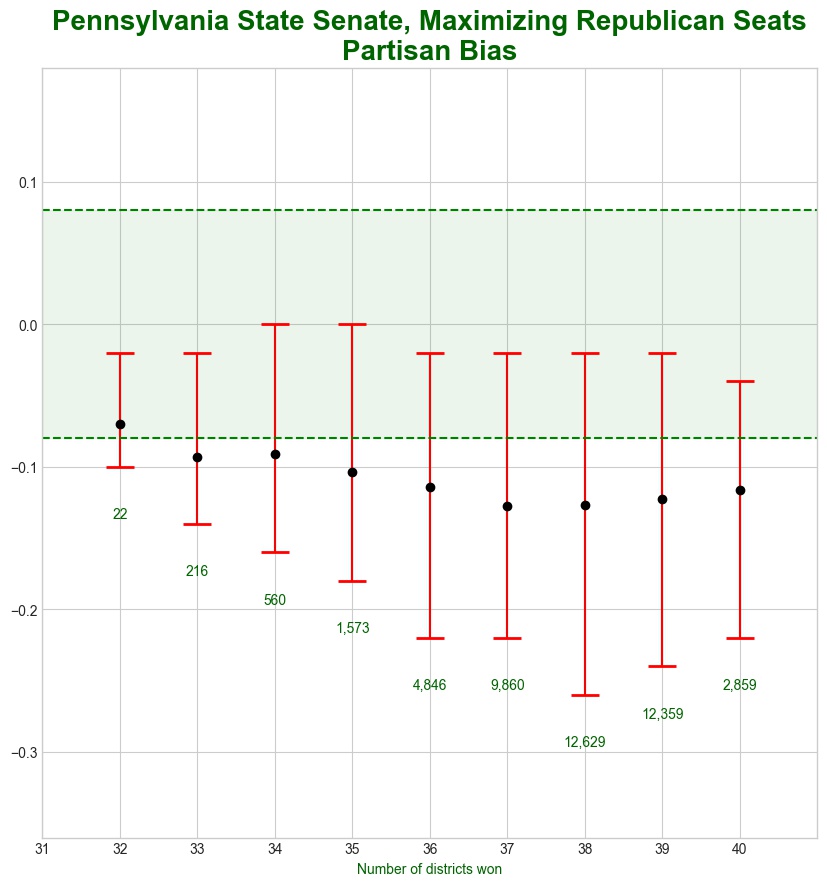}
\caption{Empirical results for Pennsylvania's State Senate map and 2016 US Senate election data, searching for maps with as many Republican-won districts as possible.  Horizontal axis is number of districts won, vertical axis is metric value ranges.  The green region is from $0.16\inf(m)$ to $0.16\sup(m)$ for each metric $m$.  The small number below each metric value range is the number of maps produced that had the corresponding number of districts won.  The dot within each vertical bar is the mean value of that metric on all produced maps with the corresponding number of districts won.}
\label{fig:short_bursts_PAupperR}
\end{figure}

\begin{figure}[h]
\centering
\includegraphics[width=1.5in]{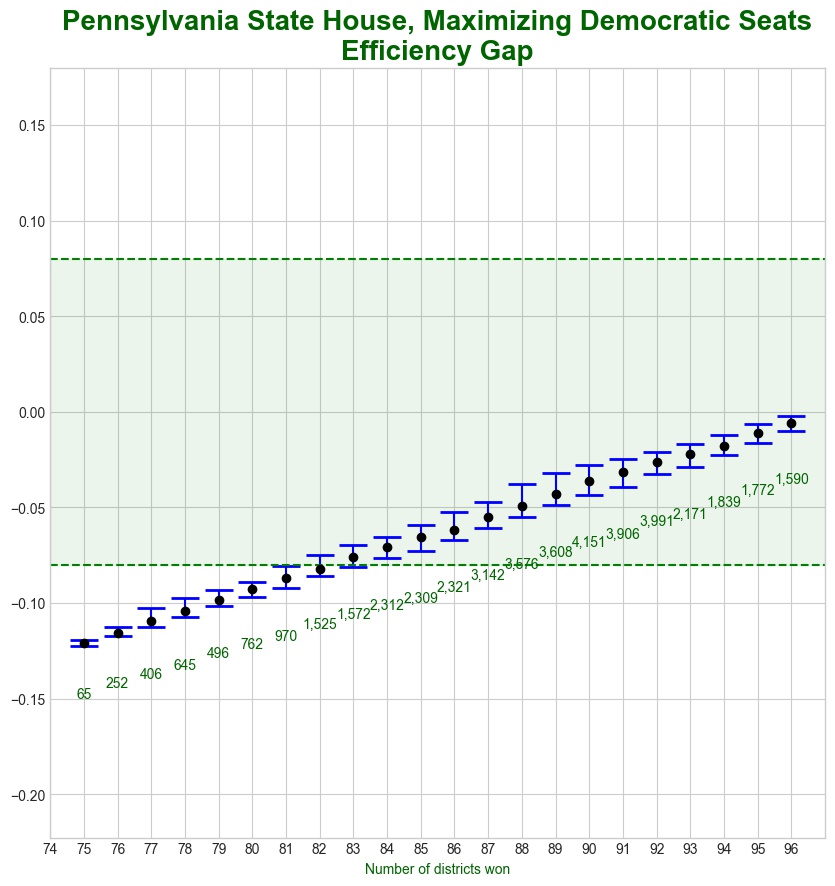}
\includegraphics[width=1.5in]{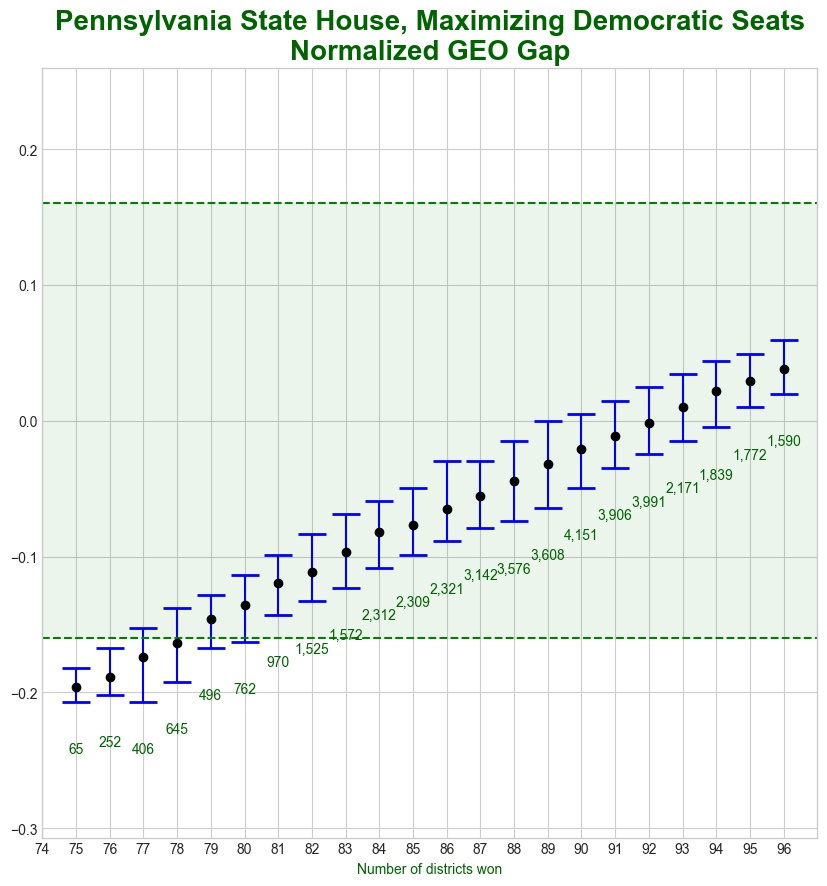}
\includegraphics[width=1.5in]{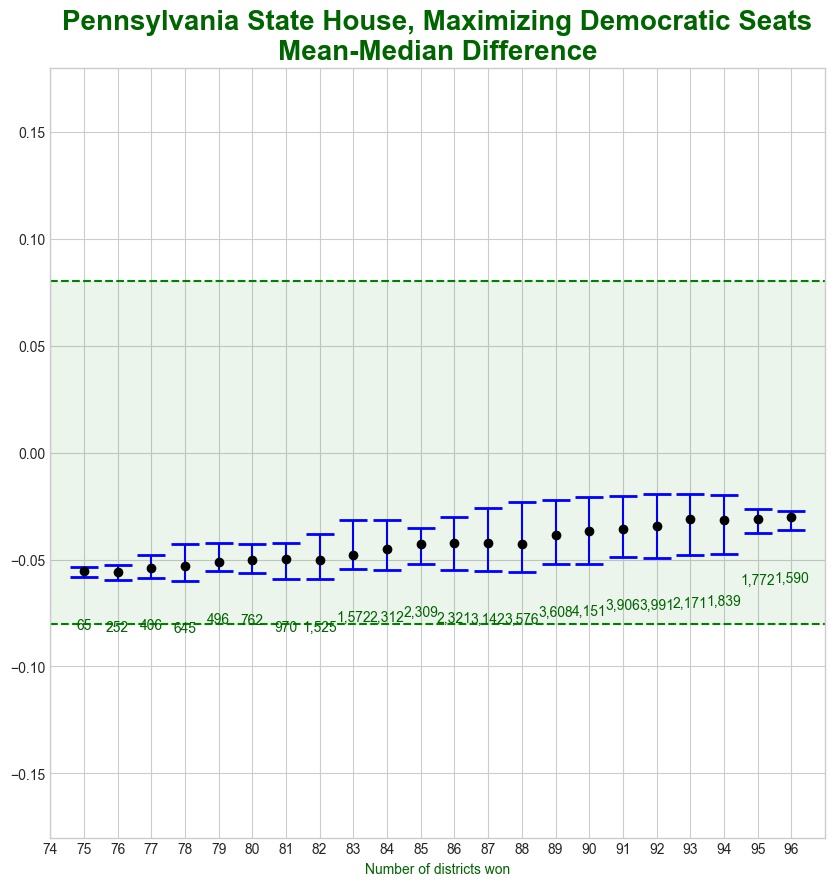}
\includegraphics[width=1.5in]{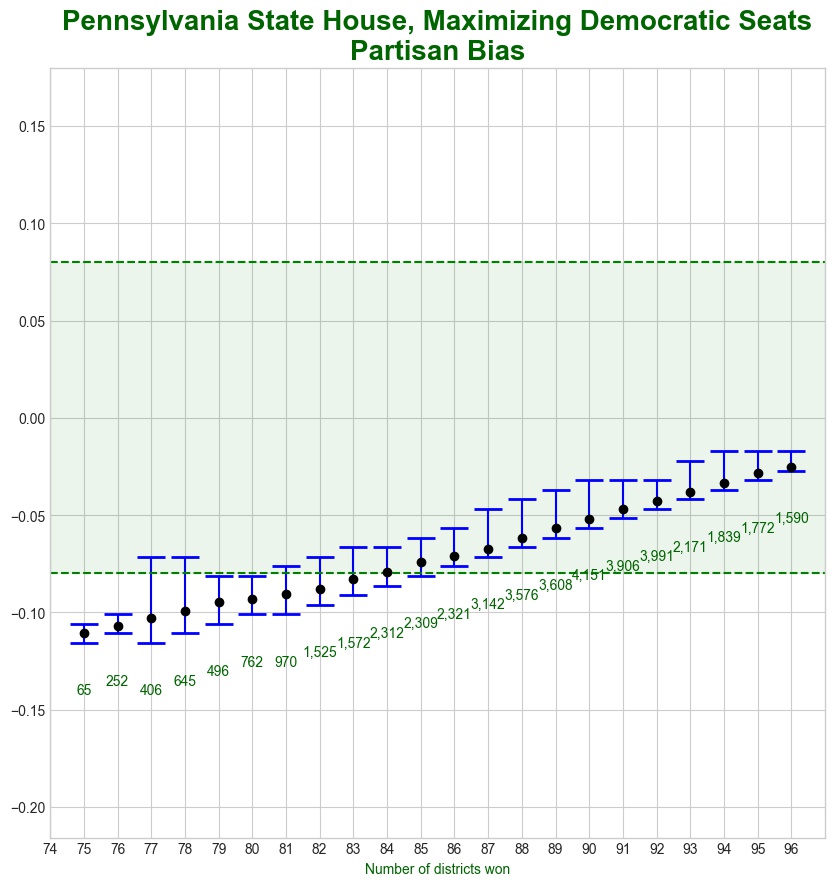}
\caption{Empirical results for Pennsylvania's State House map and 2016 US Senate election data, searching for maps with as many Democratic-won districts as possible.  Horizontal axis is number of districts won, vertical axis is metric value ranges.  The green region is from $0.16\inf(m)$ to $0.16\sup(m)$ for each metric $m$.  The small number below each metric value range is the number of maps produced that had the corresponding number of districts won.  The dot within each vertical bar is the mean value of that metric on all produced maps with the corresponding number of districts won.}
\label{fig:short_bursts_PAlowerD}
\end{figure}

\begin{figure}[h]
\centering
\includegraphics[width=1.5in]{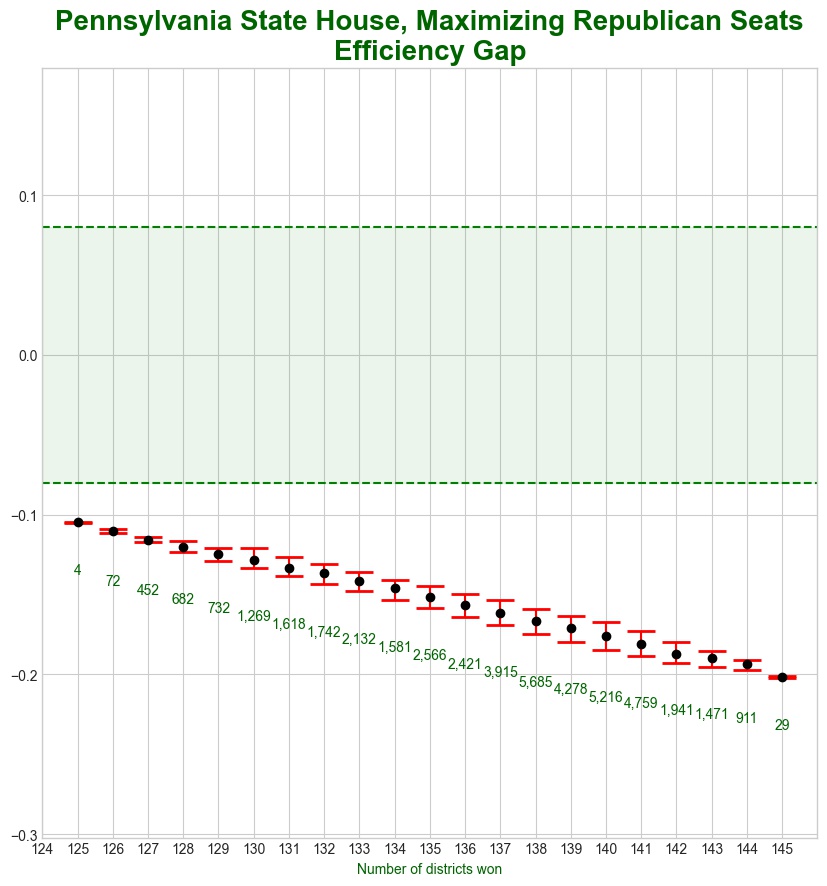}
\includegraphics[width=1.5in]{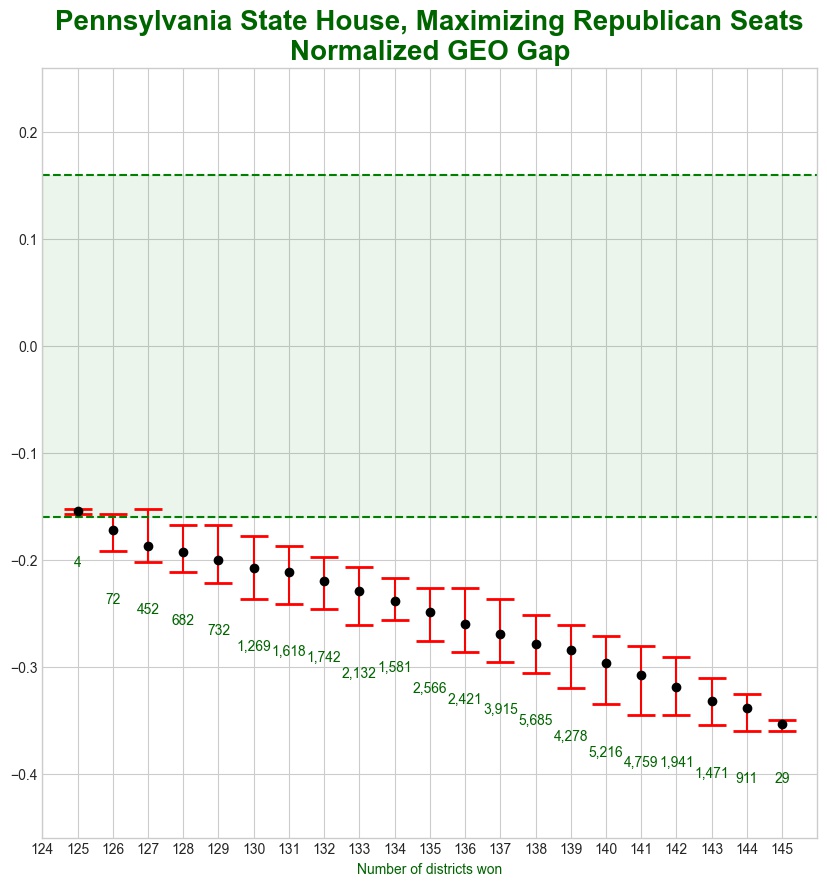}
\includegraphics[width=1.5in]{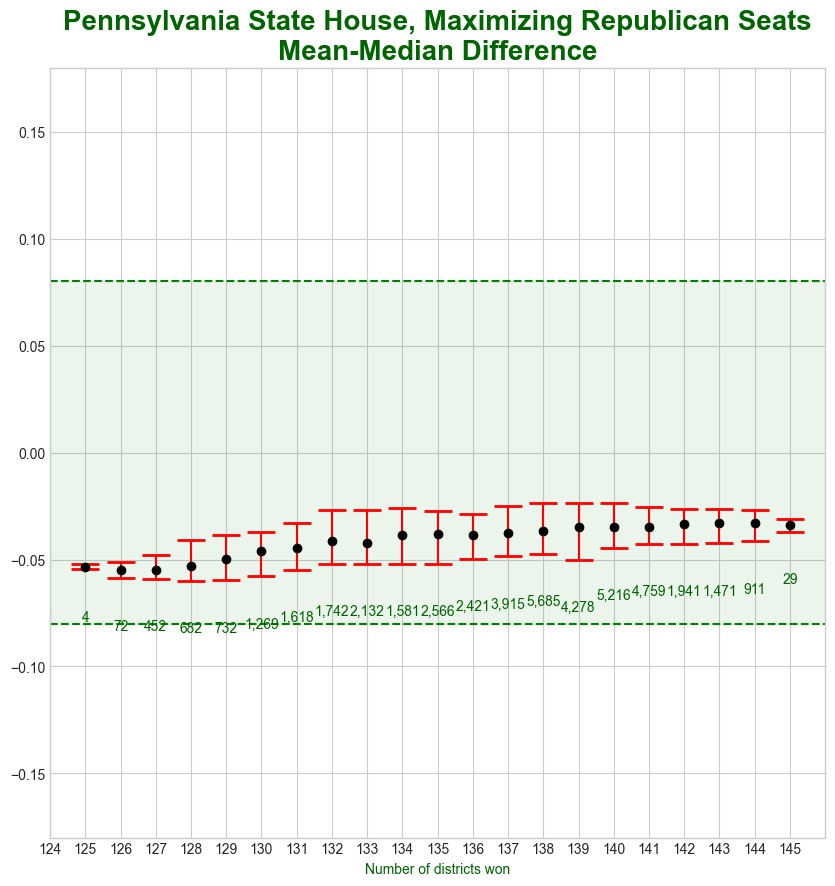}
\includegraphics[width=1.5in]{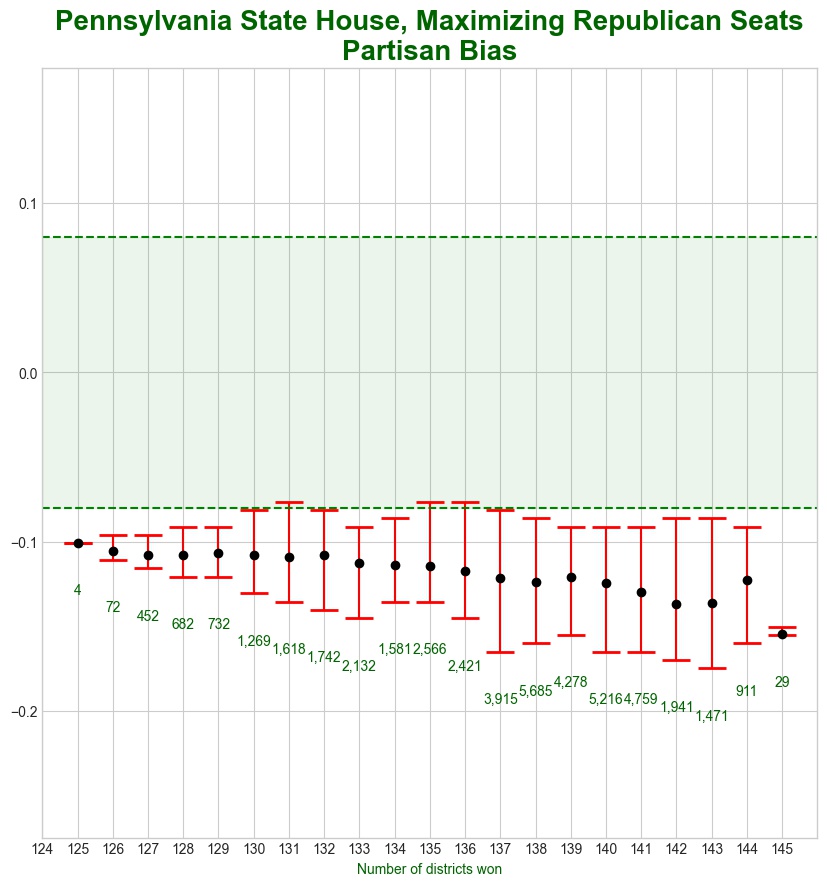}
\caption{Empirical results for Pennsylvania's State House map and 2016 US Senate election data, searching for maps with as many Republican-won districts as possible.  Horizontal axis is number of districts won, vertical axis is metric value ranges.  The green region is from $0.16\inf(m)$ to $0.16\sup(m)$ for each metric $m$.  The small number below each metric value range is the number of maps produced that had the corresponding number of districts won.  The dot within each vertical bar is the mean value of that metric on all produced maps with the corresponding number of districts won.}
\label{fig:short_bursts_PAlowerR}
\end{figure}

\begin{figure}[h]
\centering
\includegraphics[width=1.5in]{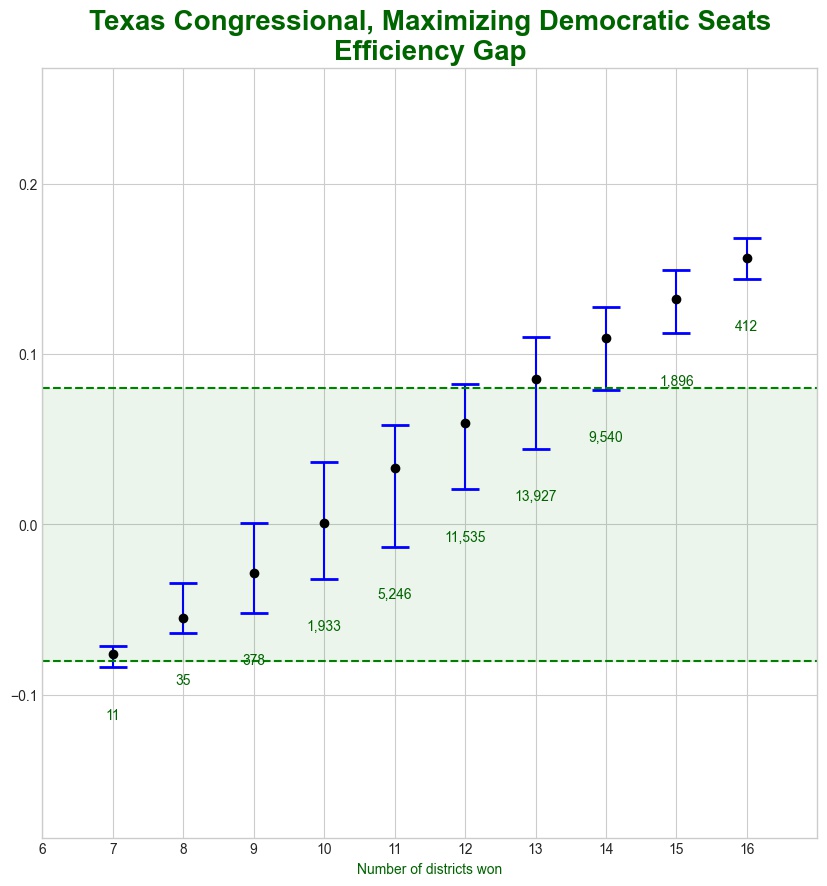}
\includegraphics[width=1.5in]{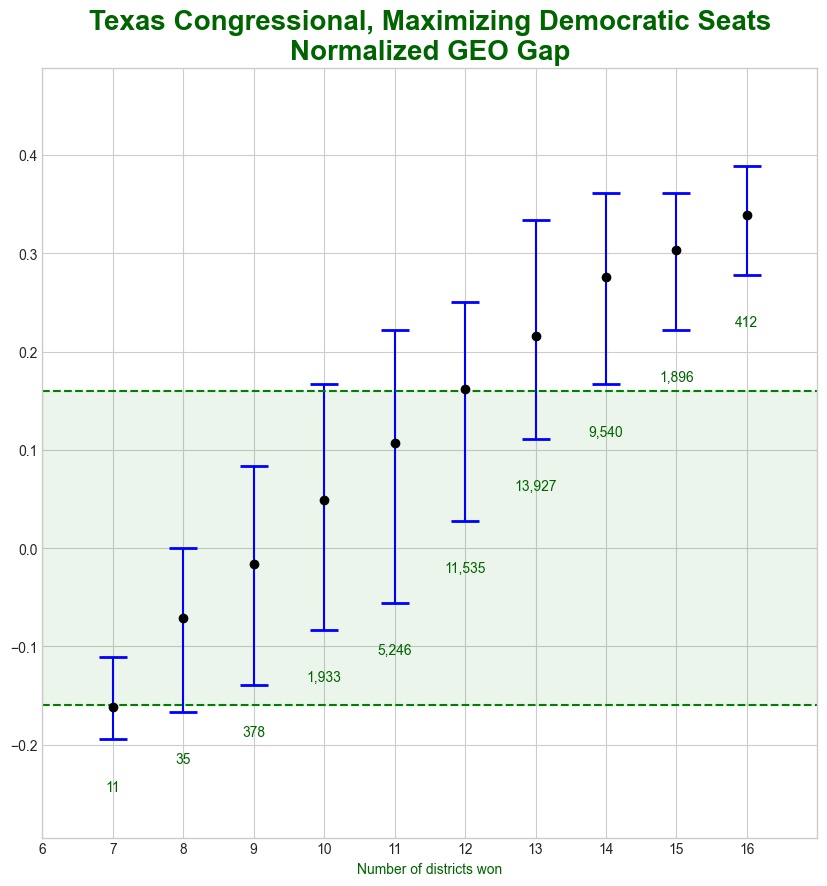}
\includegraphics[width=1.5in]{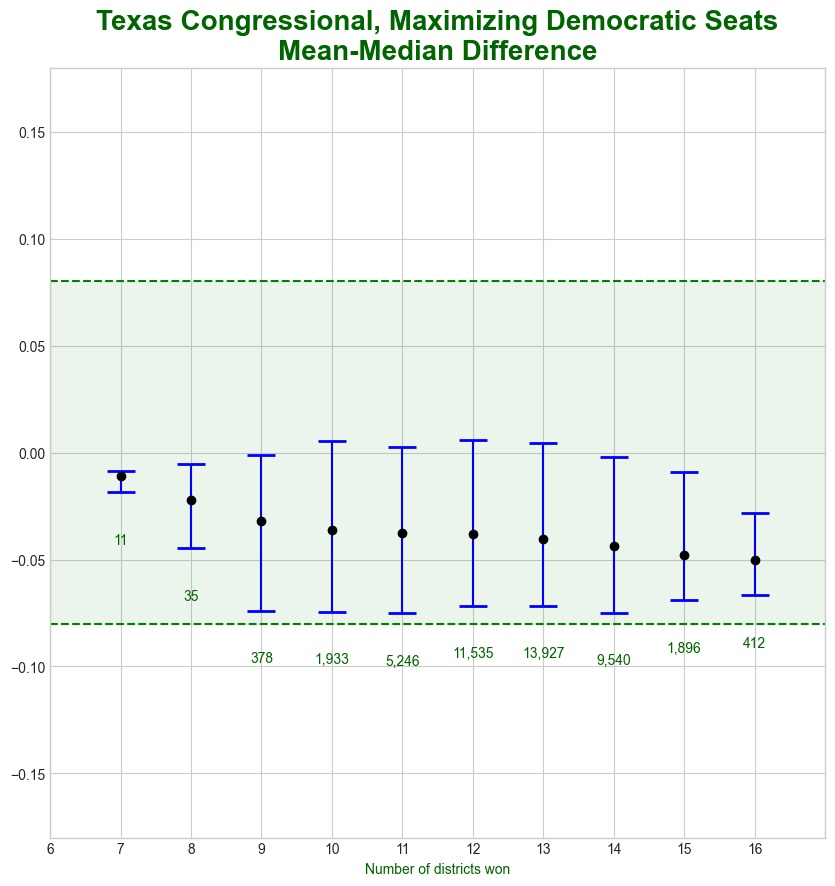}
\includegraphics[width=1.5in]{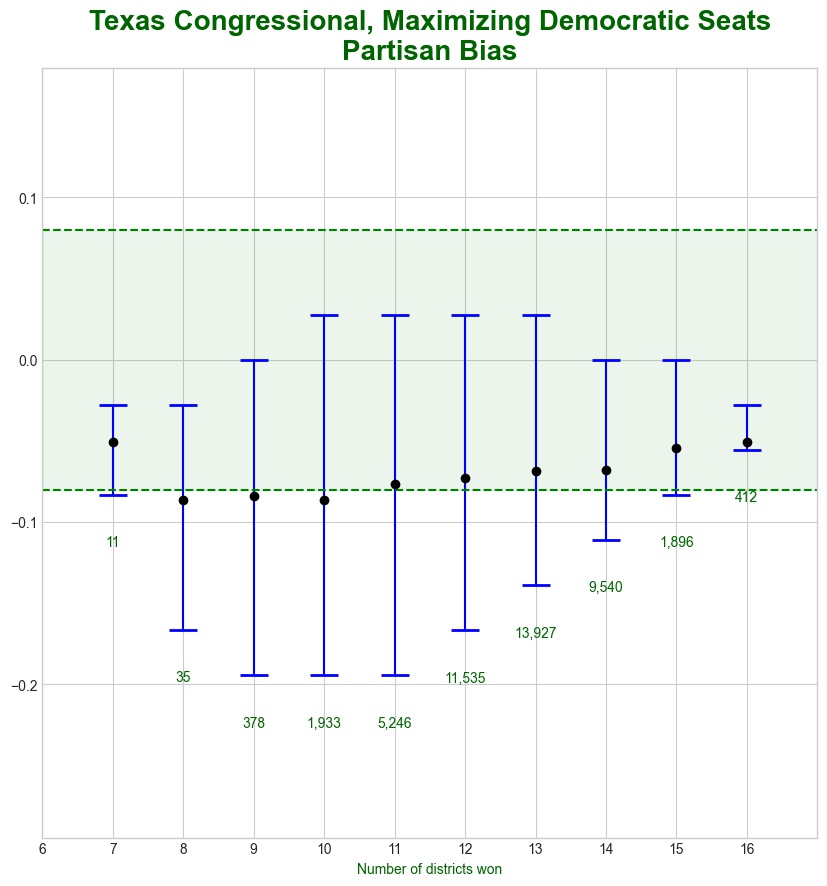}
\caption{Empirical results for Texas's Congressional map and 2014 US Senate election data, searching for maps with as many Democratic-won districts as possible.  Horizontal axis is number of districts won, vertical axis is metric value ranges.  The green region is from $0.16\inf(m)$ to $0.16\sup(m)$ for each metric $m$.  The small number below each metric value range is the number of maps produced that had the corresponding number of districts won.  The dot within each vertical bar is the mean value of that metric on all produced maps with the corresponding number of districts won.}
\label{fig:short_bursts_TXcongD}
\end{figure}

\begin{figure}[h]
\centering
\includegraphics[width=1.5in]{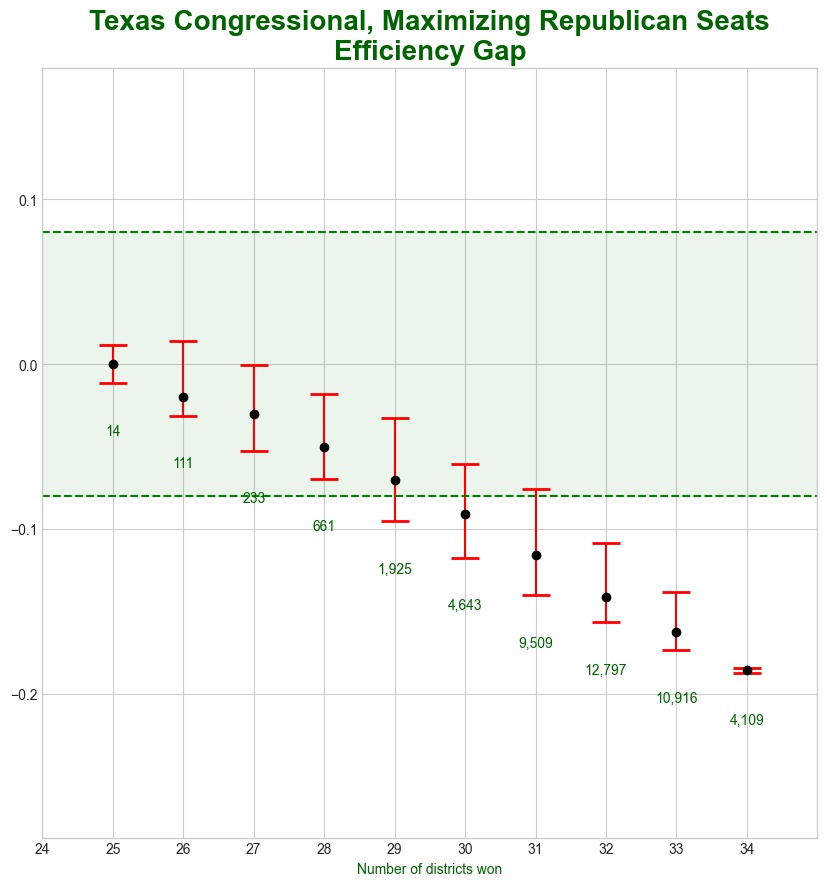}
\includegraphics[width=1.5in]{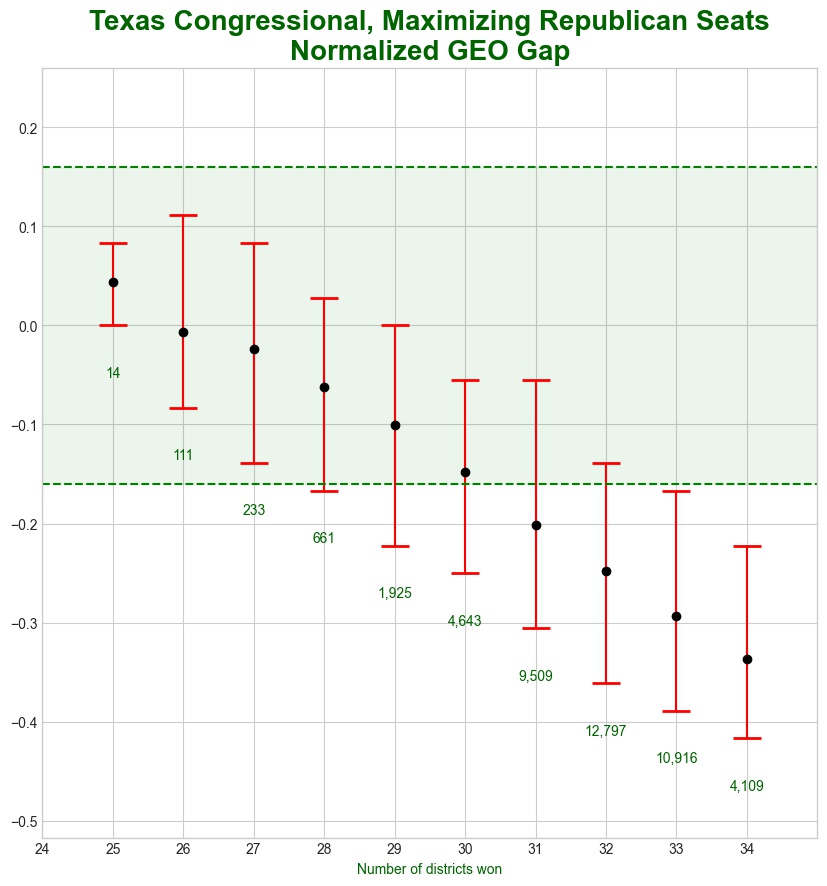}
\includegraphics[width=1.5in]{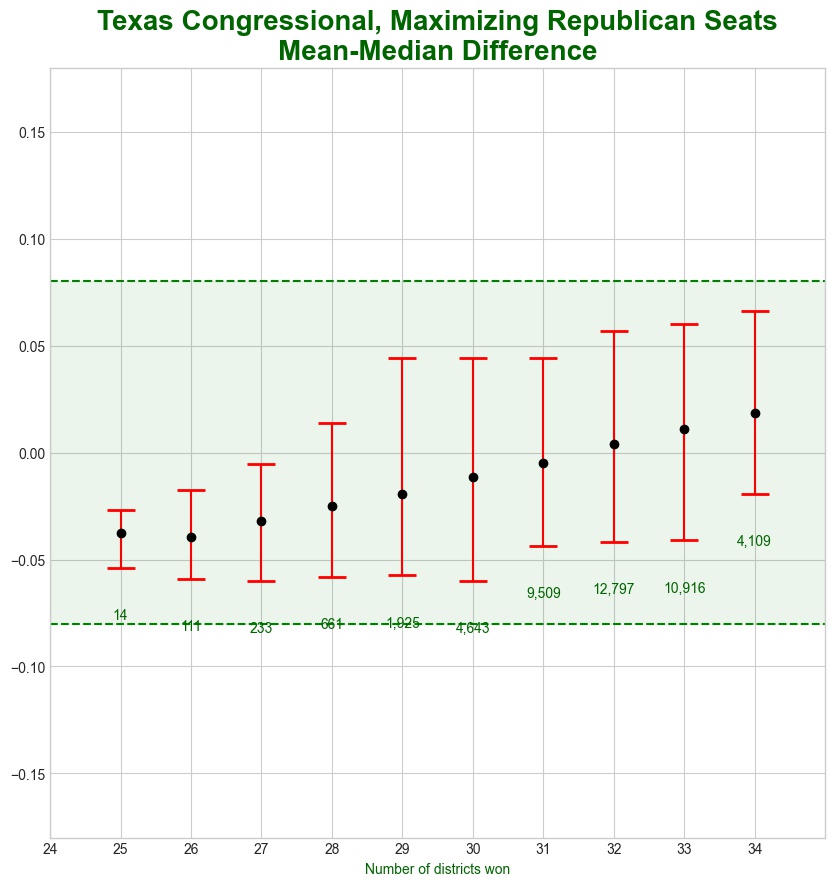}
\includegraphics[width=1.5in]{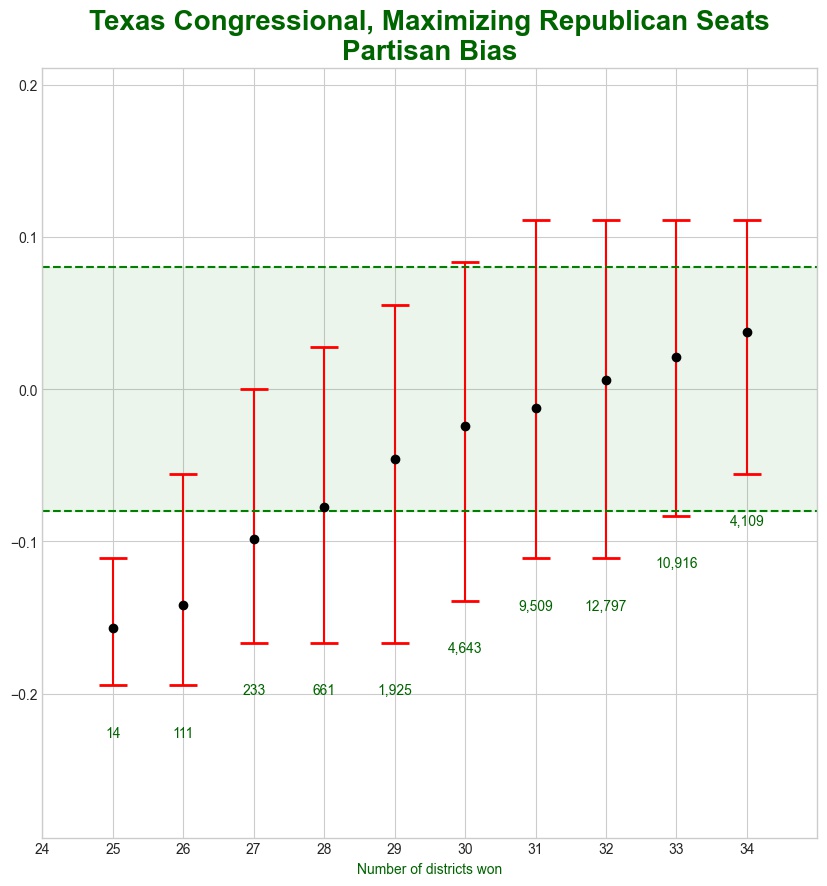}
\caption{Empirical results for Texas's Congressional map and 2014 US Senate election data, searching for maps with as many Republican-won districts as possible.  Horizontal axis is number of districts won, vertical axis is metric value ranges.  The green region is from $0.16\inf(m)$ to $0.16\sup(m)$ for each metric $m$.  The small number below each metric value range is the number of maps produced that had the corresponding number of districts won.  The dot within each vertical bar is the mean value of that metric on all produced maps with the corresponding number of districts won.}
\label{fig:short_bursts_TXcongR}
\end{figure}

\begin{figure}[h]
\centering
\includegraphics[width=1.5in]{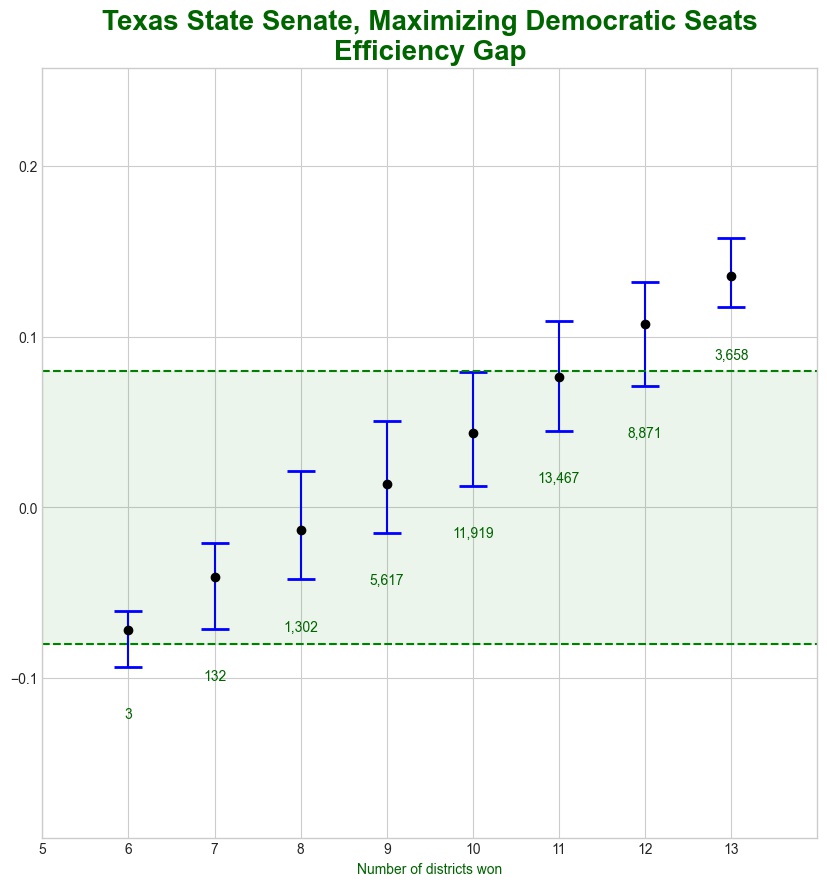}
\includegraphics[width=1.5in]{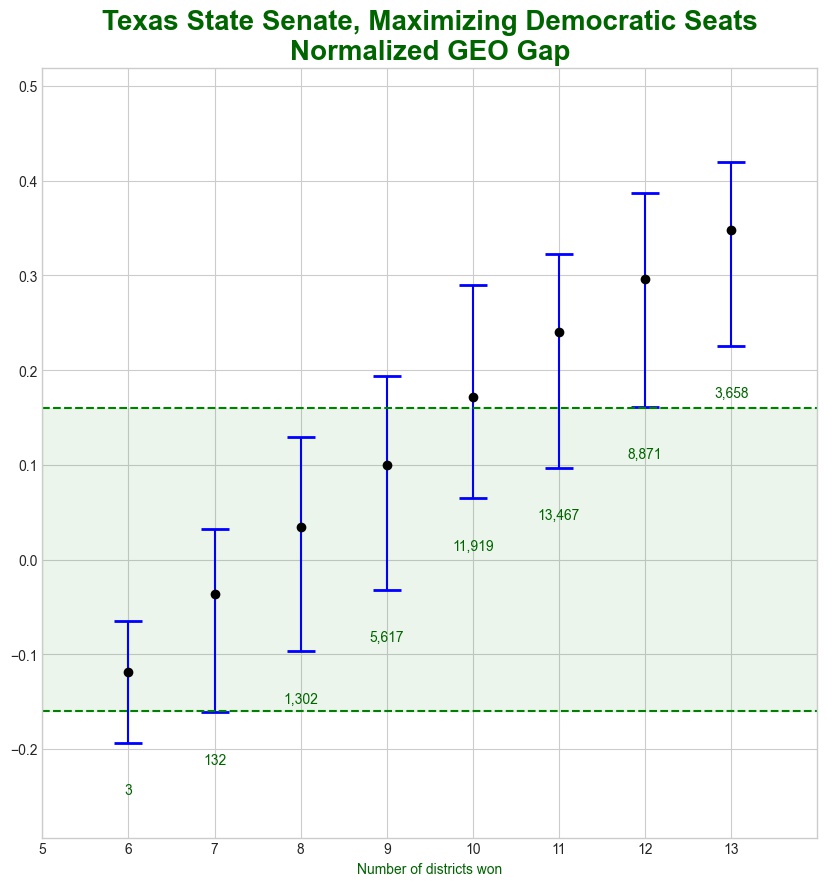}
\includegraphics[width=1.5in]{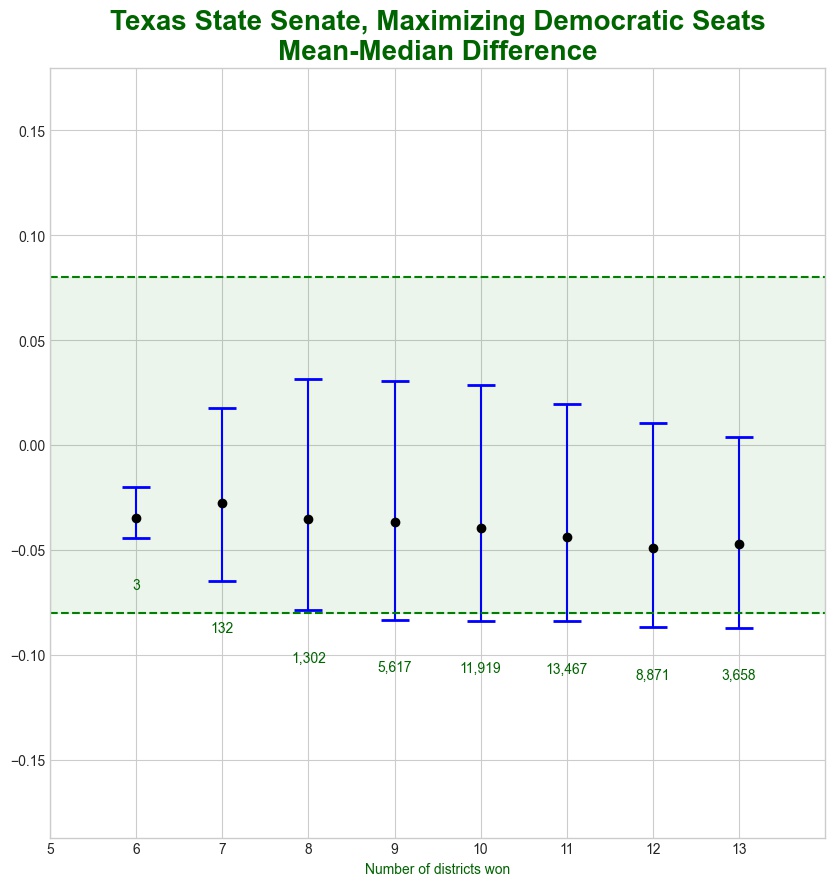}
\includegraphics[width=1.5in]{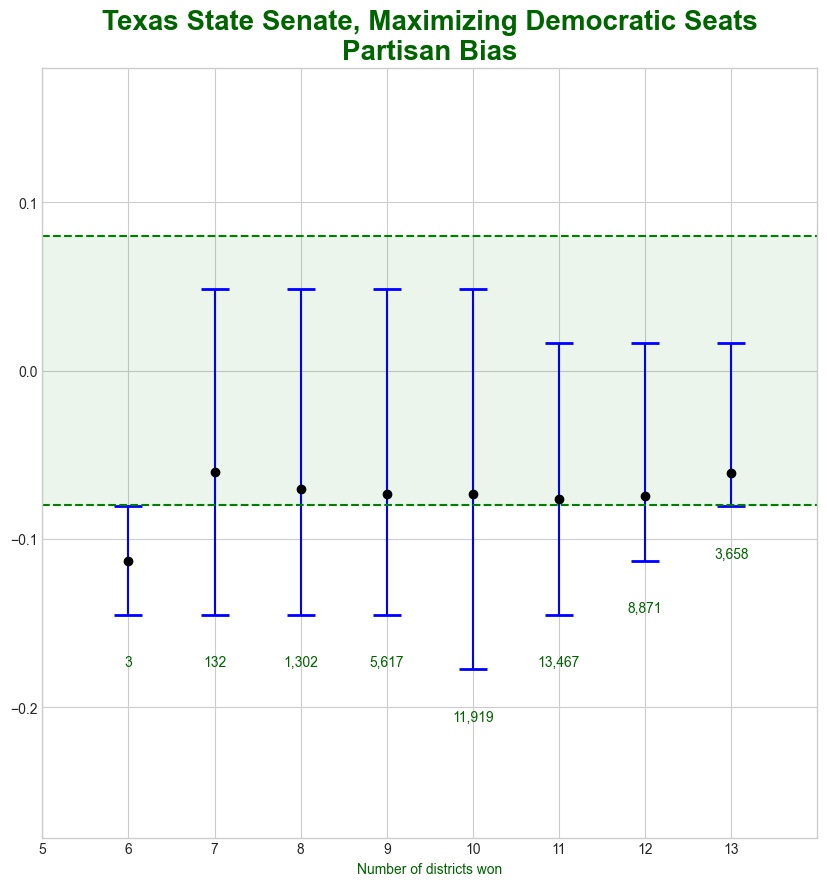}
\caption{Empirical results for Texas's State Senate map and 2014 US Senate election data, searching for maps with as many Democratic-won districts as possible.  Horizontal axis is number of districts won, vertical axis is metric value ranges.  The green region is from $0.16\inf(m)$ to $0.16\sup(m)$ for each metric $m$.  The small number below each metric value range is the number of maps produced that had the corresponding number of districts won.  The dot within each vertical bar is the mean value of that metric on all produced maps with the corresponding number of districts won.}
\label{fig:short_bursts_TXupperD}
\end{figure}

\begin{figure}[h]
\centering
\includegraphics[width=1.5in]{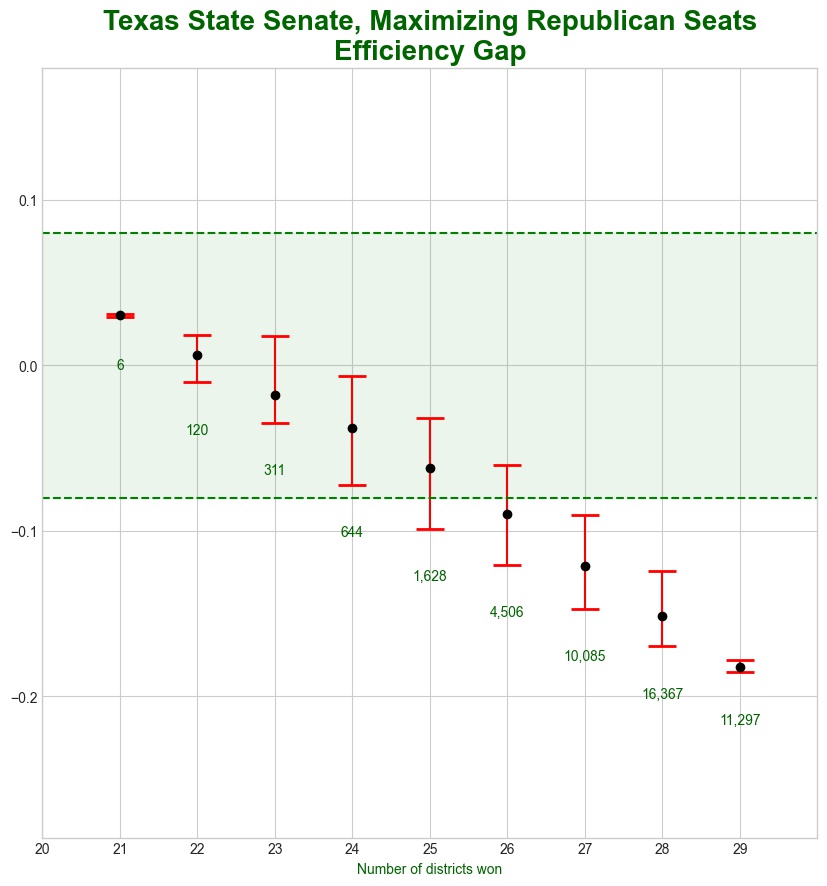}
\includegraphics[width=1.5in]{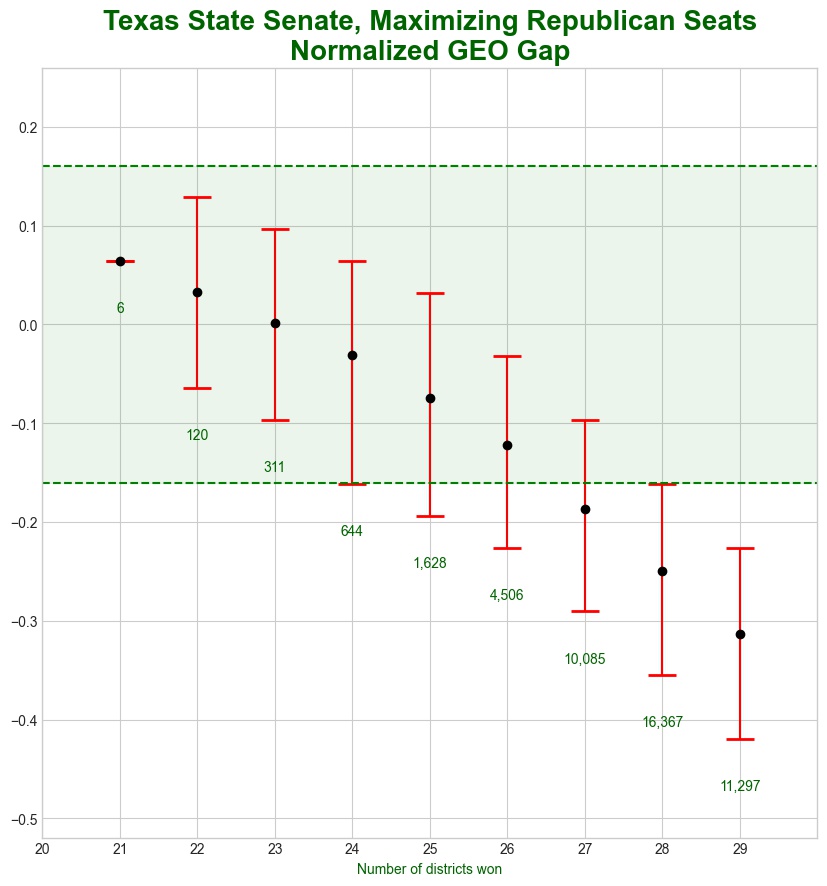}
\includegraphics[width=1.5in]{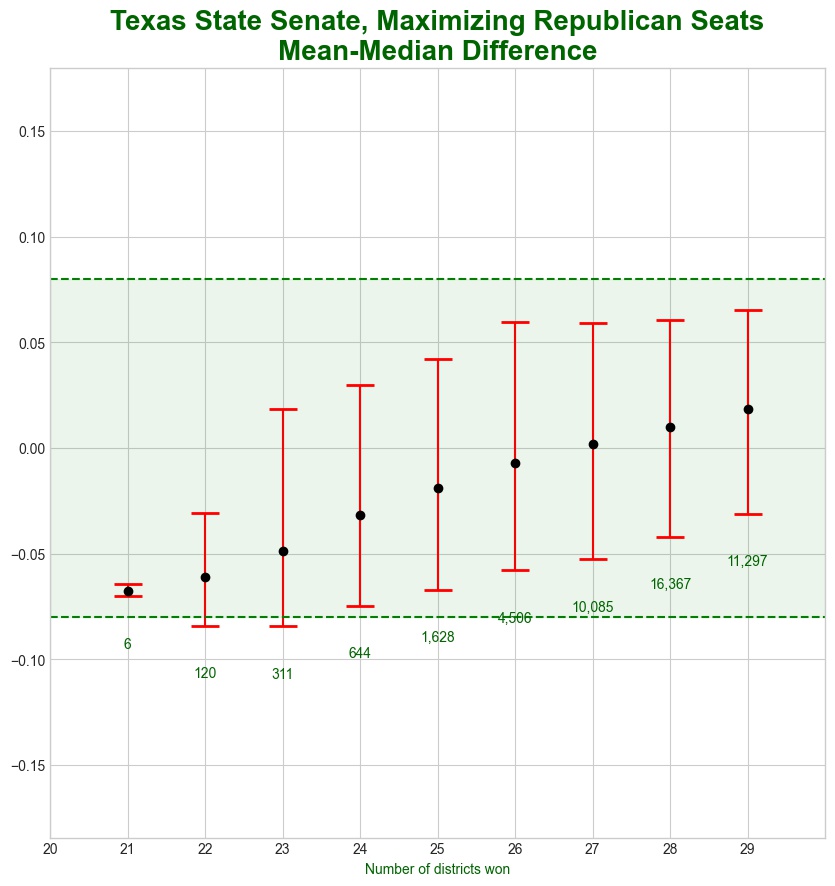}
\includegraphics[width=1.5in]{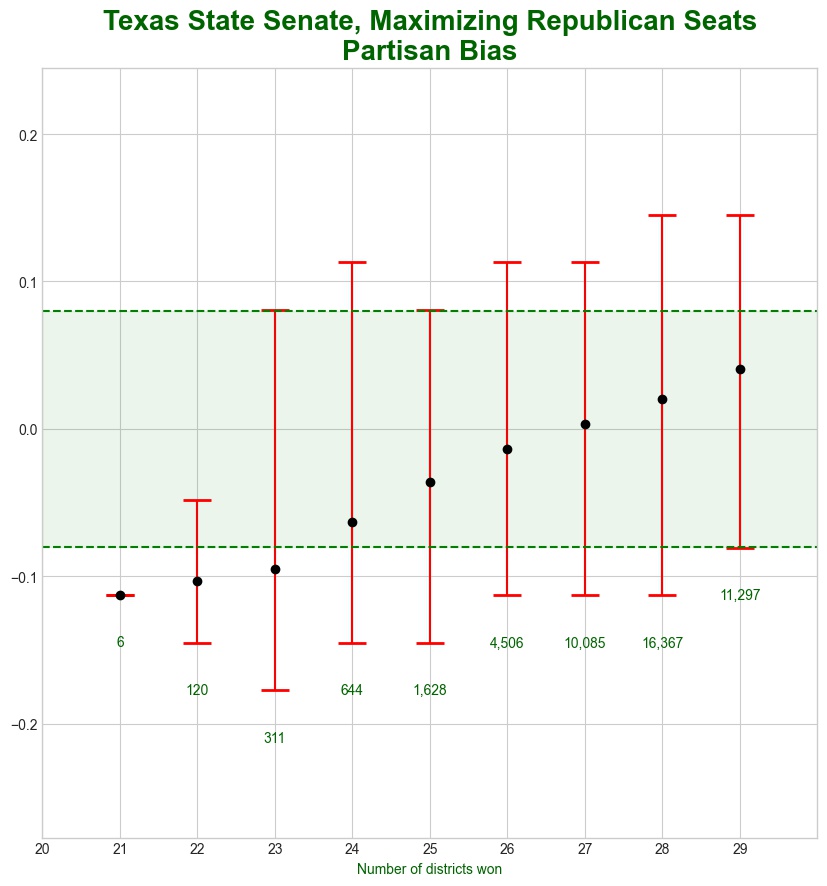}
\caption{Empirical results for Texas's State Senate map and 2014 US Senate election data, searching for maps with as many Republican-won districts as possible.  Horizontal axis is number of districts won, vertical axis is metric value ranges.  The green region is from $0.16\inf(m)$ to $0.16\sup(m)$ for each metric $m$.  The small number below each metric value range is the number of maps produced that had the corresponding number of districts won.  The dot within each vertical bar is the mean value of that metric on all produced maps with the corresponding number of districts won.}
\label{fig:short_bursts_TXupperR}
\end{figure}

\begin{figure}[h]
\centering
\includegraphics[width=1.5in]{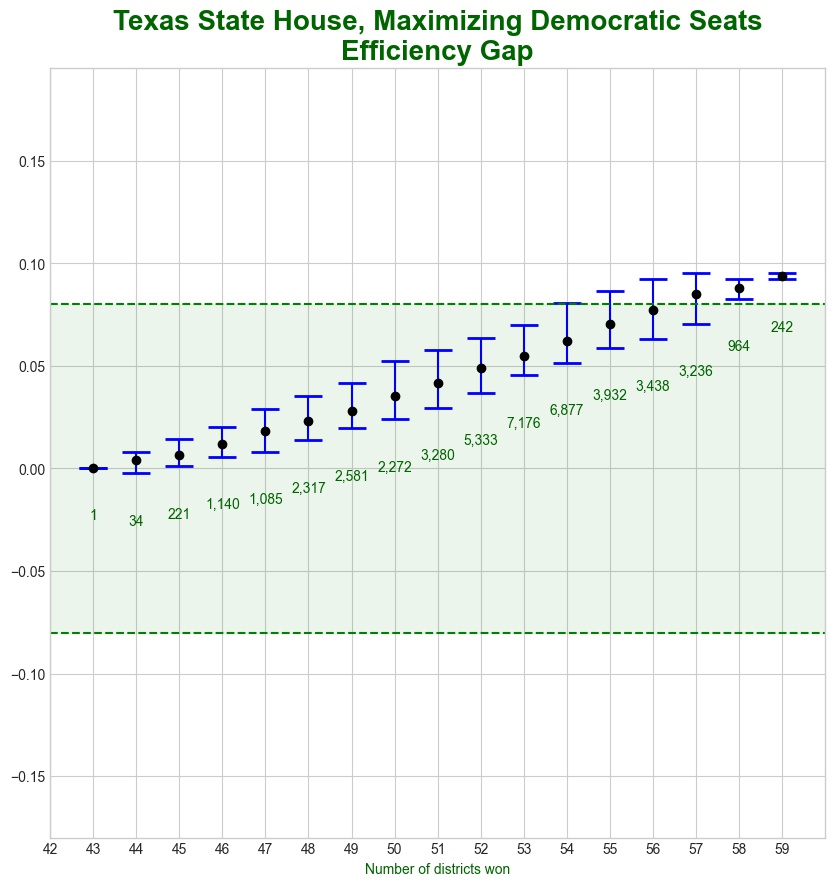}
\includegraphics[width=1.5in]{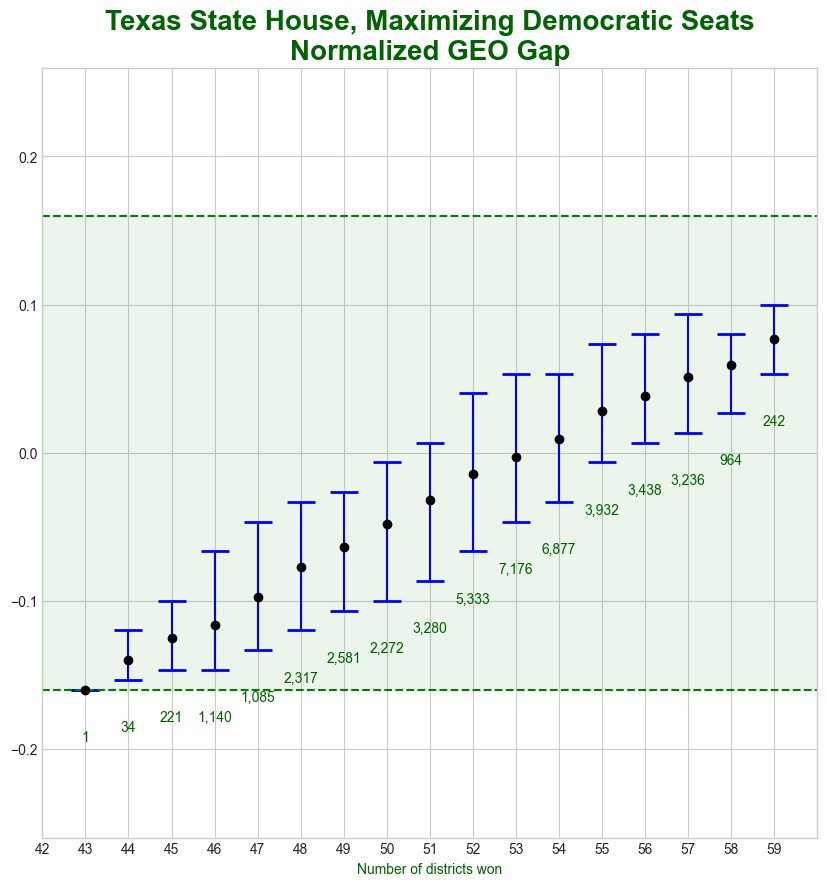}
\includegraphics[width=1.5in]{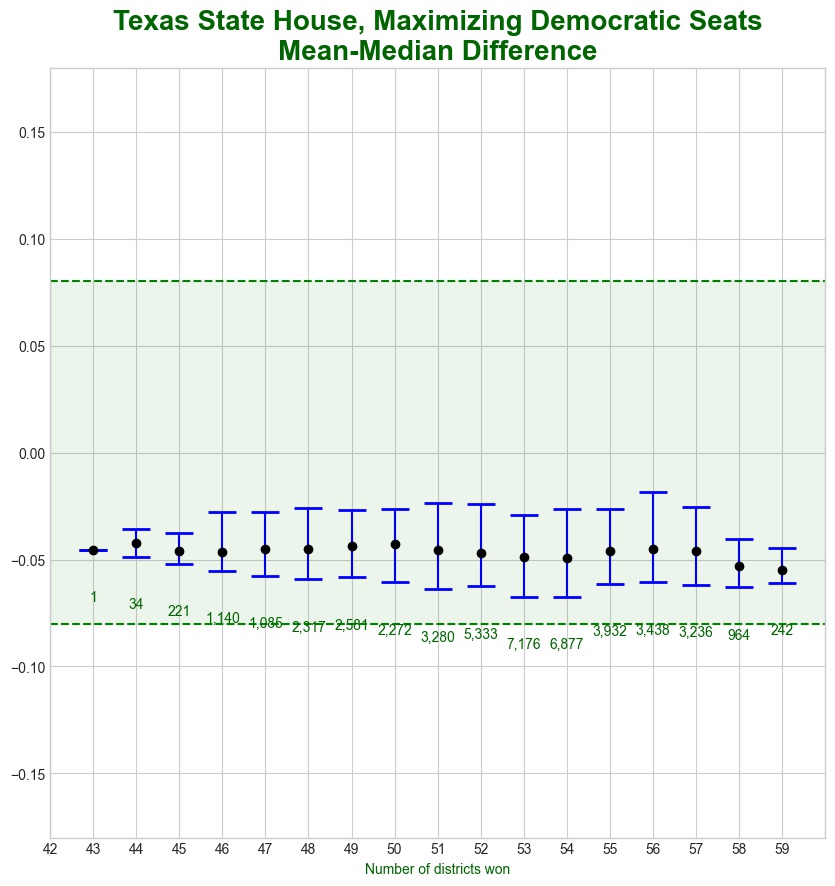}
\includegraphics[width=1.5in]{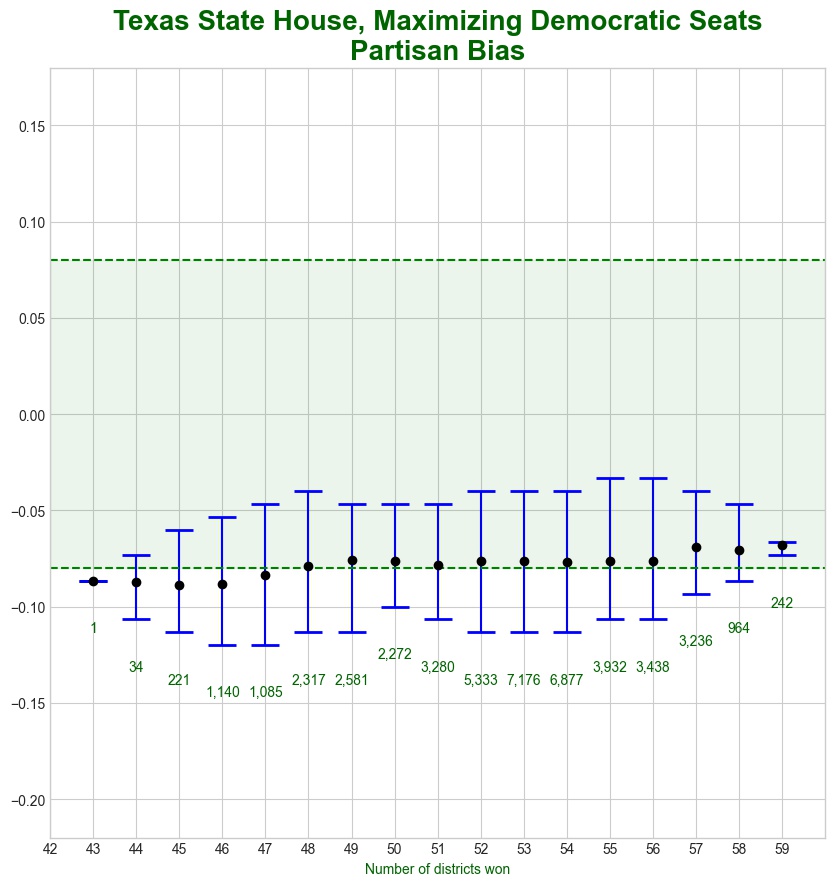}
\caption{Empirical results for Texas's State House map and 2014 US Senate election data, searching for maps with as many Democratic-won districts as possible.  Horizontal axis is number of districts won, vertical axis is metric value ranges.  The green region is from $0.16\inf(m)$ to $0.16\sup(m)$ for each metric $m$.  The small number below each metric value range is the number of maps produced that had the corresponding number of districts won.  The dot within each vertical bar is the mean value of that metric on all produced maps with the corresponding number of districts won.}
\label{fig:short_bursts_TXlowerD}
\end{figure}

\begin{figure}[h]
\centering
\includegraphics[width=1.5in]{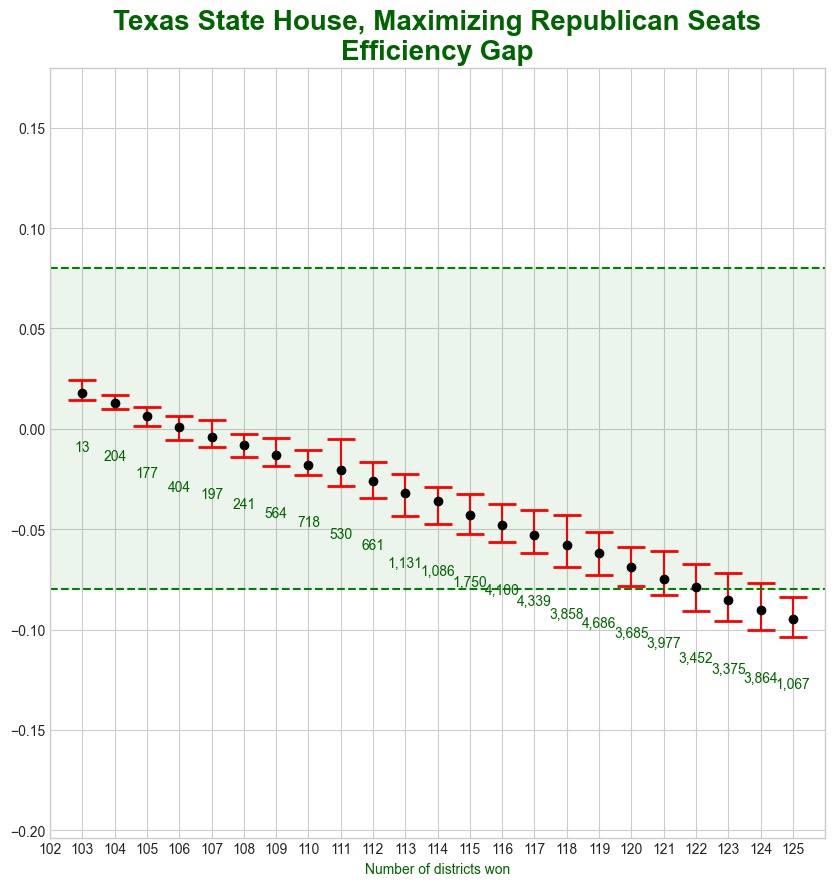}
\includegraphics[width=1.5in]{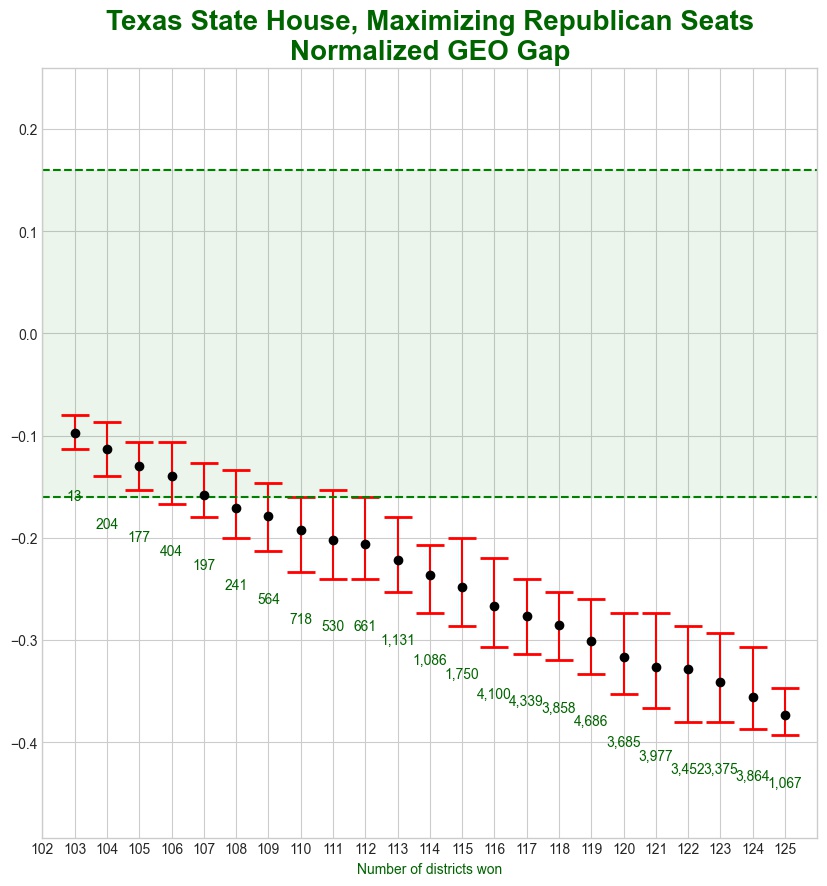}
\includegraphics[width=1.5in]{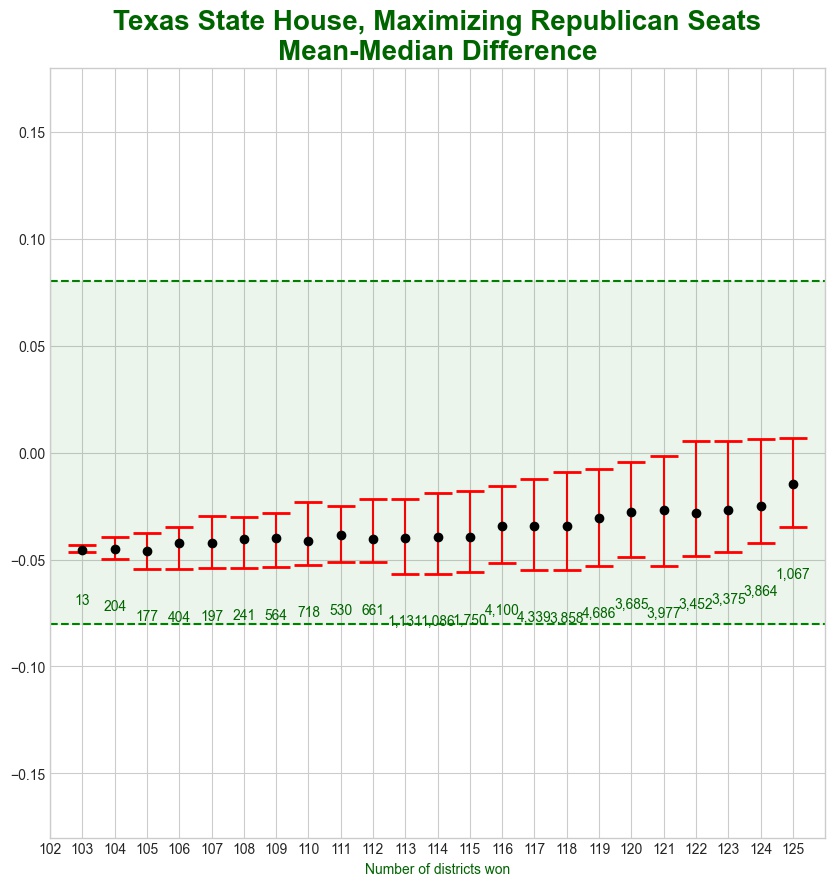}
\includegraphics[width=1.5in]{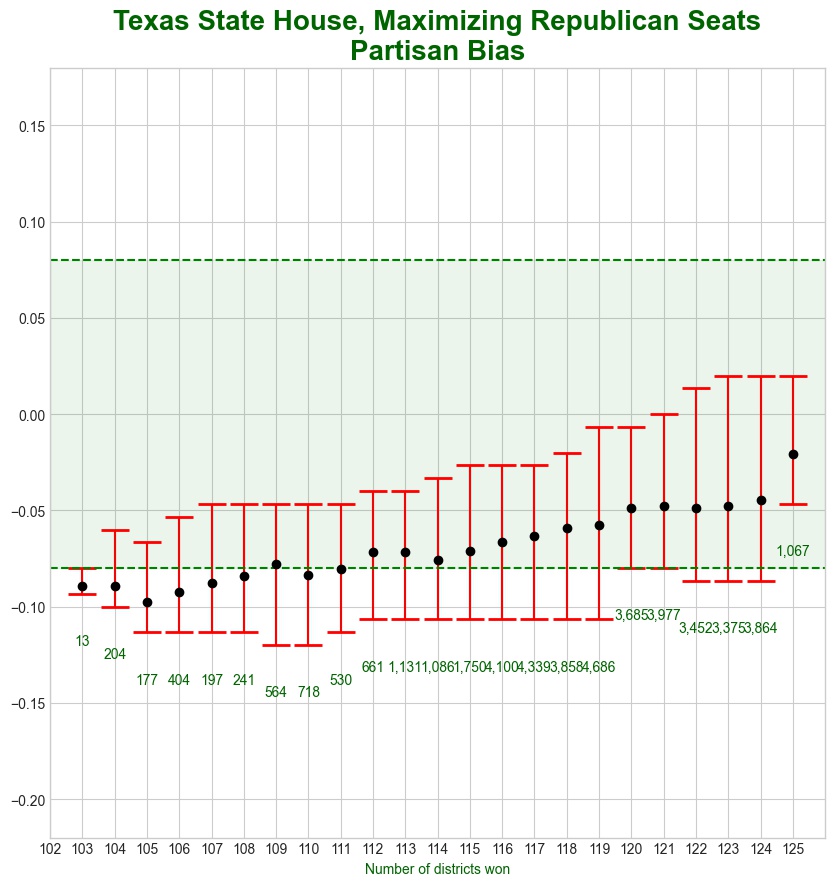}
\caption{Empirical results for Texas's State House map and 2014 US Senate election data, searching for maps with as many Republican-won districts as possible.  Horizontal axis is number of districts won, vertical axis is metric value ranges.  The green region is from $0.16\inf(m)$ to $0.16\sup(m)$ for each metric $m$.  The small number below each metric value range is the number of maps produced that had the corresponding number of districts won.  The dot within each vertical bar is the mean value of that metric on all produced maps with the corresponding number of districts won.}
\label{fig:short_bursts_TXlowerR}
\end{figure}

\end{document}